\documentclass[11pt]{amsart}
\usepackage{geometry}                
\usepackage[latin1]{inputenc}   
\usepackage{mathpazo}  
\usepackage[english]{babel}
\usepackage[usenames,dvipsnames,cmyk]{xcolor}
\usepackage{url}
\usepackage{amsmath,amssymb,latexsym,mathrsfs,amssymb,amsthm}
\usepackage{enumerate}
\usepackage{mathtools} 
\usepackage[all]{xy}
\usepackage{setspace}
\usepackage{tikz, tkz-euclide}
\usepackage[colorlinks=true, linkcolor=black, citecolor=DarkOrchid, urlcolor=DarkOrchid]{hyperref}

\usetikzlibrary{angles,quotes}

\newtheorem{theorem}{Theorem}[section]
\newtheorem{lemma}[theorem]{Lemma}
\newtheorem{proposition}[theorem]{Proposition}
\newtheorem{corollary}[theorem]{Corollary}

\theoremstyle{definition}
\newtheorem{definition}[theorem]{Definition}

\theoremstyle{remark}
\newtheorem{remark}[theorem]{Remark}
\newtheorem{rhp}{Riemann-Hilbert Problem}
\newcommand{\rhref}[1]{Riemann-Hilbert Problem~\ref{#1}}

\let\Re=\undefined\DeclareMathOperator{\Re}{Re}
\let\Im=\undefined\DeclareMathOperator{\Im}{Im}

\DeclareMathOperator{\rem}{Rem}

\DeclareMathOperator{\sign}{sign}

\newcommand{\channels}{\ensuremath{\mathcal{C}}}
\newcommand{\shelves}{\ensuremath{\mathcal{S}}}
\newcommand{\exterior}{\ensuremath{\mathcal{E}}}
\newcommand{\dd}{\ensuremath{\,\mathrm{d}}}
\newcommand{\ii}{\ensuremath{\mathrm{i}}}
\newcommand{\ee}{\ensuremath{\,\mathrm{e}}}
\newcommand{\defeq}{\vcentcolon=}

\makeatletter
\renewcommand*\env@matrix[1][\arraystretch]{%
  \edef\arraystretch{#1}%
  \hskip -\arraycolsep
  \let\@ifnextchar\new@ifnextchar
  \array{*\c@MaxMatrixCols c}}
\makeatother

\let\originalleft\left
\let\originalright\right
\renewcommand{\left}{\mathopen{}\mathclose\bgroup\originalleft}
\renewcommand{\right}{\aftergroup\egroup\originalright}

\title{Extreme Superposition:  High-Order Fundamental Rogue Waves in the Far-Field Regime}

\author{Deniz Bilman}
\address{Deniz Bilman:  Department of Mathematical Sciences, University of Cincinnati, Cincinnati, OH, USA}
\email{bilman@uc.edu}
\author{Peter D.~Miller}
\address{Peter D. Miller:  Department of Mathematics, University of Michigan, Ann Arbor, MI, USA}
\email{millerpd@umich.edu}

\thanks{The authors wish to thank Liming Ling and Alex Tovbis for useful discussions during the early stages of this project.
Bilman's work was partially supported by a research fellowship from Charles Phelps Taft Research Center.
Miller was supported by the National Science Foundation under grant number DMS-1812625.
}

\date{\today}

\begin{document}

\begin{abstract}
We study fundamental rogue-wave solutions of the focusing nonlinear Schr\"odinger equation in the limit that the order of the rogue wave is large and the independent variables $(x,t)$ are proportional to the order (the far-field limit).  We first formulate a Riemann-Hilbert representation of these solutions that allows the order to vary continuously rather than by integer increments.  The intermediate solutions in this continuous family include also soliton solutions for zero boundary conditions spectrally encoded by  a single complex-conjugate pair of poles of arbitrary order, as well as other solutions having nonzero boundary conditions matching those of the rogue waves albeit with far slower decay as $x\to\pm\infty$.  The large-order far-field asymptotic behavior of the solution depends on which of three disjoint regions $\channels$, $\shelves$, and $\exterior$ contains the rescaled variables.  On the regions $\channels$ and $\shelves$ we show that the asymptotic behavior is the same for all continuous orders, while in the region $\exterior$ the discrete sequence of rogue-wave orders produces distinctive asymptotic behavior that is different from other cases.
\end{abstract}

\maketitle

\section{Introduction}
This paper is a continuation of a study, begun in \cite{BilmanLM20}, of high-order rogue-wave solutions of the focusing nonlinear Schr\"odinger equation.  As in \cite{BilmanLM20}, the starting point is a Riemann-Hilbert problem characterization of the fundamental rogue wave of order $k$, $k\geq 0$, which was originally obtained in \cite{BilmanM19} and which we now describe. Let $\Sigma_\mathrm{c}$ denote a Schwarz-symmetric simple arc connecting endpoints $\lambda=\pm\ii$ with upward orientation, let $\rho:\mathbb{C}\setminus\Sigma_\mathrm{c}\to\mathbb{C}$ be the analytic function satisfying $\rho(\lambda)^2=\lambda^2+1$ and $\rho(\lambda)=\lambda+O(\lambda^{-1})$ as $\lambda\to\infty$, and let $\Sigma_\circ$ denote a Schwarz-symmetric Jordan curve with $\Sigma_\mathrm{c}$ in its interior and let $\Sigma_\circ$ have clockwise orientation.  
Also, let
\begin{equation}
\mathbf{Q}:=\frac{1}{\sqrt{2}}\begin{bmatrix} 1 & -1 \\ 1 & 1\end{bmatrix},
\label{eq:Q-def}
\end{equation}
and let $\mathbf{E}(\lambda)$ denote the matrix function defined for $\lambda\in\mathbb{C}\setminus\Sigma_\mathrm{c}$ by
\begin{equation}
\mathbf{E}(\lambda):= f(\lambda)\begin{bmatrix}1 & \ii (\lambda-\rho(\lambda))\\\ii(\lambda-\rho(\lambda)) & 1\end{bmatrix},\quad\lambda\in\mathbb{C}\setminus\Sigma_\mathrm{c},
\label{eq:E-def}
\end{equation}
where $f(\lambda)$ is the function analytic for $\lambda\in\mathbb{C}\setminus\Sigma_\mathrm{c}$ that satisfies 
\begin{equation}
f(\lambda)^2=\frac{\lambda+\rho(\lambda)}{2\rho(\lambda)}\quad\text{and $f(\lambda)\to 1$ as $\lambda\to\infty$.}
\label{eq:f-squared}
\end{equation}
This matrix $\mathbf{E}(\lambda)$ is analytic in its domain of definition and has unit determinant. 
It is convenient to introduce a sign $s=(-1)^k$ and express the order $k\in\mathbb{Z}_{\ge 0}$ in terms of another integer $n\in\mathbb{Z}_{\ge 0}$ and $s$ by
\begin{equation}
k=2n+\frac{1}{2}(s-1)\quad\Longleftrightarrow\quad n=\frac{1}{4}(2k+1-s).
\label{eq:k-vs-n}
\end{equation}
Each value of $n\in\mathbb{Z}_{>0}$ corresponds to two consecutive values of $k$, one of each parity; however $n=0$ corresponds to $k=0$ only.  Finally, let $B(\lambda)$ denote the elementary Blaschke factor
\begin{equation}
B(\lambda):=\frac{\lambda-\ii}{\lambda+\ii}.
\label{eq:Blaschke}
\end{equation}
In the following problem as in the rest of the paper, boundary values taken from the left/right are denoted with a subscript $+$/$-$, and $\sigma_3$ denotes one of the Pauli matrices:
\begin{equation}
\sigma_1:=\begin{bmatrix}0&1\\1 & 0\end{bmatrix},\quad
\sigma_2:=\begin{bmatrix}0 & -\ii\\\ii & 0\end{bmatrix},\quad
\sigma_3:=\begin{bmatrix}1 & 0\\0 & -1\end{bmatrix}.
\end{equation} 
\begin{rhp}[Rogue wave of order $k$]
Let $(x,t)\in\mathbb{R}^2$ be arbitrary parameters, and let $k\in\mathbb{Z}_{\ge 0}$.  Find a $2\times 2$ matrix $\mathbf{M}^{(k)}(\lambda;x,t)$ with the following properties:
\begin{itemize}
\item[]\textbf{Analyticity:}  $\mathbf{M}^{(k)}(\lambda;x,t)$ is analytic in $\lambda$ for $\lambda\in\mathbb{C}\setminus(\Sigma_\circ\cup\Sigma_\mathrm{c})$, and it takes continuous boundary values on $\Sigma_\circ\cup\Sigma_\mathrm{c}$.
\item[]\textbf{Jump conditions:}  The boundary values on the jump contour $\Sigma_\circ\cup\Sigma_\mathrm{c}$ are related as follows:
\begin{equation}
\mathbf{M}_+^{(k)}(\lambda;x,t)=\mathbf{M}_-^{(k)}(\lambda;x,t)\ee^{2\ii\rho_+(\lambda)(x+\lambda t)\sigma_3},\quad \lambda\in\Sigma_\mathrm{c},
\label{eq:jump-cut}
\end{equation}
and
\begin{equation}
\mathbf{M}_+^{(k)}(\lambda;x,t)=\mathbf{M}_-^{(k)}(\lambda;x,t)\ee^{-\ii\rho(\lambda)(x+\lambda t)\sigma_3}\mathbf{Q}
B(\lambda)^{sn\sigma_3}
\mathbf{Q}^{-1}\mathbf{E}(\lambda)\ee^{\ii\rho(\lambda)(x+\lambda t)\sigma_3},\quad
\lambda\in\Sigma_\circ,
\end{equation}
where $s=\pm 1$ is the parity index of $k$, and $n$ is given by \eqref{eq:k-vs-n}.
\item[]\textbf{Normalization:}  $\mathbf{M}^{(k)}(\lambda;x,t)\to\mathbb{I}$ as $\lambda\to\infty$. 
\end{itemize}
\label{rhp:rogue-wave}
\end{rhp}
The \emph{fundamental rogue wave of order $k$} is then defined by the limit 
\begin{equation}
\psi(x,t)=\psi_k(x,t)\defeq 2\ii\lim_{\lambda\to\infty}\lambda M_{12}^{(k)}(\lambda;x,t), \quad k\in\mathbb{Z}_{\ge 0},
\label{eq:psi-from-M}
\end{equation}
and it is a rational solution of the focusing nonlinear Schr\"odinger equation in the form 
\begin{equation}
\ii\psi_t +\tfrac{1}{2}\psi_{xx}+(|\psi|^2-1)\psi=0,
\label{eq:NLS}
\end{equation}
that tends to the background solution $\psi=\psi_0(x,t)\equiv 1$ as $(x,t)\to\infty$ in $\mathbb{R}^2$.  It is this feature of simultaneous spatio-temporal localization that explains the terminology of \emph{rogue waves} for such solutions.  

As the parameter $k$ increases, the fundamental rogue wave has increasing amplitude (see \cite[Proposition 2]{BilmanLM20} and also \cite{AAS09, WangYWH17}).  This large maximum amplitude is achieved exactly at the origin $(x,t)=(0,0)$, and the aim of the previous paper \cite{BilmanLM20} was to study the fundamental rogue wave of order $k$ in a small neighborhood of this amplitude peak.  It was discovered in \cite{BilmanLM20} that for fixed $s=(-1)^k$, $sn^{-1}\psi_k(n^{-1}X,n^{-2}T)$ converges as $n\to+\infty$ to a limiting function $\Psi(X,T)$, the \emph{rogue wave of infinite order}, that solves the focusing nonlinear Schr\"odinger equation in the form $\ii\Psi_T+\tfrac{1}{2}\Psi_{XX}+|\Psi|^2\Psi=0$.  This limiting function is a highly-transcendental solution having a number of remarkable properties described in \cite{BilmanLM20}, for instance:  (i) it satisfies also ordinary differential equations of Painlev\'e type in the two independent variables, (ii) it has its own Riemann-Hilbert representation, and (iii) $\Psi(X,T)\to 0$ for large $X$ and $T$ (even though $\psi_k\to 1$ for large $x$ and $t$).  The decay for large $X$ is sufficient for the function $\Psi(\cdot,T)$ to lie in $L^2(\mathbb{R})$ for every $T\in\mathbb{R}$, but $\Psi(\cdot,T)\not\in L^1(\mathbb{R})$, and the decay in $T$ is even slower.    The function $\Psi(X,T)$ has recently also been shown to be important in several other problems; for the same equation it describes also high-order multiple-pole soliton solutions \cite{BilmanB19} and self-similar focusing in the setting of weak dispersion \cite{Suleimanov17,BuckinghamJM21}, and for the sharp-line Maxwell-Bloch system in characteristic coordinates it models initial/boundary layers \cite{LiM21}.

\subsection{Reformulated characterization of fundamental rogue waves}
The purpose of this paper is to describe the fundamental rogue-wave solution of high order $k$ in a different regime for the independent variables on which both $x$ and $t$ are instead proportional to $k$.  To this end, in
place of $\mathbf{M}^{(k)}(\lambda;x,t)$, consider the matrix $\mathbf{P}^{(k)}(\lambda;x,t)$ defined by
\begin{multline}
\mathbf{P}^{(k)}(\lambda;x,t)\defeq
\ee^{\frac{1}{2}\ii t\sigma_3}\mathbf{M}^{(k)}(\lambda;x,t)\\{}\cdot
\begin{cases}
\ee^{-\ii\rho(\lambda)(x+\lambda t)\sigma_3}\mathbf{Q}^{s}\ee^{\ii(\lambda x+\lambda^2 t)\sigma_3},&\quad
\text{$\lambda$ inside $\Sigma_\circ$},\\ \displaystyle 
\ee^{\ii[\lambda x+\lambda^2t-\rho(\lambda)(x+\lambda t)]\sigma_3}
B(\lambda)^{-n\sigma_3}
\omega(\lambda)^{-s\sigma_3},&\quad 
\text{$\lambda$ exterior to $\Sigma_\circ$},
\end{cases}
\label{eq:M-P-bulk}
\end{multline}
where we recall that $s=\pm 1$ is the parity index of $k$, where $n$ is defined by \eqref{eq:k-vs-n}, 
and where
\begin{equation}
\omega(\lambda)\defeq f(\lambda)(1+\ii (\lambda-\rho(\lambda))).
\label{eq:omega-def}
\end{equation}
An alternate formula for $\omega(\lambda)$ 
can be found as follows.  First we observe that $\omega(\lambda)$ is analytic for $\lambda\in\mathbb{C}\setminus\Sigma_\mathrm{c}$ and satisfies $\omega(\lambda)\to 1$ as $\lambda\to\infty$.  Using \eqref{eq:f-squared} and $\rho(\lambda)^2=\lambda^2+1$, we easily calculate that
\begin{equation}
\omega(\lambda)^4 = B(\lambda).
\label{eq:omega-fourth-power}
\end{equation}
In particular, it follows from this that, recalling the upward orientation of $\Sigma_\mathrm{c}$,
\begin{equation}
\omega_+(\lambda)=\ii\omega_-(\lambda),\quad\lambda\in\Sigma_\mathrm{c}.
\label{eq:omega-jump}
\end{equation}

It is easy to check that $\mathbf{P}^{(k)}(\lambda;x,t)$ is an analytic function of $\lambda$ for $\lambda\in\mathbb{C}\setminus\Sigma_\circ$, i.e., the jump of $\mathbf{M}^{(k)}(\lambda;x,t)$ across the cut $\Sigma_\mathrm{c}$ between $\pm\ii$ is removed by the substitution, and no additional singularities are introduced.  Since $\lambda x+\lambda^2t-\rho(\lambda)(x+\lambda t) = -\tfrac{1}{2}t+O(\lambda^{-1})$ as $\lambda\to\infty$, it is easy to check that $\mathbf{P}^{(k)}(\lambda;x,t)\to\mathbb{I}$ in the same limit.  One directly calculates that the jump condition satisfied by $\mathbf{P}^{(k)}(\lambda;x,t)$ across the closed curve $\Sigma_\circ$ with clockwise orientation is then
\begin{multline}
\mathbf{P}^{(k)}_+(\lambda;x,t)=\\
\mathbf{P}^{(k)}_-(\lambda;x,t)\ee^{-\ii(\lambda x+\lambda^2t)\sigma_3}\mathbf{Q}^{-s}\mathbf{Q}
B(\lambda)^{sn\sigma_3}
\mathbf{Q}^{-1}\mathbf{E}(\lambda)\omega(\lambda)^{-s\sigma_3}
B(\lambda)^{-n\sigma_3}
\ee^{\ii(\lambda x+\lambda^2t)\sigma_3},\\
\lambda\in\Sigma_\circ.
\end{multline}
But the eigenvalues of $\mathbf{E}(\lambda)$ are precisely $\omega(\lambda)^{\pm 1}$ and $\mathbf{E}(\lambda)$ is diagonalized by the constant orthogonal eigenvector matrix $\mathbf{Q}$, so $\mathbf{Q}^{-1}\mathbf{E}(\lambda)=\omega(\lambda)^{\sigma_3}\mathbf{Q}^{-1}$.  Using this identity as well as $\mathbf{Q}^2=-\ii\sigma_2$ along with \eqref{eq:k-vs-n} and \eqref{eq:omega-fourth-power}, we see that $\mathbf{P}(\lambda;x,t,\mathbf{G},M)=\mathbf{P}^{(k)}(\lambda;x,t)$ solves the following Riemann-Hilbert problem with matrix $\mathbf{G}$ and positive parameter $M$ determined from $k\in\mathbb{Z}_{\ge 0}$ by
\begin{equation}
\mathbf{G}\defeq \mathbf{Q}^{-s}\quad\text{and}\quad  M\defeq n+\tfrac{1}{4}s=\tfrac{1}{2}k+\tfrac{1}{4}.
\label{eq:G-and-M-RogueWaves}
\end{equation}
\begin{rhp}[Reformulated problem for rogue waves]
Let $(x,t)\in\mathbb{R}^2$ and $M\in\mathbb{R}$ be arbitrary parameters, and let $\mathbf{G}$ be a $2\times 2$ matrix satisfying $\det(\mathbf{G})=1$ and $\mathbf{G}^*=\sigma_2\mathbf{G}\sigma_2$.  Find a $2\times 2$ matrix $\mathbf{P}(\lambda)=\mathbf{P}(\lambda;x,t,\mathbf{G},M)$ with the following properties:
\begin{itemize}
\item[]\textbf{Analyticity:}  $\mathbf{P}(\lambda)$ is analytic in $\lambda$ for $\lambda\in\mathbb{C}\setminus\Sigma_\circ$, and it takes continuous boundary values on $\Sigma_\circ$.
\item[]\textbf{Jump conditions:}  The boundary values on the jump contour $\Sigma_\circ$ are related as follows:
\begin{equation}
\mathbf{P}_+(\lambda)=
\mathbf{P}_-(\lambda)
\ee^{-\ii (\lambda x+\lambda^2t)\sigma_3}
B(\lambda)^{M\sigma_3}
\mathbf{G}
B(\lambda)^{-M\sigma_3}
\ee^{\ii(\lambda x+\lambda^2t)\sigma_3},\quad 
\lambda\in\Sigma_\circ,
\label{eq:P-bulk-jump}
\end{equation}
where scalar powers of the Blaschke factor $B(\lambda)$ are analytic for $\lambda\in\mathbb{C}\setminus\Sigma_\mathrm{c}$ and tend to $1$ as $\lambda\to\infty$.
\item[]\textbf{Normalization:}  $\mathbf{P}(\lambda)\to\mathbb{I}$ as $\lambda\to\infty$. 
\end{itemize}
\label{rhp:rogue-wave-reformulation}
\end{rhp}
It follows from \eqref{eq:psi-from-M} and the substitution \eqref{eq:M-P-bulk} that $\psi_k(x,t)$ can be recovered from $\mathbf{P}^{(k)}(\lambda;x,t)$ by the formula
\begin{equation}
\psi_k(x,t)=2\ii\ee^{-\ii t}\lim_{\lambda\to\infty} \lambda P^{(k)}_{12}(\lambda;x,t).
\end{equation}
Note that to prove Theorem~\ref{thm:exterior} below, it will be useful to work with a limiting case for the Jordan curve $\Sigma_\circ$ in which it is squeezed into a dumbbell shape; on the ``neck'' of the dumbbell there is then a different form of the jump condition.  See Section~\ref{sec:dumbbell}.

\subsection{Continuous interpolation between rogue waves and multiple-pole solitons of arbitrary orders}
\label{sec:M-arbitrary}
Even though it is only related to fundamental rogue waves when the parameters $\mathbf{G}$ and $M$ are related to the order $k\in\mathbb{Z}_{\ge 0}$ by \eqref{eq:G-and-M-RogueWaves}, more generally it follows from the vanishing lemma \cite{Zhou89} that Riemann-Hilbert Problem~\ref{rhp:rogue-wave-reformulation} is uniquely solvable globally in $(x,t)\in\mathbb{R}^2$ for any $M\in\mathbb{R}$ and matrix $\mathbf{G}$ with $\det(\mathbf{G})=1$ and $\mathbf{G}=\sigma_2\mathbf{G}^*\sigma_2$.  From the dressing method it then follows that the function 
\begin{equation}
q=q(x,t;\mathbf{G},M)\defeq 2\ii\lim_{\lambda\to\infty}\lambda P_{12}(\lambda;x,t,\mathbf{G},M)
\label{eq:q-define}
\end{equation}
is a well-defined solution of the focusing nonlinear Schr\"odinger equation in the form
\begin{equation}
\ii q_t+\tfrac{1}{2}q_{xx}+|q|^2q=0.
\label{eq:NLS-ZBC}
\end{equation}
This implies, in particular, that $q(x,t;\mathbf{G},M)$ provides a continuous interpolation via solutions of \eqref{eq:NLS-ZBC} of fundamental rogue waves of different (integral) orders.  The intermediate interpolating solutions can be of independent interest.
For instance, noting that a general matrix $\mathbf{G}$ satisfying $\det(\mathbf{G})=1$ and $\mathbf{G}=\sigma_2\mathbf{G}^*\sigma_2$ can be written in the form
\begin{equation}
\mathbf{G}=\frac{1}{\sqrt{|a|^2+|b|^2}}\begin{bmatrix}a & b^*\\-b & a^*\end{bmatrix},\quad a,b\in\mathbb{C},
\label{eq:G-form}
\end{equation}
comparing with \cite{BilmanBW19} one sees that if $M\in\mathbb{Z}_{>0}$, then $q(x,t;\mathbf{G},M)$ is a multiple-pole soliton solution of \eqref{eq:NLS-ZBC} of order $2M$, which satisfies quite different boundary conditions than do rogue waves.  In fact, it is easy to see directly that the jump matrix in Riemann-Hilbert Problem~\ref{rhp:rogue-wave-reformulation} is single-valued meromorphic if only $M\in\tfrac{1}{2}\mathbb{Z}$, with poles of order $2|M|$ at $\lambda=\pm\ii$.  This immediately allows the problem to be reduced to the solution of a finite-dimensional linear system for all such $M$, and hence $q(x,t;\mathbf{G},\tfrac{1}{2}k)$ is a $k^\mathrm{th}$ order pole soliton solution for $k\in\mathbb{Z}_{\ge 0}$.  
In this way, we see that as $M>0$ continuously increases, $q(x,t;\mathbf{G},M)$ remains a solution of the same equation \eqref{eq:NLS-ZBC} that satisfies zero boundary conditions for $M\in\tfrac{1}{2}\mathbb{Z}_{\ge 0}$ and satisfies constant-amplitude nonzero boundary conditions for $M\in\tfrac{1}{2}\mathbb{Z}_{\ge 0}+\tfrac{1}{4}$.  This proves the following.
\begin{theorem}
Let $\mathbf{G}$ be a $2\times 2$ constant matrix with $\det(\mathbf{G})=1$ and $\mathbf{G}=\sigma_2\mathbf{G}^*\sigma_2$, and let $M>0$ be arbitrary.  Then the function $q(x,t;\mathbf{G},M)$ given in terms of the well-defined solution of Riemann-Hilbert Problem~\ref{rhp:rogue-wave-reformulation} by \eqref{eq:q-define} is a global solution of the focusing nonlinear Schr\"odinger equation in the form \eqref{eq:NLS-ZBC} that is a rogue wave of order $k\in\mathbb{Z}_{\ge 0}$ whenever $M=\tfrac{1}{2}k+\tfrac{1}{4}$ and that is a multiple-pole soliton solution of order $k$ for $k\in\mathbb{Z}_{\ge 0}$ whenever $M=\tfrac{1}{2}k$.
\label{thm:solution-family}
\end{theorem}

This strikes us as a remarkable result.  For instance, it asserts that in a precise sense the famous Peregrine solution $\psi_1(x,t)$ can be regarded as a soliton of order $\tfrac{3}{2}$, because $M=\tfrac{3}{4}$ (Peregrine) lies halfway between $M=\tfrac{1}{2}$ (stationary simple-pole soliton for zero boundary conditions) and $M=1$ (stationary double-pole soliton for zero boundary conditions).  
For values of $M\ge 0$ corresponding to neither solitons ($M\in\tfrac{1}{2}\mathbb{Z}_{\ge 0}$) nor rogue waves ($M\in\tfrac{1}{2}\mathbb{Z}_{\ge 0}+\tfrac{1}{4}$), $q(x,t;\mathbf{G},M)$ satisfies the same nonzero boundary conditions as $|x|\to\infty$ as in the rogue-wave case, \emph{except} that the decay to the background is so slow that the difference is not even in $L^2(\mathbb{R})$; by contrast it is well-known that for rogue waves the difference is in $L^1(\mathbb{R})$.  We will give the proof of this slow decay in a subsequent paper devoted to the study of the solutions for general $M\ge 0$.

Despite the fact that the boundary conditions are quite different, because the solitons and rogue waves have now been placed within the same family of solutions, they have certain properties in common.  From \cite{BilmanB19,BilmanLM20} it is known that both types of solutions exhibit the same asymptotic behavior in the large-$M$ limit near the peak amplitude point.  Choosing $\Sigma_\circ$ in Riemann-Hilbert Problem~\ref{rhp:rogue-wave-reformulation} to be a circle of radius $M$ and scaling $(x,t)$ by $x=M^{-1}X$ and $t=M^{-2}T$ produces a limiting jump condition in the $\Lambda=M^{-1}\lambda$ plane that shows immediately that the same limiting behavior near the peak in terms of the rogue wave of infinite order is also valid in the limit $M\to\infty$ along any sequence, so the ``near field'' behavior is universal with respect to $M$.  We will show in this paper that this common asymptotic behavior for the whole solution family extends to a large region of the $(x,t)$-plane, expanding in size as $M\to+\infty$ at a rate proportional to $M$.  Within this region,
the large-$M$ asymptotic behavior of $q(x,t;\mathbf{G},M)$ is rather insensitive to any particular choice of specific unbounded and increasing sequence $\{M_k\}_{k=1}^\infty$.  On the other hand, in the complementary region one sees qualitatively different asymptotic behavior along different sequences.  See Figure~\ref{fig:roguewaves-and-solitons}.

\begin{figure}[h]
\begin{center}
\phantom{!}\hfill\includegraphics[width=0.4\linewidth]{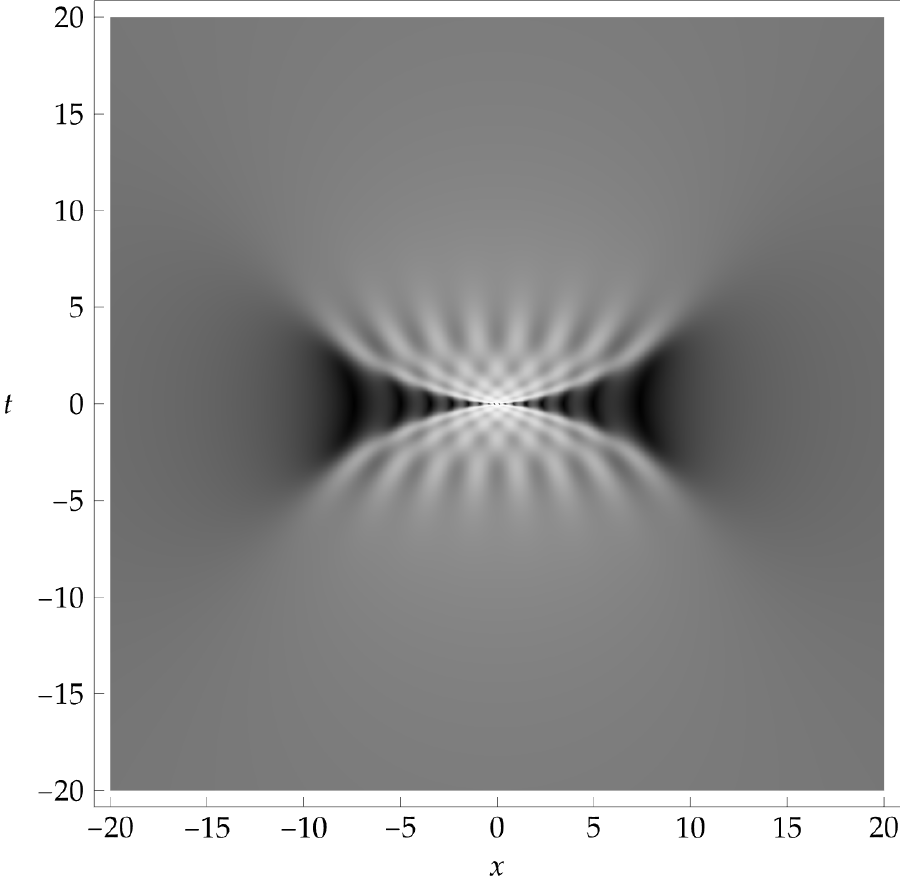}\hfill%
\includegraphics[width=0.4\linewidth]{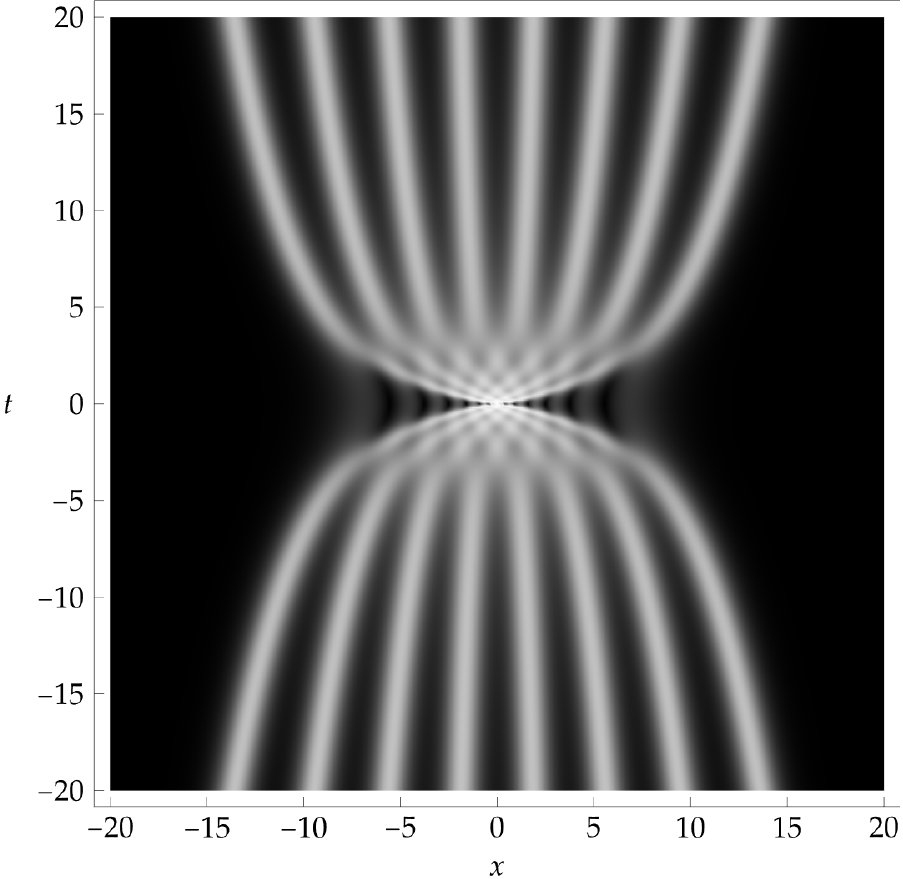}\hfill\phantom{!}\\
\phantom{!}\hfill\includegraphics[width=0.4\linewidth]{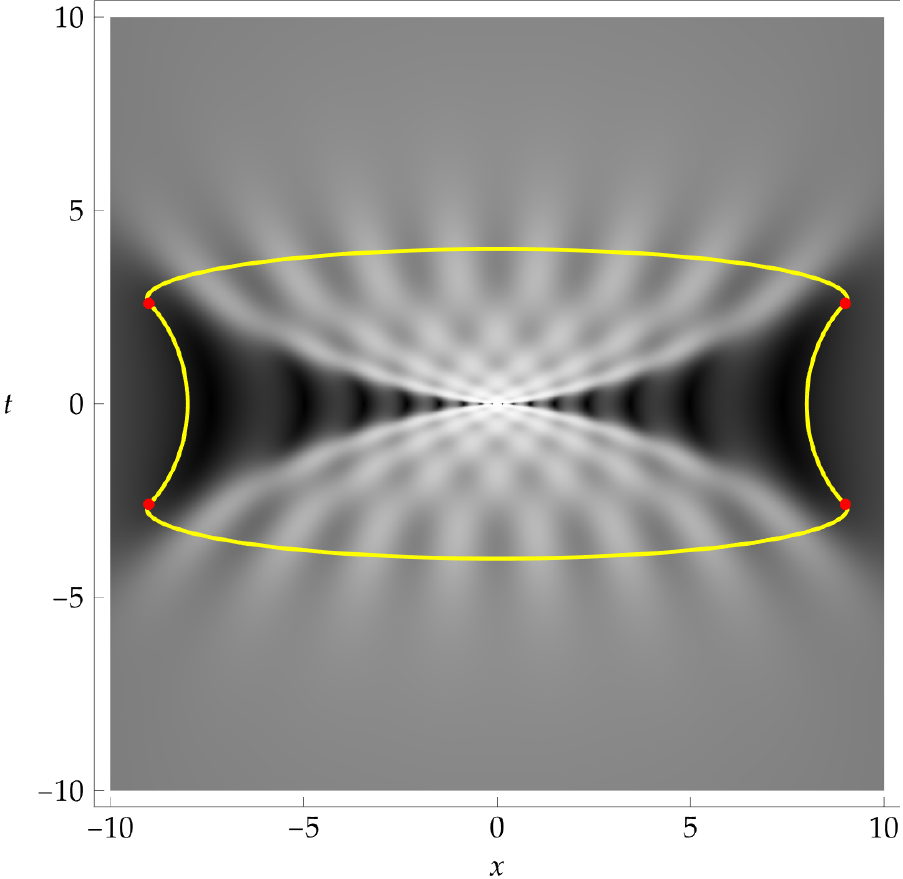}\hfill%
\includegraphics[width=0.4\linewidth]{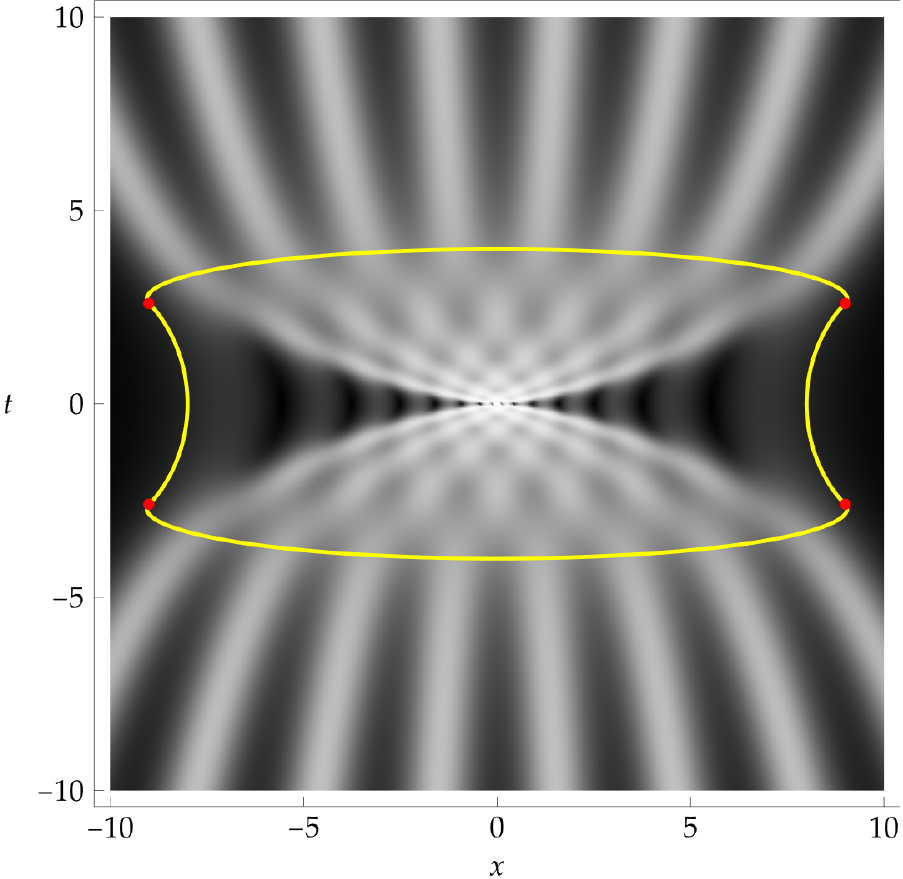}\hfill\phantom{!}
\end{center}
\caption{Top row:  amplitude density plots of the fundamental rogue wave of order $k=8$ (left) and a multiple-pole soliton of order $k=8$ (right).  Bottom row:  as in the top row, but closeup plots showing the region (bounded by yellow curves with red vertices) on which we prove common asymptotic behavior as $k\to+\infty$ for both types of solutions (see Theorem~\ref{thm:channels} and Theorem~\ref{thm:shelves} below).}
\label{fig:roguewaves-and-solitons}
\end{figure}

\subsection{Symmetry assumptions}
The function $q(x,t;\mathbf{G},M)$ is obviously unaffected by any transformation of $\mathbf{P}(\lambda;x,t,\mathbf{G},M)$ within the interior of $\Sigma_\circ$; furthermore, it is easy to see that the form of the jump condition and the symmetry property $\mathbf{G}^*=\sigma_2\mathbf{G}\sigma_2$ are both preserved if the latter transformation is taken to be right-multiplication by $w^{\sigma_3}$ where $w$ is any constant with $|w|=1$.  Thus one sees easily that there is no loss of generality in assuming $a>0$ in the form \eqref{eq:G-form}.  Under this assumption, there are only two matrices $\mathbf{G}$ that build in additional useful symmetries, namely $\mathbf{G}=\mathbf{Q}$ and $\mathbf{G}=\mathbf{Q}^{-1}$.
\begin{proposition}
For all $M>0$ and arbitrary sign $s=\pm 1$,
\begin{equation}
q(-x,t;\mathbf{Q}^{-s},M)=q(x,t;\mathbf{Q}^{-s},M)\quad\text{and}\quad
q(x,-t;\mathbf{Q}^{-s},M)=q(x,t;\mathbf{Q}^{-s},M)^*.
\label{eq:q-symmetries}  
\end{equation}
\label{prop:symmetry}
\end{proposition}
The proof is an elementary application of the representation of $q(x,t;\mathbf{G},M)$ via Riemann-Hilbert Problem~\ref{rhp:rogue-wave-reformulation} and can be found in Appendix~\ref{A:Proofs}.
The specific choice of $\mathbf{G}=\mathbf{Q}^{-s}$ with $s=\pm 1$ in Riemann-Hilbert Problem~\ref{rhp:rogue-wave-reformulation} makes the rogue wave (for $M=\tfrac{1}{2}k+\tfrac{1}{4}$ and $k\in\mathbb{Z}_{\ge 0}$ with $s=(-1)^k$) or soliton (for $M=\tfrac{1}{2}k$ with $k\in\mathbb{Z}_{\ge 0}$ and $s=\pm 1$ arbitrary) ``fundamental''.  For rogue waves the correlation of the sign $s$ with the order $k$ is important\footnote{The alternation of sign in the exponent of $\mathbf{Q}^{-s}$ is necessary to achieve the correct boundary condition $\psi_k(x,t)\to 1$ as $(x,t)\to\infty$.  Using $\mathbf{Q}^{-1}=\ii^{\sigma_3}\mathbf{Q}\ii^{-\sigma_3}$ it is easy to see that exchanging $\mathbf{Q}$ for $\mathbf{Q}^{-1}$ at fixed $M=\tfrac{1}{2}k+\tfrac{1}{4}$ corresponds to the transformation $\mathbf{P}\mapsto \ii^{\sigma_3}\mathbf{P}\ii^{-\sigma_3}$ which implies via \eqref{eq:q-define} that $q\mapsto -q$ and hence yields a rogue wave solution satisfying $\psi_k(x,t)\to -1$ as $(x,t)\to\infty$.} to fix the boundary conditions. 

This result allows us to assume, as we do for the rest of this paper, that $x\ge 0$ and $t\ge 0$.  

\subsection{The far-field regime}
A more important reason for characterizing rogue waves and solitons via Riemann-Hilbert Problem~\ref{rhp:rogue-wave-reformulation} is that its jump condition is well-suited for steepest-descent asymptotic analysis in the large $M>0$ regime where $x$ and $t$ are proportional to $M$.  Indeed, introducing rescaled variables by setting
\begin{equation}
\chi\defeq \frac{x}{M}\quad \text{and} \quad \tau\defeq \frac{t}{M},
\end{equation}
and then defining 
\begin{equation}
\vartheta(\lambda;\chi,\tau)\defeq\chi\lambda + \tau\lambda^2 +\ii\log\left(B(\lambda)\right),
\label{eq:vartheta}
\end{equation}
in which the logarithm is taken to be the principal branch (i.e., $\log(B(\lambda))$ is analytic for $\lambda\in\mathbb{C}\setminus\Sigma_\mathrm{c}$ and $\log(B(\lambda))\to 0$ as $\lambda\to\infty$), we set
\begin{equation}
\mathbf{S}(\lambda;\chi,\tau,\mathbf{G},M)\defeq\mathbf{P}(\lambda;M\chi,M\tau,\mathbf{G},M).
\label{eq:S-from-P}
\end{equation}
Then the jump condition for $\mathbf{S}(\lambda;\chi,\tau,\mathbf{G},M)$ on the jump contour $\Sigma_\circ$ reads
\begin{equation}
\mathbf{S}_+(\lambda;\chi,\tau,\mathbf{G},M)=\mathbf{S}_-(\lambda;\chi,\tau,\mathbf{G},M)\ee^{-\ii M\vartheta(\lambda;\chi,\tau)\sigma_3}\mathbf{G}\ee^{\ii M\vartheta(\lambda;\chi,\tau)\sigma_3},\quad\lambda\in\Sigma_\circ.
\label{eq:S-jump}
\end{equation}
Thus the large parameter $M\gg 1$ enters only via an exponential conjugation.  In general, a solution $q(x,t;\mathbf{G},M)$ of \eqref{eq:NLS-ZBC} is obtained from $\mathbf{S}(\lambda;\chi,\tau,\mathbf{G},M)$ via
\begin{equation}
q(M\chi,M\tau;\mathbf{G},M)=2\ii\lim_{\lambda\to\infty} \lambda S_{12}(\lambda;\chi,\tau,\mathbf{G},M).
\label{eq:q-S}
\end{equation}
To obtain the fundamental rogue wave of order $k$ we tie $\mathbf{G}$ and $M$ to $k$ via \eqref{eq:G-and-M-RogueWaves} and include an additional exponential factor:
\begin{equation}
\psi_k(M\chi,M\tau)=2\ii\ee^{-\ii M\tau}\lim_{\lambda\to\infty}\lambda S_{12}(\lambda;\chi,\tau,\mathbf{Q}^{-s},M),\quad s=(-1)^k,\quad M=\tfrac{1}{2}k+\tfrac{1}{4}.
\label{eq:psi-k-S}
\end{equation}
The regime in which the independent variables $(x,t)$ are proportional to the order $k$ (or more generally, to the parameter $M$) when the latter is large is called the \emph{far-field regime}.  The near-field regime where $x$ and $t$ are small when $k$ or $M$ is large was studied for high-order multiple-pole solitons in \cite{BilmanB19} and for fundamental rogue waves in \cite{BilmanLM20}.  It is important to observe that the near-field and far-field regimes \emph{do not} actually overlap.  There is, however, no expectation of any new phenomena occurring in the intermediate region; the near-field and far-field asymptotic formul\ae\ extend consistently to an expected overlap domain, but the conclusion of common validity over such a domain does not follow from the proofs we will give below.

\subsection{The basic exponent function $\vartheta(\lambda;\chi,\tau)$ and the domain $\channels$}
\label{sec:basic-exponent}
The exponent function $\vartheta(\lambda;\chi,\tau)$ has been studied before in the context of high-order multiple-pole soliton solutions of the focusing nonlinear Schr\"odinger equation \cite{BilmanBW19}; in the notation of that reference, we have $\varphi(\lambda;\chi,\tau,\ii)=\ii\vartheta(\lambda;\chi,\tau)$.  In particular, it is known that $\vartheta(\lambda;\chi,\tau)$ has simple critical points except when $(\chi,\tau)\in\mathbb{R}_{\ge 0}\times\mathbb{R}_{\ge 0}$ are related by the equation
\begin{equation}
16\tau^4 + (8\chi^2-72\chi+108)\tau^2 +\chi^4-2\chi^3=0.
\label{eq:boundary-curve}
\end{equation}
Clearly we can only have $\tau=0$ for $\chi\ge 0$ if $\chi=0$ or $\chi=2$.  
Solving for $\tau^2$ gives
\begin{equation}
\tau^2 
=\tfrac{1}{8}\left[-2\chi^2+18\chi-27\pm 2(9-4\chi)^\frac{3}{2}\right].
\label{eq:tau-squared}
\end{equation}
Reality of $\tau^2$ for $\chi\ge 0$ requires $0\le\chi\le\tfrac{9}{4}$.  If $2\le\chi\le\tfrac{9}{4}$, then both solutions for $\tau^2$ are non-negative.  The values of $\tau^2$ coincide only at the upper endpoint $\chi=\tfrac{9}{4}$ with common value $\tau^2=\tfrac{27}{64}$, and at the lower endpoint $\chi=2$ the smaller value of $\tau^2$ changes sign.  On the interval $0\le\chi< 2$, only the branch of $\tau^2$ with the ``$+$'' sign in \eqref{eq:tau-squared} is nonnegative (and strictly positive except at the lower endpoint $\chi=0$).  Counting with multiplicity, $\vartheta(\lambda;\chi,\tau)$ has three critical points for $\tau\neq 0$, two critical points for $\tau=0$ and $\chi>0$, and no critical points for $\tau=\chi=0$.  The critical points $\lambda$ satisfy the cubic equation
\begin{equation}
2\tau\lambda^3+\chi\lambda^2 + 2\tau\lambda +\chi-2=0,
\end{equation}
and having real coefficients the roots are in general either all real or form a conjugate pair and an isolated real root.  However, in the special case that $\tau=0$ and $0\le\chi\le 2$, there are only two roots, and the critical points are exactly the opposite real numbers
\begin{equation}
\lambda = \pm\sqrt{\frac{2}{\chi}-1},\quad 0\le\chi\le 2,\quad \tau=0.
\label{eq:tau-zero-critical-points}
\end{equation}
It follows that the graphs of the positive square roots of the positive branches of \eqref{eq:tau-squared} border a bounded and relatively open subset $\channels$ of the quadrant $(\chi,\tau)\in\mathbb{R}_{\ge 0}\times\mathbb{R}_{\ge 0}$ such that $(\chi,\tau)\in \channels$ implies that all critical points of $\vartheta(\lambda;\chi,\tau)$ are real and distinct.  In \cite{BilmanBW19} $\channels$ is called the ``algebraic-decay region''.  The same graphs border on the exterior an unbounded and relatively open subset of $\mathbb{R}_{\ge 0}\times\mathbb{R}_{\ge 0}$ on which $\vartheta(\lambda;\chi,\tau)$ has a conjugate pair of critical points with nonzero imaginary part.  The boundary of $\channels$ (shown with a red curve in Figure~\ref{fig:RegionsPlot} below) defined by the relation \eqref{eq:boundary-curve} or \eqref{eq:tau-squared} is smooth except for one point $(\chi^\sharp,\tau^\sharp)$ with coordinates
\begin{equation}
(\chi^\sharp,\tau^\sharp)\defeq\left(\tfrac{9}{4},\tfrac{3\sqrt{3}}{8}\right).
\label{eq:corner-point}
\end{equation}

Although it plays no role in the analysis of high-order fundamental rogue waves, on the exterior of $\channels$ there is a distinguished curve emanating from $(\chi^\sharp,\tau^\sharp)$ that we denote by $\ell_\mathrm{sol}$ along which the level set $\mathrm{Re}(\ii\vartheta(\lambda;\chi,\tau))=0$ is connected. This curve is determined by the condition 
\begin{equation}
\ell_\mathrm{sol}:  \mathrm{Re}\left(\int_\Gamma\ii\vartheta'(\lambda;\chi,\tau)\,\dd\lambda\right)=0,
\label{eq:DegenerateBoutroux}
\end{equation}
where $\Gamma$ is any Schwarz-symmetric contour avoiding $\lambda=\pm\ii$ and having endpoints equal to the complex-conjugate critical points of $\vartheta(\lambda;\chi,\tau)$.  The curve $\ell_\mathrm{sol}$ is shown with a black dotted line in Figure~\ref{fig:RegionsPlot} below; it is important in the asymptotic description of $q(x,t;\mathbf{G},M)$ for large $M\not\in \tfrac{1}{2}\mathbb{Z}_{\ge 0}+\tfrac{1}{4}$.

\subsection{Genus-zero modification of $\vartheta(\lambda;\chi,\tau)$ and the regions $\shelves$ and $\exterior$}
\label{sec:h-intro}
When $(\chi,\tau)\in (\mathbb{R}_{\ge 0}\times\mathbb{R}_{\ge 0})\setminus\overline{\channels}$ it will be necessary to modify the phase $\vartheta(\lambda;\chi,\tau)$ with a genus-zero $g$-function.  Let $\Sigma_g$ be a Schwarz-symmetric sub-arc of 
the jump contour for $\mathbf{S}(\lambda;\chi,\tau,\mathbf{G},M)$ with complex-conjugate endpoints $\lambda_0(\chi,\tau)=A(\chi,\tau)+\ii B(\chi,\tau)$ and $\lambda_0(\chi,\tau)^*=A(\chi,\tau)-\ii B(\chi,\tau)$, and let $g(\lambda;\chi,\tau)$ be bounded and analytic for $\lambda\in\mathbb{C}\setminus\Sigma_g$ with $g(\lambda;\chi,\tau)\to 0$ as $\lambda\to\infty$.  
Consider the matrix $\mathbf{T}(\lambda;\chi,\tau,\mathbf{G},M)$ defined in terms of $g$ and $\mathbf{S}(\lambda;\chi,\tau,\mathbf{G},M)$
by the formula
\begin{equation}
\mathbf{T}(\lambda;\chi,\tau,\mathbf{G},M)\defeq\mathbf{S}(\lambda;\chi,\tau,\mathbf{G},M)\ee^{\ii Mg(\lambda;\chi,\tau)\sigma_3}.
\label{eq:T-to-S}
\end{equation}
Then using \eqref{eq:S-jump}, we see that on the jump contour we have
\begin{equation}
\mathbf{T}_+(\lambda;\chi,\tau,\mathbf{G},M)=\mathbf{T}_-(\lambda;\chi,\tau,\mathbf{G},M)\ee^{-\ii Mh_-(\lambda;\chi,\tau)\sigma_3}
\mathbf{G}
\ee^{\ii Mh_+(\lambda;\chi,\tau)\sigma_3},
\label{eq:T-jump}
\end{equation}
where
\begin{equation}
h(\lambda;\chi,\tau)\defeq \vartheta(\lambda;\chi,\tau)+g(\lambda;\chi,\tau),\quad\lambda\in\mathbb{C}\setminus(\Sigma_g\cup\Sigma_\mathrm{c})
\label{eq:h-define}
\end{equation}
is the modification of $\vartheta(\lambda;\chi,\tau)$ referred to in the section title.  
We impose the additional condition that the sum of boundary values $h_+(\lambda;\chi,\tau)+h_-(\lambda;\chi,\tau)$ is independent of $\lambda\in\Sigma_g$, which simplifies the jump condition \eqref{eq:T-jump} for $\lambda\in\Sigma_g\subset\Sigma_\circ$.  Thus, $h'(\lambda;\chi,\tau)$ is analytic for $\lambda\in\mathbb{C}\setminus(\Sigma_g\cup\{\ii,-\ii\})$, satisfies $h_+'(\lambda;\chi,\tau)+h_-'(\lambda;\chi,\tau)=0$ on $\Sigma_g$, has simple poles inherited from $\vartheta'(\lambda;\chi,\tau)$ at $\lambda=\pm\ii$ with
\begin{equation}
\mathop{\mathrm{Res}}_{\lambda=\pm \ii}h'(\lambda;\chi,\tau)=\pm \ii,
\label{eq:hprime-residues}
\end{equation}
and has the large-$\lambda$ expansion
\begin{equation}
h'(\lambda;\chi,\tau)=2\tau\lambda+\chi + O(\lambda^{-2}),\quad\lambda\to\infty.
\label{eq:hprime-expansion}
\end{equation}
Letting $R(\lambda;\chi,\tau)$ be the analytic function for $\lambda\in\mathbb{C}\setminus\Sigma_g$ satisfying 
\begin{equation}
\begin{split}
R(\lambda;\chi,\tau)^2&=(\lambda-\lambda_0(\chi,\tau))(\lambda-\lambda_0(\chi,\tau)^*)\\
&=(\lambda-A(\chi,\tau))^2+B(\chi,\tau)^2,\quad\text{and $R(\lambda;\chi,\tau)=\lambda + O(1)$ as $\lambda\to\infty$}, 
\end{split}
\label{eq:R-define}
\end{equation}
it follows that $h'(\lambda;\chi,\tau)$ necessarily has the form
\begin{equation}
h'(\lambda;\chi,\tau)=\frac{2\tau\lambda^2+u(\chi,\tau)\lambda+v(\chi,\tau)}{\lambda^2+1}R(\lambda;\chi,\tau),
\label{eq:hprime-formula}
\end{equation}
where $A(\chi,\tau)\in\mathbb{R}$, $B(\chi,\tau)^2>0$, $u(\chi,\tau)\in\mathbb{R}$, and $v(\chi,\tau)\in\mathbb{R}$ are to be determined (uniquely, see Section~\ref{sec:g-function}) so that $h'(\lambda;\chi,\tau)$ has the desired residues \eqref{eq:hprime-residues} and large-$\lambda$ expansion \eqref{eq:hprime-expansion}.  This determination also places conditions on the location of the branch cut $\Sigma_g$ relative to the points $\lambda=\pm\ii$; see Remark~\ref{rem:Sigma_g}.

It turns out that the boundary curve \eqref{eq:boundary-curve} reappears in the analysis of the modified phase function $h(\lambda;\chi,\tau)$ as the condition that $B(\chi,\tau)^2=0$.  
In other words,
the roots of $R(\lambda;\chi,\tau)^2$ form a well-defined conjugate pair for all $(\chi,\tau)$ in the part of the first quadrant complementary to the domain $\channels$ on which the unmodified phase $\vartheta(\lambda;\chi,\tau)$ has three real critical points, and both $A(\chi,\tau)$ and $B(\chi,\tau)>0$ are real analytic functions of $(\chi,\tau)\in (\mathbb{R}_{\ge 0}\times\mathbb{R}_{\ge 0})\setminus\overline{\channels}$.  It is easy to show from the construction of $h'(\lambda;\chi,\tau)$ in Section~\ref{sec:g-function} that if one introduces polar coordinates via $\chi=r\cos(\theta)$ and $\tau=r\sin(\theta)$, then $A(\chi,\tau)\pm\ii B(\chi,\tau)\to\pm\ii$ as $r\to\infty$ uniformly with respect to $\theta$.  Also, $A(\chi,\tau)\pm\ii B(\chi,\tau)\to\infty$ as $(\chi,\tau)\to 0$ from the exterior of $\channels$.

In the study of high-order multiple-pole soliton solutions of the focusing nonlinear Schr\"odinger equation carried out in \cite{BilmanBW19}, the exterior of $\channels$ is further divided into three sub-regions, two unbounded and one bounded, on each of which a different modified phase function is needed (trivial modification, genus zero as described above, and genus one).  The rogue wave problem is simpler in that the single genus-zero phase function $h(\lambda;\chi,\tau)$ suffices to control the large-$k$ asymptotics throughout the exterior of $\channels$; however on the bounded component of the exterior identified in \cite{BilmanBW19} (where it is called the ``non-oscillatory region'' and which we denote by $\shelves$) (i) an additional $O(1)$ contribution to the phase appears in the leading term and (ii) the higher-order correction takes a different form than on the remaining unbounded component of the exterior, which we denote by $\exterior$.  The latter effect is observable in plots for finite order $k$.  The domain $\shelves$ abuts the domain $\channels$ along
the curve given by \eqref{eq:tau-squared} taken with the ``$+$'' sign, and the other part of its boundary in the first quadrant consists of a curve connecting the point $(\chi^\sharp,\tau^\sharp)$ defined in \eqref{eq:corner-point} with $(0,1)$.  While on the interior of $\shelves$ the roots of the quadratic factor in the numerator of \eqref{eq:hprime-formula} are real and distinct, denoted by $a(\chi,\tau)<b(\chi,\tau)$, the quadratic discriminant vanishes on this second boundary curve, which is shown with a solid blue line in Figure~\ref{fig:RegionsPlot}.

\begin{figure}[h]
\begin{center}
\includegraphics[width=0.5\linewidth]{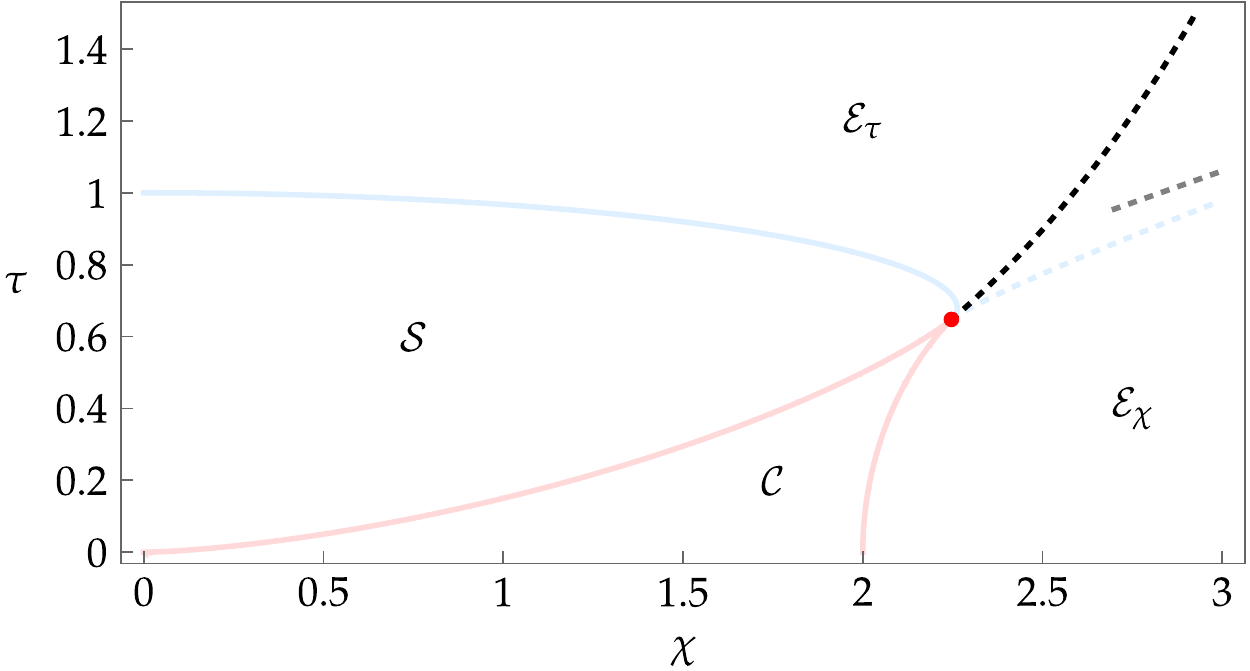}
\end{center}
\caption{The first quadrant in the $(\chi,\tau)$-plane and the regions $\channels$, $\shelves$, and $\exterior$ (which is further divided into $\exterior_\tau$ and $\exterior_\chi$).  The red curve is given by \eqref{eq:boundary-curve} or \eqref{eq:tau-squared}.  The distinguished point $(\chi^\sharp,\tau^\sharp)$ defined by \eqref{eq:corner-point} is indicated with a red dot.  Along the solid and dotted blue curves, $h'(\lambda;\chi,\tau)$ given by \eqref{eq:hprime-formula} has a (real) double root.  Important phase transitions for high-order fundamental rogue waves occur along the red curve and solid blue curve; the unbounded dotted blue curve in $\exterior$ separating $\exterior_\tau$ from $\exterior_\chi$ is of mere technical significance in our analysis.  The dotted gray line $\chi=\sqrt{8}\tau$ is an asymptote for large $\chi$ to the unbounded branch (see Section~\ref{sec:critical-points}).  The curve $\ell_\mathrm{sol}$ emanating from $(\chi^\sharp,\tau^\sharp)$ into $\exterior_\tau$ and shown with a black dotted line is of no importance at all for high-order fundamental rogue waves but it is crucial in the study of high-order multiple-pole solitons and of secondary importance in the study of $q(x,t;\mathbf{G},M)$ for other solutions characterized by Riemann-Hilbert Problem~\ref{rhp:rogue-wave-reformulation}.}
\label{fig:RegionsPlot}
\end{figure}

Exiting $\shelves$ through that curve, the roots of the quadratic factor become a complex-conjugate pair.  There is one additional unbounded curve emanating from $(\chi^\sharp,\tau^\sharp)$ into the exterior (denoted $\exterior$) of $\overline{\channels\cup\shelves}$ along which the discriminant vanishes again.  Crossing this curve (shown as a dotted blue curve in Figure~\ref{fig:RegionsPlot}), the roots of the quadratic factor become real once again.  We refer to the two unbounded components of $\exterior$ separated by this curve as $\exterior_\chi$ (the unbounded component abutting the positive $\chi$-axis for $\chi>2$) and $\exterior_\tau$ (the unbounded component abutting the positive $\tau$-axis for $\tau>1$).
The roots of the quadratic factor $2\tau\lambda^2+u(\chi,\tau)\lambda+v(\chi,\tau)$ appearing in \eqref{eq:hprime-formula} are real when $(\chi,\tau)\in \exterior_\chi\cup \shelves$ and form a complex-conjugate pair when $(\chi,\tau)\in \exterior_\tau$.

\begin{remark}
Although the bounded domain $\shelves$ coincides exactly with the ``non-oscillatory region'' identified in \cite{BilmanBW19}, the curve separating $\exterior_\tau$ from $\exterior_\chi$ is not the same as the curve $\ell_\mathrm{sol}$ separating the two unbounded components of $\exterior$ (the ``oscillatory region'' and the ``exponential-decay region'') identified in \cite{BilmanBW19} and relevant for the study of $q(x,t;\mathbf{G},M)$ for large $M\in\tfrac{1}{2}\mathbb{Z}_{\ge 0}$.  The latter curve is shown with a black dotted line in Figure~\ref{fig:RegionsPlot}.
\end{remark}

A discussion of the qualitative features of high-order fundamental rogue waves can be found in \cite[Section 1.1]{BilmanLM20}.  Near the origin in the $(x,t)$-plane one observes a narrow wedge-shaped region centered on each half of the $x$-axis containing small-amplitude oscillations and larger complementary regions centered on each half of the $t$-axis containing waves of higher amplitude.  In \cite[Section 1.1]{BilmanLM20} these types of regions near the origin were called ``channels'' and ``shelves'' respectively.  The channels and shelves were proven in \cite{BilmanLM20} to have significance for the asymptotic behavior of the rogue wave of infinite order.  In this paper we show that the channels and shelves extend also to the macroscopic regime of bounded $(\chi,\tau)$ as $\channels$ and $\shelves$ respectively. As the paper \cite{BilmanLM20} was concerned with fundamental rogue waves in a neighborhood of the origin only, the identification of the exterior domain $\exterior$ is new in this work.  Note that by definition $\channels$, $\shelves$, and $\exterior$ are all relatively open pairwise disjoint subsets of the closed first quadrant, whose union excludes only the boundary curves shown with solid lines in Figure~\ref{fig:RegionsPlot}.  Likewise $\exterior_\chi$ and $\exterior_\tau$ are relatively open disjoint subsets of $\exterior$, whose union excludes only the dotted blue curve shown in Figure~\ref{fig:RegionsPlot}.

The significance of these regions for high-order fundamental rogue waves can be seen in Figure~\ref{fig:2D-Plots}.  
\begin{figure}[h]
\begin{center}
\phantom{!}\hfill%
\includegraphics[width=0.33\linewidth]{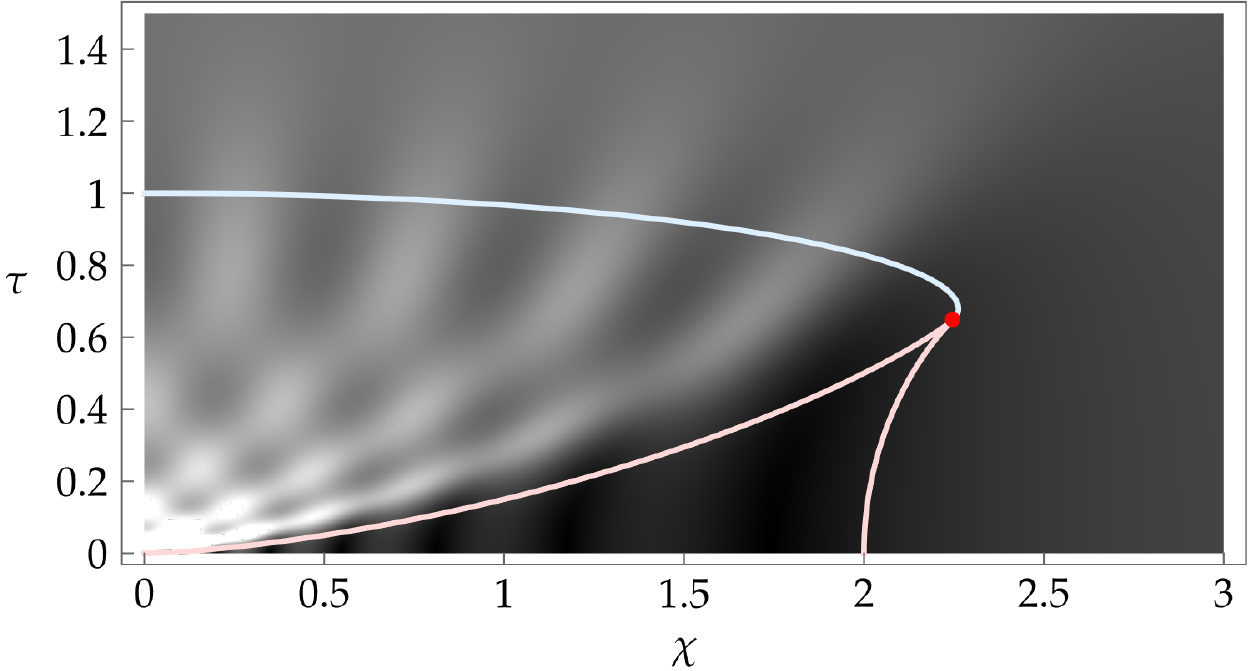}%
\hfill%
\includegraphics[width=0.33\linewidth]{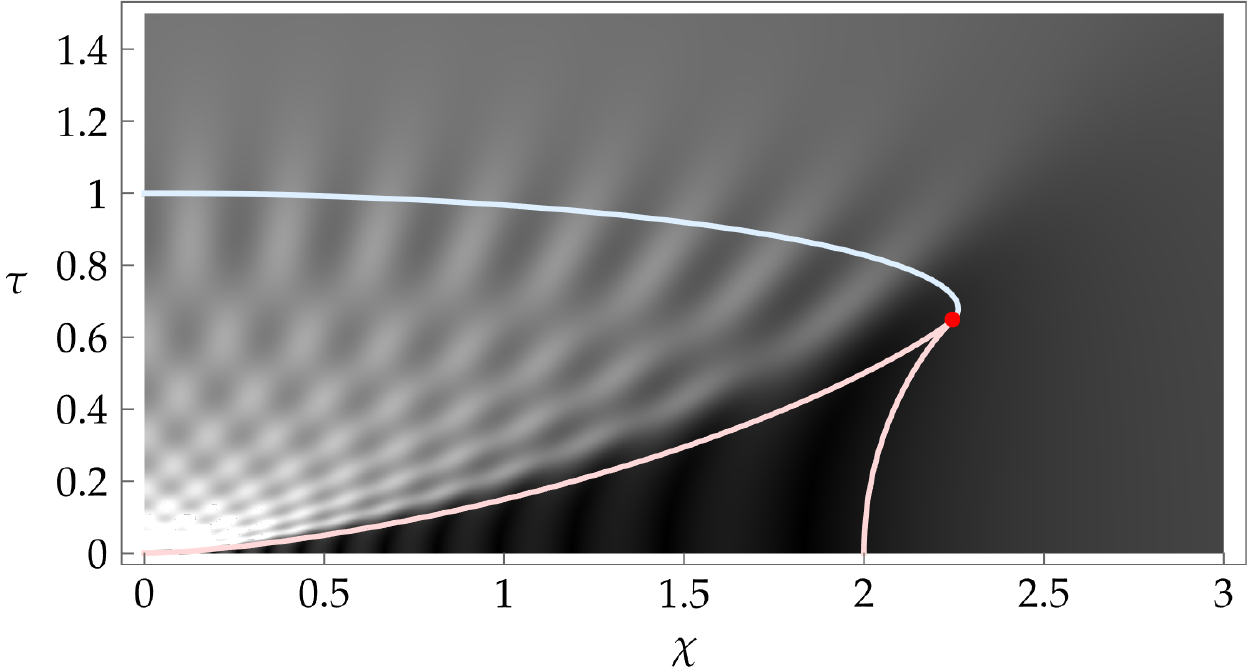}%
\hfill%
\includegraphics[width=0.33\linewidth]{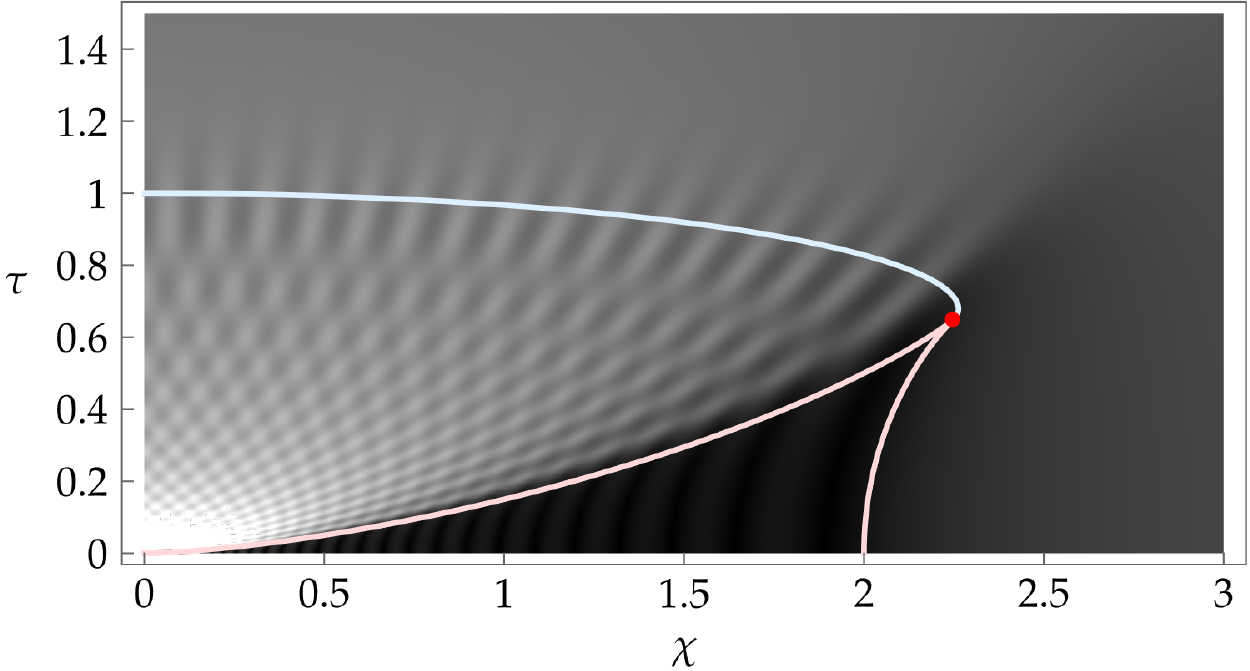}%
\hfill\phantom{!}%
\end{center}
\caption{Density plots of $|\psi_k(M\chi,M\tau)|$ with the region boundaries superimposed for $k=8$ and $M=4.25$ (left), for $k=16$ and $M=8.25$ (center), and for $k=32$ and $M=16.25$ (right).}
\label{fig:2D-Plots}
\end{figure}

\subsection{Results}
\subsubsection{Asymptotic behavior of $q(x,t;\mathbf{Q}^{-s},M)$ and fundamental rogue waves for $(\chi,\tau)\in\channels$}
\label{sec:results-channels}
Recall from Section~\ref{sec:basic-exponent} that when $(\chi,\tau)\in \channels$, the phase $\vartheta(\lambda;\chi,\tau)$ defined in \eqref{eq:vartheta} has only real and simple critical points.  When also $\tau=0$, there are precisely two of them, given by \eqref{eq:tau-zero-critical-points}.  We denote the unique continuation of these critical points to the domain $\channels$ by $a=a(\chi,\tau)$ and $b=b(\chi,\tau)$ with $a<b$.  For $\tau\neq 0$ there is a third critical point born from $\lambda=\infty$, with leading asymptotic $\lambda=-\chi/(2\tau)+ O(1)$ as $\tau\to 0$.  Hence this critical point lies to the left of $\lambda=a$ for $\tau>0$ and to the right of $\lambda=b$ for $\tau<0$, since $\chi>0$ holds throughout $\channels$.  
For $M>0$, define real phases $\Theta_a^{[\channels]}(\chi,\tau;M)$ and $\Theta_b^{[\channels]}(\chi,\tau;M)$ by 
\begin{equation}
\begin{split}
\Theta_a^{[\channels]}(\chi,\tau;M)&\defeq M\Phi_a^{[\channels]}(\chi,\tau)-\ln(M)\frac{\ln(2)}{2\pi}+\eta_a^{[\channels]}(\chi,\tau),\\
\Theta_b^{[\channels]}(\chi,\tau;M)&\defeq M\Phi_b^{[\channels]}(\chi,\tau)+\ln(M)\frac{\ln(2)}{2\pi} + \eta_b^{[\channels]}(\chi,\tau),
\end{split}
\label{eq:channels-phases}
\end{equation}
in which, noting that $\vartheta(a(\chi,\tau);\chi,\tau)$ and $\vartheta(b(\chi,\tau);\chi,\tau)$ are both real,
\begin{equation}
\begin{split}
\Phi_a^{[\channels]}(\chi,\tau)&\defeq -2\vartheta(a(\chi,\tau);\chi,\tau),\\
\Phi_b^{[\channels]}(\chi,\tau)&\defeq -2\vartheta(b(\chi,\tau);\chi,\tau)
\end{split}
\label{eq:channels-principal-phases}
\end{equation}
and, noting that $\vartheta''(a(\chi,\tau);\chi,\tau)<0$ and $\vartheta''(b(\chi,\tau);\chi,\tau)>0$ (derivatives with respect to $\lambda$),
\begin{equation}
\begin{split}
\eta_a^{[\channels]}(\chi,\tau)&\defeq -\frac{\ln(2)}{2\pi}\ln\left(-(b(\chi,\tau)-a(\chi,\tau))^2\vartheta''(a(\chi,\tau);\chi,\tau)\right)\\&\qquad\qquad\qquad\qquad{}-\frac{\ln(2)^2}{2\pi}-\frac{1}{4}\pi+\arg\left(\Gamma\left(\frac{\ii\ln(2)}{2\pi}\right)\right),\\
\eta_b^{[\channels]}(\chi,\tau)&\defeq \frac{\ln(2)}{2\pi}\ln\left((b(\chi,\tau)-a(\chi,\tau))^2\vartheta''(b(\chi,\tau);\chi,\tau)\right)\\
&\qquad\qquad\qquad\qquad{}+\frac{\ln(2)^2}{2\pi}+\frac{1}{4}\pi-\arg\left(\Gamma\left(\frac{\ii\ln(2)}{2\pi}\right)\right).
\end{split}
\label{eq:channels-lower-order-phases}
\end{equation}
Also, define two positive amplitudes by
\begin{equation}
\begin{split}
F_a^{[\channels]}(\chi,\tau)&\defeq \sqrt{-\frac{\ln(2)}{\pi\vartheta''(a(\chi,\tau);\chi,\tau)}},\\
F_b^{[\channels]}(\chi,\tau)&\defeq \sqrt{\frac{\ln(2)}{\pi\vartheta''(b(\chi,\tau);\chi,\tau)}}.
\end{split}
\label{eq:channels-amplitudes}
\end{equation}
Our first result is then the following:  
\begin{theorem}[Far-field asymptotics of $q(x,t;\mathbf{Q}^{-s},M)$ for $(\chi,\tau)\in\channels$]
Let $s=\pm 1$ be arbitrary.  Then, as $M\to+\infty$ through an arbitrary sequence of values,
$q(M\chi,M\tau;\mathbf{Q}^{-s},M)=\mathfrak{L}_s^{[\channels]}(\chi,\tau;M)+O(M^{-\frac{3}{2}})$, where
\begin{equation}
\mathfrak{L}_s^{[\channels]}(\chi,\tau;M)\defeq sM^{-\frac{1}{2}}\left[F_a^{[\channels]}(\chi,\tau)\ee^{\ii\Theta_a^{[\channels]}(\chi,\tau;M)} + F_b^{[\channels]}(\chi,\tau)\ee^{\ii\Theta_b^{[\channels]}(\chi,\tau;M)}\right],
\label{eq:leading-term-channels-q}
\end{equation}
and where the error term is uniform for $(\chi,\tau)$ in any compact subset of $\channels$.
\label{thm:channels}
\end{theorem}
We present the proof in Section~\ref{sec:channels}.  Note that as this result allows for $M$ to take any positive values tending to $+\infty$, it describes both high-order multiple-pole soliton solutions and fundamental rogue waves over the same domain $\channels$ (as well as many other families of solutions interpolating between the two types); hence $\channels$ with its reflections in the coordinate axes forms a component of the region bounded by the yellow curves in Figure~\ref{fig:roguewaves-and-solitons}.   For the high-order multiple-pole soliton case corresponding to large $M\in\tfrac{1}{2}\mathbb{Z}_{\ge 0}$, it implies one of the results in \cite{BilmanBW19}, although we sharpen the error estimate from $O(M^{-1})$ to $O(M^{-\frac{3}{2}})$.  For the rogue wave case of most interest to us here, we need to correlate the values of $M$ and the index $s$ to the order $k$ and include an additional exponential factor:
\begin{corollary}
The fundamental rogue wave of order $k\in\mathbb{Z}_{>0}$ satisfies $\psi_k(M\chi,M\tau)=L_k^{[\channels]}(\chi,\tau)+O(k^{-\frac{3}{2}})$, where
\begin{equation}
L_k^{[\channels]}(\chi,\tau)\defeq 
\ee^{-\ii M\tau}\mathfrak{L}_s^{[\channels]}(\chi,\tau;M),
\quad s=(-1)^k,\quad M=\tfrac{1}{2}k+\tfrac{1}{4},
\label{eq:leading-term-channels}
\end{equation}
in which $\mathfrak{L}_s^{[\channels]}(\chi,\tau;M)$ is given by \eqref{eq:leading-term-channels-q} 
and where the error term is uniform for $(\chi,\tau)$ in any compact subset of $\channels$.
\label{cor:rogue-wave-channels}
\end{corollary}
When $\tau=0$, the two oscillations in the leading term have a common amplitude, and the formula simplifies further.  Indeed,
\begin{equation}
L_k^{[\channels]}(\chi,0)=\sqrt{\frac{\ln(2)}{\pi M}}\frac{2}{\chi^\frac{3}{4}(2-\chi)^\frac{1}{4}}\cos\left(2MF(\chi)-\frac{\ln(2)}{2\pi}\ln(M)-\Omega(\chi)+\Phi_0\right),\quad 0<\chi<2,
\label{eq:leading-term-channels-tau-zero}
\end{equation}
where
\begin{equation}
F(\chi)\defeq\chi\sqrt{\frac{2}{\chi}-1}+\pi-2\tan^{-1}\left(\sqrt{\frac{2}{\chi}-1}\right),\quad
\Omega(\chi)\defeq\frac{\ln(2)}{\pi}\ln(\chi)+\frac{3\ln(2)}{2\pi}\ln\left(\sqrt{\frac{2}{\chi}-1}\right),
\end{equation}
and
\begin{equation}
\Phi_0\defeq\left(k-\frac{1}{4}\right)\pi-\frac{3\ln(2)^2}{2\pi}+\arg\left(\Gamma\left(\frac{\ii\ln(2)}{2\pi}\right)\right).
\label{eq:leading-term-channels-tau-zero-last}
\end{equation}

\subsubsection{Asymptotic behavior of fundamental rogue waves for $(\chi,\tau)\in\exterior$}
\label{sec:Results-Exterior}
Unlike the analysis for $(\chi,\tau)\in\channels$, our next result pertains to fundamental rogue waves only, i.e., we cannot allow $M$ to tend to $\infty$ in an arbitrary fashion without making substantial modifications that are beyond the scope of this work.  Let $\gamma(\chi,\tau)$ be defined on $\exterior$ by
\begin{equation}
\gamma(\chi,\tau)\defeq \chi A(\chi,\tau)+\tau(A(\chi,\tau)^2-\tfrac{1}{2}B(\chi,\tau)^2)+\ii\int_{-\ii}^{\lambda_0(\chi,\tau)^*}\frac{\dd\lambda}{R(\lambda;\chi,\tau)}+\ii\int_{\lambda_0(\chi,\tau)}^{\ii}
\frac{\dd\lambda}{R(\lambda;\chi,\tau)},
\label{eq:gamma-formula-intro}
\end{equation}
in which the path of integration in each integral is arbitrary in the part of the upper/lower half-plane complementary to $\Sigma_g$.  (In practice, to compute $\gamma(\chi,\tau)$ for $\chi>0$ in $\exterior$ it suffices to let $R(\lambda;\chi,\tau)$ have a vertical branch cut connecting $\lambda_0(\chi,\tau)$ and $\lambda_0(\chi,\tau)^*$; recall that $A(\chi,\tau)=\mathrm{Re}(\lambda_0(\chi,\tau))$ and $B(\chi,\tau)=\mathrm{Im}(\lambda_0(\chi,\tau))$.)  
\begin{theorem}[Far-field asymptotics of $\psi_k(x,t)$ for $(\chi,\tau)\in\exterior$]
The fundamental rogue wave $\psi_k(x,t)$ of order $k\in\mathbb{Z}_{>0}$ satisfies $\psi_k(M\chi,M\tau)=L_k^{[\exterior]}(\chi,\tau)+O(k^{-1})$, where
\begin{equation}
L_k^{[\exterior]}(\chi,\tau)\defeq B(\chi,\tau)\ee^{-\ii M\tau}\ee^{-2\ii M\gamma(\chi,\tau)},\quad M=\tfrac{1}{2}k+\tfrac{1}{4},
\label{eq:leading-term-exterior}
\end{equation}
and where the error term is uniform for $(\chi,\tau)$ in compact subsets of $\exterior$.
\label{thm:exterior}
\end{theorem}
This result shows a marked difference between high-order fundamental rogue waves and high-order soliton solutions of the focusing nonlinear Schr\"odinger equation.  Indeed, for $(\chi,\tau)\in\exterior$, fundamental rogue waves behave like a slowly-modulated plane-wave solution of the same equation.  By contrast, high-order multiple-pole solitons behave like a slowly-modulated elliptic function solution or decay exponentially to zero on complementary subregions of $\exterior$ \cite{BilmanBW19}. It is not difficult to show that $\gamma(\chi,\tau)+\tfrac{1}{2}\tau\to 0$ as $(\chi,\tau)\to\infty$ in $\exterior$.  In conjunction with the fact that $A(\chi,\tau)\pm\ii B(\chi,\tau)\to \pm\ii$ as $(\chi,\tau)\to\infty$ in $\exterior$, this shows that the leading term $L_k^{[\exterior]}(\chi,\tau)$ tends to the background solution $\psi_0\equiv 1$ as $(\chi,\tau)\to\infty$ in $\exterior$, a result that is consistent with the known asymptotic $\psi_k(x,t)\to 1$ as $(x,t)\to\infty$ in $\mathbb{R}^2$, although our proof of Theorem~\ref{thm:exterior} as given in Section~\ref{sec:Schi-Stau} does not allow $(\chi,\tau)$ to become unbounded.  

It is also worth noting that the leading term $L_k^{[\exterior]}(\chi,\tau)$ becomes explicit if $\tau=0$.  Indeed, using \eqref{eq:AB-tau-small} below, one sees that $A(\chi,0)=0$ and $B(\chi,\tau)^2 = 1-4/\chi^2$ for $\chi>2$, and hence also from \eqref{eq:gamma-formula-intro} one obtains $\gamma(\chi,0)=0$.  Therefore,
\begin{equation}
L_k^{[\exterior]}(\chi,0)=\sqrt{1-\frac{4}{\chi^2}},\quad \chi>2,
\end{equation}
a formula that, in light of Theorem~\ref{thm:exterior}, describes precisely how $\psi_k(M\chi,0)$ rises from being small of size $k^{-\frac{1}{2}}$ for $0<\chi<2$ (as given in Corollary~\ref{cor:rogue-wave-channels} and \eqref{eq:leading-term-channels-tau-zero}--\eqref{eq:leading-term-channels-tau-zero-last}) to ultimately approach the unit background value for large $x$.

\subsubsection{Asymptotic behavior of $q(x,t;\mathbf{Q}^{-s},M)$ and fundamental rogue waves for $(\chi,\tau)\in \shelves$}
\label{sec:Results-Shelves}
The next results again allow $M$ to become large in an arbitrary fashion, and they concern the asymptotic description of $q(x,t;\mathbf{Q}^{-s},M)$ for rescaled coordinates $(\chi,\tau)\in \shelves$.  An obvious feature of the plots shown in Figure~\ref{fig:2D-Plots} as well as similar plots of high-order multiple-pole solitons \cite{BilmanBW19} is that in the domain $\shelves$ there are evidently amplitude oscillations of small (on the scale of $\chi$ and $\tau$) wavelength and period.  To capture these oscillations it is necessary to include both a leading term and a first error term in an asymptotic formula for $q(x,t;\mathbf{Q}^{-s},M)$.  To formulate our result, we first define some quantities.  Recall that for $(\chi,\tau)\in \shelves$, the function $h'(\lambda;\chi,\tau)$ has two real simple zeros $a(\chi,\tau)<b(\chi,\tau)$, as well as a conjugate pair $A(\chi,\tau)\pm\ii B(\chi,\tau)$ of branch points.  For such $(\chi,\tau)$ we assume that the Schwarz-symmetric logarithmic branch cut $\Sigma_\mathrm{c}$ connecting $\pm\ii$ with upward orientation crosses the real axis at a unique point between $a(\chi,\tau)$ and $b(\chi,\tau)$.  First, set
\begin{equation}
\kappa(\chi,\tau)\defeq \chi A(\chi,\tau)+\tau(A(\chi,\tau)^2-\tfrac{1}{2}B(\chi,\tau)^2)+\ii\int_{\Sigma_\mathrm{c}}\frac{\dd\lambda}{R(\lambda;\chi,\tau)},\quad (\chi,\tau)\in \shelves,
\label{eq:kappa-formula}
\end{equation}
which is well defined under the assumption that the path of integration $\Sigma_\mathrm{c}$ lies to the right of the Schwarz-symmetric branch cut $\Sigma_g$ of $R$, which we assume crosses the real axis only at $\lambda=a(\chi,\tau)$.  Then define
\begin{equation}
\mu(\chi,\tau)\defeq\frac{\ln(2)}{2\pi}\int_{a(\chi,\tau)}^{b(\chi,\tau)}\frac{\dd\lambda}{R(\lambda;\chi,\tau)}>0,
\label{eq:mu-formula-intro}
\end{equation}
where the integration is on the real line where the integrand is strictly positive.  Next, we set
\begin{equation}
K_a(\chi,\tau)\defeq \frac{\ln(2)}{\pi}\frac{|a(\chi,\tau)-\lambda_0(\chi,\tau)|}{2\pi\ii}
\int_{C}\log\left(\frac{\lambda-a(\chi,\tau)}{\lambda-b(\chi,\tau)}\right)\frac{\dd\lambda}{R(\lambda;\chi,\tau)(\lambda-a(\chi,\tau))}
\label{eq:intro-Ka}
\end{equation}
and
\begin{equation}
K_b(\chi,\tau)\defeq \frac{\ln(2)}{\pi}\frac{|b(\chi,\tau)-\lambda_0(\chi,\tau)|}{2\pi\ii}
\int_{C}\log\left(\frac{\lambda-a(\chi,\tau)}{\lambda-b(\chi,\tau)}\right)\frac{\dd\lambda}{R(\lambda;\chi,\tau)(\lambda-b(\chi,\tau))},
\label{eq:intro-Kb}
\end{equation}
where the contour $C$ lies to the left of $\Sigma_g$ with the same endpoints and orientation, and where the logarithm is cut on the real line in $[a(\chi,\tau),b(\chi,\tau)]$ and tends to zero as $\lambda\to\infty$.  Now we define real phases $\Theta_a^{[\shelves]}(\chi,\tau;M)$ and $\Theta_b^{[\shelves]}(\chi,\tau;M)$ by (compare with \eqref{eq:channels-phases})
\begin{equation}
\begin{split}
\Theta_a^{[\shelves]}(\chi,\tau;M)&\defeq M\Phi_a^{[\shelves]}(\chi,\tau)-\ln(M)\frac{\ln(2)}{2\pi}+\eta_a^{[\shelves]}(\chi,\tau)
\\
\Theta_b^{[\shelves]}(\chi,\tau;M)&\defeq M\Phi_b^{[\shelves]}(\chi,\tau)+\ln(M)\frac{\ln(2)}{2\pi} +
\eta_b^{[\shelves]}(\chi,\tau) 
\end{split}
\label{eq:Thetas-shelves}
\end{equation}
in which, noting that $h_-(a(\chi,\tau);\chi,\tau)$ and $h(b(\chi,\tau);\chi,\tau)$ are both real and comparing with \eqref{eq:channels-principal-phases},
\begin{equation}
\begin{split}
\Phi_a^{[\shelves]}(\chi,\tau)&\defeq -2h_-(a(\chi,\tau);\chi,\tau)\\
\Phi_b^{[\shelves]}(\chi,\tau)&\defeq -2h(b(\chi,\tau);\chi,\tau)
\end{split}
\label{eq:Phis-shelves}
\end{equation}
and, noting that $h''_-(a(\chi,\tau);\chi,\tau)<0$ and $h''(b(\chi,\tau);\chi,\tau)>0$ and comparing with \eqref{eq:channels-lower-order-phases},
\begin{equation}
\begin{split}
\eta_a^{[\shelves]}(\chi,\tau)&\defeq -\frac{\ln(2)}{2\pi}\ln\left(-(b(\chi,\tau)-a(\chi,\tau))^2h''_-(a(\chi,\tau);\chi,\tau)\right)\\
&\qquad\qquad\qquad\qquad{}-\frac{\ln(2)^2}{2\pi}-\frac{1}{4}\pi+\arg\left(\Gamma\left(\frac{\ii\ln(2)}{2\pi}\right)\right),\\
\eta_b^{[\shelves]}(\chi,\tau)&\defeq \frac{\ln(2)}{2\pi}\ln\left((b(\chi,\tau)-a(\chi,\tau))^2h''(b(\chi,\tau);\chi,\tau)\right)\\
&\qquad\qquad\qquad\qquad{}+\frac{\ln(2)^2}{2\pi}+\frac{1}{4}\pi-\arg\left(\Gamma\left(\frac{\ii\ln(2)}{2\pi}\right)\right).
\end{split}
\end{equation}
We also define additional real phases by
\begin{equation}
\begin{split}
\delta_a(\chi,\tau)&\defeq \pi-2(K_a(\chi,\tau)+\mu(\chi,\tau)),\\
\delta_b(\chi,\tau)&\defeq -2(K_b(\chi,\tau)+\mu(\chi,\tau)).
\end{split}
\end{equation}
By analogy with \eqref{eq:channels-amplitudes} define positive amplitudes by
\begin{equation}
\begin{split}
F_a^{[\shelves]}(\chi,\tau)&\defeq \sqrt{-\frac{\ln(2)}{\pi h''_-(a(\chi,\tau);\chi,\tau)}},\\
F_b^{[\shelves]}(\chi,\tau)&\defeq \sqrt{\frac{\ln(2)}{\pi h''(b(\chi,\tau);\chi,\tau)}}.
\end{split}
\end{equation}
Finally, define four positive modulation factors with range $[0,1]$ by
\begin{equation}
\begin{split}
m_a^\pm(\chi,\tau)&\defeq\tfrac{1}{2}\left(1\pm\cos\left(\arg\left(a(\chi,\tau)-\lambda_0(\chi,\tau)\right)\right)\right),\\
m_b^\pm(\chi,\tau)&\defeq\tfrac{1}{2}\left(1\pm\cos\left(\arg\left(b(\chi,\tau)-\lambda_0(\chi,\tau)\right)\right)\right).
\end{split}
\label{eq:m-a-b-shelves}
\end{equation}
Our main result for the region $\shelves$ is then the following.
\begin{theorem}[Far-field asymptotics of $q(x,t;\mathbf{Q}^{-s},M)$ for $(\chi,\tau)\in\shelves$]
Let $s=\pm 1$ be arbitrary.  Then, as $M\to+\infty$ through an arbitrary sequence of values,
$q(M\chi,M\tau;\mathbf{Q}^{-s};M)=\mathfrak{L}_s^{[\shelves]}(\chi,\tau;M)+\mathfrak{S}_s^{[\shelves]}(\chi,\tau;M) + O(M^{-1})$, where
\begin{equation}
\mathfrak{L}_s^{[\shelves]}(\chi,\tau;M)\defeq B(\chi,\tau)\ee^{-2\ii (M\kappa(\chi,\tau)+\mu(\chi,\tau)+\frac{1}{4}s\pi)},
\label{eq:leading-term-shelves-q}
\end{equation}
and 
\begin{multline}
\mathfrak{S}_s^{[\shelves]}(\chi,\tau;M)\defeq sM^{-\frac{1}{2}}\Big[m_a^+(\chi,\tau)\ee^{\ii\delta_a(\chi,\tau)}F_a^{[\shelves]}(\chi,\tau)\ee^{\ii\Theta_a^{[\shelves]}(\chi,\tau;M)} \\
{}+ m_b^+(\chi,\tau)\ee^{\ii\delta_b(\chi,\tau)}F_b^{[\shelves]}(\chi,\tau)\ee^{\ii\Theta_b^{[\shelves]}(\chi,\tau;M)} \\
{}- m_a^-(\chi,\tau)\ee^{-\ii\delta_a(\chi,\tau)}F_a^{[\shelves]}(\chi,\tau)\ee^{-\ii[\Theta_a^{[\shelves]}(\chi,\tau;M)+4M\kappa(\chi,\tau)+4\mu(\chi,\tau)]} \\
{}- m_b^-(\chi,\tau)\ee^{-\ii\delta_b(\chi,\tau)}F_b^{[\shelves]}(\chi,\tau)\ee^{-\ii[\Theta_b^{[\shelves]}(\chi,\tau;M)+4M\kappa(\chi,\tau)+4\mu(\chi,\tau)]}\Big],
\label{eq:subleading-term-shelves-q}
\end{multline}
and where the error term is uniform for $(\chi,\tau)$ in any compact subset of $\shelves$.
\label{thm:shelves}
\end{theorem}
This result therefore provides both a leading term $\mathfrak{L}_s^{[\shelves]}(\chi,\tau;M)$ (which in the case of high-order multiple-pole solitons with $M\in\mathbb{Z}_{>0}$ was obtained in \cite{BilmanBW19}) and a sub-leading term $\mathfrak{S}_s^{[\shelves]}(\chi,\tau;M)$.  As with Theorem~\ref{thm:channels}, this result applies to the full family of solutions including both solitons and rogue waves, and hence $\shelves$ with its reflections in the coordinate axes forms the remaining components of the region bounded by the yellow curves in Figure~\ref{fig:roguewaves-and-solitons}.  To write the formula in the rogue wave case requires just cosmetic modification; the analogue of Corollary~\ref{cor:rogue-wave-channels} when $(\chi,\tau)\in\shelves$ is the following.
\begin{corollary}
The fundamental rogue wave of order $k\in\mathbb{Z}_{>0}$ satisfies $\psi_k(M\chi,M\tau)=L_k^{[\shelves]}(\chi,\tau)+S_k^{[\shelves]}(\chi,\tau)+O(k^{-1})$, where
\begin{equation}
L_k^{[\shelves]}(\chi,\tau)\defeq\ee^{-\ii M\tau}\mathfrak{L}_s^{[\shelves]}(\chi,\tau;M),\quad
S_k^{[\shelves]}(\chi,\tau)\defeq\ee^{-\ii M\tau}\mathfrak{S}_s^{[\shelves]}(\chi,\tau;M),\quad
s=(-1)^k,\quad M=\tfrac{1}{2}k+\tfrac{1}{4},
\label{eq:rw-terms-shelves}
\end{equation}
in which $\mathfrak{L}_s^{[\shelves]}(\chi,\tau;M)$ and $\mathfrak{S}_s^{[\shelves]}(\chi,\tau;M)$ are given by \eqref{eq:leading-term-shelves-q} and \eqref{eq:subleading-term-shelves-q} respectively, and
where the error term is uniform for $(\chi,\tau)$ in any compact subset of $\shelves$.
\label{cor:rogue-wave-shelves}
\end{corollary}


The first correction on the domain $\shelves$ resolves the obvious oscillations visible in plots of high-order multiple-pole soliton solutions \cite{BilmanBW19} and in plots of high-order fundamental rogue waves such as those shown in Figure~\ref{fig:2D-Plots}.  On two-dimensional plots such as these, one observes that these fluctuations form a highly-regular interference pattern.  To see how Theorem~\ref{thm:shelves} yields such a pattern, we can rewrite the combination $\mathfrak{L}_s^{[\shelves]}(\chi,\tau;M)+\mathfrak{S}_s^{[\shelves]}(\chi,\tau;M)$ in a different form by factoring out a phase factor, which has the effect of producing some symmetry in the four phases present in \eqref{eq:subleading-term-shelves-q}.  Therefore, using $s=\pm 1$, we write:
\begin{multline}
\mathfrak{L}_s^{[\shelves]}(\chi,\tau;M)+\mathfrak{S}_s^{[\shelves]}(\chi,\tau;M) = s\ee^{-2\ii\phi(\chi,\tau;M)}\Big[-\ii B(\chi,\tau)\\
{}+M^{-\frac{1}{2}}\Big(m_a^+(\chi,\tau)F_a^{[\shelves]}(\chi,\tau)\ee^{\ii\phi_a(\chi,\tau;M)} -
m_a^-(\chi,\tau)F_a^{[\shelves]}(\chi,\tau)\ee^{-\ii\phi_a(\chi,\tau;M)} \\
{}+
m_b^+(\chi,\tau)F_b^{[\shelves]}(\chi,\tau)\ee^{\ii\phi_b(\chi,\tau;M)} -
m_b^-(\chi,\tau)F_b^{[\shelves]}(\chi,\tau)\ee^{-\ii\phi_b(\chi,\tau;M)}\Big)
\Big],
\label{eq:leading-and-subleading-shelves-rewritten}
\end{multline}
in which
\begin{equation}
\begin{split}
\phi(\chi,\tau;M)&\defeq M\kappa(\chi,\tau)+\mu(\chi,\tau),\\
\phi_a(\chi,\tau;M)&\defeq \Theta^{[\shelves]}_a(\chi,\tau;M)+2M\kappa(\chi,\tau) + \delta_a(\chi,\tau) + 2\mu(\chi,\tau),\\
\phi_b(\chi,\tau;M)&\defeq \Theta^{[\shelves]}_b(\chi,\tau;M)+2M\kappa(\chi,\tau) + \delta_b(\chi,\tau) + 2\mu(\chi,\tau).
\end{split}
\label{eq:symmetrical-phases}
\end{equation}
Using the fact that $m_a^+(\chi,\tau)+m_a^-(\chi,\tau)=m_b^+(\chi,\tau)+m_b^-(\chi,\tau)=1$ to expand the square modulus of the right-hand side of \eqref{eq:leading-and-subleading-shelves-rewritten} through terms proportional to $M^{-\frac{1}{2}}$, and combining with Theorem~\ref{thm:shelves} then gives the following.
\begin{corollary}
Let $s=\pm 1$ be arbitrary.  Then as $M\to+\infty$ through an arbitrary sequence of values, 
\begin{multline}
|q(M\chi,M\tau;\mathbf{Q}^{-s},M)|^2 = 
B(\chi,\tau)^2\\
{}-2M^{-\frac{1}{2}}B(\chi,\tau)\left[F_a^{[\shelves]}(\chi,\tau)\sin(\phi_a(\chi,\tau;M))+F_b^{[\shelves]}(\chi,\tau)\sin(\phi_b(\chi,\tau;M))\right]+O(M^{-1}),
\label{eq:interference}
\end{multline}
where the error is uniform for $(\chi,\tau)$ in compact subsets of $\shelves$.
\label{cor:pattern}
\end{corollary}
Since $|q|^2=|\psi_k|^2$ when $s=(-1)^k$ and $M=\tfrac{1}{2}k+\tfrac{1}{4}$, this result explains the interference pattern seen in amplitude plots of high-order fundamental rogue waves such as in \cite[Figure 2]{BilmanLM20} and in Figure~\ref{fig:2D-Plots} of this paper.  However, as it is valid for arbitrary $M\to+\infty$, the same formula also explains the similar patterns observed in plots of $k^\mathrm{th}$-order pole solitons for $M=\tfrac{1}{2}k$ large such as can be found in \cite{BilmanB19,BilmanBW19}.  It is equally valid for all other increasing sequences of $M$-values that do not correspond to either type of solution.
Corollary~\ref{cor:pattern} shows that for $(\chi,\tau)\in\shelves$, the squared modulus of $q(M\chi, M \tau; \mathbf{Q}^{-s}, M)$ consists of a slowly-varying ``shelf'' of size $O(1)$ and a rapidly-varying perturbation proportional to $M^{-\frac{1}{2}}$. To leading order, this perturbation is a superposition of two sine functions with different phases $\phi_a(\chi,\tau;M)$ and $\phi_b(\chi,\tau;M)$ whose derivatives are large for $M\gg 1$ due to the presence of the terms $2M[\kappa(\chi,\tau)-h_-(a(\chi,\tau);\chi,\tau)]$ and $2M[\kappa(\chi,\tau)-h(b(\chi,\tau);\chi,\tau)]$ respectively, see \eqref{eq:Thetas-shelves}, \eqref{eq:Phis-shelves}, and \eqref{eq:symmetrical-phases}.  
Since $F_a^{[\shelves]}(\chi,\tau)$ and $F_b^{[\shelves]}(\chi,\tau)$ are both positive, the two terms proportional to $M^{-\frac{1}{2}}$ in \eqref{eq:interference} are individually maximized when $\phi_a(\chi,\tau;M)\in (-\tfrac{1}{2}+2\mathbb{Z})\pi$ and where $\phi_b(\chi,\tau;M)\in (-\tfrac{1}{2}+2\mathbb{Z})\pi$, each condition of which produces a ($M$-dependent) system of curves that can be plotted over the region $\shelves$ in the $(\chi,\tau)$-plane.  Provided that $\nabla\phi_a(\chi,\tau;M)$ and $\nabla\phi_b(\chi,\tau;M)$ are linearly-independent vectors near a given point $(\chi,\tau)\in\shelves$, the two systems of maximizing curves will intersect each other transversely and there will be isolated local maxima of $|q|^2$ that form a locally-regular parallelogram lattice of increasing density as $M\to+\infty$.  
When $M>0$ is large, the gradient vectors of the phases $\phi_a(\chi,\tau;M)$ and $\phi_b(\chi,\tau;M)$ are dominated by the terms proportional to $M$.  Then, since in the limit that $(\chi,\tau)$ approaches the common boundary of $\shelves$ and $\exterior$ the real critical points $a(\chi,\tau)<b(\chi,\tau)$ coalesce, one can see that these leading terms coincide at the boundary curve, implying that the systems of maximizing curves nearly coincide at this boundary of $\shelves$.  Therefore, in this limit, the lattice of local maxima degenerates into a pattern of stripes instead, such as can be seen along the blue curves in Figure~\ref{fig:2D-Plots}.  For high-order multiple-pole solitons, the stripes in the square modulus $|q|^2$ grow as $(\chi,\tau)$ exits $\shelves$ and form a stripe pattern of $O(1)$ size that are modeled by an elliptic function in the ``oscillatory region'' that is a proper subset of $\exterior$ abutting $\shelves$ \cite{BilmanBW19}; for high-order fundamental rogue waves the stripes instead decay away as $(\chi,\tau)$ exits $\shelves$, leaving only the slowly-varying background amplitude $B(\chi,\tau)>0$ as described on the whole of $\exterior$ by Theorem~\ref{thm:exterior}. See also Figure~\ref{fig:roguewaves-and-solitons}.

As a final corollary of Theorem~\ref{thm:shelves}, we present a space-time localized asymptotic formula for $q(M\chi,M\tau;\mathbf{Q}^{-s},M)$.
\begin{corollary}
Let $s=\pm 1$ be arbitrary, and fix $(\chi_0,\tau_0)\in\shelves$.  Then, as $M\to+\infty$ through an arbitrary sequence of values,
\begin{equation}
q(M\chi_0+\Delta x,M\tau_0+\Delta t;\mathbf{Q}^{-s},M)=Q(\Delta x,\Delta t)\left(1+M^{-\frac{1}{2}}\left(p_a(\Delta x,\Delta t)+p_b(\Delta x,\Delta t)\right)\right)+O(M^{-1})
\label{eq:Q-perturbation-shelves}
\end{equation}
holds uniformly for bounded $(\Delta x,\Delta t)$, where 
\begin{equation}
Q(\Delta x,\Delta t)\defeq \mathcal{A}\ee^{\ii(\xi_0\Delta x-\Omega_0\Delta t)},\quad \mathcal{A}\defeq -\ii s\ee^{-2\ii\phi(\chi_0,\tau_0;M)}B(\chi_0,\tau_0),
\label{eq:leading-plane-wave-intro}
\end{equation}
and
\begin{equation}
\begin{split}
p_a(\Delta x,\Delta t)&\defeq \ii\frac{F_a^{[\shelves]}(\chi_0,\tau_0)}{B(\chi_0,\tau_0)}\Big[
m_a^+(\chi_0,\tau_0)\ee^{\ii\phi_a(\chi_0,\tau_0;M)}\ee^{\ii(\xi_a\Delta x-\Omega_a\Delta t)} \\
&\qquad\qquad\qquad {}- 
m_a^-(\chi_0,\tau_0)\ee^{-\ii\phi_a(\chi_0,\tau_0;M)}\ee^{-\ii(\xi_a\Delta x-\Omega_a\Delta t)}\Big],\\
p_b(\Delta x,\Delta t)&\defeq \ii\frac{F_b^{[\shelves]}(\chi_0,\tau_0)}{B(\chi_0,\tau_0)}\Big[
m_b^+(\chi_0,\tau_0)\ee^{\ii\phi_b(\chi_0,\tau_0;M)}\ee^{\ii(\xi_b\Delta x-\Omega_b\Delta t)} \\
&\qquad\qquad\qquad {}- 
m_b^-(\chi_0,\tau_0)\ee^{-\ii\phi_b(\chi_0,\tau_0;M)}\ee^{-\ii(\xi_b\Delta x-\Omega_b\Delta t)}\Big],
\end{split}
\label{eq:p-a-b-shelves}
\end{equation}
in which real local wavenumbers are defined by
\begin{equation}
\begin{split}
\xi_0&\defeq -2\kappa_\chi(\chi_0,\tau_0),\\
\xi_a&\defeq 2(\kappa_\chi(\chi_0,\tau_0)-h_{\chi-}(a(\chi_0,\tau_0);\chi_0,\tau_0)),\\
\xi_b&\defeq 2(\kappa_\chi(\chi_0,\tau_0)-h_{\chi}(b(\chi_0,\tau_0);\chi_0,\tau_0)),
\end{split}
\label{eq:wavenumbers-intro}
\end{equation}
and real local frequencies are defined by
\begin{equation}
\begin{split}
\Omega_0&\defeq 2\kappa_\tau(\chi_0,\tau_0),\\
\Omega_a&\defeq -2(\kappa_\tau(\chi_0,\tau_0)-h_{\tau-}(a(\chi_0,\tau_0);\chi_0,\tau_0)),\\
\Omega_b&\defeq -2(\kappa_\tau(\chi_0,\tau_0)-h_{\tau}(b(\chi_0,\tau_0);\chi_0,\tau_0)).
\end{split}
\label{eq:frequencies-intro}
\end{equation}
Moreover, $Q(\Delta x,\Delta t)$ is a plane-wave solution of the focusing nonlinear Schr\"odinger equation in the form
\begin{equation}
\ii Q_{\Delta t} + \tfrac{1}{2} Q_{\Delta x\Delta x} + |Q|^2Q=0,
\label{eq:NLS-Deltas}
\end{equation}
and both $p_a(\Delta x,\Delta t)$ and $p_b(\Delta x,\Delta t)$ are particular plane-wave solutions of the formal linearization of \eqref{eq:NLS-Deltas} about $Q(\Delta x,\Delta t)$ written in the frame rotating with the phase of that solution:
\begin{equation}
\ii p_{\Delta t} + \ii\xi_0p_{\Delta x} + \tfrac{1}{2}p_{\Delta x\Delta x} +|\mathcal{A}|^2(p+p^*)=0.
\label{eq:linearization-intro}
\end{equation}
The relative wavenumbers $\xi_a$ and $\xi_b$ also satisfy the inequalities
\begin{equation}
\xi_a^2>4|\mathcal{A}|^2\quad\text{and}\quad \xi_b^2>4|\mathcal{A}|^2.
\label{eq:relative-wavenumber-inequalities}
\end{equation}
\label{cor:shelves-local}
\end{corollary}
Note that in defining the local wavenumbers and frequencies, it makes no difference whether one first evaluates $h(\lambda;\chi,\tau)$ at $\lambda=b(\chi,\tau)$ or $h_-(\lambda;\chi,\tau)$ at $\lambda=a(\chi,\tau)\in\Sigma_g$ and then differentiates with respect to $\chi$ or $\tau$, or the other way around.  This is because $\lambda=a(\chi,\tau)$ and $\lambda=b(\chi,\tau)$ are the roots of the quadratic factor in the  numerator of \eqref{eq:hprime-formula}.

The well-known theory of plane-wave solutions of the focusing nonlinear Schr\"odinger equation 
of arbitrary amplitude $|\mathcal{A}|$ and the formal linearized theory of their perturbations is briefly summarized in Appendix~\ref{A:perturbations}.  A key result of that theory is the existence of an unstable band of relative wavenumbers $\xi$ given by the inequality $\xi^2\le 4|\mathcal{A}|^2$.  It follows from \eqref{eq:relative-wavenumber-inequalities} that the solutions $p_a(\Delta x,\Delta t)$ and $p_b(\Delta x,\Delta t)$ are linearly stable perturbations of the underlying plane wave $Q(\Delta x,\Delta t)$.

The proof of Corollary~\ref{cor:shelves-local} is given in Section~\ref{sec:wave-theoretic-interpretation} below.

\subsubsection{Relations between asymptotic formul\ae\ for $q(M\chi,M\tau;\mathbf{Q}^{-s},M)$ on $\channels$ and $\shelves$}
The asymptotic description of $q(M\chi,M\tau;\mathbf{Q}^{-s},M)$ when $(\chi,\tau)\in\shelves$ given in Theorem~\ref{thm:shelves} is substantially more complicated than for $(\chi,\tau)\in\channels$ (cf., Theorem~\ref{thm:channels}). 
However, comparing \eqref{eq:leading-term-channels-q} and \eqref{eq:subleading-term-shelves-q}, one notices that the part of $\mathfrak{S}_s^{[\shelves]}(\chi,\tau;M)$ written on the first two lines of \eqref{eq:subleading-term-shelves-q} bears a striking resemblance to the leading term $\mathfrak{L}_s^{[\channels]}(\chi,\tau;M)$ valid on the other side of the $\shelves$--$\channels$ boundary curve.  Indeed, $h(\lambda;\chi,\tau)$ degenerates at this curve into the unmodified phase $\vartheta(\lambda;\chi,\tau)$, making the indicated terms match except for the slowly-varying complex factors $m_{a,b}^+(\chi,\tau)\ee^{\ii\delta_{a,b}(\chi,\tau)}$ present within $\shelves$.  Approaching this same curve from $\shelves$, $B(\chi,\tau)\to 0$, so it is also true that the leading term $\mathfrak{L}_s^{[\shelves]}(\chi,\tau;M)$ vanishes in the limit.  However, it is difficult to compare the two asymptotic formul\ae\ quantitatively near the $\shelves$--$\channels$ boundary because $\vartheta''(a(\chi,\tau);\chi,\tau)$ and $h''(a(\chi,\tau);\chi,\tau)$ both vanish as the boundary curve is approached from $\channels$ and from $\shelves$, respectively (we also note that $(\chi,\tau)\mapsto a(\chi,\tau)$ denotes two different real-analytic functions on $\channels$ and $\shelves$ that happen to agree along the common boundary curve).  This makes one of the terms in $\mathfrak{L}^{[\channels]}_s(\chi,\tau;M)$ and two of the terms in $\mathfrak{S}^{[\shelves]}_s(\chi,\tau;M)$ blow up at the boundary curve.  Of course, neither Theorem~\ref{thm:channels} nor Theorem~\ref{thm:shelves} accurately describes $q(M\chi,M\tau;\mathbf{Q}^{-s},M)$ near this curve, so this blow up merely signals the need for further double-scaling asymptotic analysis to resolve the wave field in its vicinity.

\subsubsection{Relations between the asymptotic formula for $\psi_k(M\chi,M\tau)$ on $\exterior$ with those valid on $\channels$ and $\shelves$}
To discuss the region $\exterior$ in light of Theorem~\ref{thm:exterior}, we need to restrict attention to the fundamental rogue-wave solutions $\psi_k(M\chi,M\tau)$ where $M=\tfrac{1}{2}k+\tfrac{1}{4}$.  As the region $\exterior$ abuts both $\channels$ and $\shelves$, it is interesting and useful to compare asymptotic formul\ae\ for $\psi_k(M\chi,M\tau)$ valid on all three regions.  

The simplest observation is that since $B(\chi,\tau)\downarrow 0$ as $(\chi,\tau)$ approaches $\channels$ from anywhere in the exterior,  in particular from $\exterior$, $L^{[\exterior]}_k(\chi,\tau)\to 0$ also in this limit.  This fact is consistent with the fact that $L^{[\channels]}_k(\chi,\tau)$ is small of order $k^{-\frac{1}{2}}$.  However, we note that neither Corollary~\ref{cor:rogue-wave-channels} nor Theorem~\ref{thm:exterior} is valid on a neighborhood of any common boundary point of $\channels$ and $\exterior$.  Like the problem of studying $q(M\chi,M\tau;\mathbf{Q}^{-s},M)$ near the common boundary of $\channels$ and $\shelves$, some new phenomena may be uncovered by a suitable double-scaling analysis to zoom in on points on the curve separating $\exterior$ from $\channels$.  

We can give a more quantitative comparison between the asymptotic formul\ae\ for $\psi_k(M\chi,M\tau)$ on the domains $\exterior$ and $\shelves$.  First, note that the integral in \eqref{eq:gamma-formula-intro} originally defined for $(\chi,\tau)\in\exterior$ admits continuation to $(\mathbb{R}_{\ge 0}\times\mathbb{R}_{\ge 0})\setminus\overline{\channels}$ as a real analytic function, and the latter domain contains also $\shelves$.  Thus for $(\chi,\tau)\in\shelves$, $\gamma(\chi,\tau)$ and $\kappa(\chi,\tau)$ given by \eqref{eq:kappa-formula} can be compared.  Indeed, deforming the integration path $\Sigma_\mathrm{c}$ in \eqref{eq:kappa-formula} leftward to lie partly along the right edge of the branch cut $\Sigma_g$ by its upward orientation, one can replace the resulting integral along $\Sigma_g$ of $1/R_-(\lambda;\chi,\tau)$ by half of the integral of $1/R(\lambda;\chi,\tau)$ over a positively oriented loop enclosing $\Sigma_g$.  Evaluating the latter integral by residues using $R(\lambda;\chi,\tau)=\lambda + O(1)$ as $\lambda\to\infty$ and comparing with \eqref{eq:gamma-formula-intro} one obtains the following identity:
\begin{equation}
\kappa(\chi,\tau)=\gamma(\chi,\tau)-\pi,
\quad (\chi,\tau)\in\shelves.
\label{eq:kappa-gamma}
\end{equation}
We then have the following, which uses the fact that $M=\tfrac{1}{2}k+\tfrac{1}{4}$ and $s=(-1)^k$ for the fundamental rogue wave of order $k$.
\begin{corollary}
The phase and amplitude of the leading term $L_k^{[\exterior]}(\chi,\tau)$ admit real analytic continuation from $\exterior$ into $\shelves$, in which the following identity holds:
\begin{equation}
L_k^{[\shelves]}(\chi,\tau)=\ee^{-2\ii\mu(\chi,\tau)}L_k^{[\exterior]}(\chi,\tau),\quad (\chi,\tau)\in\shelves.
\end{equation}
Therefore, for fundamental rogue waves of high order $k$, the leading terms agree for $(\chi,\tau)\in\exterior$ and for $(\chi,\tau)\in\shelves$, up to a phase $-2\mu(\chi,\tau)$ that vanishes as the common boundary is approached from $\shelves$.
\label{cor:leading-term-match-shelves-exterior}
\end{corollary}
The amplitude $|\psi_k(M\chi,M\tau)|$ is compared with that of the common leading term, namely $B(\chi,\tau)$, on the exterior of $\channels$ in Figure~\ref{fig:AbsLeadingTermShelves}.  
\begin{figure}[h]
\begin{center}
\phantom{!}\hfill\includegraphics[width=0.4\linewidth]{AltScaling-amplitude-regions-plot-orderk-32.pdf}\hfill%
\includegraphics[width=0.4\linewidth]{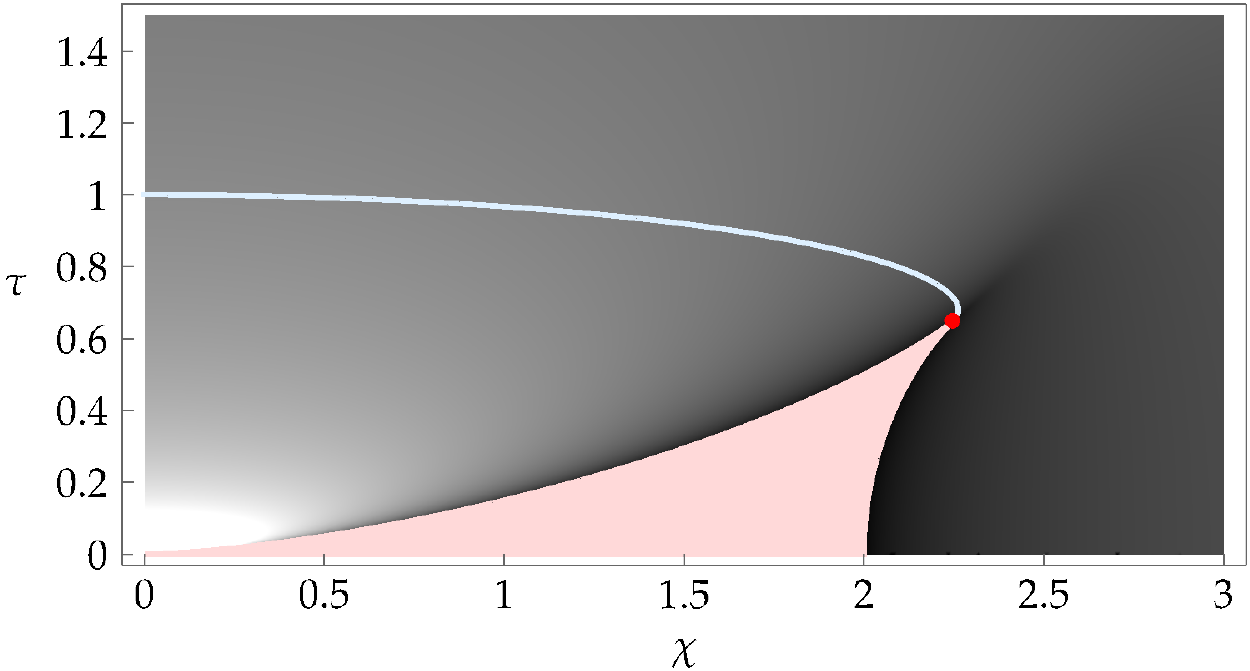}\hfill\phantom{!}%
\end{center}
\caption{Left:  same as the right-hand panel of Figure~\ref{fig:2D-Plots}.  Right:  amplitude $B(\chi,\tau)>0$ of the leading term as a function of $(\chi,\tau)\in\shelves\cup\exterior$.  Both plots employ the same colormap.}
\label{fig:AbsLeadingTermShelves}
\end{figure}

In terms of derivatives with respect to $(\chi,\tau)$ rather than $(x,t)$, the focusing nonlinear Schr\"odinger equation \eqref{eq:NLS-ZBC} satisfied by $q$ takes the rescaled ``semiclassical'' form:
\begin{equation}
\ii\epsilon q_\tau +\tfrac{1}{2}\epsilon^2q_{\chi\chi} + |q|^2q=0,\quad\epsilon=\frac{1}{M}\ll 1.
\label{eq:NLS-semiclassical}
\end{equation}
To study this equation for small $\epsilon$ it is convenient to introduce in place of $q$ Madelung's real variables $\rho$ and $U$ given by
\begin{equation}
\rho(\chi,\tau)\defeq |q|^2\quad\text{and}\quad U(\chi,\tau)\defeq \epsilon\mathrm{Im}\left(\frac{q_\chi}{q}\right).
\label{eq:Madelung-transform}
\end{equation}
Then, without approximation \eqref{eq:NLS-semiclassical} can be written in the form
\begin{equation}
\rho_\tau+(\rho U)_\chi=0\quad\text{and}\quad U_\tau + \left(\tfrac{1}{2}U^2-\rho\right)_\chi = \tfrac{1}{2}\epsilon^2 F[\rho]_{\chi\chi},\quad\text{where}\quad
F[\rho]\defeq\frac{\rho_{\chi\chi}}{2\rho}-\left(\frac{\rho_\chi}{2\rho}\right)^2.
\label{eq:Madelung-exact}
\end{equation}
Assuming that $\rho\neq 0$, it may appear reasonable to neglect the formally small term in \eqref{eq:Madelung-exact} and hence obtain the approximating system
\begin{equation}
\rho_\tau+(\rho U)_\chi=0\quad\text{and}\quad U_\tau + \left(\tfrac{1}{2}U^2-\rho\right)_\chi = 0.
\label{eq:dispersionless-NLS}
\end{equation}
This is an elliptic quasilinear system on $\rho$ and $U$ known as the (focusing) \emph{dispersionless nonlinear Schr\"odinger system}.  Now, observing that $\rho$ and $U$ defined by \eqref{eq:Madelung-transform} are invariant under $q\mapsto \psi\defeq \ee^{-\ii\tau/\epsilon}q$, we may apply these definitions to the leading terms $L_k^{[\exterior]}(\chi,\tau)$ and $L_k^{[\shelves]}(\chi,\tau)$ of $\psi_k$ on $\exterior$ and $\shelves$, respectively.  Up to a correction term in $U$ proportional to  $\epsilon$ that is only present when $(\chi,\tau)\in\shelves$ (originating from the phase correction $-2\mu(\chi,\tau)$), the formul\ae\ in both regions read
\begin{equation}
\rho(\chi,\tau)=B(\chi,\tau)^2\quad\text{and}\quad U(\chi,\tau)=-2\gamma_\chi(\chi,\tau).
\label{eq:Madelung-transform-on-leading-term}
\end{equation}
\begin{corollary}
The expressions \eqref{eq:Madelung-transform-on-leading-term} satisfy the dispersionless nonlinear Schr\"odinger system \eqref{eq:dispersionless-NLS} for $(\chi,\tau)\in(\mathbb{R}_{\ge 0}\times\mathbb{R}_{\ge 0})\setminus\overline{\channels}$, i.e., for $(\chi,\tau)\in\exterior$ or $(\chi,\tau)\in\shelves$, or on the common boundary curve.
\label{cor:Whitham}
\end{corollary}
Note that the elliptic nature of \eqref{eq:dispersionless-NLS} is consistent with the real analyticity of $A(\chi,\tau)$ and $B(\chi,\tau)$ (and, via \eqref{eq:gamma-formula-intro}, $\gamma(\chi,\tau)$).  The dispersionless nonlinear Schr\"odinger system \eqref{eq:dispersionless-NLS} is also sometimes called the genus-zero Whitham modulation system.  In the Whitham modulation theory it arises from an ansatz of a solution of \eqref{eq:NLS-semiclassical} in the form of a modulated plane wave.  The proof of Corollary~\ref{cor:Whitham} relies on the identity $\gamma_\chi(\chi,\tau)=\kappa_\chi(\chi,\tau)=A(\chi,\tau)$ that is established in Lemma~\ref{lemma:g-derivatives}, and is given in Remark~\ref{rem:Whitham} below.

\begin{remark}
A version of Corollary~\ref{cor:Whitham} holds more generally for $\rho$ and $U$ extracted from the leading term $\mathfrak{L}_s^{[\shelves]}(\chi,\tau;M)$ valid as an approximation for $q(M\chi,M\tau;\mathbf{Q}^{-s},M)$ for general large $M$ when $(\chi,\tau)\in\shelves$.  One need only replace $\gamma_\chi(\chi,\tau)$ with $\kappa_\chi(\chi,\tau)$ in the definition \eqref{eq:Madelung-transform-on-leading-term}.  As such it also holds for the high-order multiple-pole soliton solutions studied in \cite{BilmanBW19}.
\end{remark}

\subsection{The behavior of $q(M\chi,M\tau;\mathbf{Q}^{-s},M)$ for $(\chi,\tau)\in\exterior$ for general large $M$}
\label{sec:general}
It is shown in \cite{BilmanBW19} that in the unbounded domain $\exterior$, high-order multiple-pole solitons behave quite differently from high-order fundamental rogue waves as reported in Theorem~\ref{thm:exterior}.  For the soliton solutions, the domain $\exterior$ is divided into two components by the curve $\ell_\mathrm{sol}$ described by \eqref{eq:DegenerateBoutroux} and shown with a dotted black line in Figure~\ref{fig:RegionsPlot}.  On the component adjacent to the positive $\chi$-axis, the solution is exponentially small when $M=\tfrac{1}{2}k\to+\infty$.  This is consistent with the exponential decay of $q(x,t;\mathbf{Q}^{-s},\tfrac{1}{2}k)$ as $x\to \pm\infty$ for fixed $t$, although for technical reasons the proof given in \cite{BilmanBW19} does not allow $(\chi,\tau)$ to become unbounded.  On the complementary component adjacent to the positive $\tau$-axis, the solution behaves completely differently.  Here $q(\tfrac{1}{2}k\chi,\tfrac{1}{2}k\tau;\mathbf{Q}^{-s},\tfrac{1}{2}k)$ is approximated by a modulated elliptic function of amplitude asymptotically independent of $M=\tfrac{1}{2}k$.  The elliptic modulus $m=m(\chi,\tau)$ approaches $m=0$ as $(\chi,\tau)$ approaches the common boundary with $\shelves$, while it approaches $m=1$ instead as $(\chi,\tau)$ approaches the curve $\ell_\mathrm{sol}$.  In the former limit the elliptic wave degenerates onto the trigonometric plane-wave leading term given in all cases of $M\to\infty$ in Theorem~\ref{thm:shelves}, and in the latter limit the elliptic wave degenerates into a train of isolated solitons (which explains our notation $\ell_\mathrm{sol}$).   

When we consider solutions $q(x,t;\mathbf{Q}^{-s},M)$ that do not fit into either family, we see both common features and substantial differences comparing with the special cases of solitons and rogue waves.  The way to take the limit $M\to\infty$ in this situation is to represent $M$ in modular form as $M=\tfrac{1}{2}k+r$ with quotient $k\in\mathbb{Z}_{\ge 0}$ and remainder $0\le r<\tfrac{1}{2}$.  Then we fix the remainder and let $k\to+\infty$.  Of course the soliton case is $r=0$, and the rogue wave case is $r=\tfrac{1}{4}$.  When $r\neq 0$ and $r\neq\tfrac{1}{4}$, the large-$M$ asymptotic behavior of $q(M\chi,M\tau;\mathbf{Q}^{-s},M)$ for bounded $(\chi,\tau)$ in $\exterior$ depends on whether $(\chi,\tau)\in\exterior_\chi$ or $(\chi,\tau)\in\exterior_\tau$ (see Figure~\ref{fig:RegionsPlot}).  Because we think it will be interesting for the reader, in the following paragraphs we describe what we have learned about these solutions; however full details and proofs will be given in a subsequent paper devoted to the case of general $M\ge 0$.

If $(\chi,\tau)\in\exterior_\chi$, then a version of Theorem~\ref{thm:exterior} applies to $\ee^{-\ii M\tau}q(M\chi,M\tau;\mathbf{Q}^{-s},M)$ for $r\neq 0$ and $r\neq\tfrac{1}{4}$, in which the leading term is multiplied by an $M$-independent but $(\chi,\tau)$-dependent phase factor, and in which the error term is larger, of size $O(M^{-\frac{1}{2}})$.  The sub-leading term proportional to $M^{-\frac{1}{2}}$ is simpler than for $(\chi,\tau)\in\shelves$, consisting of only one of the two waves present for instance in \eqref{eq:p-a-b-shelves}; this means that the amplitude fluctuations will form a stripe pattern rather than an interference pattern such as occurs in $\shelves$.  The approximation of $\ee^{-\ii M\tau}q(M\chi,M\tau;\mathbf{Q}^{-s},M)$ tends to the background plane wave $\psi=\pm 1$ as $(\chi,\tau)\to\infty$ in $\exterior_\chi$, which is consistent with the exact boundary conditions satisfied by the fundamental rogue-wave solutions occurring for $r=\tfrac{1}{4}$; however the proof we have in mind of this result is not valid for technical reasons when $(\chi,\tau)$ become unbounded.  Nonetheless, it follows from a different proof that, like the rogue-wave solutions, all solutions $q(x,t;\mathbf{Q}^{-s},M)$ for remainder $r\neq 0$ satisfy nonzero boundary conditions with unit limiting amplitude as $x\to\pm\infty$; however for $r\neq\tfrac{1}{4}$ the decay is so slow that the difference between $q$ and the background does not even lie in $L^2(\mathbb{R})$.

On the other hand, if $(\chi,\tau)\in\exterior_\tau$, then as $M\to\infty$ with $r\neq 0$ and $r\neq\tfrac{1}{4}$ fixed, $q(M\chi,M\tau;\mathbf{Q}^{-s},M)$ is approximated by a modulated elliptic function of amplitude neither small nor large.  In the part of $\exterior_\tau$ above (i.e., for larger $\tau$) the curve $\ell_\mathrm{sol}$, the leading term of the approximation differs from that valid in the same region for the multiple-pole soliton case of $r=0$ only in phase modifications that are independent of $M\gg 1$.  However the error term is of order $O(M^{-\frac{1}{2}})$ rather than $O(M^{-1})$.  In the part of $\exterior_\tau$ lying below the curve $\ell_\mathrm{sol}$, the solution $q(M\chi,M\tau;\mathbf{Q}^{-s},M)$ evidently behaves neither like the rogue-wave solutions for $r=\tfrac{1}{4}$ (approximated by modulated plane waves) nor like the multiple-pole soliton solutions for $r=0$ (exponentially small).  The elliptic modulus varies with $(\chi,\tau)$ from $m=1$ on the curve $\ell_\mathrm{sol}$ to $m=0$ on the curve $\partial\exterior_\chi\cap\partial\exterior_\tau$ (the blue dotted curve in Figure~\ref{fig:RegionsPlot}).  

The asymptotic description of the solution $q(M\chi,M\tau;\mathbf{Q}^{-s},M)$ for remainder $r\neq 0$ and $r\neq\tfrac{1}{4}$ and $(\chi,\tau)\in\exterior$ is consistent with the universal long-time asymptotics for solutions of the focusing nonlinear Schr\"odinger equation with nonzero boundary conditions at $x=\pm\infty$ established by Biondini and Mantzavinos \cite{BiondiniM17}.  These authors showed that for a wide variety of initial conditions, the solution depends asymptotically only on the ratio $\xi\defeq x/t=\chi/\tau$, and as a function of $\xi$ is approximated for $|\xi|<\sqrt{8}$ (translating to our scaling of the equation from theirs) by a modulated elliptic function solution with elliptic modulus $m(\xi)$ varying between $m(0)=1$ and $m(\sqrt{8})=0$, and approximated for $|\xi|>\sqrt{8}$ by a plane-wave solution of constant amplitude equal to that specified by the large-$x$ boundary conditions.  This is consistent with our description of $q(M\chi,M\tau;\mathbf{Q}^{-s},M)$ for remainder $r\neq 0$ and $r\neq\tfrac{1}{4}$ because 
\begin{itemize}
\item
the condition $\chi/\tau=\sqrt{8}$ is precisely the linear asymptote valid for large $\tau$ (dotted gray line in Figure~\ref{fig:RegionsPlot}) for the curve $\partial\exterior_\chi\cap\partial\exterior_\tau$ (dotted blue curve in Figure~\ref{fig:RegionsPlot}), and 
\item
the condition $m(0)=1$ is consistent with $m(\chi,\tau)\to 1$ as $(\chi,\tau)\to\ell_\mathrm{sol}$ because the latter curve, while not asymptotic to any line for large $\tau$, satisfies 
$\chi=\ln(\tau)+O(1)$ as $\tau\to +\infty$.  Hence the whole region above the curve $\ell_\mathrm{sol}$ in Figure~\ref{fig:RegionsPlot} can be found to the left of $\xi=\xi_0$ for any $\xi_0>0$, asymptotically in the large-$\tau$ limit.
\end{itemize}
On the other hand, the class of solutions considered in \cite{BiondiniM17} does not contain $q(x,t;\mathbf{Q}^{-s},M)$ for remainder $r\neq 0$ and large $M$ because increasing $M$ by half-integer increments amounts to iteration of a Darboux transformation \cite{BilmanM19} that injects solitons/rogue waves into the solution at the distinguished value of the spectral parameter corresponding to the nonzero background solution (here, $\lambda=\pm\ii$).  The slow decay to the background as $x\to\pm\infty$ for  $r\neq 0,\tfrac{1}{4}$ also obstructs analysis by inverse-scattering methods.  There are some extensions of the results of \cite{BiondiniM17} that allow for finitely many solitons with generic spectral parameters but no results for the case that the injected solitons are at the distinguished value.  It is also true that, as has been mentioned several times already, it is not possible to directly compare large-$(x,t)$ asymptotics with large-$M$ and bounded $(\chi,\tau)$ asymptotics without additional arguments that are not part of our proofs.

The reason why the solution $q(M\chi,M\tau;\mathbf{Q}^{-s},M)$ is so sensitive to the value of the remainder $r$ when $(\chi,\tau)\in\exterior$ is that in this domain we need to use the limiting form of the jump contour $\Sigma_\circ$ in which it is deformed into a dumbbell shape consisting of two loops connected by a ``neck'' that we denote by $N$ in Section~\ref{sec:dumbbell}.  When $\lambda\in N$, the algebraic form of the jump condition for this deformed problem depends explicitly on $r$; see Remark~\ref{rem:M-quantum} below.  In particular, for the cases $r=0$ and $r=\tfrac{1}{4}$ the jump matrix has two elements that vanish exactly, which prohibits the use of two of the four canonical factorizations of unit-determinant $2\times 2$ matrices:
\begin{equation}
\begin{alignedat}{3}
\begin{bmatrix} a&b\\c&d\end{bmatrix}&=\begin{bmatrix}1 & 0\\ca^{-1} & 1\end{bmatrix}a^{\sigma_3}\begin{bmatrix}1 & ba^{-1}\\0 & 1\end{bmatrix},&&\quad a\neq 0, &&\quad\text{(``LDU'')},\\
\begin{bmatrix} a&b\\c&d\end{bmatrix}&=\begin{bmatrix}1 & 0\\db^{-1}&1\end{bmatrix}\begin{bmatrix}0&b\\-b^{-1} & 0\end{bmatrix}\begin{bmatrix}1&0\\ab^{-1}&1\end{bmatrix},&&\quad b\neq 0,&&\quad\text{(``LTL'')},\\
\begin{bmatrix} a&b\\c&d\end{bmatrix}&=\begin{bmatrix}1 & ac^{-1}\\0 & 1\end{bmatrix}
\begin{bmatrix}0 & -c^{-1}\\c & 0\end{bmatrix}\begin{bmatrix}1 & dc^{-1}\\0 & 1\end{bmatrix},&&\quad c\neq 0,&&\quad\text{(``UTU'')},\\
\begin{bmatrix} a&b\\c&d\end{bmatrix}&=\begin{bmatrix}1 & bd^{-1}\\0 & 1\end{bmatrix}d^{-\sigma_3}
\begin{bmatrix}1 & 0\\cd^{-1}&1\end{bmatrix},&&\quad d\neq 0,&&\quad\text{(``UDL'')}.
\end{alignedat}
\end{equation}
It turns out that for $r=0$ (multiple-pole soliton case) only the LDU and UDL factorizations are possible and they are both trivial as the jump matrix on $N$ is already diagonal.  Similarly for $r=\tfrac{1}{4}$ (rogue wave case) only the LTL and UTU factorizations are possible and they are both trivial as the jump matrix on $N$ is already off-diagonal.  On the other hand, for $r\neq 0$ and $r\neq \tfrac{1}{4}$ there are four nonzero pivots and hence all four factorizations are admissible; moreover all four are essential to the steepest-descent arguments behind the proofs.

As they do not concern rogue waves and require substantially different proofs, all of the results reported in Section~\ref{sec:general} describing $q(M\chi,M\tau;\mathbf{Q}^{-s},M)$ for large $M$ with remainder $r\neq 0$ and $r\neq \tfrac{1}{4}$ and $(\chi,\tau)\in\exterior$ will be given in more detail and fully proven in a forthcoming paper.

\subsection{Numerical illustration of the results}

We now illustrate the accuracy of the asymptotic formul\ae{} obtained for the fundamental rogue waves $\psi_k(M \chi, M \tau)$, $M=\tfrac{1}{2}k+\tfrac{1}{4}$, for $(\chi,\tau)$ in the regions $\channels$, $\exterior$, and $\shelves$. 
In each subsubsection that follows, we first plot the exact solution $\psi_k$ against the approximation provided by the asymptotic formul\ae{} for certain values of $k$. Second, we study the trends in the relevant error sizes by comparing the approximations with the family of exact solutions $\psi_k$ as the value of $k$ increases over a set of positive integers $\mathcal{K}$.
The solutions $\psi_{k}$ are computed by numerically solving linear systems system obtained from their representations given by \rhref{rhp:rogue-wave}. We refer the reader to \cite[Section 3.5]{BilmanM19} for the derivation of the linear system used in this work. 

\subsubsection{Numerical illustration of the asymptotic formula for $\psi_k(M\chi,M\tau)$ in $\channels$} 
Here we give numerical evidence confirming Corollary~\ref{cor:rogue-wave-channels}.  Recall the leading term approximation of $\psi_k(M\chi,M\tau)$ denoted $L_k^{[\channels]}(\chi,\tau)$ and defined in \eqref{eq:leading-term-channels}.
\paragraph{\textit{Comparison plots}} We first consider $\tau = 0$, in which case $\channels$ comprises the open interval $0<\chi<2$, and we fix the proper subset $[0.25,1.75]$ of values for $\chi$. We plot $\psi_{k}(M \chi, M\tau)$ versus $L_k^{[\channels]}(\chi,\tau)$ for $k=15$ and for $k=32$ in Figure~\ref{f:comparison-channels-tau-0-k-15-32}. Note that the solution is real-valued for $\tau=0$ hence no plots for the imaginary parts are given.
\begin{figure}[h]
\includegraphics[width=0.45\textwidth]{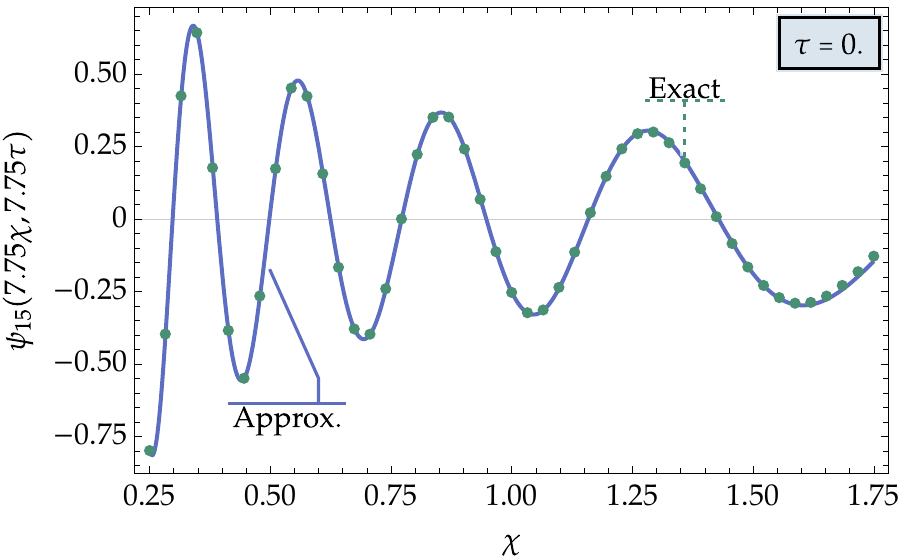}\qquad
\includegraphics[width=0.45\textwidth]{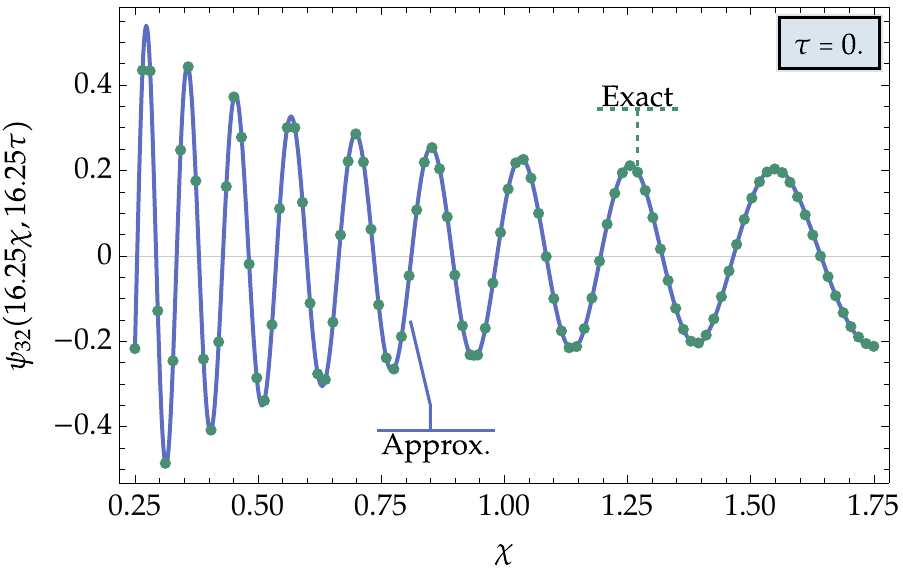}
\caption{Comparison of the exact solution $\psi_k(M \chi, M\tau)$ (dots) with the approximation $L_k^{[\channels]}(\chi,\tau)$ (solid curve) at $\tau=0$ for $0.25\leq \chi \leq 1.75$. Left panel: $k=15$, right panel: $k=32$. The solution is real-valued for $\tau=0$.}
\label{f:comparison-channels-tau-0-k-15-32}
\end{figure}
Next, we set $\tau=0.125$ and consider the range $1 \leq \chi \leq 1.75$ which lies inside the region $\channels$. Figure~\ref{f:comparison-channels-tau-0p125-k-31} presents plots of the real and imaginary parts of $\psi_{k}(M \chi, M\tau)$ with those of $L_k^{[\channels]}(\chi,\tau)$ for $k=31$ and Figure~\ref{f:comparison-channels-tau-0p125-k-32} presents the same comparisons for $k=32$.
 \begin{figure}[h]
\includegraphics[width=0.45\textwidth]{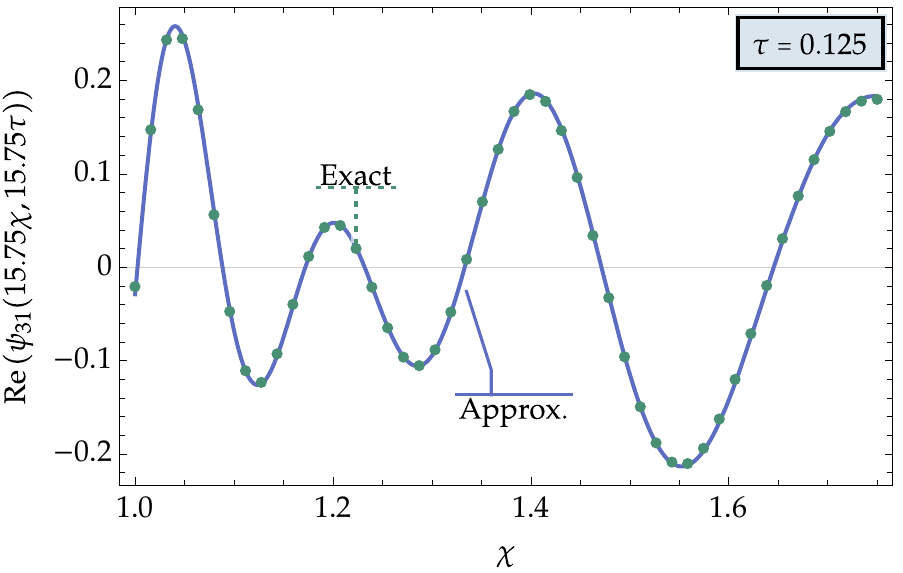}\qquad
\includegraphics[width=0.45\textwidth]{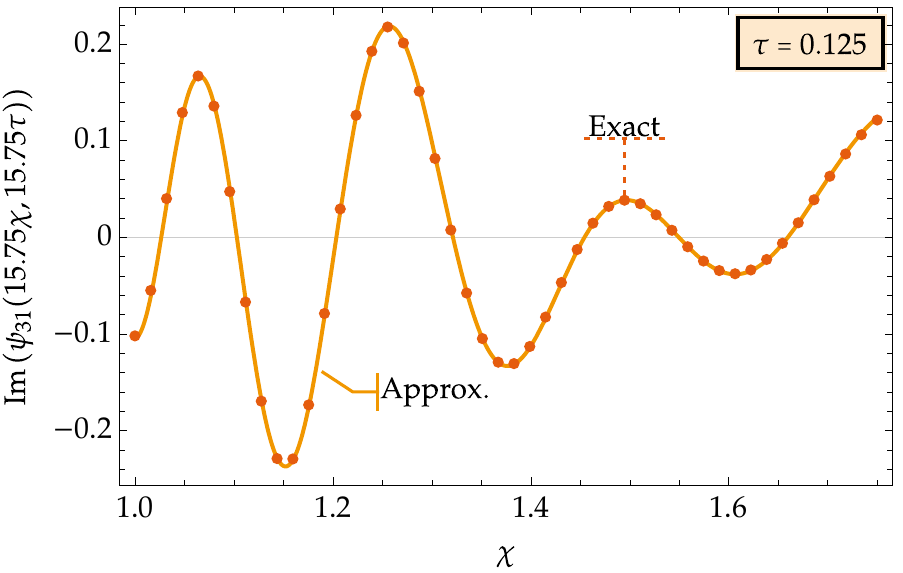}
\caption{Comparison of the exact solution $\psi_k(M \chi, M\tau)$ (dots) with the approximation $L_k^{[\channels]}(\chi,\tau)$ (solid curve) for $k=31$, at $\tau=0.125$ for $1\leq \chi \leq 1.75$. Left panel: real parts, right panel: imaginary parts.}
\label{f:comparison-channels-tau-0p125-k-31}
\end{figure}
 \begin{figure}[h]
\includegraphics[width=0.45\textwidth]{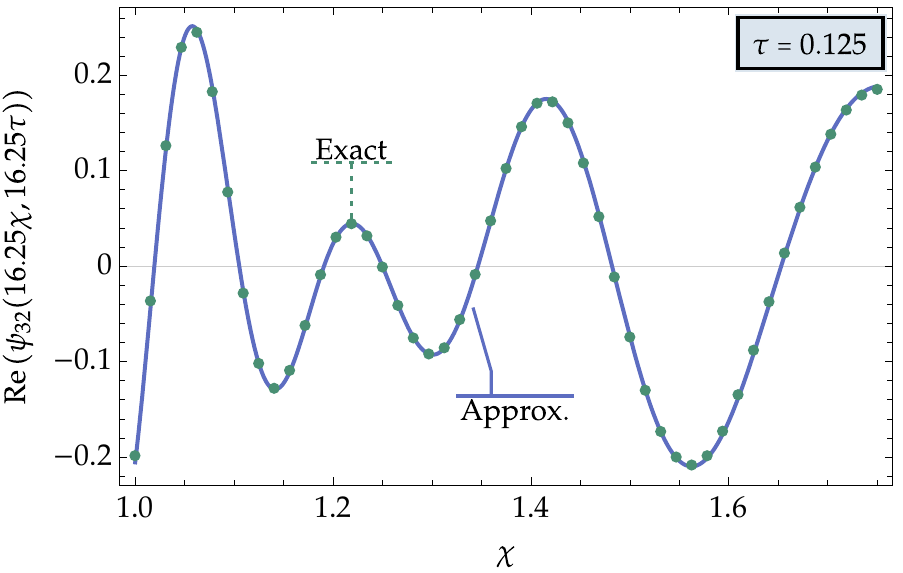}\qquad
\includegraphics[width=0.45\textwidth]{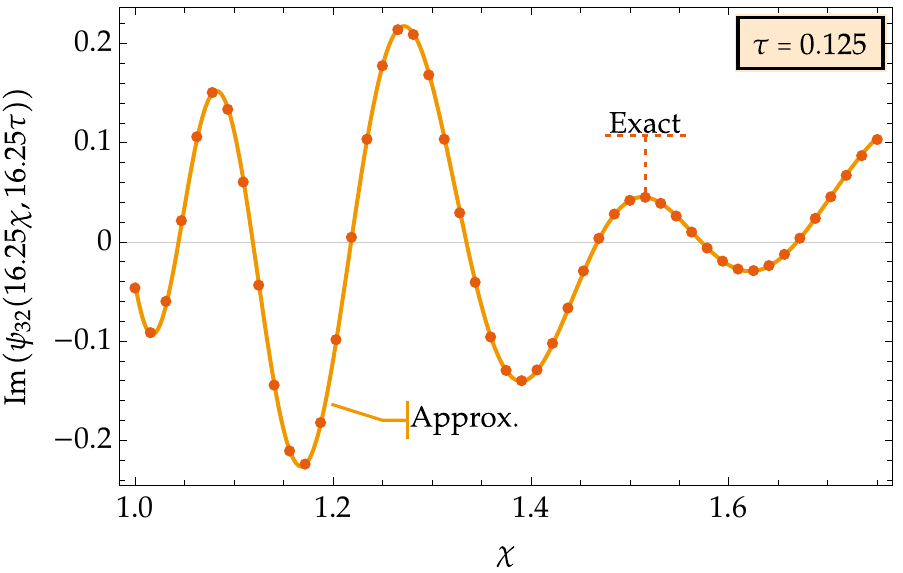}
\caption{Comparison of the exact solution $\psi_k(M \chi, M\tau)$ (dots) with the approximation $L_k^{[\channels]}(\chi,\tau)$ (solid curve) for $k=32$, at $\tau=0.125$ for $1 \leq \chi \leq 1.75$. Left panel: real parts, right panel: imaginary parts.}
\label{f:comparison-channels-tau-0p125-k-32}
\end{figure}

\paragraph{\textit{Error plots}} To validate the size of the error term $O(k^{-\frac{3}{2}})$ predicted in Corollary~\ref{cor:rogue-wave-channels}, we fix $\tau=0.125$ and the range $1\leq \chi \leq 1.75$ for the values of $\chi$. We construct a grid $\mathcal{G}_k$ on this interval with step size
$\delta\chi = (4M)^{-1}$ starting at the left endpoint $\chi=1$. 
We then compute the absolute errors made in approximating the fundamental rogue waves $\psi_k(M\chi,M\tau)$ with the leading terms $L_k^{[\channels]}(\chi,\tau)$ measured in the sup-norm over the grid $\mathcal{G}_k$, for $k$ ranging over the set $\mathcal{K} := \{16, 32, 48, 64, 80, 96 \}$:
\begin{equation}
E_k^{[\channels]}\defeq \sup_{\chi \in \mathcal{G}_k} \left| \psi_{k}(M\chi,M\tau) - L_k^{[\channels]}(\chi,\tau) \right|, \quad k\in\mathcal{K},\quad \tau=0.125.
\end{equation}
We plot $\ln(E_k^{[\channels]})$ versus $\ln(k)$ in Figure~\ref{f:channels-errors} and perform linear regression, which yields the best-fit line $\ln(E_k^{[\channels]}) =0.518758 - 1.48135 \ln(k)$ with $R$-squared value of $0.998502$. The slope of this line recovers approximately the exponent $-\tfrac{3}{2}$ in the error term predicted in Corollary~\ref{cor:rogue-wave-channels}.
\begin{figure}[h]
\includegraphics[width=0.45\textwidth]{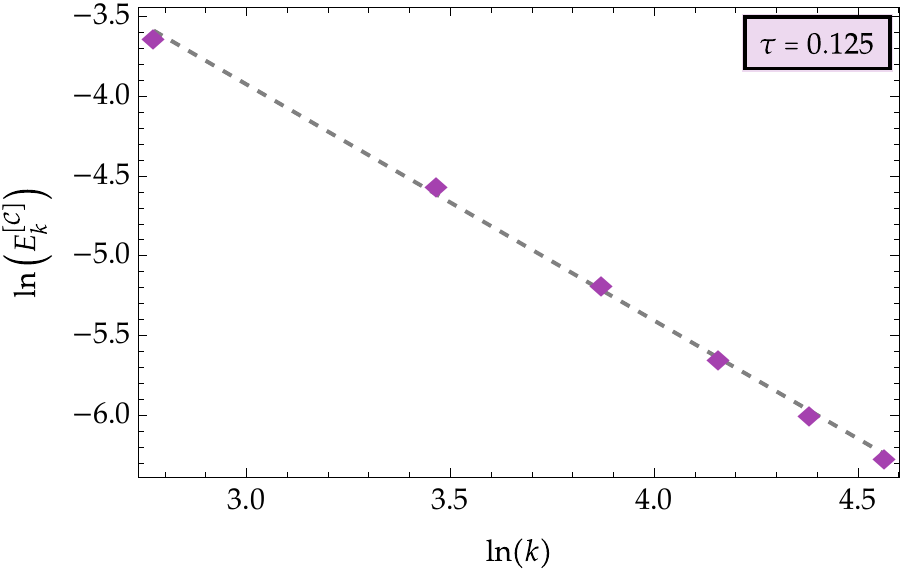}
\caption{Plot of $\ln(E_k^{[\channels]})$ versus $\ln(k)$ (diamond markers) for $k\in\mathcal{K}$. The dashed line is the best-fit line for the plotted data set.}
\label{f:channels-errors}
\end{figure}

\subsubsection{Numerical illustration of the asymptotic formula for $\psi_k(M\chi,M\tau)$ in $\exterior$} 
Next, we give numerical evidence confirming Theorem~\ref{thm:exterior}.  Recall the leading term approximation of $\psi_k(M\chi,M\tau)$ denoted $L_k^{[\exterior]}(\chi,\tau)$ and defined in \eqref{eq:leading-term-exterior}.
%
\paragraph{\textit{Comparison plots}} We fix $\tau = 1.25>1$ and consider the range $0\leq \chi \leq 4$, which contains points $(\chi,\tau)$ from $\exterior_\chi$ and from $\exterior_\tau$, and plot the real and imaginary parts of $\psi_k(M\chi,M \tau)$ and its approximation $L^{[\exterior]}_k(\chi,\tau)$ for $k=16$ in Figure~\ref{f:comparison-exterior-k-16} and for $k=31$ in Figure~\ref{f:comparison-exterior-k-31}.
\begin{figure}[h]
\includegraphics[width=0.45\textwidth]{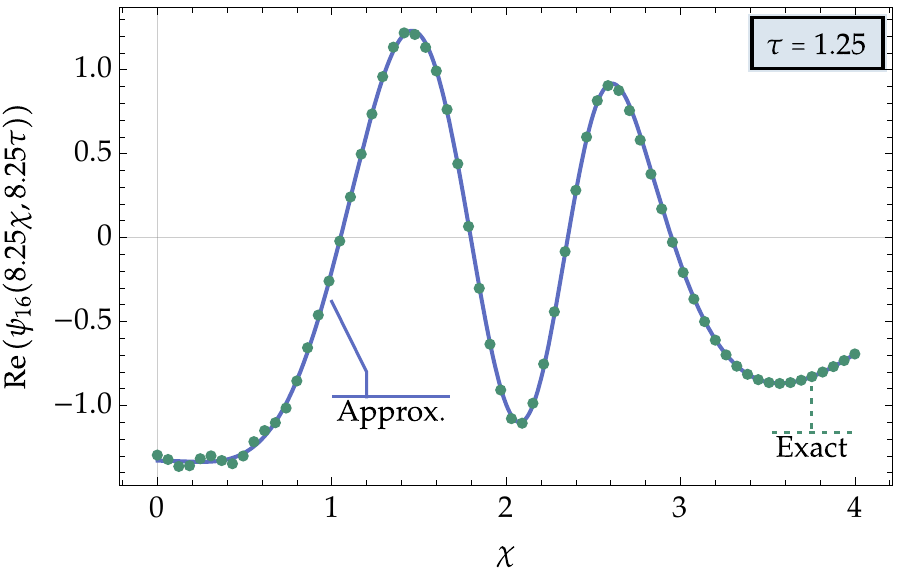}\qquad
\includegraphics[width=0.45\textwidth]{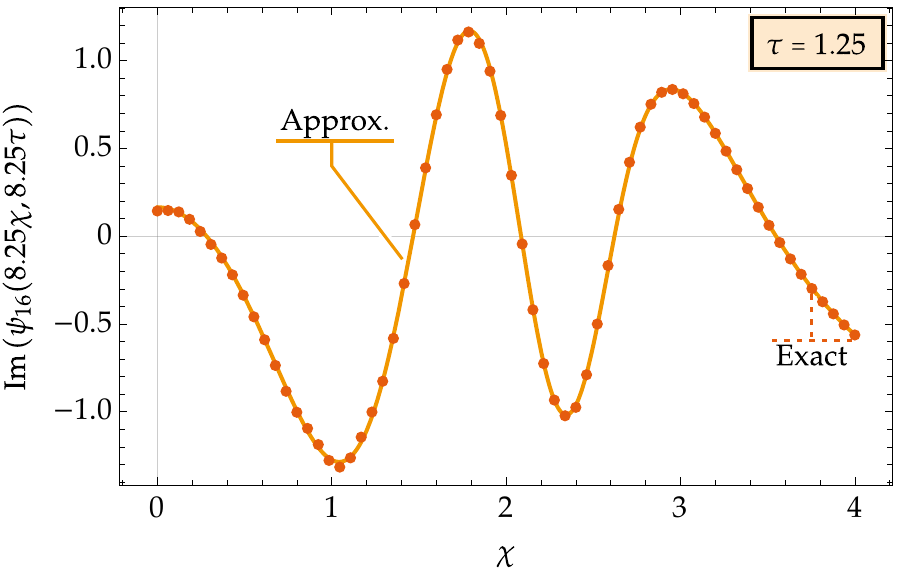}
\caption{Comparison of the exact solution $\psi_k(M \chi, M\tau)$ (dots) with the approximation $L_k^{[\exterior]}(\chi,\tau)$ (solid curve) for $k=16$, at $\tau=1.25$ for $0\leq \chi \leq 4$. Left panel: real parts, right panel: imaginary parts.}
\label{f:comparison-exterior-k-16}
\end{figure}
\begin{figure}[h]
\includegraphics[width=0.45\textwidth]{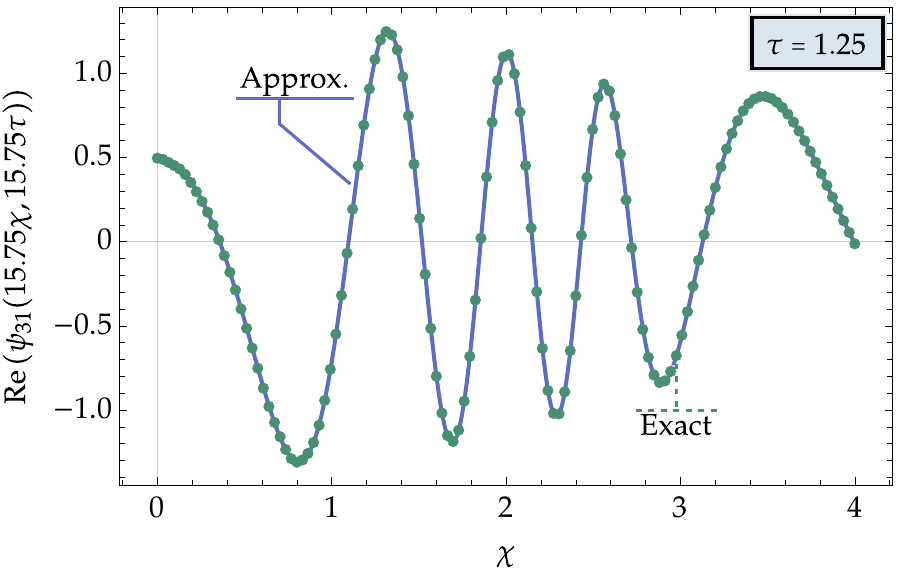}\qquad
\includegraphics[width=0.45\textwidth]{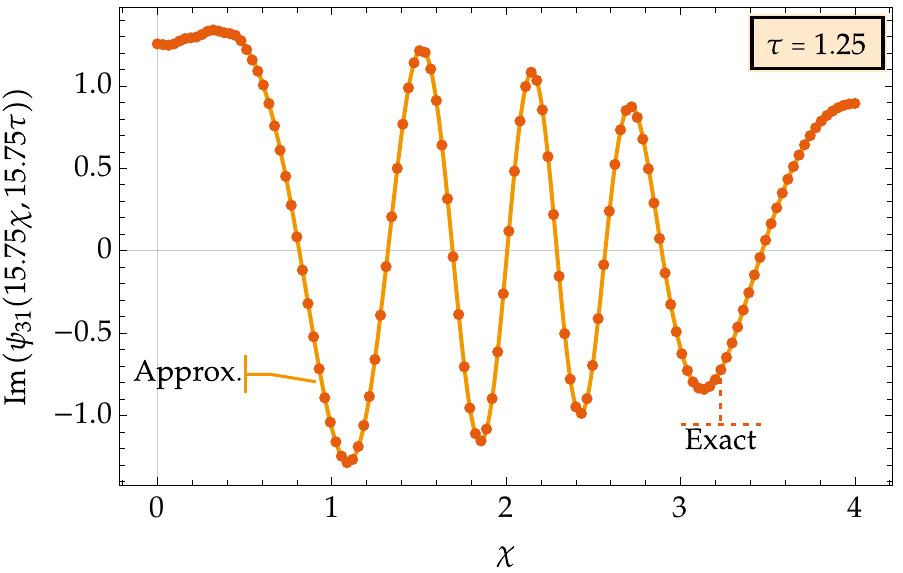}
\caption{Comparison of the exact solution $\psi_k(M \chi, M\tau)$ (dots) with the approximation $L_k^{[\exterior]}(\chi,\tau)$ (solid curve) for $k=31$, at $\tau=1.25$ for $0\leq \chi \leq 4$. Left panel: real parts, right panel: imaginary parts.}
\label{f:comparison-exterior-k-31}
\end{figure}
\paragraph{\textit{Error plots}} We now validate the size of the error term
$O(k^{-1})$ predicted in Theorem~\ref{thm:exterior}. We fix $\tau=1.25$ and the range $3\leq \chi \leq 4$ for the values of $\chi$, and again construct a grid $\mathcal{G}_k$ on this interval with step size
$\delta\chi = (4M)^{-1}$ starting at the left endpoint $\chi=3$. 
We then compute the absolute errors made in approximating the fundamental rogue waves $\psi_k(M\chi,M\tau)$ with the leading terms $L_k^{[\exterior]}(\chi,\tau)$ measured in the sup-norm over the grid $\mathcal{G}_k$, for $k$ ranging over the set $\mathcal{K} := \{16, 32, 48, 64, 80, 96 \}$:
\begin{equation}
E_k^{[\exterior]}:= \sup_{\chi \in \mathcal{G}_k} \left| \psi_{k}(M\chi,M\tau) - L_k^{[\exterior]}(\chi,\tau) \right|, \quad k\in\mathcal{K},\quad \tau=1.25.
\end{equation}
We plot $\ln(E_k^{[\exterior]})$ versus $\ln(k)$ in Figure~\ref{f:exterior-errors} and perform linear regression, which yields the best-fit line $\ln(E_k^{[\exterior]}) =-1.98314 - 0.927044 \ln(k)$ with $R$-squared value of $0.999274$. The slope of this line recovers approximately the exponent $-1$ in the error predicted in Theorem~\ref{thm:exterior}.
\begin{figure}[h]
\includegraphics[width=0.45\textwidth]{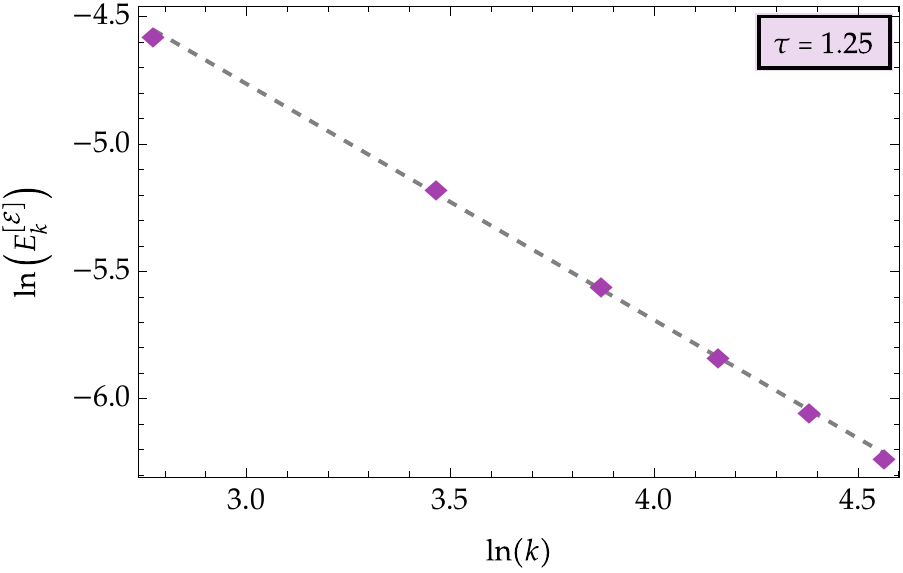}
\caption{Plot of $\ln(E_k^{[\exterior]})$ versus $\ln(k)$ (diamond markers) for $k\in\mathcal{K}$. The dashed line is the best-fit line for the plotted data set.}
\label{f:exterior-errors}
\end{figure}
\subsubsection{Numerical illustration of the asymptotic formula for $\psi_k(M\chi,M\tau)$ in $\shelves$} 
Finally, we turn to the illustration of Corollary~\ref{cor:rogue-wave-shelves}, recalling the approximate formula $L_k^{[\shelves]}(\chi,\tau)+S_k^{[\shelves]}(\chi,\tau)$ defined in \eqref{eq:rw-terms-shelves}.

\paragraph{\textit{Comparison plots}} We fix $\tau = 0.5$ for which $\shelves$ consists of the interval $(0,2)$ for the values of $\chi$. We consider the range $0.25\leq\chi\leq1.75$ which is a proper subset of $(0,2)$ and plot the real and imaginary parts of $\psi_k(M\chi, M\tau)$ and of
$ L_k^{[\shelves]}(\chi,\tau) + S_k^{[\shelves]}(\chi,\tau)$ for $k=23$ in Figure~\ref{f:comparison-bun-k-23} and for $k=32$ in Figure~\ref{f:comparison-bun-k-32}.
\begin{figure}[h]
\includegraphics[width=0.45\textwidth]{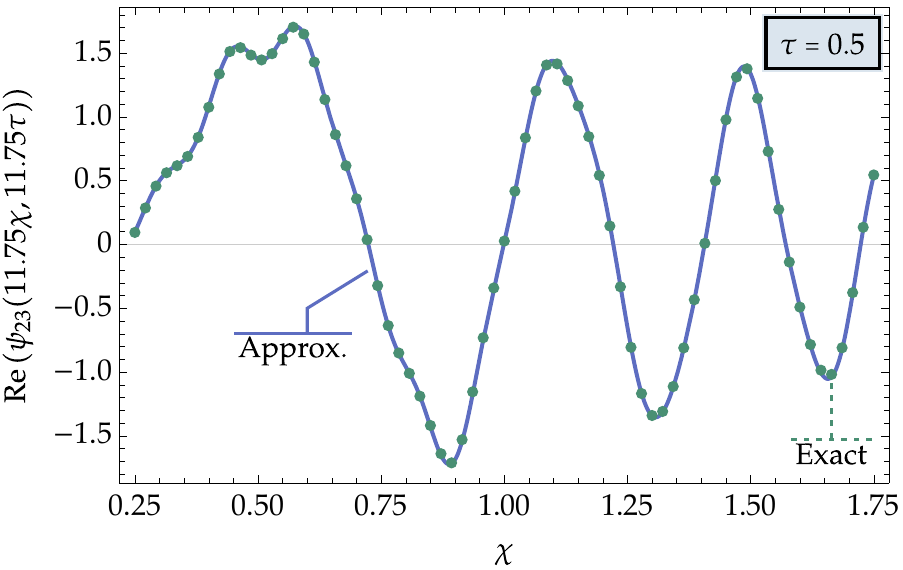}\qquad
\includegraphics[width=0.45\textwidth]{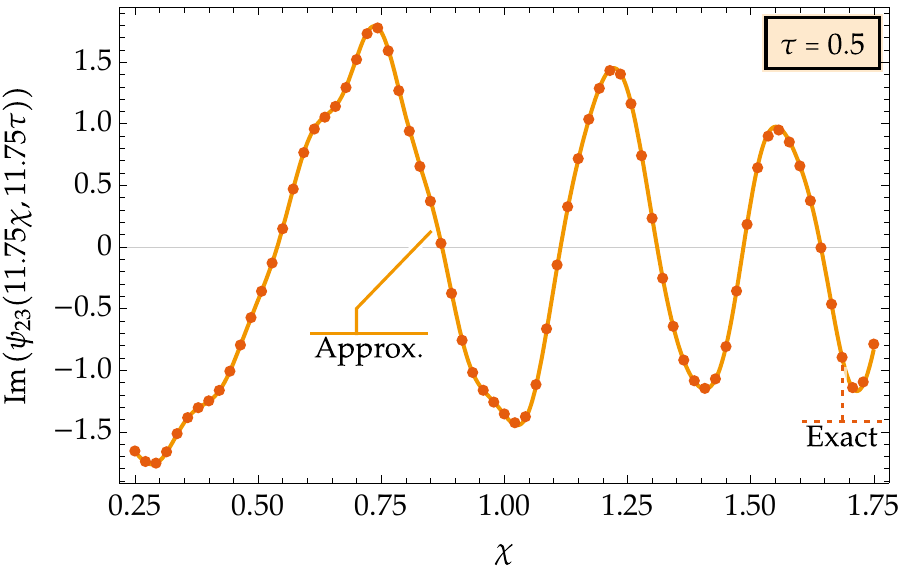}
\caption{Comparison of the exact solution $\psi_k(M \chi, M\tau)$ (dots) with the approximation $L_k^{[\shelves]}(\chi,\tau)+S_k^{[\shelves]}(\chi,\tau)$ (solid curve) for $k=23$, at $\tau=0.5$ for $0.25< \chi < 1.75$. Left panel: real parts, right panel: imaginary parts.}
\label{f:comparison-bun-k-23}
\end{figure}
\begin{figure}[h]
\includegraphics[width=0.45\textwidth]{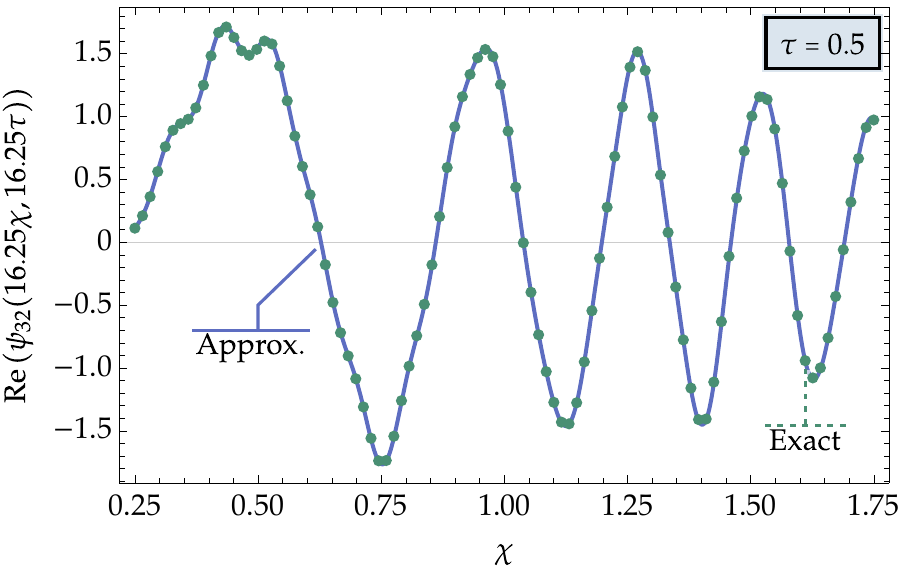}\qquad
\includegraphics[width=0.45\textwidth]{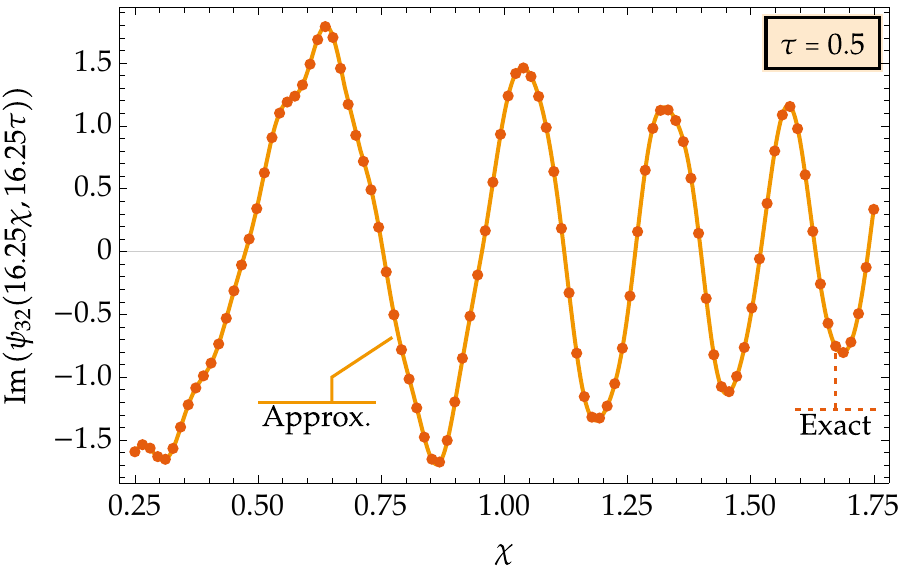}
\caption{Comparison of the exact solution $\psi_k(M \chi, M\tau)$ (dots) with the approximation $L_k^{[\shelves]}(\chi,\tau)+S_k^{[\shelves]}(\chi,\tau)$ (solid curve) for $k=32$, at $\tau=0.5$ for $0.25< \chi < 1.75$. Left panel: real parts, right panel: imaginary parts.}
\label{f:comparison-bun-k-32}
\end{figure}

\paragraph{\textit{Error plots}} To validate the size of the error term predicted in Corollary~\ref{cor:rogue-wave-shelves}, we fix $\tau=0.5$, consider the range $0.5 \leq \chi \leq 1.5$ corresponding to a line segment contained in $\shelves$, and construct a grid $\mathcal{G}_k$ on this interval with step size $\delta \chi = (4M)^{-1}$ starting at the left endpoint $\chi=0.5$. 
We then again compute the absolute errors made in approximating the fundamental rogue waves  $\psi_k(M\chi,M\tau)$ with $L_k^{[\shelves]}(\chi,\tau) + S_k^{[\shelves]}(\chi,\tau)$ measured in the sup-norm over the grid $\mathcal{G}_k$, for $k$ ranging over the set $\mathcal{K} := \{16, 32, 48, 64, 80, 96 \}$:
\begin{equation}
E_k^{[\shelves]}:= \sup_{\chi \in \mathcal{G}_k} \left| \psi_{k}(M\chi,M\tau) - \left( L_k^{[\exterior]}(\chi,\tau) + S_k^{[\shelves]}(\chi,\tau) \right)\right|, \quad k\in\mathcal{K},\quad \tau=0.5.
\end{equation}
We plot $\ln(E_k^{[\shelves]})$ versus $\ln(k)$ in Figure~\ref{f:shelves-errors} and perform linear regression, which yields the best-fit line $\ln(E_k^{[\shelves]}) =-0.366203 - 1.01873 \ln(k)$ with $R$-squared value of $0.995771$. The slope of this line recovers approximately the exponent $-1$ in the error predicted in Corollary~\ref{cor:rogue-wave-shelves}.
\begin{figure}[h]
\includegraphics[width=0.45\textwidth]{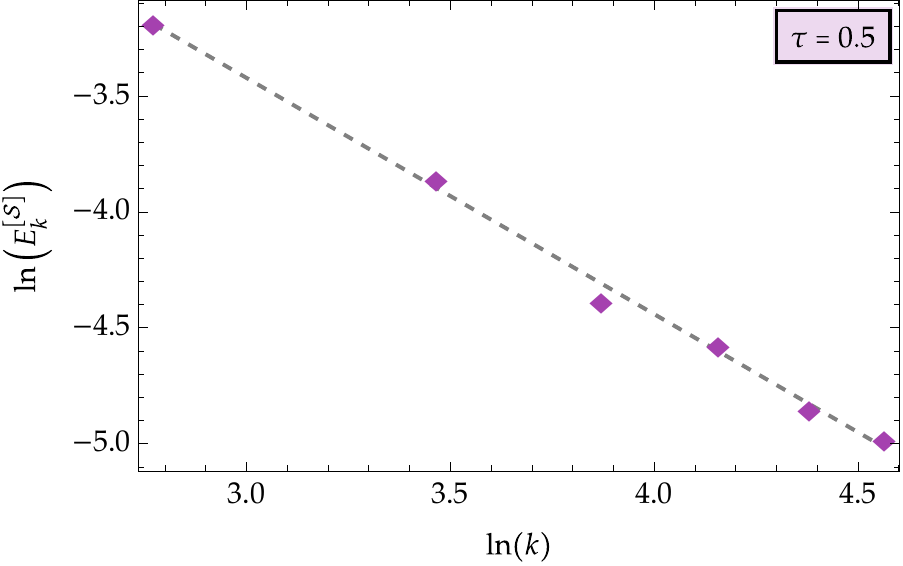}
\caption{Plot of $\ln(E_k^{[\shelves]})$ versus $\ln(k)$ (diamond markers) for $k\in\mathcal{K}$. The dashed line is the best-fit line for the plotted data set.}
\label{f:shelves-errors}
\end{figure}

\paragraph{\textit{Interference pattern}} We now illustrate how Corollary~\ref{cor:pattern} accurately predicts the complicated wave pattern seen in the plots in Figure~\ref{fig:2D-Plots} for $(\chi,\tau)\in\shelves$.  
Although Corollary~\ref{cor:pattern} applies to more general solutions, in keeping with the setting we restrict attention to special case of the fundamental rogue-wave solutions.
Thus, we compute the following $M$-dependent unions of level curves inside $\shelves$ in the $(\chi,\tau)$-plane:
\begin{align}
\mathcal{M}_a &=\bigcup_{j\in\mathbb{Z}} \{ (\chi,\tau)\in \shelves \colon \phi_a(\chi,\tau;M) = - \frac{\pi}{2} + 2\pi j \},\\
\mathcal{M}_b &=\bigcup_{j\in\mathbb{Z}} \{ (\chi,\tau)\in \shelves \colon \phi_b(\chi,\tau;M) = - \frac{\pi}{2} + 2\pi j \},
\end{align}
so that $\sin(\phi_a(\chi,\tau;M)) =  -1$ for $(\chi,\tau)\in \mathcal{M}_a$ and  $\sin(\phi_b(\chi,\tau;M)) =  -1$ for $(\chi,\tau)\in \mathcal{M}_b$. 
Accordingly, the claim is that the intersection points of $\mathcal{M}_a$ and $\mathcal{M}_b$ locate the amplitude peaks formed by $|\psi_k(M \chi, M\tau)|^2$. To verify that this is the case, we fix the box $[0.2, 1.2]\times[0.2,0.8] \subset \shelves$ and plot $|\psi_k(M \chi, M\tau)|^2$ and the set of points $\mathcal{M}_a \cup \mathcal{M}_b$. Figure~\ref{f:interference-overlays} illustrates this formation as it is described.

\begin{figure}[h]
\includegraphics[width=0.32\textwidth]{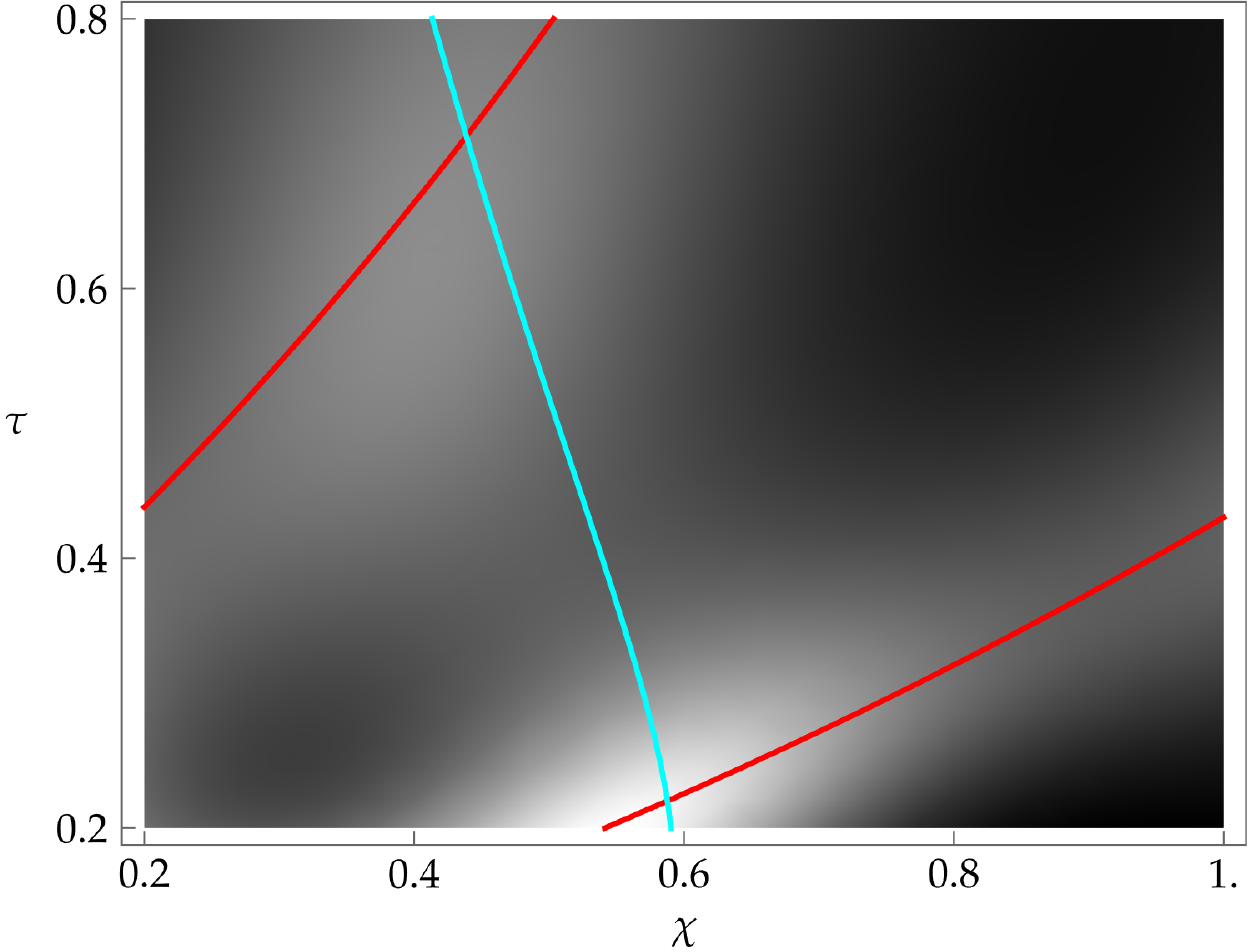}\quad \includegraphics[width=0.32\textwidth]{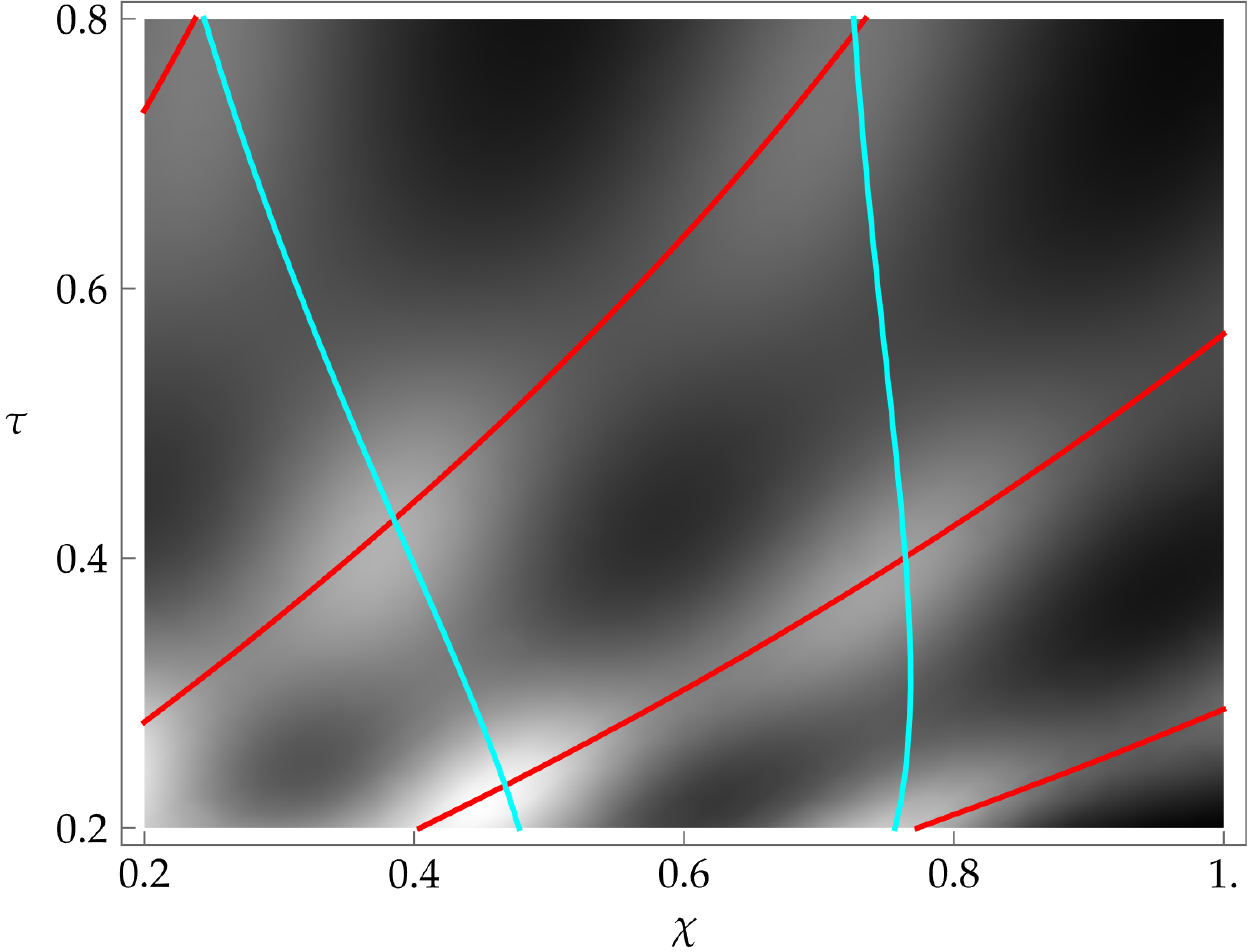}
\includegraphics[width=0.32\textwidth]{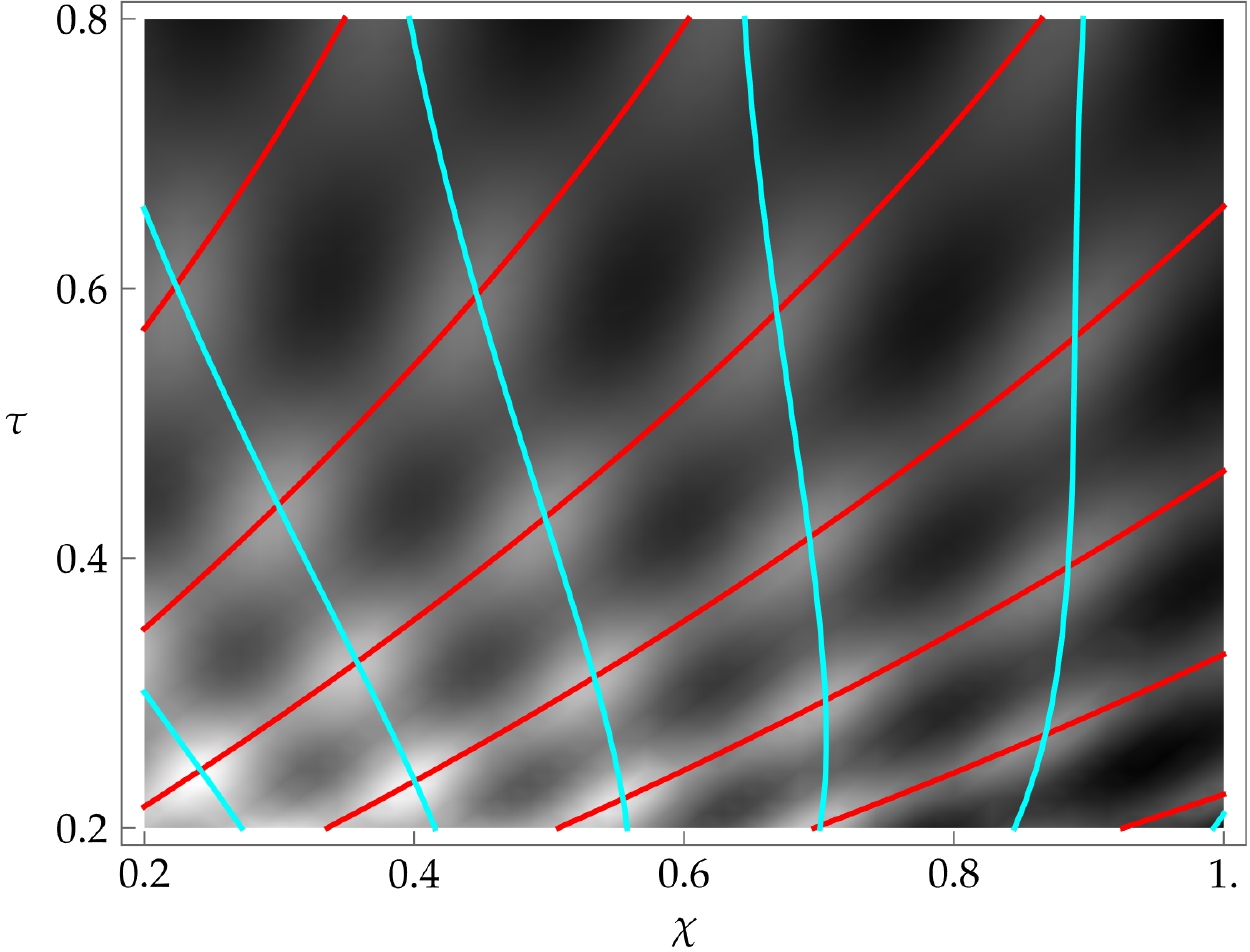}
\caption{Monochrome density plot of $|\psi_k(M \chi, M\tau)|^2$ and plots of the unions of level curves $\mathcal{M}_a$ (red contours) and $\mathcal{M}_b$ (cyan contours) for values of $k=4,8,16$. The brighter colors correspond to larger amplitude, and hence the white spots are where the peaks are formed. Left panel: $k=4$, center panel: $k=8$, right panel: $k=16$.}
\label{f:interference-overlays}
\end{figure}


\section{Far-Field Asymptotic Behavior in the Domain $\channels$}
\label{sec:channels}
In this section we prove Theorem~\ref{thm:channels}.  
Our analysis is driven by the sign chart of $\mathrm{Re}(\ii\vartheta(\lambda;\chi,\tau))$ in the $\lambda$-plane.  The function $\lambda\mapsto\mathrm{Re}(\ii\vartheta(\lambda;\chi,\tau))$ is odd with respect to Schwarz reflection $\lambda\mapsto \lambda^*$.  It follows that the whole real $\lambda$-axis is a component of the zero level curve $\mathrm{Re}(\ii\vartheta(\lambda;\chi,\tau))=0$, so all three critical points lie on the zero level.  Since each critical point is simple when $(\chi,\tau)\in\channels$, from each of them a unique  arc of the zero level curve emanates locally into the upper half-plane with a vertical tangent.  From \eqref{eq:vartheta} one sees easily that as $\chi>0$ holds in $\channels$, $\mathrm{Re}(\ii\vartheta(\lambda;\chi,0))$ is negative (resp., positive) for sufficiently large $\lambda$ in the upper (resp., lower) half-plane; on the other hand $\mathrm{Re}(\ii\vartheta(\lambda;\chi,\tau))$ is always positive (resp., negative) near $\lambda=\ii$ (resp., near $\lambda=-\ii$).  From this and the fact that $\mathrm{Re}(\ii\vartheta(\lambda;\chi,\tau))$ is harmonic away from $\lambda=\pm\ii$ it follows that when $\tau=0$ the two arcs of the zero level curve emanating into the upper half-plane from $\lambda=a(\chi,0)$ and $\lambda=b(\chi,0)$ actually coincide and close around the singularity at $\lambda=\ii$.  This structure persists under perturbation for $\tau\neq 0$, as the arc of the zero level curve emanating into the upper half-plane from the newly-born large critical point must tend to $\lambda=\infty$ vertically without intersecting the arc we denote by $\Gamma^+=\Gamma^+(\chi,\tau)$ joining $a(\chi,\tau)$ and $b(\chi,\tau)$ in the upper half-plane. Therefore, for all $(\chi,\tau)\in \channels$, the zero level curve of $\lambda\mapsto\mathrm{Re}(\ii\vartheta(\lambda;\chi,\tau))$ is the disjoint union $\mathbb{R}\sqcup\Gamma^+\sqcup\Gamma^-\sqcup \ell^+\sqcup \ell^-$, where $\ell^+$ denotes the unbounded arc in the upper half-plane emanating from the third critical point that is large when $\tau\neq 0$ is small (we take $\ell^+=\emptyset$ when $\tau=0$) and $\Gamma^-$ and $\ell^-$ are the Schwarz reflections of $\Gamma^+$ and $\ell^+$ respectively.

\subsection{Steepest descent deformation of the Riemann-Hilbert problem}
Since $\overline{\Gamma^+\cup\Gamma^-}$ is a simple closed curve with the points $\lambda=\pm\ii$ in its interior, we use this curve as $\Sigma_\circ$ in the formulation of Riemann-Hilbert Problem~\ref{rhp:rogue-wave-reformulation}.  As $\Sigma_\circ$ has clockwise orientation, we assume that $\Gamma^+$ is oriented from $a$ to $b$ while $\Gamma^-$ is oriented from $b$ to $a$ in the lower half-plane. In the jump condition \eqref{eq:S-jump} for the matrix $\mathbf{S}(\lambda;\chi,\tau,\mathbf{Q}^{-s},M)$ equivalent to $\mathbf{P}(\lambda;x,t,\mathbf{Q}^{-s},M)$ by \eqref{eq:S-from-P}, we factor the matrix $\mathbf{Q}^{-s}$, $s=\pm 1$, as
\begin{equation}
\mathbf{Q}^{-s}=\begin{cases}
2^{\frac{1}{2}\sigma_3}\begin{bmatrix}1 & \tfrac{1}{2}s\\0 & 1\end{bmatrix}\begin{bmatrix}1 & 0\\-s & 1\end{bmatrix},\quad& \lambda\in\Gamma^+,\\
2^{-\frac{1}{2}\sigma_3}\begin{bmatrix}1 & 0\\-\tfrac{1}{2}s & 1\end{bmatrix}\begin{bmatrix}1 & s\\0 & 1\end{bmatrix},\quad& \lambda\in\Gamma^-.
\end{cases}
\label{eq:Q-factorizations}
\end{equation}
Based on these two factorizations, we define a new unknown $\mathbf{W}(\lambda)=\mathbf{W}(\lambda;\chi,\tau,\mathbf{Q}^{-s},M)$ related to $\mathbf{S}(\lambda;\chi,\tau,\mathbf{Q}^{-s},M)$ by first introducing ``lens'' domains $L^\pm$ and $R^\pm$ to the left and right respectively of $\Gamma^\pm$ (so thin as to exclude the points $\pm\ii$ and to support a fixed sign of $\mathrm{Re}(\ii\vartheta(\lambda;\chi,\tau))$) and we let $\Omega^\pm$ denote the domain between $R^\pm$ and the real line.  See Figure~\ref{fig:Channels1}, left panel.  
\begin{figure}[h]
\begin{center}
\includegraphics{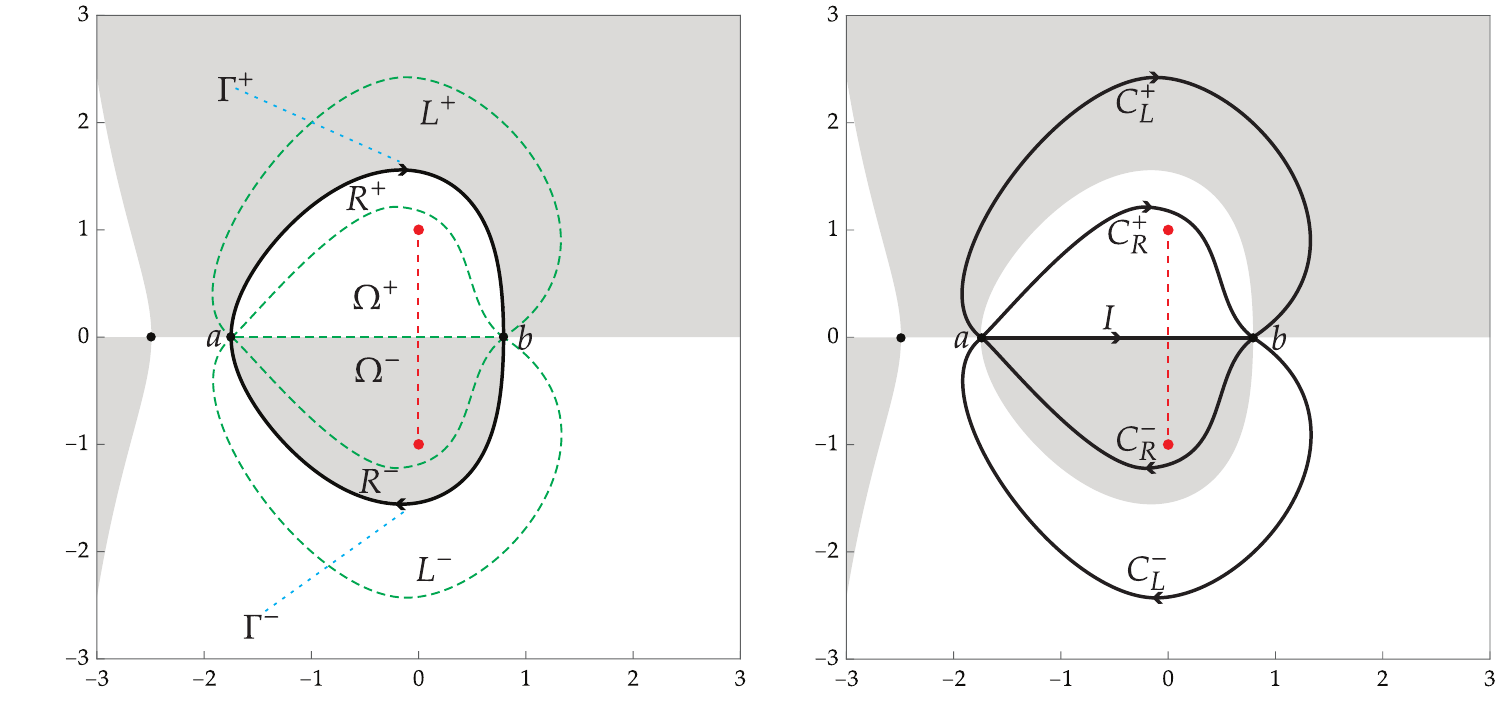}
\end{center}
\caption{Left:  for $(\chi,\tau)=(1,0.145)\in \channels$, the regions in the $\lambda$-plane where $\mathrm{Re}(\ii\vartheta(\lambda;\chi,\tau))<0$ (shaded) and $\mathrm{Re}(\ii\vartheta(\lambda;\chi,\tau))>0$ (unshaded), and the curve $\Sigma_\circ=\Gamma^+\cup\Gamma^-$.  The jump contour $\Sigma_\mathrm{c}$ for $\vartheta(\lambda;\chi,\tau)$ is indicated with a red dashed line terminating at the endpoints $\lambda=\pm\ii$.  Critical points of $\vartheta(\lambda;\chi,\tau)$ are shown with black dots.  Also shown are the ``lens'' regions $L^\pm$ and $R^\pm$ lying to the left and right respectively of $\Gamma^\pm$, and the domains $\Omega^\pm$ lying between $R^\pm$ and the real axis and containing the points $\lambda=\pm\ii$.  Right:  the jump contour for $\mathbf{W}(\lambda)$.}
\label{fig:Channels1}
\end{figure}
Then, we define $\mathbf{W}(\lambda)$ by 
\begin{equation}
\mathbf{W}(\lambda)\defeq\mathbf{S}(\lambda;\chi,\tau,\mathbf{Q}^{-s},M)\begin{bmatrix}1 & 0\\
s\ee^{2\ii M\vartheta(\lambda;\chi,\tau)} & 1\end{bmatrix},\quad\lambda\in L^+,
\label{eq:T-S-L-plus-ALT}
\end{equation}
\begin{equation}
\mathbf{W}(\lambda)\defeq\mathbf{S}(\lambda;\chi,\tau,\mathbf{Q}^{-s},M)2^{\frac{1}{2}\sigma_3}\begin{bmatrix}1 & \tfrac{1}{2}s\ee^{-2\ii M\vartheta(\lambda;\chi,\tau)}\\ 0 & 1\end{bmatrix},\quad\lambda\in R^+,
\label{eq:T-S-R-plus-ALT}
\end{equation}
\begin{equation}
\mathbf{W}(\lambda)\defeq\mathbf{S}(\lambda;\chi,\tau,\mathbf{Q}^{-s},M)2^{\frac{1}{2}\sigma_3},\quad
\lambda\in\Omega^+,
\end{equation}
\begin{equation}
\mathbf{W}(\lambda)\defeq\mathbf{S}(\lambda;\chi,\tau,\mathbf{Q}^{-s},M)2^{-\frac{1}{2}\sigma_3},\quad
\lambda\in\Omega^-,
\end{equation}
\begin{equation}
\mathbf{W}(\lambda)\defeq\mathbf{S}(\lambda;\chi,\tau,\mathbf{Q}^{-s},M)2^{-\frac{1}{2}\sigma_3}\begin{bmatrix} 1 & 0\\-\tfrac{1}{2}s\ee^{2\ii M\vartheta(\lambda;\chi,\tau)} & 1\end{bmatrix},\quad\lambda\in R^-,\quad\text{and}
\label{eq:T-S-R-minus-ALT}
\end{equation}
\begin{equation}
\mathbf{W}(\lambda)\defeq\mathbf{S}(\lambda;\chi,\tau,\mathbf{Q}^{-s},M)\begin{bmatrix}
1 & -s\ee^{-2\ii M\vartheta(\lambda;\chi,\tau)}\\ 0 & 1\end{bmatrix},\quad
\lambda\in L^-,
\label{eq:T-S-L-minus-ALT}
\end{equation}
and we take $\mathbf{W}(\lambda)\defeq\mathbf{S}(\lambda;\chi,\tau,\mathbf{Q}^{-s},M)$ whenever $\lambda\in \mathbb{C}\setminus\overline{L^+\cup R^+\cup\Omega^+\cup\Omega^-\cup R^-\cup L^-}$.  Then it is easy to check that $\mathbf{W}(\lambda)$ may be defined for $\lambda\in\Gamma^+\cup\Gamma^-$ to be analytic there, so that $\mathbf{W}(\lambda)$ is analytic in the complement of the jump contour $C_L^+\cup C_R^+\cup I\cup C_R^-\cup C_L^-$ shown in Figure~\ref{fig:Channels1}, right panel.  The jump conditions satisfied by $\mathbf{W}(\lambda)$ on the arcs of this jump contour are then:
\begin{equation}
\mathbf{W}_+(\lambda)=\mathbf{W}_-(\lambda)\begin{bmatrix}1 & 0\\
-s\ee^{2\ii M\vartheta(\lambda;\chi,\tau)} & 1\end{bmatrix},\quad\lambda\in C_L^+,
\label{eq:Tjump-channels-CLplus-ALT}
\end{equation}
\begin{equation}
\mathbf{W}_+(\lambda)=\mathbf{W}_-(\lambda)\begin{bmatrix}1 & \tfrac{1}{2}s\ee^{-2\ii M\vartheta(\lambda;\chi,\tau)} \\ 0 & 1\end{bmatrix},\quad\lambda\in C_R^+,
\label{eq:Tjump-channels-CRplus-ALT}
\end{equation}
\begin{equation}
\mathbf{W}_+(\lambda)=\mathbf{W}_-(\lambda)2^{\sigma_3},\quad\lambda\in I,
\label{eq:Channels-T-I-jump}
\end{equation}
\begin{equation}
\mathbf{W}_+(\lambda)=\mathbf{W}_-(\lambda)\begin{bmatrix}1 & 0\\
-\tfrac{1}{2}s\ee^{2\ii M\vartheta(\lambda;\chi,\tau)} & 1\end{bmatrix},\quad
\lambda\in C_R^-,\quad\text{and}
\label{eq:Tjump-channels-CRminus-ALT}
\end{equation}
\begin{equation}
\mathbf{W}_+(\lambda)=\mathbf{W}_-(\lambda)\begin{bmatrix} 1 & s\ee^{-2\ii M\vartheta(\lambda;\chi,\tau)}\\ 0 & 1\end{bmatrix},\quad\lambda\in C_L^-.
\label{eq:Tjump-channels-CLminus-ALT}
\end{equation}
It follows from the sign chart of $\mathrm{Re}(\ii\vartheta(\lambda;\chi,\tau))$ as shown in Figure~\ref{fig:Channels1} that as $n\to+\infty$, the jump matrix for $\mathbf{W}(\lambda)$ is an exponentially small perturbation of the identity everywhere on the jump contour except on the interval $I=[a,b]$ and in neighborhoods of its endpoints.  

\subsection{Parametrix construction}
\label{sec:channels-parametrix}
To deal with those jump matrices that are not near-identity, we first construct an \emph{outer parametrix}
$\dot{\mathbf{W}}^{\mathrm{out}}(\lambda)$ by setting
\begin{equation}
\dot{\mathbf{W}}^{\mathrm{out}}(\lambda)\defeq\left(\frac{\lambda-b(\chi,\tau)}{\lambda-a(\chi,\tau)}\right)^{-\ii p\sigma_3},\quad p\defeq\frac{\ln(2)}{2\pi}>0,\quad\lambda\in\mathbb{C}\setminus I.
\label{eq:Channels-Tout}
\end{equation}
Here, the power function is the principal branch, making $\dot{\mathbf{W}}^\mathrm{out}(\lambda)$ analytic in the indicated domain.  Furthermore it is clear that $\dot{\mathbf{W}}_+^\mathrm{out}(\lambda)=\dot{\mathbf{W}}_-^\mathrm{out}(\lambda)2^{\sigma_3}$ holds for $\lambda\in I$, so the jump condition in \eqref{eq:Channels-T-I-jump} is satisfied exactly by the outer parametrix, which also tends to the identity as $\lambda\to\infty$.  However, $\dot{\mathbf{W}}^\mathrm{out}(\lambda)$ is discontinuous near the endpoints of $I$, making the outer parametrix a poor model for $\mathbf{W}(\lambda)$ near these points.  

We can construct \emph{inner parametrices} near $\lambda=a,b$ that locally satisfy the jump conditions for $\mathbf{W}(\lambda)$ exactly.  Let $D_a(\delta)$ and $D_b(\delta)$ be disks of radius $\delta$ centered at $\lambda=a,b$ respectively, where $\delta>0$ is sufficiently small but independent of $n$.  We first define conformal coordinates $f_a(\lambda;\chi,\tau)$ and $f_b(\lambda;\chi,\tau)$ in these disks by setting 
\begin{equation}
f_a(\lambda;\chi,\tau)^2=2[\vartheta_a(\chi,\tau)-\vartheta(\lambda;\chi,\tau)]\quad\text{and}\quad
f_b(\lambda;\chi,\tau)^2=2[\vartheta(\lambda;\chi,\tau)-\vartheta_b(\chi,\tau)],
\end{equation}
where $\vartheta_a(\chi,\tau)\defeq\vartheta(a(\chi,\tau);\chi,\tau)$ and $\vartheta_b(\chi,\tau)\defeq\vartheta(b(\chi,\tau);\chi,\tau)$, 
and then taking analytic square roots in each case so that the inequalities $f_a'(a(\chi,\tau);\chi,\tau)<0$ and $f_b'(b(\chi,\tau);\chi,\tau)>0$ both hold.  This is possible because $a$ and $b$ are simple critical points of $\vartheta(\lambda;\chi,\tau)$, with $\vartheta''_a(\chi,\tau)\defeq\vartheta''(a(\chi,\tau);\chi,\tau)<0$ and $\vartheta''_b(\chi,\tau)\defeq\vartheta''(b(\chi,\tau);\chi,\tau)>0$.  In fact, one has the formul\ae\
\begin{equation}
f_a'(a(\chi,\tau);\chi,\tau)=-\sqrt{-\vartheta''_a(\chi,\tau)}\quad\text{and}\quad
f_b'(b(\chi,\tau);\chi,\tau)=\sqrt{\vartheta''_b(\chi,\tau)}.
\label{eq:Channels-fafb-Derivs}
\end{equation}
Next, define $M$-independent holomorphic matrix valued functions in $D_a(\delta)$ and $D_b(\delta)$ by
\begin{equation}
\mathbf{H}^a(\lambda)\defeq\left(\frac{f_a(\lambda;\chi,\tau)}{a(\chi,\tau)-\lambda}\right)^{-\ii p\sigma_3}(b(\chi,\tau)-\lambda)^{-\ii p\sigma_3}(\ii\sigma_2),\quad\lambda\in D_a(\delta)
\label{eq:Channels-Ha}
\end{equation}
and
\begin{equation}
\mathbf{H}^b(\lambda)\defeq\left(\frac{f_b(\lambda;\chi,\tau)}{\lambda-b(\chi,\tau)}\right)^{\ii p\sigma_3}(\lambda-a(\chi,\tau))^{\ii p\sigma_3},\quad\lambda\in D_b(\delta).  
\label{eq:Channels-Hb}
\end{equation}
Note that in both cases, the diagonal prefactor is an analytic function nonvanishing in the relevant disk for $\delta$ sufficiently small.  In particular,
\begin{equation}
\begin{split}
\mathbf{H}^a(a(\chi,\tau))&=\left(-f_a'(a(\chi,\tau);\chi,\tau)\right)^{-\ii p\sigma_3}(b(\chi,\tau)-a(\chi,\tau))^{-\ii p\sigma_3}(\ii\sigma_2)\\
&=\left(-\vartheta_a''(\chi,\tau)\right)^{-\frac{1}{2}\ii p\sigma_3}(b(\chi,\tau)-a(\chi,\tau))^{-\ii p\sigma_3}(\ii\sigma_2)
\end{split}
\label{eq:Channels-Ha-center}
\end{equation}
and
\begin{equation}
\begin{split}
\mathbf{H}^b(b(\chi,\tau))&=f_b'(b(\chi,\tau);\chi,\tau)^{\ii p\sigma_3}(b(\chi,\tau)-a(\chi,\tau))^{\ii p\sigma_3}\\
&=\vartheta''_b(\chi,\tau)^{\frac{1}{2}\ii p\sigma_3}(b(\chi,\tau)-a(\chi,\tau))^{\ii p\sigma_3},
\end{split}
\label{eq:Channels-Hb-center}
\end{equation}
where on the second line in each case we used \eqref{eq:Channels-fafb-Derivs}.
Letting 
$\zeta_{a,b}=M^{\frac{1}{2}}f_{a,b}(\lambda;\chi,\tau)$
denote rescalings of the conformal coordinates, 
we then define the inner parametrices by setting
\begin{equation}
\dot{\mathbf{W}}^a(\lambda)\defeq M^{-\frac{1}{2}\ii p\sigma_3}\ee^{-\ii M\vartheta_a(\chi,\tau)\sigma_3}\ii^{\frac{1}{2}(1-s)\sigma_3}\mathbf{H}^a(\lambda)
\mathbf{U}(\zeta_a)(\ii\sigma_2)^{-1}\ii^{-\frac{1}{2}(1-s)\sigma_3}\ee^{\ii M\vartheta_a(\chi,\tau)\sigma_3},\quad\lambda\in D_a(\delta)
\label{eq:Channels-Ta-ALT}
\end{equation}
and
\begin{equation}
\dot{\mathbf{W}}^b(\lambda)\defeq M^{\frac{1}{2}\ii p\sigma_3}\ee^{-\ii M\vartheta_b(\chi,\tau)\sigma_3}\ii^{\frac{1}{2}(1-s)\sigma_3}\mathbf{H}^b(\lambda)
\mathbf{U}(\zeta_b)\ii^{-\frac{1}{2}(1-s)\sigma_3}\ee^{\ii M\vartheta_b(\chi,\tau)\sigma_3},\quad\lambda\in D_b(\delta).
\label{eq:Channels-Tb-ALT}
\end{equation}
Here the factors to the left of $\mathbf{U}(\zeta_{a,b})$ in each case are analytic on the relevant disk and therefore have no effect on the jump conditions, and the matrix function $\mathbf{U}(\zeta)$ is defined in terms of parabolic cylinder functions as the solution of Riemann-Hilbert Problem 5 in \cite{BilmanLM20} (for example; a development of the solution of this problem is given in \cite[Appendix A]{Miller18} taking $\tau=1$ in the notation of that reference).  The main properties of $\mathbf{U}(\zeta)$ that we need to refer to here are
\begin{itemize}
\item $\mathbf{U}(\zeta)$ is analytic for $|\arg(\zeta)|<\tfrac{1}{4}\pi$, $\tfrac{1}{4}\pi<|\arg(\zeta)|<\tfrac{3}{4}\pi$, and $\tfrac{3}{4}\pi<|\arg(\zeta)|<\pi$ (five sectors);
\item $\mathbf{U}(\zeta)$ takes continuous boundary values from each of the five sectors related by jump conditions $\mathbf{U}_+(\zeta)=\mathbf{U}_-(\zeta)\mathbf{V}^\mathrm{PC}(\zeta)$, where $\mathbf{V}^\mathrm{PC}(\zeta)$ is defined in terms of the exponentials $\ee^{\pm\ii\zeta^2}$ on the five complementary oriented boundary rays as shown in \cite[Figure 9]{BilmanLM20};
\item $\mathbf{U}(\zeta)$ has uniform asymptotics in all directions of the complex plane given by
\begin{equation}
\mathbf{U}(\zeta)\zeta^{\ii p\sigma_3}=\mathbb{I}+\frac{1}{2\ii\zeta}\begin{bmatrix}0 & \alpha\\-\beta & 0\end{bmatrix}+\begin{bmatrix}O(\zeta^{-2}) & O(\zeta^{-3})\\O(\zeta^{-3}) & O(\zeta^{-2})\end{bmatrix},\quad\zeta\to\infty,
\label{eq:PCU-asymp}
\end{equation}
where
\begin{equation}
\alpha\defeq\frac{2^\frac{3}{4}\sqrt{2\pi}}{\Gamma(\ii p)}\ee^{\frac{1}{4}\ii\pi}\ee^{2\pi\ii p^2}
= \sqrt{\frac{\ln(2)}{\pi}}\ee^{\ii(\frac{1}{4}\pi+2\pi p^2-\arg(\Gamma(\ii p)))},\quad\beta\defeq -\alpha^*.
\label{eq:Channels-alpha-beta}
\end{equation}
\end{itemize}
In particular, the analyticity and jump conditions satisfied by $\mathbf{U}(\zeta)$ imply that the inner parametrices $\dot{\mathbf{W}}^a(\lambda)$ and $\dot{\mathbf{W}}^b(\lambda)$ exactly satisfy the jump conditions for $\mathbf{W}(\lambda)$ within their respective disks of definition (here we assume that the jump contours for $\mathbf{W}(\lambda)$ within each disk have been deformed to agree with preimages under $\lambda\mapsto \zeta_{a,b}$ of the straight rays across which $\mathbf{U}(\zeta)$ has jump discontinuities).  

A \emph{global parametrix} is then constructed from the outer and inner parametrices as follows:
\begin{equation}
\dot{\mathbf{W}}(\lambda)\defeq\begin{cases}
\dot{\mathbf{W}}^a(\lambda),&\quad\lambda\in D_a(\delta),\\
\dot{\mathbf{W}}^b(\lambda),&\quad\lambda\in D_b(\delta),\\
\dot{\mathbf{W}}^\mathrm{out}(\lambda),&\quad\lambda\in\mathbb{C}\setminus (I\cup D_a(\delta)\cup D_b(\delta)).
\end{cases}
\end{equation}

\subsection{Small norm problem for the error and large-$M$ expansion}
\label{sec:small-norm-channels}
We now compare the (unknown) matrix $\mathbf{W}(\lambda)=\mathbf{W}(\lambda;\chi,\tau,\mathbf{Q}^{-s},M)$ with its global parametrix by defining the \emph{error} as
\begin{equation}
\mathbf{F}(\lambda)\defeq\mathbf{W}(\lambda)\dot{\mathbf{W}}(\lambda)^{-1}.
\end{equation}
Since the parametrix is an exact solution of the Riemann-Hilbert jump conditions for $\mathbf{W}(\lambda)$ within the disks $D_{a,b}(\delta)$ and across the part of $I=[a,b]$ exterior to these disks, 
$\mathbf{F}(\lambda)$ can be extended to an analytic function of $\lambda\in\mathbb{C}$ with the exception of the arcs of $C_L^\pm$ and $C_R^\pm$ lying outside of the disks $D_{a,b}(\delta)$, and the boundaries $\partial D_{a,b}(\delta)$, which we take to have clockwise orientation.  Because $\delta$ is fixed 
as $M\to+\infty$,
and since $\dot{\mathbf{W}}^\mathrm{out}(\lambda)$ is independent of 
$M$,
there is a positive constant $\nu>0$ such that 
$\mathbf{F}_+(\lambda)=\mathbf{F}_-(\lambda)(\mathbb{I}+O(\ee^{-\nu M}))$ holds uniformly on the jump contour for $\mathbf{F}(\lambda)$ except on the circles $\partial D_{a,b}(\delta)$.  On the circles, we calculate the jump matrix for $\mathbf{F}(\lambda)$ as follows:
\begin{equation}
\mathbf{F}_+(\lambda)=\mathbf{F}_-(\lambda)\cdot\dot{\mathbf{W}}^{a,b}(\lambda)\dot{\mathbf{W}}^\mathrm{out}(\lambda)^{-1},\quad\lambda\in\partial D_{a,b}(\delta),
\end{equation}
because $\mathbf{W}(\lambda)$ is continuous across $\partial D_{a,b}(\delta)$.  Now we use the fact that by comparing the definition \eqref{eq:Channels-Tout} of the outer parametrix 
$\dot{\mathbf{W}}^\mathrm{out}(\lambda)$ with the definitions \eqref{eq:Channels-Ha}--\eqref{eq:Channels-Hb} of $\mathbf{H}^a(\lambda)$ and $\mathbf{H}^b(\lambda)$, we have
\begin{multline}
\dot{\mathbf{W}}^\mathrm{out}(\lambda)\ee^{-\ii M\vartheta_a(\chi,\tau)\sigma_3}\ii^{\frac{1}{2}(1-s)\sigma_3}(\ii\sigma_2)
=M^{-\frac{1}{2}\ii p\sigma_3}\ee^{-\ii M\vartheta_a(\chi,\tau)\sigma_3}\ii^{\frac{1}{2}(1-s)\sigma_3}\mathbf{H}^a(\lambda)\zeta_a^{-\ii p\sigma_3},\\
\lambda\in D_a(\delta)\setminus I
\end{multline}
and
\begin{equation}
\dot{\mathbf{W}}^\mathrm{out}(\lambda)\ee^{-\ii M\vartheta_b(\chi,\tau)\sigma_3}\ii^{\frac{1}{2}(1-s)\sigma_3}
=M^{\frac{1}{2}\ii p\sigma_3}\ee^{-\ii M\vartheta_b(\chi,\tau)\sigma_3}\ii^{\frac{1}{2}(1-s)\sigma_3}\mathbf{H}^b(\lambda)\zeta_b^{-\ii p\sigma_3},\quad
\lambda\in D_b(\delta)\setminus I.
\end{equation}
Therefore, using \eqref{eq:Channels-Ta-ALT} and \eqref{eq:PCU-asymp} and the fact that 
$\zeta_a=M^\frac{1}{2}f_a(\lambda;\chi,\tau)$
while $f_a(\lambda;\chi,\tau)$ is bounded away from zero on $\partial D_a(\delta)$ for $\delta$ sufficiently small independent of 
$M$,
\begin{multline}
\mathbf{F}_+(\lambda)=\mathbf{F}_-(\lambda)
M^{-\frac{1}{2}\ii p\sigma_3}\ee^{-\ii M\vartheta_a(\chi,\tau)\sigma_3}\ii^{\frac{1}{2}(1-s)\sigma_3}\mathbf{H}^a(\lambda)\\
\cdot 
\left(\mathbb{I}+\frac{1}{2\ii M^{\frac{1}{2}}f_a(\lambda;\chi,\tau)}\begin{bmatrix}0 & \alpha\\-\beta & 0\end{bmatrix}+\begin{bmatrix}O(M^{-1}) & O(M^{-\frac{3}{2}})\\O(M^{-\frac{3}{2}}) & O(M^{-1})\end{bmatrix}\right)\\
\cdot\mathbf{H}^a(\lambda)^{-1}\ii^{-\frac{1}{2}(1-s)\sigma_3}\ee^{\ii M\vartheta_a(\chi,\tau)\sigma_3}M^{\frac{1}{2}\ii p\sigma_3},\quad\lambda\in\partial D_a(\delta).
\label{eq:Channels-VF-partialDa-ALT}
\end{multline}
Likewise, using \eqref{eq:Channels-Tb-ALT} and the fact that 
$\zeta_b=M^\frac{1}{2}f_b(\lambda;\chi,\tau)$
with $f_b(\lambda;\chi,\tau)$ bounded away from zero on $\partial D_b(\delta)$, 
\begin{multline}
\mathbf{F}_+(\lambda)=\mathbf{F}_-(\lambda)
M^{\frac{1}{2}\ii p\sigma_3}\ee^{-\ii M\vartheta_b(\chi,\tau)\sigma_3}\ii^{\frac{1}{2}(1-s)\sigma_3}\mathbf{H}^b(\lambda)\\
\cdot 
\left(\mathbb{I}+\frac{1}{2\ii M^{\frac{1}{2}}f_b(\lambda;\chi,\tau)}\begin{bmatrix}0 & \alpha\\-\beta & 0\end{bmatrix}+\begin{bmatrix}O(M^{-1}) & O(M^{-\frac{3}{2}})\\O(M^{-\frac{3}{2}}) & O(M^{-1})\end{bmatrix}\right)\\
\cdot\mathbf{H}^b(\lambda)^{-1}\ii^{-\frac{1}{2}(1-s)\sigma_3}\ee^{\ii M\vartheta_b(\chi,\tau)\sigma_3}M^{-\frac{1}{2}\ii p\sigma_3},\quad\lambda\in\partial D_b(\delta).
\label{eq:Channels-VF-partialDb-ALT}
\end{multline}
In particular, it follows that 
$\mathbf{F}_+(\lambda)=\mathbf{F}_-(\lambda)(\mathbb{I}+O(M^{-\frac{1}{2}}))$ 
holds uniformly on the compact jump contour for $\mathbf{F}(\lambda)$, which otherwise is analytic and tends to $\mathbb{I}$ as $\lambda\to\infty$.  By small-norm theory for such Riemann-Hilbert problems, it follows that 
$\mathbf{F}_-(\cdot)=\mathbb{I}+O(M^{-\frac{1}{2}})$ 
holds in the $L^2$ sense on the jump contour, in the limit 
$M\to+\infty$.

\subsection{Asymptotic formula for 
$q(M\chi,M\tau;\mathbf{Q}^{-s},M)$
for $(\chi,\tau)\in \channels$}
Beginning with \eqref{eq:q-S} and using the facts that $\mathbf{S}(\lambda;\chi,\tau,\mathbf{Q}^{-s},M)=\mathbf{W}(\lambda)=\mathbf{W}(\lambda;\chi,\tau,\mathbf{Q}^{-s},M)$ and $\dot{\mathbf{W}}(\lambda)=\dot{\mathbf{W}}^\mathrm{out}(\lambda)$ both hold for $|\lambda|$ sufficiently large, we obtain the exact formula
\begin{equation}
\begin{split}
q(M\chi,M\tau;\mathbf{Q}^{-s},M)&=2\ii\lim_{\lambda\to\infty}\lambda W_{12}(\lambda)\\
&=
2\ii\lim_{\lambda\to\infty}\lambda\left[F_{11}(\lambda)\dot{W}^\mathrm{out}_{12}(\lambda)+F_{12}(\lambda)\dot{W}^\mathrm{out}_{22}(\lambda)\right].
\end{split}
\label{eq:psi-k-exact-channels-ALT}
\end{equation}
Since $\dot{\mathbf{W}}^\mathrm{out}(\lambda)$ is a diagonal matrix tending to $\mathbb{I}$ as $\lambda\to\infty$, this formula simplifies to
\begin{equation}
q(M\chi,M\tau;\mathbf{Q}^{-s},M)=
2\ii\lim_{\lambda\to\infty}\lambda F_{12}(\lambda).
\end{equation}
If $\mathbf{V}^\mathbf{F}(\lambda)$ denotes the jump matrix for $\mathbf{F}(\lambda)$, i.e., $\mathbf{F}_+(\lambda)=\mathbf{F}_-(\lambda)\mathbf{V}^\mathbf{F}(\lambda)$ holds on the jump contour $\Sigma_\mathbf{F}$, then it follows from the Plemelj formula that
\begin{equation}
\mathbf{F}(\lambda)=\mathbb{I}+\frac{1}{2\pi\ii}\int_{\Sigma_\mathbf{F}}\frac{\mathbf{F}_-(\eta)(\mathbf{V}^\mathbf{F}(\eta)-\mathbb{I})}{\eta-\lambda}\,\dd\eta,\quad
\lambda\in\mathbb{C}\setminus\Sigma_\mathbf{F},
\label{eq:F-Cauchy-channels}
\end{equation}
and therefore
\begin{equation}
q(M\chi,M\tau;\mathbf{Q}^{-s},M)=-\frac{1}{\pi}\int_{\Sigma_\mathbf{F}}\left[F_{11-}(\eta)V^\mathbf{F}_{12}(\eta)+F_{12-}(\eta)(V^\mathbf{F}_{22}(\eta)-1)\right]\,\dd\eta.
\end{equation}
Since 
$V^\mathbf{F}_{22}(\cdot)-1=O(M^{-1})$
holds uniformly on $\Sigma_\mathbf{F}$, as $\Sigma_\mathbf{F}$ is compact we also have 
$V^\mathbf{F}_{22}(\cdot)-1=O(M^{-1})$
in $L^2(\Sigma_\mathbf{F})$.  Using that 
$F_{12-}(\cdot)=O(M^{-\frac{1}{2}})$
in $L^2(\Sigma_\mathbf{F})$ as well, by Cauchy-Schwarz,
\begin{equation}
q(M\chi,M\tau;\mathbf{Q}^{-s},M)=-\frac{1}{\pi}\int_{\Sigma_\mathbf{F}}F_{11-}(\eta)V^\mathbf{F}_{12}(\eta)\,\dd\eta + O(M^{-\frac{3}{2}}).
\end{equation}
A similar argument allows $F_{11-}(\eta)$ to be replaced with $1$ at the cost of an error term of the same order.  Indeed, taking a boundary value on $\Sigma_\mathbf{F}$ in \eqref{eq:F-Cauchy-channels} gives for $\varphi(\lambda)\defeq F_{11-}(\lambda)-1$ the integral equation
\begin{equation}
\varphi(\lambda)-\frac{1}{2\pi\ii}\int_{\Sigma_\mathbf{F}}\frac{\varphi(\eta)(V^\mathbf{F}_{11}(\eta)-1)}{\eta-\lambda_-}\,\dd\eta = f(\lambda),\quad\lambda\in\Sigma_\mathbf{F},
\label{eq:phi-integral-equation}
\end{equation}
where 
\begin{equation}
f(\lambda)\defeq \frac{1}{2\pi\ii}\int_{\Sigma_\mathbf{F}}\frac{V_{11}^\mathbf{F}(\eta)-1}{\eta-\lambda_-}\,\dd\eta + \frac{1}{2\pi\ii}\int_{\Sigma_\mathbf{F}}\frac{F_{12-}(\eta)V_{21}^\mathbf{F}(\eta)}{\eta-\lambda_-}\,\dd\eta,\quad\lambda\in\Sigma_\mathbf{F}.
\label{eq:phi-integral-equation-RHS}
\end{equation}
The small-norm theory is fundamentally based on the fact that the Cauchy projection operator 
\begin{equation}
m(\lambda)\mapsto\frac{1}{2\pi\ii}\int_{\Sigma_\mathbf{F}}\frac{m(\eta)\,\dd\eta}{\eta-\lambda_-},\quad\lambda\in\Sigma_\mathbf{F}
\end{equation}
is bounded on $L^2(\Sigma_\mathbf{F})$ with norm depending only on the geometry of the contour $\Sigma_\mathbf{F}$, which is independent of any large parameter.  Since 
$V_{11}^\mathbf{F}(\cdot)-1=O(M^{-1})$
in $L^\infty(\Sigma_\mathbf{F})$ it follows easily from \eqref{eq:phi-integral-equation} that $\varphi(\cdot)=O(f(\cdot))$ in $L^2(\Sigma_\mathbf{F})$ as 
$M\to\infty$.
Likewise, from \eqref{eq:phi-integral-equation-RHS} we see that $f(\cdot)=O(V^\mathbf{F}_{11}(\cdot)-1)+O(F_{12-}(\cdot)V_{21}^\mathbf{F}(\cdot))$ in $L^2(\Sigma_\mathbf{F})$.  Since 
$V^\mathbf{F}_{11}(\cdot)-1=O(M^{-1})$
in $L^\infty(\Sigma_\mathbf{F})$, compactness of $\Sigma_\mathbf{F}$ implies that 
$V^\mathbf{F}_{11}(\cdot)-1=O(M^{-1})$
in $L^2(\Sigma_\mathbf{F})$.  Also, since 
$V_{21}^\mathbf{F}(\cdot)=O(M^{-\frac{1}{2}})$
in $L^\infty(\Sigma_\mathbf{F})$ while 
$F_{12-}(\cdot)=O(M^{-\frac{1}{2}})$
in $L^2(\Sigma_\mathbf{F})$, we consequently have 
$F_{12-}(\cdot)V_{21}^\mathbf{F}(\cdot)=O(M^{-1})$
in $L^2(\Sigma_\mathbf{F})$ as well.  Therefore 
$\varphi(\cdot)=F_{11-}(\cdot)-1=O(M^{-1})$
in $L^2(\Sigma_\mathbf{F})$.  As 
$V_{12}^\mathbf{F}(\cdot)=O(M^{-\frac{1}{2}})$
in $L^\infty(\Sigma_\mathbf{F})$ and hence also in $L^2(\Sigma_\mathbf{F})$ it then follows by Cauchy-Schwarz that
\begin{equation}
q(M\chi,M\tau;\mathbf{Q}^{-s},M)=-\frac{1}{\pi}\int_{\Sigma_\mathbf{F}}V^\mathbf{F}_{12}(\eta)\,\dd\eta + O(M^{-\frac{3}{2}}).
\end{equation}
The dominant contribution to the integral comes from $\partial D_a(\delta)\cup\partial D_b(\delta)$ where $V^\mathbf{F}_{12}(\cdot)$ is proportional to 
$M^{-\frac{1}{2}}$,
while contributions from the rest of $\Sigma_\mathbf{F}$ are uniformly exponentially small.  Therefore, we may modify the integration contour to consist of just two small circles:
\begin{equation}
q(M\chi,M\tau;\mathbf{Q}^{-s},M)=-\frac{1}{\pi}\int_{\partial D_a(\delta)\cup\partial D_b(\delta)}V^\mathbf{F}_{12}(\eta)\,\dd\eta + O(M^{-\frac{3}{2}}).
\end{equation}
Now, using the jump conditions \eqref{eq:Channels-VF-partialDa-ALT}--\eqref{eq:Channels-VF-partialDb-ALT} and the fact that $\mathbf{H}^a(\cdot)$ is off-diagonal while $\mathbf{H}^b(\cdot)$ is diagonal, one easily finds that
\begin{equation}
V_{12}^\mathbf{F}(\eta)=\frac{M^{-\ii p}\ee^{-2\ii M\vartheta_a(\chi,\tau)}(-1)^{\frac{1}{2}(1-s)}}{2\ii M^\frac{1}{2}f_a(\eta;\chi,\tau)}
\beta H^a_{12}(\eta)^2 + O(M^{-\frac{3}{2}}),\quad
\eta\in\partial D_a(\delta),
\end{equation}
\begin{equation}
V_{12}^\mathbf{F}(\eta)=\frac{M^{\ii p}\ee^{-2\ii M\vartheta_b(\chi,\tau)}(-1)^{\frac{1}{2}(1-s)}}{2\ii M^\frac{1}{2}f_b(\eta;\chi,\tau)}
\alpha H^b_{11}(\eta)^2 + O(M^{-\frac{3}{2}}),\quad
\eta\in\partial D_b(\delta).
\end{equation}
Therefore, since $f_{a,b}(\cdot;\chi,\tau)$ are analytic functions with simple zeros at $a$ and $b$ respectively, a residue calculation gives
\begin{multline}
q(M\chi,M\tau;\mathbf{Q}^{-s},M)=\frac{(-1)^{\frac{1}{2}(1-s)}}{M^\frac{1}{2}}\left[M^{-\ii p}\ee^{-2\ii M\vartheta_a(\chi,\tau)}\frac{\beta H_{12}^a(a(\chi,\tau))^2}{f_a'(a(\chi,\tau);\chi,\tau)}\right.\\
\left. {}+M^{\ii p}\ee^{-2\ii M\vartheta_b(\chi,\tau)}\frac{\alpha H_{11}^b(b(\chi,\tau))^2}{f_b'(b(\chi,\tau);\chi,\tau)}\right]+O(M^{-\frac{3}{2}}).
\end{multline}
Since $s=\pm 1$,
we then use \eqref{eq:Channels-fafb-Derivs}, \eqref{eq:Channels-Ha-center}--\eqref{eq:Channels-Hb-center}, and \eqref{eq:Channels-alpha-beta} to obtain
\begin{multline}
q(M\chi,M\tau;\mathbf{Q}^{-s},M)=\frac{s}{M^\frac{1}{2}}\sqrt{\frac{\ln(2)}{\pi}}
\left[
\ee^{\ii\phi}
\frac{\ee^{-2\ii M\vartheta_a(\chi,\tau)}(-\vartheta''_a(\chi,\tau))^{-\ii p}}{(-\vartheta''_a(\chi,\tau))^\frac{1}{2}}\right.\\
\left.{}+\ee^{-\ii\phi}
\frac{\ee^{-2\ii M\vartheta_b(\chi,\tau)}\vartheta''_b(\chi,\tau)^{\ii p}}{\vartheta''_b(\chi,\tau)^\frac{1}{2}}\right]+O(M^{-\frac{3}{2}}),
\end{multline}
where, recalling the value of $p$ from \eqref{eq:Channels-Tout}, a real angle $\phi$ is defined by 
\begin{equation}
\phi\defeq -\frac{\ln(2)}{2\pi}\ln(M)-\frac{\ln(2)}{\pi}\ln(b(\chi,\tau)-a(\chi,\tau))-\frac{\ln(2)^2}{2\pi}-\frac{1}{4}\pi+\arg\left(\Gamma\left(\frac{\ii\ln(2)}{2\pi}\right)\right).
\label{eq:phi-def}
\end{equation}
We may further observe that the numerator of each of the fractions in square brackets above has unit modulus, so upon identifying the angles of those phase factors
the proof of Theorem~\ref{thm:channels} is complete, with a standard argument to supply the local uniformity of the error estimate for $(\chi,\tau)$ in compact subsets of $\channels$ (which can include points on the positive $\chi$-axis).

\subsection{Simplification for $\tau=0$}
%
The further simplification mentioned at the end of Section~\ref{sec:results-channels}, so that $(\chi,\tau)\in \channels$ with $\tau=0$ means $0<\chi<2$, is accomplished by noting that the phase function $\vartheta(\lambda;\chi,0)$ defined by \eqref{eq:vartheta} is an odd function of $\lambda$ for each $\chi\in (0,2)$, and we recall that the critical points $\lambda=a,b$ in this case are given by \eqref{eq:tau-zero-critical-points}: 
\begin{equation}
b(\chi,0)=\sqrt{\frac{2}{\chi}-1},\quad a(\chi,0)=-b(\chi,0).
\end{equation}
A computation then shows that
\begin{equation}
\vartheta_b(\chi,0)=\vartheta(b(\chi,0);\chi,0)=\chi\sqrt{\frac{2}{\chi}-1}+\pi-2\tan^{-1}\left(\sqrt{\frac{2}{\chi}-1}\right),\quad\vartheta_a(\chi,0)=-\vartheta_b(\chi,0),
\end{equation}
and that
\begin{equation}
\vartheta_b''(\chi,0)=\vartheta''(b(\chi,0);\chi,0)=\chi^2\sqrt{\frac{2}{\chi}-1},\quad\vartheta_a''(\chi,0)=-\vartheta_b''(\chi,0).
\end{equation}
Therefore, in this special case, the leading term denoted $L^{[\channels]}_k(\chi,\tau)$ in \eqref{eq:leading-term-channels} reduces for $\tau=0$ and $0<\chi<2$ to \eqref{eq:leading-term-channels-tau-zero}.

\section{Properties of $h(\lambda;\chi,\tau)$ for $(\chi,\tau)\in\exterior\cup\shelves$}
\label{sec:GenusZeroModification}

\subsection{Unique determination of $h'(\lambda;\chi,\tau)$ for $(\chi,\tau)\in\exterior\cup\shelves$}
\label{sec:g-function}
Here we show how $(\chi,\tau)\in\overline{\exterior\cup\shelves}$ determines a unique function $h'(\lambda;\chi,\tau)$ of the form \eqref{eq:hprime-formula} that satisfies the residue and asymptotic conditions \eqref{eq:hprime-residues} and \eqref{eq:hprime-expansion} respectively.  
 
We first use \eqref{eq:hprime-expansion} with \eqref{eq:hprime-formula} to explicitly eliminate $A(\chi,\tau)$ and $B(\chi,\tau)^2$ in favor of $u(\chi,\tau)$ and $v(\chi,\tau)$:
\begin{equation}
\begin{split}
A(\chi,\tau)&=\frac{u(\chi,\tau)-\chi}{2\tau},\\
A(\chi,\tau)^2+B(\chi,\tau)^2&=\frac{3u(\chi,\tau)^2}{4\tau^2}-\frac{v(\chi,\tau)}{\tau}+2-\frac{\chi u(\chi,\tau)}{\tau^2}+\frac{\chi^2}{4\tau^2}.
\end{split}
\label{eq:eliminate-AB}
\end{equation}
Then, the residue conditions \eqref{eq:hprime-residues} become $(-2\tau\pm\ii u(\chi,\tau)+v(\chi,\tau))R(\pm\ii;\chi,\tau)=-2$.  Imposing instead the \emph{squares} of these conditions\footnote{Later, getting the signs right for the residues is accomplished by choosing the location of the branch cut $\Sigma_g$ in relation to the points $\lambda=\pm\ii$.  See Remark~\ref{rem:Sigma_g}.} one arrives at two complex-conjugate equations, which amount to two real equations by taking real and imaginary parts.  The real part equation reads $\mathcal{R}=0$, where
\begin{multline}
\mathcal{R}\defeq 3u^4-3u^2v^2-4\chi u^3+4\tau v^3+4\chi uv^2+\chi^2u^2-\chi^2v^2-8\chi\tau uv+8\tau^2u^2-20\tau^2v^2\\{}+4\chi^2\tau v+32\tau^3v-16\tau^4+16\tau^2-4\chi^2\tau^2,
\end{multline}
and the imaginary part equation reads $\mathcal{I}_1\mathcal{I}_2=0$, where
\begin{equation}
\mathcal{I}_1\defeq u^2-\chi u-2\tau v+4\tau^2\quad\text{and}\quad
\mathcal{I}_2\defeq 3uv-4\tau u-\chi v+2\chi\tau.
\end{equation}
Note that for $\tau=0$,
\begin{equation}
\left.\mathcal{R}\right|_{\tau=0}=(u+v)(u-v)(3u-\chi)(u-\chi),\quad
\left.\mathcal{I}_1\right|_{\tau=0}=(u-\chi)u,\quad\text{and}\quad
\left.\mathcal{I}_2\right|_{\tau=0}=(3u-\chi)v,
\end{equation}
so one solution is to choose $u(\chi,0)=\chi$ and $v(\chi,0)=0$.  In order to apply the implicit function theorem to continue this solution to $\tau\neq 0$, it is necessary to discard the factor $\mathcal{I}_1$ and enforce only the conditions $\mathcal{R}=0$ and $\mathcal{I}_2=0$.  Then a calculation shows that the Jacobian is
\begin{equation}
\left.\det\begin{bmatrix} \mathcal{R}_u & \mathcal{R}_v\\
\mathcal{I}_{2u} & \mathcal{I}_{2v}\end{bmatrix}\right|_{\tau=0,u=\chi,v=0}=4\chi^4,
\end{equation}
which is nonzero for $\chi>2$.  Moreover, the equation $\mathcal{I}_2=0$ can be used to explicitly eliminate $v$ by
\begin{equation}
\mathcal{I}_2=0\quad\Leftrightarrow\quad v=2\tau\frac{2u-\chi}{3u-\chi}.
\label{eq:eliminate-v}
\end{equation}
With $v$ eliminated, the equation $\mathcal{R}=0$ reads $P(u;\chi,\tau)=0$, where $P(u;\chi,\tau)$ is the septic polynomial
\begin{multline}
P(u;\chi,\tau)\defeq 81u^7-189\chi u^6 + (162\chi^2+72\tau^2)u^5-(66\chi^2+120\tau^2)\chi u^4 \\
{}+ (13\chi^4+56\chi^2\tau^2+16\tau^4+432\tau^2)u^3 -(\chi^4+8\chi^2\tau^2+16\tau^4+432\tau^2)\chi u^2 \\
{}+ 144\chi^2\tau^2u-16\chi^3\tau^2.
\end{multline}
Note that $P(u;\chi,0)=(u-\chi)(3u-\chi)^4u^2$, so if $\chi>0$, $u(\chi,0)=\chi$ is a simple root hence continuable to $\tau>0$ (sufficiently small, given $\chi>0$) by the implicit function theorem.  In the limit $\tau\downarrow 0$ we can compute as many terms in the Taylor expansion of $u(\chi,\tau)$ about $u(\chi,0)=\chi$ as we like; in particular it is easy to see that
\begin{equation}
u(\chi,\tau)=\chi -\frac{8\tau^2}{\chi^3} + O(\tau^4),\quad\tau\downarrow 0,\quad \chi>2,
\end{equation}
which implies via \eqref{eq:eliminate-v} that 
\begin{equation}
v(\chi,\tau)=\tau-\frac{4\tau^3}{\chi^4}+O(\tau^5),\quad\tau\downarrow 0,\quad\chi>2.
\end{equation}
From \eqref{eq:eliminate-AB} we then also find that
\begin{equation}
A(\chi,\tau)=-\frac{4\tau}{\chi^3}+O(\tau^3)\quad\text{and}\quad
B(\chi,\tau)^2=1-\frac{4}{\chi^2}+O(\tau^2),\quad\tau\downarrow 0,\quad \chi>2.
\label{eq:AB-tau-small}
\end{equation}
In the special case that $\tau=0$ and $\chi>2$, it follows that $u(\chi,0)=\chi$, $v(\chi,0)=0$, $A(\chi,0)=0$ and $B(\chi,0)^2=1-4/\chi^2<1$.
We claim that this solution can be uniquely continued not just locally near $\tau=0$ but also to the entire unbounded exterior region $\exterior$ as well as through its common boundary with the bounded region $\shelves$ into that entire region.  We have the following result, the proof of which can be found in Appendix~\ref{A:Proofs}.
\begin{proposition}
Let $(\chi,\tau)\in \overline{\exterior\cup\shelves}$.  Then $P(u;\chi,\tau)$ has a unique real root of odd multiplicity, denoted $u=u(\chi,\tau)$ with $u(0,\tau)=0$ for $\tau>0$ and $u(\chi,0)=\chi$ for $\chi>2$.  There exists a value $\tau_1>0$ such that except for $\chi=0$ and possibly three or fewer points $(\chi,\tau)$ with $\chi>0$ and $\tau=\tau_1$, $u(\chi,\tau)$ is the only real root of $P(u;\chi,\tau)$ and it is simple.  
\label{prop:u}
\end{proposition}

\begin{remark}
In the case $\tau=0$ and $\chi>2$, $\Sigma_g$ is an arc connecting the two points $\lambda=\pm\ii B(\chi,0)$.  If we take $\Sigma_g$ to be the purely imaginary straight-line segment connecting these points, then from the prescribed large-$\lambda$ asymptotic behavior of $R(\lambda;\chi,0)$ we find that $R(\pm\ii;\chi,0)=\pm 2\ii/\chi$, from which it follows directly via the formula \eqref{eq:hprime-formula} that the residue conditions \eqref{eq:hprime-residues} hold, so the signs of the residues which had been conflated in squaring the residue conditions are indeed correctly resolved with the indicated choice of $\Sigma_g$.  To ensure that this successful resolution is maintained upon continuation of the solution from $\tau=0$ it is then sufficient that $\Sigma_g$ deform continuously with $(\chi,\tau)$ without ever contacting the poles $\lambda=\pm\ii$.  This is feasible because the endpoints $A(\chi,\tau)\pm\ii B(\chi,\tau)$ lie in the left half-plane for all $(\chi,\tau)\in\overline{\exterior\cup\shelves}$ with $\chi>0$ and $\tau>0$; indeed from \eqref{eq:eliminate-AB}, $A(\chi,\tau)=0$ holds if and only if $u(\chi,\tau)=\chi$ and $P(\chi;\chi,\tau)=128\tau^2\chi^3$ which vanishes only on the coordinate axes.  Therefore $A(\chi,\tau)$ has one sign on the interior of $\overline{\exterior\cup\shelves}$, and by \eqref{eq:AB-tau-small} we see that $A(\chi,\tau)<0$.  Note that when $\chi\downarrow 0$ for given $\tau>0$, the proof of Proposition~\ref{prop:u} shows that $B(\chi,\tau)^2$ tends to a value strictly greater than $1$ while $A(\chi,\tau)\to 0$, so in this limiting situation we should choose $\Sigma_g$ to lie in the left half-plane except for its endpoints.
\label{rem:Sigma_g}
\end{remark}

We are now in a position to show that, as claimed in Section~\ref{sec:h-intro}, $B(\chi,\tau)^2>0$ holds for all $(\chi,\tau)\in\exterior\cup\shelves\cup(\partial\exterior\cap\partial\shelves)$ and $B(\chi,\tau)^2=0$ holds on the common boundary of the union with $\channels$. 
Indeed,
expressing $B^2$ explicitly in terms of $u$, $\chi$, and $\tau$ using \eqref{eq:eliminate-AB} and \eqref{eq:eliminate-v}, one finds that $B^2=0$ implies $3u^3-4\chi u^2+(4\tau^2+\chi^2) u=0$.  The resultant between this equation and $P(u;\chi,\tau)$ vanishes on the open quadrant $(\chi,\tau)\in\mathbb{R}_{>0}\times\mathbb{R}_{>0}$ exactly where \eqref{eq:boundary-curve} holds.

\subsubsection{Critical points of $h'(\lambda;\chi,\tau)$ for $(\chi,\tau)\in\exterior\cup\shelves$}
\label{sec:critical-points}
Observe that while the coefficients $u(\chi,\tau)$ and $v(\chi,\tau)$ in the quadratic factor in the numerator of $h'(\lambda;\chi,\tau)$ defined in \eqref{eq:hprime-formula} depend real-analytically on $(\chi,\tau)\in (\mathbb{R}_{\ge 0}\times\mathbb{R}_{\ge 0})\setminus\overline{\channels}$, the quadratic discriminant vanishes to first order along two curves in this domain so the roots undergo bifurcation upon crossing these curves.  Eliminating $v$ via \eqref{eq:eliminate-v}, the quadratic discriminant $u^2-8\tau v$ is seen to vanish only if $3u^3-\chi u^2-32\tau^2 u+16\chi\tau^2=0$.  The resultant of this cubic polynomial with $P(u;\chi,\tau)$ vanishes for $(\chi,\tau)\in\mathbb{R}_{>0}\times\mathbb{R}_{>0}$ exactly where $D(\chi,\tau)\defeq H_{10}(\chi,\tau)+H_8(\chi,\tau)+H_6(\chi,\tau) + H_4(\tau)=0$, in which the $H_j$ are homogeneous polynomials
\begin{equation}
\begin{split}
H_{10}(\chi,\tau)&\defeq -(8\tau^2-\chi^2)(100\tau^2+\chi^2)^4,\\
H_8(\chi,\tau)&\defeq 2\chi^8+1040\chi^6\tau^2+1741728\chi^4\tau^4-125516800\chi^2\tau^6+730880000\tau^8,\\
H_6(\chi,\tau)&\defeq \chi^6+504\chi^4\tau^2+3103488\chi^2\tau^4+67627008\tau^6,\\
H_4(\tau)&\defeq 1492992\tau^4.
\end{split}
\label{eq:H-polynomials}
\end{equation}
There is one unbounded arc in the first quadrant where this condition holds (see the dotted blue curve in Figure~\ref{fig:RegionsPlot}) and it is governed far from the origin by the highest-order homogeneous terms $H_{10}(\chi,\tau)\approx 0$; this arc is therefore asymptotic to the line $\chi=\sqrt{8}\tau$ (see the dotted gray line in Figure~\ref{fig:RegionsPlot}).  The cusp point $(\chi,\tau)=(\chi^\sharp,\tau^\sharp)$ (see \eqref{eq:corner-point}) on the boundary of $\channels$ is a non-smooth point on the locus $D(\chi,\tau)=0$ because the gradient vector vanishes there as well.  
In fact, setting $(\chi,\tau)=(\chi^\sharp+\Delta\chi,\tau^\sharp+\Delta\tau)$, one computes that
\begin{equation}
D(\chi,\tau)=\frac{3948901875}{256}\left(\frac{2}{\sqrt{3}}\Delta \tau-\Delta \chi\right)^2 + O((\Delta\chi^2+\Delta\tau^2)^\frac{3}{2}).
\end{equation}
The leading terms describe two curves tangent to the line $\Delta\chi=\tfrac{2}{\sqrt{3}}\Delta\tau$, and along this line the cubic correction terms are proportional to $\Delta\tau^3$ by a negative coefficient.  Therefore the two curves both emanate from the cusp point $(\Delta\chi,\Delta\tau)=(0,0)$ along this tangent line in the direction $\Delta\tau>0$, entering the exterior of $\channels$ from the cusp point.  
Since $D(0,\tau)=2048(1-\tau^2)\tau^4(625\tau^2+27)^2$, an arc along which this condition holds exits the quadrant $(\chi,\tau)\in \mathbb{R}_{>0}\times\mathbb{R}_{>0}$ on the $\tau$-axis at the point $(\chi,\tau)=(0,1)$.  This is the arc separating $\shelves$ from $\exterior$, and is shown as a solid blue curve in Figure~\ref{fig:RegionsPlot}.


\subsubsection{Construction of $g(\lambda;\chi,\tau)$ when $\Sigma_g\cap\Sigma_\mathrm{c}=\emptyset$}
\label{sec:g-function-loop}
Since $h'(\lambda;\chi,\tau)$ is now well-defined for all $(\chi,\tau)\in \exterior\cup\shelves$, we have $g'(\lambda;\chi,\tau)=h'(\lambda;\chi,\tau)-\vartheta'(\lambda;\chi,\tau)$ which has removable singularities at $\lambda=\pm\ii$ according to \eqref{eq:hprime-residues}
and hence is an analytic function for $\lambda\in\mathbb{C}\setminus\Sigma_g$ with, according to \eqref{eq:hprime-expansion}, the asymptotic behavior $g'(\lambda;\chi,\tau)=O(\lambda^{-2})$ as $\lambda\to\infty$.  Since $g'(\lambda;\chi,\tau)$ is integrable at $\lambda=\infty$, the contour integral
\begin{equation}
g(\lambda;\chi,\tau)\defeq\int_\infty^\lambda g'(\eta;\chi,\tau)\,\dd\eta
\label{eq:g-integral}
\end{equation}
is independent of path in the domain $\mathbb{C}\setminus\Sigma_g$ and defines the unique antiderivative analytic in the same domain that satisfies the condition $g(\infty;\chi,\tau)=0$.  It is easy to check that $g(\lambda^*;\chi,\tau)=g(\lambda;\chi,\tau)^*$ holds for each $\lambda\in\Sigma_g$ and $(\chi,\tau)\in \exterior\cup\shelves$.  Obtaining $g(\lambda;\chi,\tau)$ from \eqref{eq:g-integral} is a bit of a subtle calculation, because the integrability at $\eta=\infty$ and $\eta=\pm\ii$ relies on cancellations arising from the equations satisfied by the parameters $u,v,A,B$.  Another approach is to assume that $\Sigma_g$ is determined by solving those equations, and then to note that $g(\lambda;\chi,\tau)$ is a function analytic for $\lambda\in\mathbb{C}\setminus\Sigma_g$ with $g(\lambda;\chi,\tau)=O(\lambda^{-1})$ as $\lambda\to\infty$ and whose boundary values on $\Sigma_g$ satisfy 
\begin{equation}
\begin{split}
g_+(\lambda;\chi,\tau)+g_-(\lambda;\chi,\tau)&=h_+(\lambda;\chi,\tau)+h_-(\lambda;\chi,\tau)-2\vartheta(\lambda;\chi,\tau)\\
&=2 \kappa(\chi,\tau)-2\vartheta(\lambda;\chi,\tau),\quad \lambda\in\Sigma_g,
\end{split}
\label{eq:hpm-kappa}
\end{equation}
for some integration constant $\kappa(\chi,\tau)$ (because the sum of boundary values of $h$ is constant along $\Sigma_g$).  This formula \eqref{eq:hpm-kappa} assumes that $\vartheta(\lambda;\chi,\tau)$ is analytic on the subset $\Sigma_g$ of the jump contour for $\mathbf{S}(\lambda;\chi,\tau,\mathbf{G},M)$.  As the jump contour for $\vartheta(\lambda;\chi,\tau)$ is $\Sigma_\mathrm{c}$, we are assuming that the latter is contained in the interior of the Jordan curve $\Sigma_\circ$, which guarantees that $\Sigma_g\cap\Sigma_\mathrm{c}=\emptyset$.  Another situation, in which $\Sigma_\circ$ is deformed into a dumbbell-shaped jump contour with a Schwarz-symmetric neck that is necessarily a subset of $\Sigma_\mathrm{c}$, will be required to prove Theorem~\ref{thm:exterior}.  We will describe how the procedure needs to be modified for that case in Section~\ref{sec:g-function-dumbbell} below.

Returning to \eqref{eq:hpm-kappa}, to determine the constant $\kappa(\chi,\tau)$ and simultaneously obtain $g(\lambda;\chi,\tau)$ without using \eqref{eq:g-integral}, we represent $g(\lambda;\chi,\tau)$ in the form $g(\lambda;\chi,\tau)=R(\lambda;\chi,\tau)k(\lambda;\chi,\tau)$, from which it follows that $k(\lambda;\chi,\tau)$ is analytic for $\lambda\in\mathbb{C}\setminus\Sigma_g$, with bounded boundary values except at the branch points $\lambda=A\pm\ii B$ where it is only required that the product $R(\lambda;\chi,\tau)k(\lambda;\chi,\tau)$ is bounded.  We also require that $k(\lambda;\chi,\tau)=O(\lambda^{-2})$ as $\lambda\to\infty$, and that the boundary values taken by $k(\lambda;\chi,\tau)$ along $\Sigma_g$ are related by
\begin{equation}
k_+(\lambda;\chi,\tau)-k_-(\lambda;\chi,\tau)=\frac{2 \kappa(\chi,\tau)-2\vartheta(\lambda;\chi,\tau)}{R_+(\lambda;\chi,\tau)},\quad\lambda\in\Sigma_g,
\end{equation}
as implied by \eqref{eq:hpm-kappa}.
It follows that $k(\lambda;\chi,\tau)$ is necessarily given by the Plemelj formula:
\begin{equation}
k(\lambda;\chi,\tau)=\frac{1}{\ii\pi}\int_{\Sigma_g}\frac{\kappa(\chi,\tau)-\vartheta(\eta;\chi,\tau)}{R_+(\eta;\chi,\tau)(\eta-\lambda)}\,\dd\eta,\quad\lambda\in\mathbb{C}\setminus\Sigma_g.
\label{eq:k-formula}
\end{equation}
It is not hard to see that this formula automatically gives the condition that $R(\lambda;\chi,\tau)k(\lambda;\chi,\tau)$ is bounded at the branch points $\lambda=A\pm\ii B$.  However the condition $k(\lambda;\chi,\tau)=O(\lambda^{-2})$ as $\lambda\to\infty$ remains to be enforced, and this will determine the integration constant $\kappa(\chi,\tau)$.  Indeed, the coefficient of the leading term proportional to $\lambda^{-1}$ in the Laurent expansion of $k(\lambda;\chi,\tau)$ about $\lambda=\infty$ must vanish, i.e.,
\begin{equation}
\int_{\Sigma_g}\frac{\kappa(\chi,\tau)-\vartheta(\lambda;\chi,\tau)}{R_+(\lambda;\chi,\tau)}\,\dd\lambda = 0.
\label{eq:K-integral}
\end{equation}
Note that letting $L$ denote any clockwise-oriented loop surrounding the branch cut $\Sigma_g$ of $R(\lambda;\chi,\tau)$, a residue computation at $\lambda=\infty$ where $R(\lambda;\chi,\tau)=\lambda+O(1)$ shows that
\begin{equation}
\int_{\Sigma_g}\frac{\dd\lambda}{R_+(\lambda;\chi,\tau)} = \frac{1}{2}\oint_L\frac{\dd\lambda}{R(\lambda;\chi,\tau)} =-\ii\pi.
\label{eq:integral-R-plus}
\end{equation}
This being nonzero shows that $\kappa(\chi,\tau)$ will indeed be determined by the condition \eqref{eq:K-integral}.  Then,
\begin{equation}
\int_{\Sigma_g}\frac{\vartheta(\lambda;\chi,\tau)}{R_+(\lambda;\chi,\tau)}\,\dd\lambda = I_1(\chi,\tau)+I_2(\chi,\tau),
\end{equation}
where
\begin{equation}
I_1(\chi,\tau)\defeq\int_{\Sigma_g}\frac{\chi\lambda+\tau\lambda^2}{R_+(\lambda;\chi,\tau)}\,\dd\lambda\quad\text{and}\quad
I_2(\chi,\tau)\defeq\ii\int_{\Sigma_g}
\frac{\log(B(\lambda))}{R_+(\lambda;\chi,\tau)}\,\dd\lambda.
\label{eq:I1-I2}
\end{equation}
A similar residue calculation using two more terms in the large-$\lambda$ expansion of $R(\lambda;\chi,\tau)$, specifically that $R(\lambda;\chi,\tau)^{-1}=\lambda^{-1}+A\lambda^{-2}+(A^2-\tfrac{1}{2}B^2)\lambda^{-3}+O(\lambda^{-4})$ shows that
\begin{equation}
I_1(\chi,\tau)=-\ii\pi (\chi A +\tau(A^2-\tfrac{1}{2}B^2)),\quad A=A(\chi,\tau),\quad B^2=B(\chi,\tau)^2.
\end{equation}
Now assuming that the loop $L$ excludes the branch cut $\Sigma_\mathrm{c}$ of the logarithm and that $L'$ is a counter-clockwise oriented contour that encircles $\Sigma_\mathrm{c}$ but that excludes $\Sigma_g$, we use the fact that the integrand for $I_2$ is integrable at $\lambda=\infty$ to obtain
\begin{equation}
I_2(\chi,\tau)=\frac{1}{2}\ii\oint_L
\frac{\log(B(\lambda))}{R(\lambda;\chi,\tau)}\,\dd\lambda = \frac{1}{2}\ii\oint_{L'}
\frac{\log(B(\lambda))}{R(\lambda;\chi,\tau)}\,\dd\lambda.
\label{eq:I2-identities}
\end{equation}
Then, collapsing $L'$ to both sides of $\Sigma_\mathrm{c}$, where $R(\lambda;\chi,\tau)$ is analytic but the boundary values of the logarithm differ by $2\pi\ii$, 
\begin{equation}
I_2(\chi,\tau)=\pi\int_{\Sigma_\mathrm{c}}\frac{\dd\lambda}{R(\lambda;\chi,\tau)},
\end{equation}
where we recall that $\Sigma_\mathrm{c}$ is a Schwarz-symmetric arc oriented from $-\ii$ to $\ii$.  Therefore, $I_2(\chi,\tau)$ is
purely imaginary and is computable in terms of $A(\chi,\tau)$ and $B(\chi,\tau)^2$ via hyperbolic functions.  We have therefore obtained a formula for the integration constant $\kappa(\chi,\tau)$ in the form \eqref{eq:kappa-formula} written in Section~\ref{sec:Results-Shelves}.
According to \eqref{eq:k-formula} and $g(\lambda;\chi,\tau)=R(\lambda;\chi,\tau)k(\lambda;\chi,\tau)$ we have (evaluating the term proportional to $\kappa(\chi,\tau)$ by residues):
\begin{equation}
g(\lambda;\chi,\tau)= \kappa(\chi,\tau)-\frac{R(\lambda;\chi,\tau)}{\ii\pi}\int_{\Sigma_g}\frac{\vartheta(\eta;\chi,\tau)\,\dd\eta}{R_+(\eta;\chi,\tau)(\eta-\lambda)}.
\end{equation}

\subsubsection{Construction of $g(\lambda;\chi,\tau)$ when $\Sigma_g\subset\Sigma_\mathrm{c}$}
\label{sec:g-function-dumbbell}
If the Schwarz-symmetric jump contour $\Sigma_g$ is to be taken as a subset of $\Sigma_\mathrm{c}$,
then some modification of the construction of $g(\lambda;\chi,\tau)$ is needed.  Indeed, in this situation the phase function $\vartheta(\lambda;\chi,\tau)$ defined in \eqref{eq:vartheta} takes two distinct boundary values at each point of $\Sigma_g$, so instead of \eqref{eq:hpm-kappa} the condition that the sum of boundary values of $h(\lambda;\chi,\tau)$ is constant along $\Sigma_g$ now reads
\begin{equation}
\begin{split}
g_+(\lambda;\chi,\tau)+g_-(\lambda;\chi,\tau)&=h_+(\lambda;\chi,\tau)+h_-(\lambda;\chi,\tau)-\vartheta_+(\lambda;\chi,\tau)-\vartheta_-(\lambda;\chi,\tau)\\
&=2\gamma(\chi,\tau)-\vartheta_+(\lambda;\chi,\tau)-\vartheta_-(\lambda;\chi,\tau),\quad
\lambda\in\Sigma_g\subset\Sigma_\mathrm{c},
\end{split}
\label{eq:hpm-gamma}
\end{equation}
where the constant value of $h_++h_-$ is now denoted $2\gamma(\chi,\tau)$.  As before, we write $g(\lambda;\chi,\tau)=R(\lambda;\chi,\tau)k(\lambda;\chi,\tau)$ and find that the analogue of \eqref{eq:k-formula} reads
\begin{equation}
k(\lambda;\chi,\tau)=\frac{1}{\ii\pi}\int_{\Sigma_g}\frac{\gamma(\chi,\tau)-\tfrac{1}{2}\vartheta_+(\eta;\chi,\tau)-\tfrac{1}{2}\vartheta_-(\eta;\chi,\tau)}{R_+(\eta;\chi,\tau)(\eta-\lambda)}\,\dd\eta,\quad\lambda\in\mathbb{C}\setminus\Sigma_g,
\label{eq:k-formula-gamma}
\end{equation}
where $\gamma(\chi,\tau)$ is to be chosen to enforce the condition analogous to \eqref{eq:K-integral}: 
\begin{equation}
\int_{\Sigma_g}\frac{\gamma(\chi,\tau)-\tfrac{1}{2}\vartheta_+(\lambda;\chi,\tau)-\tfrac{1}{2}\vartheta_-(\lambda;\chi,\tau)}{R_+(\lambda;\chi,\tau)}\,\dd\lambda = 0.
\label{eq:K-integral-gamma}
\end{equation}
By residue calculations, this can be written in the form
\begin{equation}
\gamma(\chi,\tau)=\chi A(\chi,\tau)+\tau(A(\chi,\tau)^2-\tfrac{1}{2}B(\chi,\tau)^2)-\frac{1}{\ii\pi}I_2(\chi,\tau),
\end{equation}
where now $I_2(\chi,\tau)$ is given by a modification of the formula in \eqref{eq:I1-I2}:
\begin{equation}
I_2(\chi,\tau)\defeq\ii\int_{\Sigma_g}\frac{\tfrac{1}{2}\log_+(B(\lambda))+\tfrac{1}{2}\log_-(B(\lambda))}{R_+(\lambda;\chi,\tau)}\,\dd\lambda.
\end{equation}
Taking $L$ to be a clockwise-oriented loop passing through the endpoints $A\pm\ii B$ of $\Sigma_g$ and enclosing $\Sigma_g$ but with the two arcs of $\Sigma_\mathrm{c}\setminus\Sigma_g$ in its exterior, and taking $L'$ to be a pair of counterclockwise-oriented loops each enclosing one of the arcs of $\Sigma_\mathrm{c}\setminus\Sigma_g$ and passing through the corresponding endpoint of $\Sigma_g$, we again arrive at the identities \eqref{eq:I2-identities}.  Then collapsing the loops of $L'$ to both sides of the arcs of $\Sigma_\mathrm{c}\setminus\Sigma_g$ we obtain
\begin{equation}
I_2(\chi,\tau)=\pi\int_{\Sigma_\mathrm{c}\setminus\Sigma_g}\frac{\dd\lambda}{R(\lambda;\chi,\tau)},
\end{equation}
leading to the analogue of \eqref{eq:kappa-formula}:
\begin{equation}
\gamma(\chi,\tau)=\chi A+\tau(A^2-\tfrac{1}{2}B^2) +\ii\int_{\Sigma_\mathrm{c}\setminus\Sigma_g}\frac{\dd\lambda}{R(\lambda;\chi,\tau)}.
\label{eq:gamma-formula}
\end{equation}
This is equivalent to the form written in \eqref{eq:gamma-formula-intro} in Section~\ref{sec:Results-Exterior}.

\subsubsection{Structure of the zero level curve $\mathrm{Re}(\ii h(\lambda;\chi,\tau))=0$}
\label{sec:zero-level-curve}
A consequence of the choice of integration constant to ensure that $g(\lambda;\chi,\tau)\to 0$ as $\lambda\to\infty$ is that both $g$ and $h$ have even Schwarz symmetry for all $(\chi,\tau)\in \overline{\exterior\cup\shelves}$:
\begin{equation}
g(\lambda^*;\chi,\tau)^*=g(\lambda;\chi,\tau)\quad\text{and}\quad h(\lambda^*;\chi,\tau)^*=h(\lambda;\chi,\tau),\quad (\chi,\tau)\in \overline{\exterior\cup\shelves}.
\label{eq:g-h-Schwarz}
\end{equation}
It follows that $\mathrm{Re}(\ii h(\lambda;\chi,\tau))=0$ holds for all $\lambda\in\mathbb{R}$, $\lambda\not\in\Sigma_g$.  We also have the following.
\begin{lemma}
For all $(\chi,\tau)\in \overline{\exterior\cup\shelves}$, $\mathrm{Re}(\ii h(A\pm\ii B;\chi,\tau))=0$, where $A\pm\ii B=A(\chi,\tau)\pm\ii B(\chi,\tau)$ are the complex-conjugate endpoints of $\Sigma_g$.
\label{lem:h-imaginary-at-branch-points}
\end{lemma}
\begin{proof}
Let $\lambda_\mathbb{R}\in\mathbb{R}$ with $\lambda_\mathbb{R}\not\in\Sigma_g$.
Since $h(\lambda_\mathbb{R};\chi,\tau)$ is purely real,
\begin{equation}
\begin{split}
\mathrm{Re}(\ii h(A+\ii B;\chi,\tau))&=\mathrm{Re}\left(\ii \int_{\lambda_\mathbb{R}}^{A+\ii B}h'(\lambda;\chi,\tau)\,\dd\lambda\right)\\ &= \frac{1}{2}\ii \int_{\lambda_\mathbb{R}}^{A+\ii B}h'(\lambda;\chi,\tau)\,\dd\lambda -\frac{1}{2}\ii\left[\int_{\lambda_\mathbb{R}}^{A+\ii B}h'(\lambda;\chi,\tau)\,\dd\lambda\right]^*,
\end{split}
\end{equation}
where due to \eqref{eq:hprime-residues} the path of integration $L:\lambda_\mathbb{R}\to A+\ii B$ is arbitrary in the upper half plane, except that it is chosen so that the pole at $\lambda=\ii$ does not lie between $L$ and $\Sigma_g$.  Using the even Schwarz symmetry of $h'(\lambda;\chi,\tau)$ and a contour integral reparametrization,
\begin{equation}
\begin{split}
\left[\int_{\lambda_\mathbb{R}}^{A+\ii B}h'(\lambda;\chi,\tau)\,\dd\lambda\right]^*&=\int_{\lambda_\mathbb{R}}^{A+\ii B}h'(\lambda;\chi,\tau)^*\,\dd\lambda^*\\
&=\int_{\lambda_\mathbb{R}}^{A+\ii B}h'(\lambda^*;\chi,\tau)\,\dd\lambda^*\\
&=-\int_{A-\ii B}^{\lambda_\mathbb{R}}h'(\lambda;\chi,\tau)\,\dd\lambda,
\end{split}
\end{equation}
where in the final integral the path of integration is $L^*$ but with opposite orientation.  Combining these results, and taking into account that $h'(\lambda;\chi,\tau)$ changes sign across $\Sigma_g$, we have
\begin{equation}
\mathrm{Re}(\ii h(A+\ii B;\chi,\tau))=\frac{1}{4}\ii \oint_O h'(\lambda;\chi,\tau)\,\dd\lambda,
\end{equation}
where $O$ is a simple closed contour enclosing $\Sigma_g$ but excluding $\lambda=\pm\ii$, the orientation of which depends on whether $\lambda_\mathbb{R}$ lies to the left or right of the point where $\Sigma_g$ intersects the real axis.  Using \eqref{eq:hprime-residues}, without changing the value of the integral we may replace $O$ by another contour surrounding $\Sigma_g$ with the same orientation but now also enclosing $\lambda=\pm\ii$.  Since there are no longer any singularities of $h'(\lambda;\chi,\tau)$ outside of $O$, we may evaluate the integral over $O$ by residues at $\lambda=\infty$.  Using \eqref{eq:hprime-expansion} one sees that the residue of $h'(\lambda;\chi,\tau)$ at $\lambda=\infty$ vanishes, so we conclude that $\mathrm{Re}(\ii h(A+\ii B;\chi,\tau))=0$.  Using \eqref{eq:g-h-Schwarz} then gives also $\mathrm{Re}(\ii h(A-\ii B;\chi,\tau))=0$.
\end{proof}
This result implies that the level curve $\mathrm{Re}(\ii h(\lambda;\chi,\tau))=0$ does not depend substantially on the choice of branch cut $\Sigma_g$.  Indeed, the differential $\ii h'(\lambda;\chi,\tau)\,\dd\lambda$ can be extended from $\lambda\in\mathbb{C}\setminus\Sigma_g$ to the hyperelliptic Riemann surface $\mathcal{R}$ of the equation $R^2=(\lambda-A)^2+B^2$ just by adding a second copy of $\mathbb{C}\setminus\Sigma_g$ on which $R(\lambda;\chi,\tau)$ is replaced with $-R(\lambda;\chi,\tau)$.  Since $\mathcal{R}$ has genus zero and hence has trivial homology, and since the residues of $h'(\lambda;\chi,\tau)$ (see \eqref{eq:hprime-residues}--\eqref{eq:hprime-expansion}) are imaginary, the real part of an antiderivative of $\ii h'(\lambda;\chi,\tau)\,\dd\lambda$ is well defined up to a constant as a harmonic function on $\mathcal{R}$ with the four points corresponding to $\lambda=\pm\ii$ omitted.  By 
Lemma~\ref{lem:h-imaginary-at-branch-points}, if the constant of integration is determined by fixing the base point of integration to be one of the two branch points, the real part vanishes at both branch points and on the principal sheet of $\mathcal{R}$ this function coincides with $\mathrm{Re}(\ii h(\lambda;\chi,\tau))$ while on the auxiliary sheet it coincides with $-\mathrm{Re}(\ii h(\lambda;\chi,\tau))$.  It follows that the projection from each sheet of $\mathcal{R}$ to $\mathbb{C}$ of the zero level is exactly the same.  Since the choice of branch cut $\Sigma_g$ for $R(\lambda;\chi,\tau)$ only affects the value of $\mathrm{Re}(\ii h(\lambda;\chi,\tau))$ up to a sign, the zero level curve is essentially independent of the location of $\Sigma_g$ (technically, $\mathrm{Re}(\ii h(\lambda;\chi,\tau))$ is undefined on $\Sigma_g$, but the zero level curve can be extended unambiguously to $\Sigma_g$). 

As noted above, the zero level set $\mathrm{Re}(\ii h(\lambda;\chi,\tau))=0$ always contains the real axis as a proper subset, as well as the branch points $\lambda=A\pm\ii B$.  Since $\ii h(\lambda;\chi,\tau)=\ii \vartheta(\lambda;\chi,\tau)+\ii g(\lambda;\chi,\tau)=\ii \vartheta(\lambda;\chi,\tau)+O(\lambda^{-1})=\ii \chi\lambda+\ii \tau\lambda^2+O(\lambda^{-1})$ as $\lambda\to\infty$, for $\tau\neq 0$ in $\overline{\exterior\cup\shelves}$ there is exactly one Schwarz-symmetric pair of arcs of the zero level set that are asymptotically vertical, one in each half-plane.  All other arcs of the level set in $\mathbb{C}\setminus\mathbb{R}$ are bounded.  These arcs are necessarily ``horizontal'' trajectories of the rational quadratic differential $h'(\lambda;\chi,\tau)^2\,\dd\lambda^2$, i.e., curves along which $h'(\lambda;\chi,\tau)^2\,\dd\lambda^2>0$.  By Lemma~\ref{lem:h-imaginary-at-branch-points}, some of the arcs of the zero level set are so-called critical trajectories, i.e., those emanating from zeros of $h'(\lambda;\chi,\tau)^2$.  By Jenkins' three-pole theorem \cite[Theorem 3.6]{Jenkins58} and the basic structure theorem \cite[Theorem 3.5]{Jenkins58}, the union of critical trajectories of $h'(\lambda;\chi,\tau)^2\,\dd\lambda^2$ has empty interior and divides the complex $\lambda$-plane into a finite number of domains.  Two of these domains, one in each half-plane, are so-called \emph{circle domains} each containing one of the poles $\lambda=\pm\ii$ and each having at least one of the zeros of $h'(\lambda;\chi,\tau)^2$ on its boundary.  Furthermore, from each of the simple roots $\lambda=A\pm\ii B$ of $h'(\lambda;\chi,\tau)^2$ emanate locally exactly three critical trajectories, and from each of the double roots of $h'(\lambda;\chi,\tau)^2$ (i.e., the roots of $2\tau\lambda^2+u(\chi,\tau)\lambda+v(\chi,\tau)$) emanate locally exactly four critical trajectories.

Suppose first that $(\chi,\tau)\in \exterior_\chi\cup \shelves$.  Then the double roots of $h'(\lambda;\chi,\tau)$ (two for $\tau\neq 0$ and one for $\tau=0$) are real, and therefore two of the four trajectories emanating from each coincide with intervals of $\mathbb{R}$ (that are contained in the level set $\mathrm{Re}(\ii h(\lambda;\chi,\tau))=0$, and the closure of the union of which is exactly $\mathbb{R}$).  In this case, by Lemma~\ref{lem:h-imaginary-at-branch-points} all critical trajectories are included in the level set $\mathrm{Re}(\ii h(\lambda;\chi,\tau))=0$.  The level curves entering the upper and lower half-planes vertically from $\lambda=\infty$ for $\tau\neq 0$ can only terminate at one of the roots of $h'(\lambda;\chi,\tau)^2$.  These trajectories either terminate at one of the real double roots, or at the conjugate pair of simple roots $\lambda=A\pm\ii B$.  
\begin{itemize}
\item
If they terminate at one of the two real double roots, then the non-real trajectories emanating from the other real double root can only terminate at the simple roots $\lambda=A\pm\ii B$.  It follows that the two additional trajectories emanating from each of these simple roots must coincide and form a closed curve in each half-plane.  By Teichm\"uller's lemma \cite[Theorem 14.1]{Strebel84}, this curve must be the boundary of the circle domain containing the pole $\lambda=\pm\ii$. If $\tau=0$ and hence there are no unbounded arcs of the level set in the open upper and lower half-planes, then by the same arguments the non-real trajectories emanating from the unique real double root terminate at the simple roots $\lambda=A\pm\ii B$, and the remaining two trajectories from each of these coincide and enclose the poles at $\lambda=\pm\ii$.  The zero level $\mathrm{Re}(\ii h(\lambda;\chi,\tau))=0$ consists of the real line, a Schwarz-symmetric pair of arcs connecting a real double root with the conjugate pair of simple roots $\lambda=A\pm\ii B$, a Schwarz-symmetric pair of loops joining each simple root $\lambda=A\pm\ii B$ to itself and enclosing the poles at $\lambda=\pm\ii$, and (if $\tau\neq 0$) a Schwarz-symmetric pair of unbounded arcs emanating from the second real double root and tending vertically to $\lambda=\infty$.  This topological configuration of the zero level set holds on the domain $\exterior_\chi$ (as one can see from the limiting case of $\tau=0$, where the zero level set acquires additional Schwarz reflection symmetry in the imaginary axis).
\item
If they terminate at the conjugate pair of simple roots $\lambda=A\pm \ii B$, then the remaining two trajectories emanating from each simple root terminate at the two real double roots, and the boundary of the circle domain in each half-plane consists of three distinct trajectories, one of which is the interval of the real axis between the two real double roots and is common to the boundaries of both circle domains.  (The other apparent possibility, that the two additional trajectories emanating from $\lambda=A\pm\ii B$ coincide and that the two trajectories emanating into each half-plane from the two real double roots also coincide, can be ruled out by Teichm\"uller's lemma since two closed curves formed by critical trajectories would appear in each half-plane, only one of which can contain a pole.)  The zero level set $\mathrm{Re}(\ii h(\lambda;\chi,\tau))=0$ consists of the real line, a Schwarz-symmetric pair of arcs from each of the two real double roots to the conjugate pair of simple roots, and a Schwarz-symmetric pair of unbounded arcs emanating from the conjugate pair of simple roots and tending vertically to $\lambda=\infty$.  This topological configuration of the zero level set holds on the domain $\shelves$ (as one can see from the limiting case of $\chi=0$, where again the zero level set acquires additional Schwarz reflection symmetry in the imaginary axis).
\end{itemize}

Next suppose that $(\chi,\tau)\in \exterior_\tau$.  Then the double roots of $h'(\lambda;\chi,\tau)^2$ form a conjugate pair that we denote by $\lambda=C\pm\ii D$.  By Lemma~\ref{lem:h-imaginary-at-branch-points}, the simple roots $\lambda=A\pm\ii B$ are on the zero level of $\mathrm{Re}(\ii h(\lambda;\chi,\tau))$ and therefore at most one trajectory from each can be unbounded.  If none of the three trajectories emanating from $\lambda=A\pm\ii B$ is unbounded, then at least one of them must terminate at $\lambda=C\pm\ii D$ implying that $\mathrm{Re}(\ii h(C\pm\ii D;\chi,\tau))=0$ and hence the unbounded arc of the level curve in each half-plane terminates at this point as well.  If it is exactly one trajectory from $\lambda=A\pm\ii B$ that terminates at $\lambda=C\pm\ii D$, then the other two coincide forming a loop, and the remaining two bounded trajectories emanating from the latter must coincide forming a second loop; however only one of these loops can contain the pole at $\lambda=\pm\ii$ so the existence of both is ruled out by Teichm\"uller's lemma.  If it is exactly two trajectories from $\lambda=A\pm\ii B$ that terminate at $\lambda=C\pm\ii D$, then the third trajectory would have to be unbounded contradicting the assumption that all trajectories from $A\pm\ii B$ are bounded.  If all three trajectories from $\lambda=A\pm\ii B$ terminate at $\lambda=C\pm\ii D$, then we again form two domains bounded by trajectories only one of which can contain a pole leading to a contradiction with Teichm\"uller's lemma.  We conclude that exactly one of the trajectories emanating from each simple root $\lambda=A\pm\ii B$ is unbounded.  
It then follows that the other two trajectories emanating from $\lambda=A\pm\ii B$ must coincide.  Indeed, otherwise they must both terminate at the double root $\lambda=C\pm\ii D$ in the same half-plane from which we learn that $\mathrm{Re}(\ii h(C\pm\ii D;\chi,\tau))=0$ which implies that neither of the remaining two trajectories from $\lambda=C\pm\ii D$ can be unbounded or terminate at $A\pm\ii B$, so they must coincide.  It is then apparent that each half-plane contains a domain bounded by the two curves connecting $A\pm\ii B$ with $C\pm\ii D$ and a domain bounded by the trajectory joining $C\pm\ii D$ to itself; however the pole $\lambda=\pm\ii$ can only lie in one of these two domains, so the existence of the other leads to a contradiction with Teichm\"uller's lemma. The zero level set $\mathrm{Re}(\ii h(\lambda;\chi,\tau))=0$ is then the disjoint union of three components:  the real line and a Schwarz-symmetric pair of components each consisting of a loop trajectory joining $\lambda=A\pm\ii B$ to itself and surrounding $\lambda=\pm\ii$ and an unbounded trajectory emanating from $\lambda=A\pm\ii B$.  In this case, the double roots $\lambda=C\pm\ii D$ do not lie on the zero level of $\mathrm{Re}(\ii h(\lambda;\chi,\tau))$, and the level set is not connected. 

\section{Far-Field Asymptotic Behavior of Rogue Waves in the Domain $\exterior$}
\label{sec:Schi-Stau}
In this section, we prove Theorem~\ref{thm:exterior}.  Since that result is specialized to the case of fundamental rogue waves of order $k\in\mathbb{Z}_{>0}$ for which $\mathbf{G}=\mathbf{Q}^{-s}$ with $s=(-1)^k$ and $M=\tfrac{1}{2}k+\tfrac{1}{4}$ (assumptions that are essential to the proof), in this section we will write $\mathbf{S}^{(k)}(\lambda;\chi,\tau)=\mathbf{S}(\lambda;\chi,\tau,\mathbf{Q}^{-s},M)$.
\subsection{Deformation to a dumbbell-shaped contour}
\label{sec:dumbbell}
When $(\chi,\tau)\in \exterior$, we will find it useful to begin by replacing the Jordan jump contour $\Sigma_\circ$ for $\mathbf{S}^{(k)}(\lambda;\chi,\tau)$ with a dumbbell-shaped contour consisting of a closed loop $\Gamma^+$ in the upper half-plane surrounding the point $\lambda=\ii$ in the clockwise sense, its reflection $\Gamma^-$ in the real axis (also oriented in the clockwise sense), and a ``neck'' $N$ consisting of an upward-oriented arc against the left side of the branch cut $\Sigma_\mathrm{c}$ for $\vartheta(\lambda;\chi,\tau)$ and a downward-oriented arc against the right side of the same cut.  Combining these two jump conditions with the jump discontinuity of the function $\vartheta(\lambda;\chi,\tau)$ across the central arc $\Sigma_\mathrm{c}$ of the neck, we can write a single jump condition for $\mathbf{S}^{(k)}(\lambda;\chi,\tau)$ across $N$, which we take to be oriented in the upward direction.  For this calculation, we assume that initially the Jordan curve $\Sigma_\circ$ contains $\Gamma^+\cup N\cup \Gamma^-$ in its interior and we introduce a substitution by setting
\begin{equation}
\tilde{\mathbf{S}}^{(k)}(\lambda;\chi,\tau)\defeq
\mathbf{S}^{(k)}(\lambda;\chi,\tau)\ee^{-\ii M\vartheta(\lambda;\chi,\tau)\sigma_3}\mathbf{Q}^{-s}\ee^{\ii M\vartheta(\lambda;\chi,\tau)\sigma_3},
\label{eq:S-Stilde-ALT}
\end{equation}
for $\lambda$ between $\Sigma_\circ$ and $\Gamma^+\cup N\cup \Gamma^-$, and we set $\tilde{\mathbf{S}}^{(k)}(\lambda;\chi,\tau)\defeq\mathbf{S}^{(k)}(\lambda;\chi,\tau)$ elsewhere, i.e., in the exterior of $\Sigma_\circ$ and in the interior of $\Gamma^+$ and of $\Gamma^-$.  
Dropping the tilde, the jump contour for $\mathbf{S}^{(k)}(\lambda;\chi,\tau)$ becomes $\Gamma^+\cup N\cup \Gamma^-$.  The jump condition for $\mathbf{S}^{(k)}(\lambda;\chi,\tau)$ across $\Gamma^+$ and $\Gamma^-$ reads exactly the same as the original jump condition \eqref{eq:S-jump} across $\Sigma_\circ$.  
To compute the jump of the redefined $\mathbf{S}^{(k)}(\lambda;\chi,\tau)$ across $N$,
we start from its definition and using the fact that $\vartheta(\lambda;\chi,\tau)$ takes distinct boundary values on $N$ from either side we get
\begin{multline}
\mathbf{S}^{(k)}_+(\lambda;\chi,\tau)=\mathbf{S}_-^{(k)}(\lambda;\chi,\tau)\ee^{-\ii M\vartheta_-(\lambda;\chi,\tau)\sigma_3}\mathbf{Q}^s\ee^{\ii M\vartheta_-(\lambda;\chi,\tau)\sigma_3}
\ee^{-\ii M\vartheta_+(\lambda;\chi,\tau)\sigma_3}\mathbf{Q}^{-s}\ee^{\ii M\vartheta_+(\lambda;\chi,\tau)\sigma_3}\\
{}=\mathbf{S}_-^{(k)}(\lambda;\chi,\tau)\frac{1}{2}
\begin{bmatrix} 1+\ee^{2\ii M(\vartheta_+(\lambda;\chi,\tau)-\vartheta_-(\lambda;\chi,\tau))} & 
s\left(\ee^{-2\ii M\vartheta_+(\lambda;\chi,\tau)}-\ee^{-2\ii M\vartheta_-(\lambda;\chi,\tau)}\right)\\
-s\left(\ee^{2\ii M\vartheta_+(\lambda;\chi,\tau)}-\ee^{2\ii M\vartheta_-(\lambda;\chi,\tau)}\right) & 
1+\ee^{-2\ii M(\vartheta_+(\lambda;\chi,\tau)-\vartheta_-(\lambda;\chi,\tau))}\end{bmatrix},\\
\quad\quad\quad\quad\quad\quad\lambda\in N.
\end{multline}
But by \eqref{eq:vartheta} we have $\vartheta_+(\lambda;\chi,\tau)-\vartheta_-(\lambda;\chi,\tau)=-2\pi$.  Since $M=\tfrac{1}{2}k+\tfrac{1}{4}$ for $k\in\mathbb{Z}_{>0}$, this easily reduces to
\begin{equation}
\begin{split}
\mathbf{S}^{(k)}_+(\lambda;\chi,\tau)&=\mathbf{S}_-^{(k)}(\lambda;\chi,\tau)\begin{bmatrix}0 & s\ee^{-2\ii M\vartheta_+(\lambda;\chi,\tau)}\\-s\ee^{2\ii M\vartheta_+(\lambda;\chi,\tau)} & 0\end{bmatrix}\\
&=\mathbf{S}_-^{(k)}(\lambda;\chi,\tau)\begin{bmatrix}0 & -s\ee^{-2\ii M\vartheta_-(\lambda;\chi,\tau)}\\
s\ee^{2\ii M\vartheta_-(\lambda;\chi,\tau)} & 0\end{bmatrix},\quad\lambda\in N.
\end{split}
\label{eq:S-N-jump-ALT}
\end{equation}

\begin{remark}
The fact that the jump matrix on $N$ is off-diagonal is a consequence of the quantization of $M>0$ via $M=\tfrac{1}{2}k+\tfrac{1}{4}$, $k\in\mathbb{Z}_{>0}$, and the choice of ``core'' matrix $\mathbf{G}=\mathbf{Q}^{-s}$ for $s=(-1)^k$.  
More generally, if we express $M\ge 0$ in the modular form $M=\tfrac{1}{2}k+r$ with $k\in\mathbb{Z}_{\ge 0}$ and $0\le r<\tfrac{1}{2}$, then for $\mathbf{G}=\mathbf{Q}^{-s}$ with $s=\pm 1$ arbitrary we obtain
\begin{equation}
\mathbf{S}_+(\lambda;\chi,\tau,\mathbf{Q}^{-s},M)=
\mathbf{S}_-(\lambda;\chi,\tau,\mathbf{Q}^{-s},M)
\ee^{-\ii M\vartheta_-(\lambda;\chi,\tau)\sigma_3}\mathbf{Z}
\ee^{\ii M\vartheta_+(\lambda;\chi,\tau)\sigma_3},\quad\lambda\in N
\end{equation}
in place of \eqref{eq:S-N-jump-ALT}, where $\mathbf{Z}$ is the constant matrix
\begin{equation}
\mathbf{Z}\defeq\begin{bmatrix}
(-1)^k\cos(2\pi r) & s (-1)^k\ii\sin(2\pi r)\\ s(-1)^k\ii\sin(2\pi r) & (-1)^k\cos(2\pi r)\end{bmatrix}.
\end{equation}
It is then clear that the \emph{only} values of $M\ge 0$ for which $\mathbf{Z}$ is off-diagonal are those corresponding to rogue waves.  
This is the reason why fundamental rogue waves behave differently for $(\chi,\tau)\in \exterior$ than other solutions obtained from Riemann-Hilbert Problem~\ref{rhp:rogue-wave-reformulation} for different parameters as described in Section~\ref{sec:M-arbitrary}, such as the high-order multiple-pole solitons for which $M\in\tfrac{1}{2}\mathbb{Z}_{\ge 0}$.  The latter solutions are special once again, in that they are precisely the solutions for which $\mathbf{Z}$ is diagonal (in fact $\mathbf{Z}=(-1)^k\mathbb{I}$).
In the general case, all four entries of $\mathbf{Z}$ are nonzero and hence available for use as pivots in matrix factorizations, and this distinguishes the asymptotic behavior on $\exterior$ from both special cases as described in Section~\ref{sec:general}.
\label{rem:M-quantum}
\end{remark}

Next, we explain how the contours $\Gamma^+$ and $N$ should be chosen (recall that $\Gamma^-$ is the Schwarz reflection of $\Gamma^+$ with clockwise orientation).  Recall from Section~\ref{sec:zero-level-curve} that as $(\chi,\tau)$ ranges over $\exterior$, there exists a simple closed curve surrounding the point $\lambda=\ii$ and passing through the point $\lambda=A+\ii B$ such that all roots of $2\tau\lambda^2+u\lambda+v=0$ are in the exterior of this curve, and importantly, such that $h'(\lambda;\chi,\tau)\,\dd\lambda$ is purely real along the curve.  In other words, the circle domain for the rational quadratic differential $h'(\lambda;\chi,\tau)^2\,\dd\lambda^2$ containing the pole $\lambda=\ii$ (reality of the residue due to the condition \eqref{eq:hprime-residues} guarantees that this point is indeed contained in a circle domain) has only the critical point $\lambda=A+\ii B$ on its boundary.  We take the boundary curve, which is a critical trajectory for $h'(\lambda;\chi,\tau)^2\,\dd\lambda^2$, to be the loop $\Gamma^+$.  Then we choose $N$ to be any Schwarz-symmetric arc from $\lambda=A-\ii B$ to $\lambda=A+\ii B$ that lies in the exterior of both loops $\Gamma^+\cup \Gamma^-$.  Later we will fix its direction near the endpoints of $N$.  See the left-hand panels of Figures~\ref{fig:Schi1}--\ref{fig:Stau1}.

\begin{figure}[h]
\begin{center}
\includegraphics{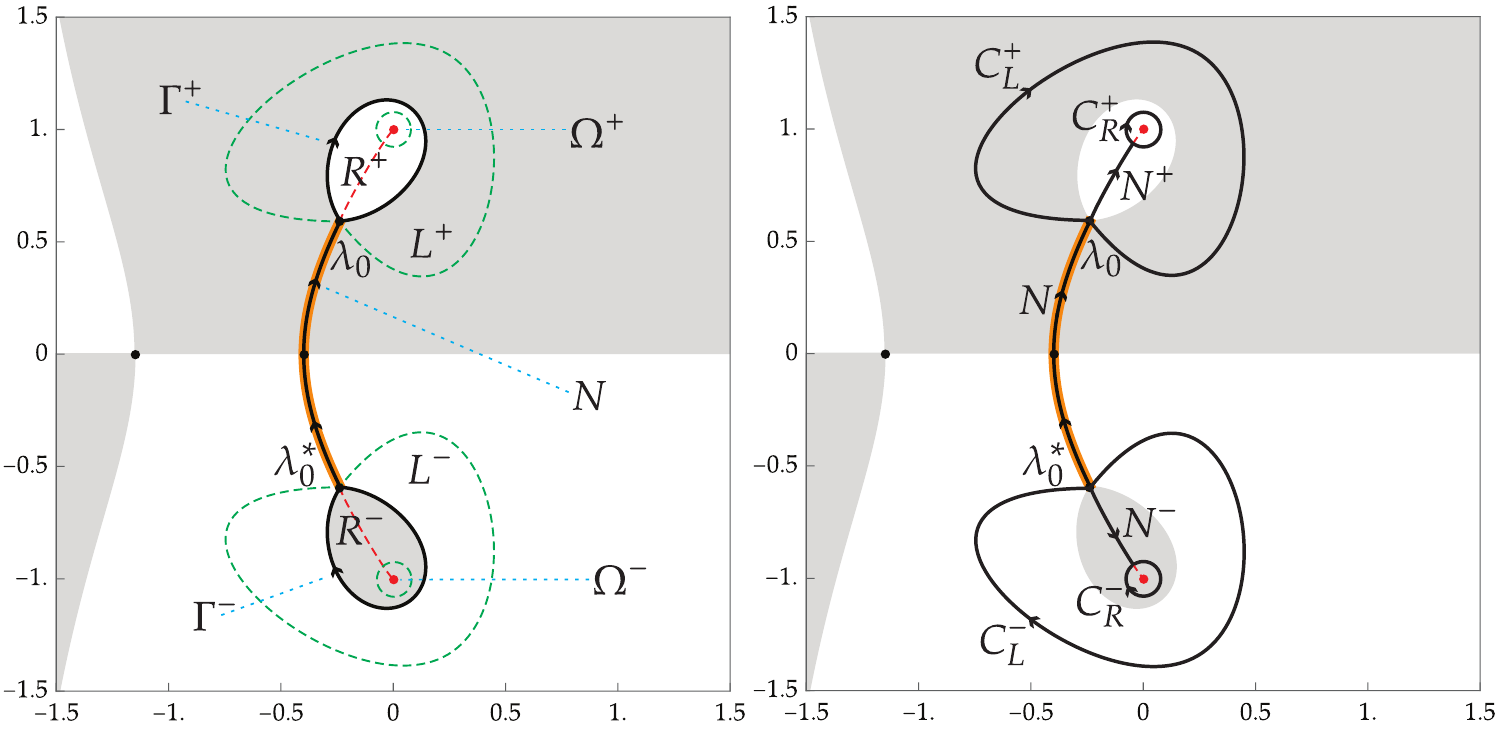}
\end{center}
\caption{Left:  for $(\chi,\tau)=(2.5,0.7)\in \exterior_\chi$, the regions in the $\lambda$-plane where $\mathrm{Re}(\ii h(\lambda;\chi,\tau))<0$ (shaded) and $\mathrm{Re}(\ii h(\lambda;\chi,\tau))>0$ (unshaded), and the modified jump contour $\Gamma^+\cup N\cup\Gamma^-$.  The jump contour $\Sigma_\mathrm{c}$ for $\vartheta(\lambda;\chi,\tau)$ consists of the union of $N$ and the dashed red arcs terminating at $\lambda=\pm\ii$ (red dots).  Critical points of $h(\lambda;\chi,\tau)$ are shown with black dots.  Also shown are the ``lens'' regions $L^\pm$ and $R^\pm$ lying to the left and right respectively of $\Gamma^\pm$.  Right:  the jump contour for $\mathbf{W}^{(k)}(\lambda;\chi,\tau)$.  Note that for $(\chi,\tau)\in \exterior_\chi$ we may choose the branch cut $N=\Sigma_g$ (highlighted in orange) to coincide with a level curve of $\mathrm{Re}(\ii h(\lambda;\chi,\tau))$ and with this choice $\mathrm{Re}(\ii h(\lambda;\chi,\tau))$ is a continuous function with the exception of the points $\lambda=\pm \ii$.}
\label{fig:Schi1}
\end{figure}

\begin{figure}[h]
\begin{center}
\includegraphics{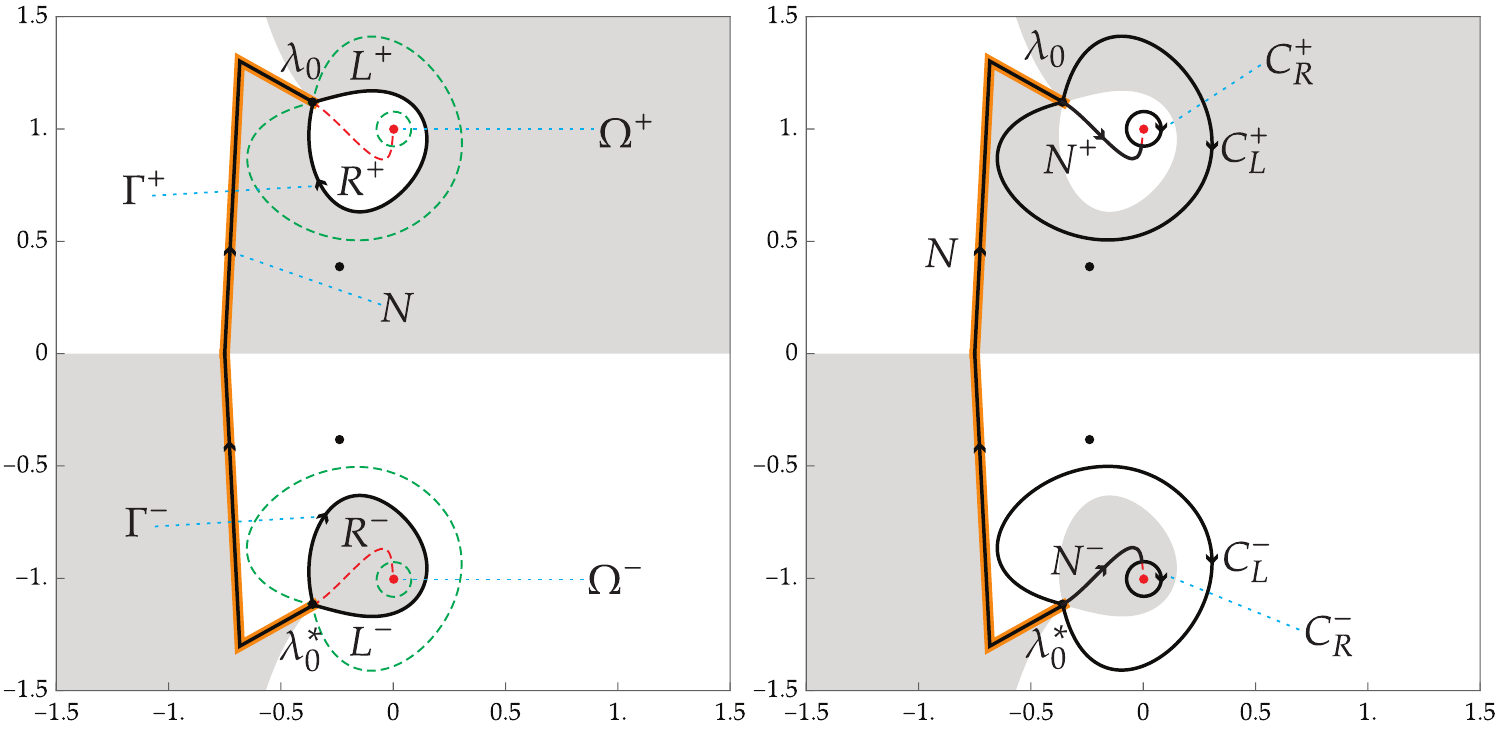}
\end{center}
\caption{As in Figure~\ref{fig:Schi1} but now for $(\chi,\tau)=(2.0,1.2)\in \exterior_\tau$.  In this case a pair of real critical points of $h(\lambda;\chi,\tau)$ on $\exterior_\chi$ have merged and split into a conjugate pair that is necessarily on a nonzero level of $\mathrm{Re}(\ii h(\lambda;\chi,\tau))$.  The consequence is that it is no longer possible on $\exterior_\tau$ to choose the branch cut $N=\Sigma_g$ to be a level curve of $\mathrm{Re}(\ii h(\lambda;\chi,\tau))$ which therefore experiences a jump discontinuity across the cut.  We illustrate this fact in this figure by taking $N$ as a somewhat arbitrary union of straight line segments instead of any natural trajectory of $h'(\lambda;\chi,\tau)^2\,\dd\lambda^2$.  Crucially, this issue plays no role in the subsequent analysis, because it is never necessary to factor the jump matrix carried by $N$.}
\label{fig:Stau1}
\end{figure}

\subsection{Introduction of $g$ and steepest descent deformation of the Riemann-Hilbert problem}
Now with the contours $\Gamma^\pm$ and $N$ set up in this way, we introduce the $g$-function via the transformation \eqref{eq:T-to-S} 
taking $\mathbf{S}^{(k)}(\lambda;\chi,\tau)$ to $\mathbf{T}^{(k)}(\lambda;\chi,\tau)$.  We assume that the Schwarz-symmetric arc $\Sigma_g$ where $g(\lambda;\chi,\tau)$ fails to be analytic coincides with $N$, which in turn is a sub-arc of $\Sigma_\mathrm{c}$.  Therefore, we need the version of the construction of $g(\lambda;\chi,\tau)$ described in Section~\ref{sec:g-function-dumbbell}.  Recalling from \eqref{eq:hpm-gamma} that $h_+(\lambda;\chi,\tau)+h_-(\lambda;\chi,\tau)=2\gamma(\chi,\tau)$ for $\lambda\in\Sigma_g=N$, where 
$\gamma(\chi,\tau)$
is a real quantity 
given by \eqref{eq:gamma-formula}
we obtain from \eqref{eq:S-N-jump-ALT} the jump condition for $\mathbf{T}^{(k)}(\lambda;\chi,\tau)$ along $N$ in the form
\begin{equation}
\mathbf{T}^{(k)}_+(\lambda;\chi,\tau)=\mathbf{T}^{(k)}_-(\lambda;\chi,\tau)\begin{bmatrix}
0 & \ii\ee^{-2\ii M\gamma(\chi,\tau)}\\\ii\ee^{2\ii M\gamma(\chi,\tau)} & 0\end{bmatrix},\quad\lambda\in N.
\label{eq:T-jump-N-Schi-Stau-ALT}
\end{equation}
We next take advantage of the fact that $\mathrm{Re}(\ii h(\lambda;\chi,\tau))=0$ on $\Gamma^\pm$ to transform $\mathbf{T}^{(k)}(\lambda;\chi,\tau)$ explicitly into $\mathbf{W}^{(k)}(\lambda;\chi,\tau)$ by a substitution based on the same elementary factorization \eqref{eq:Q-factorizations} used in Section~\ref{sec:channels}.  Let $\Omega^\pm$ denote small disks centered at $\lambda=\pm\ii$ and enclosed by $\Gamma^\pm$ respectively, let $R^\pm$ denote the interior of $\Gamma^\pm$ with the closure of $\Omega^\pm$ excluded, and let $L^\pm$ denote lens-shaped regions on the exterior of $\Gamma^\pm$ as shown in the left-hand panels of Figures~\ref{fig:Schi1}--\ref{fig:Stau1}.  Then, we make the following definition (compare with \eqref{eq:T-S-L-plus-ALT}--\eqref{eq:T-S-L-minus-ALT}):
\begin{equation}
\mathbf{W}^{(k)}(\lambda;\chi,\tau)\defeq\mathbf{T}^{(k)}(\lambda;\chi,\tau)\begin{bmatrix}
1 & 0\\ s\ee^{2\ii Mh(\lambda;\chi,\tau)} & 1\end{bmatrix},\quad\lambda\in L^+,
\label{eq:W-def-Schi-Stau-Lplus-ALT}
\end{equation}
\begin{equation}
\mathbf{W}^{(k)}(\lambda;\chi,\tau)\defeq\mathbf{T}^{(k)}(\lambda;\chi,\tau)2^{\frac{1}{2}\sigma_3}\begin{bmatrix}1 & \tfrac{1}{2}s\ee^{-2\ii M h(\lambda;\chi,\tau)}\\ 0 & 1\end{bmatrix},\quad\lambda\in R^+,
\label{eq:W-def-Schi-Stau-Rplus-ALT}
\end{equation}
\begin{equation}
\mathbf{W}^{(k)}(\lambda;\chi,\tau)\defeq\mathbf{T}^{(k)}(\lambda;\chi,\tau)2^{\frac{1}{2}\sigma_3},\quad
\lambda\in\Omega^+,
\end{equation}
\begin{equation}
\mathbf{W}^{(k)}(\lambda;\chi,\tau)\defeq\mathbf{T}^{(k)}(\lambda;\chi,\tau)2^{-\frac{1}{2}\sigma_3},\quad
\lambda\in\Omega^-,
\end{equation}
\begin{equation}
\mathbf{W}^{(k)}(\lambda;\chi,\tau)\defeq\mathbf{T}^{(k)}(\lambda;\chi,\tau)2^{-\frac{1}{2}\sigma_3}\begin{bmatrix} 1 & 0\\-\tfrac{1}{2}s\ee^{2\ii Mh(\lambda;\chi,\tau)} & 1\end{bmatrix},\quad\lambda\in R^-,\quad\text{and}
\label{eq:W-def-Schi-Stau-Rminus-ALT}
\end{equation}
\begin{equation}
\mathbf{W}^{(k)}(\lambda;\chi,\tau)\defeq\mathbf{T}^{(k)}(\lambda;\chi,\tau)\begin{bmatrix}
1 & -s\ee^{-2\ii Mh(\lambda;\chi,\tau)}\\ 0 & 1\end{bmatrix},\quad
\lambda\in L^-,
 \label{eq:W-def-Schi-Stau-Lminus-ALT}
\end{equation}
and elsewhere that $\mathbf{T}^{(k)}(\lambda;\chi,\tau)$ is defined we set $\mathbf{W}^{(k)}(\lambda;\chi,\tau)\defeq\mathbf{T}^{(k)}(\lambda;\chi,\tau)$.  One can check that $\mathbf{W}^{(k)}(\lambda;\chi,\tau)$ can be defined on $\Gamma^\pm$ to be analytic there.  Taking into account that 
$\ee^{\pm 2\ii Mh(\lambda;\chi,\tau)}$
has jump discontinuities across arcs $N^\pm$ within the annular domains $R^\pm$, 
the jump contour for $\mathbf{W}^{(k)}(\lambda;\chi,\tau)$ is as shown in the right-hand panels of Figures~\ref{fig:Schi1}--\ref{fig:Stau1}.
The jump conditions satisfied by $\mathbf{W}^{(k)}(\lambda;\chi,\tau)$ are then the following.  Firstly, since $\mathbf{W}^{(k)}(\lambda;\chi,\tau)=\mathbf{T}^{(k)}(\lambda;\chi,\tau)$ holds for both boundary values taken along $N$, the same jump condition \eqref{eq:T-jump-N-Schi-Stau-ALT} holds for $\mathbf{W}^{(k)}(\lambda;\chi,\tau)$ also.  Next, comparing with \eqref{eq:Tjump-channels-CLplus-ALT}--\eqref{eq:Tjump-channels-CRplus-ALT} and \eqref{eq:Tjump-channels-CRminus-ALT}--\eqref{eq:Tjump-channels-CLminus-ALT}, we have
\begin{equation}
\mathbf{W}^{(k)}_+(\lambda;\chi,\tau)=\mathbf{W}^{(k)}_-(\lambda;\chi,\tau)\begin{bmatrix}1 & 0\\
-s\ee^{2\ii Mh(\lambda;\chi,\tau)} & 1\end{bmatrix},\quad\lambda\in C_L^+,
\label{eq:Wjump-exterior-CLplus}
\end{equation}
\begin{equation}
\mathbf{W}^{(k)}_+(\lambda;\chi,\tau)=\mathbf{W}^{(k)}_-(\lambda;\chi,\tau)\begin{bmatrix}1 & \tfrac{1}{2}s\ee^{-2\ii Mh(\lambda;\chi,\tau)} \\ 0 & 1\end{bmatrix},\quad\lambda\in C_R^+,
\end{equation}
\begin{equation}
\mathbf{W}^{(k)}_+(\lambda;\chi,\tau)=\mathbf{W}^{(k)}_-(\lambda;\chi,\tau)\begin{bmatrix}1 & 0\\
-\tfrac{1}{2}s\ee^{2\ii Mh(\lambda;\chi,\tau)} & 1\end{bmatrix},\quad
\lambda\in C_R^-,\quad\text{and}
\end{equation}
\begin{equation}
\mathbf{W}^{(k)}_+(\lambda;\chi,\tau)=\mathbf{W}^{(k)}_-(\lambda;\chi,\tau)\begin{bmatrix} 1 & s\ee^{-2\ii Mh(\lambda;\chi,\tau)}\\ 0 & 1\end{bmatrix},\quad\lambda\in C_L^-.
\label{eq:Wjump-exterior-CLminus}
\end{equation}
Finally, for $\lambda\in N^\pm$ we compute 
\begin{equation}
\mathbf{W}^{(k)}_+(\lambda;\chi,\tau)=\mathbf{W}_-^{(k)}(\lambda;\chi,\tau)\begin{bmatrix}1 & \tfrac{1}{2}s\left(\ee^{-2\ii Mh_+(\lambda;\chi,\tau)}-\ee^{-2\ii Mh_-(\lambda;\chi,\tau)}\right)\\0 & 1\end{bmatrix},\quad
\lambda\in N^+,\quad\text{and}
\end{equation}
\begin{equation}
\mathbf{W}^{(k)}_+(\lambda;\chi,\tau)=\mathbf{W}^{(k)}_-(\lambda;\chi,\tau)\begin{bmatrix}1 & 0\\
-\tfrac{1}{2}s\left(\ee^{2\ii Mh_+(\lambda;\chi,\tau)}-\ee^{2\ii Mh_-(\lambda;\chi,\tau)}\right) & 1\end{bmatrix},\quad\lambda\in N^-.
\end{equation}
Then, since for $\lambda\in N^\pm$, $g(\lambda;\chi,\tau)$ has no jump discontinuity and $\vartheta_+(\lambda;\chi,\tau)-\vartheta_-(\lambda;\chi,\tau)=-2\pi$, and since $M=\tfrac{1}{2}k+\tfrac{1}{4}$ for $k\in\mathbb{Z}_{>0}$, these simplify to
\begin{equation}
\mathbf{W}^{(k)}_+(\lambda;\chi,\tau)=\mathbf{W}^{(k)}_-(\lambda;\chi,\tau)\begin{bmatrix}1 & \ii\ee^{-\ii M(h_+(\lambda;\chi,\tau)+h_-(\lambda;\chi,\tau))}\\0 & 1\end{bmatrix},\quad\lambda\in N^+,\quad\text{and}
\end{equation}
\begin{equation}
\mathbf{W}^{(k)}_+(\lambda;\chi,\tau)=\mathbf{W}^{(k)}_-(\lambda;\chi,\tau)\begin{bmatrix}1 & 0\\
\ii\ee^{\ii M(h_+(\lambda;\chi,\tau)+h_-(\lambda;\chi,\tau))} & 1\end{bmatrix},\quad\lambda\in N^-.
\end{equation}

\subsection{Parametrix construction}
\label{sec:Airy-parametrix}
From the sign structure of $\mathrm{Re}(\ii h(\lambda;\chi,\tau))$ as indicated with shading in Figures~\ref{fig:Schi1}--\ref{fig:Stau1}, it is then clear that the jump matrices are exponentially small perturbations of the identity matrix except when $\lambda\in N=\Sigma_g$ and in small neighborhoods of the branch points $\lambda=A\pm \ii B$.  To deal with these, we first construct an \emph{outer parametrix} denoted 
$\dot{\mathbf{W}}^{(k),\mathrm{out}}(\lambda;\chi,\tau)$
designed to solve the jump condition 
\eqref{eq:T-jump-N-Schi-Stau-ALT} for $\lambda\in\Sigma_g$ exactly, to be analytic for $\lambda\in\mathbb{C}\setminus\Sigma_g$, and to tend to the identity as $\lambda\to\infty$.  This is easily accomplished simply by diagonalization of the constant jump matrix, the eigenvalues of which are $\pm \ii$.  All solutions of the jump condition 
\eqref{eq:T-jump-N-Schi-Stau-ALT}
have singularities at the endpoints of $\Sigma_g$, and we select the unique solution with the mildest rate of growth at these two points:
\begin{equation}
\dot{\mathbf{W}}^{(k),\mathrm{out}}(\lambda;\chi,\tau)\defeq\\
\ee^{-\ii M\gamma(\chi,\tau)\sigma_3}\mathbf{Q}y(\lambda;\chi,\tau)^{\sigma_3}\mathbf{Q}^{-1}\ee^{\ii M\gamma(\chi,\tau)\sigma_3},
\label{eq:outer-parametrix-Schi-Stau-ALT}
\end{equation}
where $\mathbf{Q}$ is the matrix defined in \eqref{eq:Q-def}, and where $y(\lambda;\chi,\tau)$ is the function analytic for $\lambda\in\mathbb{C}\setminus\Sigma_g$ determined by the conditions
\begin{equation}
y(\lambda;\chi,\tau)^4=\frac{\lambda-\lambda_0(\chi,\tau)}{\lambda-\lambda_0(\chi,\tau)^*},\quad\text{and $y(\lambda;\chi,\tau)\to 1$ as $\lambda\to\infty$}.
\label{eq:y-def}
\end{equation}
Note that the only dependence on 
$M$
enters via the oscillatory factors 
$\ee^{\pm\ii M\gamma(\chi,\tau)\sigma_3}$,
so the outer parametrix
$\dot{\mathbf{W}}^{(k),\mathrm{out}}(\lambda;\chi,\tau)$
is bounded as 
$M\to\infty$,
provided that $\lambda$ is bounded away from $\lambda_0(\chi,\tau)$ and $\lambda_0(\chi,\tau)^*$.

Next, we let $D_{\lambda_0}(\delta)$ and $D_{\lambda_0^*}(\delta)=D_{\lambda_0}(\delta)^*$ be disks of small radius $\delta$ independent of 
$M$
centered at $\lambda=\lambda_0(\chi,\tau)=A(\chi,\tau)+\ii B(\chi,\tau)$ and $\lambda=\lambda_0(\chi,\tau)^*$ respectively.  
Since $h'(\lambda;\chi,\tau)$ vanishes like a square root as $\lambda\to \lambda_0(\chi,\tau)$ and $h_+(\lambda_0(\chi,\tau);\chi,\tau)+h_-(\lambda_0(\chi,\tau);\chi,\tau)=2\gamma(\chi,\tau)$, there is a univalent function $f_{\lambda_0}(\lambda;\chi,\tau)$ defined on $D_{\lambda_0}(\delta)$ with $f_{\lambda_0}(\lambda_0(\chi,\tau);\chi,\tau)=0$ such that 
\begin{equation}
f_{\lambda_0}(\lambda;\chi,\tau)^3=-(h_+(\lambda;\chi,\tau)+h_-(\lambda;\chi,\tau)-2\gamma(\chi,\tau))^2,\quad \lambda\in D_{\lambda_0}(\delta),
\label{eq:Airy-map-Schi-Stau-ALT}
\end{equation}
in which the sum of boundary values of $h$ is analytically continued from $N^+$ to $D_{\lambda_0}(\delta)\setminus N$ by means of the identity $h_+(\lambda;\chi,\tau)-h_-(\lambda;\chi,\tau)=-2\pi$ for $\lambda\in N^+$.
Moreover, the univalent solution of \eqref{eq:Airy-map-Schi-Stau-ALT} and the jump contours $N\cap D_{\lambda_0}(\delta)$, $N^+\cap D_{\lambda_0}(\delta)$, and $C_L^+\cap D_{\lambda_0}(\delta)$ can be chosen so that $\lambda\in N\cap D_{\lambda_0}(\delta)$ implies $f_{\lambda_0}(\lambda;\chi,\tau)<0$, $\lambda\in N^+\cap D_{\lambda_0}(\delta)$ implies $f_{\lambda_0}(\lambda;\chi,\tau)>0$, and $\lambda\in C_L^+\cap D_{\lambda_0}(\delta)$ implies that either $\arg(f_{\lambda_0}(\lambda;\chi,\tau))=\tfrac{2}{3}\pi$ or $\arg(f_{\lambda_0}(\lambda;\chi,\tau))=-\tfrac{2}{3}\pi$.  
Define a matrix $\mathbf{X}^{(k)}(\lambda;\chi,\tau)$ within $D_{\lambda_0}(\delta)$ by
\begin{equation}
\mathbf{X}^{(k)}(\lambda;\chi,\tau)\defeq\mathbf{W}^{(k)}(\lambda;\chi,\tau)\ee^{-\ii M\gamma(\chi,\tau)\sigma_3}\ee^{\frac{1}{4}\ii\pi\sigma_3},\quad\lambda\in D_{\lambda_0}(\delta).
\label{eq:Y-from-X-Schi-Stau-ALT}
\end{equation}
Then, using again $M=\tfrac{1}{2}k+\tfrac{1}{4}$ for $k\in\mathbb{Z}_{>0}$, the jump conditions satisfied by 
$\mathbf{X}^{(k)}(\lambda;\chi,\tau)$
can be written in a simple form, in terms of the variable (rescaled conformal coordinate on $D_{\lambda_0}(\delta)$) 
$\zeta\defeq M^\frac{2}{3}f_{\lambda_0}(\lambda;\chi,\tau)$:
\begin{equation}
\mathbf{X}^{(k)}_+(\lambda;\chi,\tau)=\mathbf{X}^{(k)}_-(\lambda;\chi,\tau)\begin{bmatrix}1 &\ee^{-\zeta^\frac{3}{2}} \\ 0 & 1\end{bmatrix},\quad \arg(\zeta)=0,
\label{eq:Airy-jump-first}
\end{equation}
\begin{equation}
\mathbf{X}^{(k)}_+(\lambda;\chi,\tau)=\mathbf{X}^{(k)}_-(\lambda;\chi,\tau)\begin{bmatrix}1 & 0\\-\ee^{\zeta^\frac{3}{2}} & 1\end{bmatrix},\quad\arg(\zeta)=\pm\tfrac{2}{3}\pi,\quad\text{and}
\end{equation}
\begin{equation}
\mathbf{X}^{(k)}_+(\lambda;\chi,\tau)=\mathbf{X}^{(k)}_-(\lambda;\chi,\tau)\begin{bmatrix}0 & -1\\
1 & 0\end{bmatrix},\quad\arg(-\zeta)=0,
\label{eq:Airy-jump-last}
\end{equation}
where for uniformity all four rays are taken to be oriented away from the origin in the $\zeta$-plane.
There exists a unique matrix function $\mathbf{A}(\zeta)$ with the following properties:
\begin{itemize}
\item
$\mathbf{A}(\zeta)$ is analytic for $0<|\arg(\zeta)|<\tfrac{2}{3}\pi$ and $\tfrac{2}{3}\pi<|\arg(\zeta)|<\pi$ (four sectors);
\item $\mathbf{A}(\zeta)$ takes continuous boundary values from each sector satisfying the same jump conditions written in \eqref{eq:Airy-jump-first}--\eqref{eq:Airy-jump-last};
\item $\mathbf{A}(\zeta)$ has uniform asymptotics in all directions of the complex plane given by 
\begin{equation}
\mathbf{A}(\zeta)\ee^{-\frac{1}{4}\ii\pi\sigma_3}\mathbf{Q}\ee^{\frac{1}{4}\ii\pi\sigma_3}\zeta^{-\frac{1}{4}\sigma_3}=\begin{bmatrix}1+O(\zeta^{-3}) & O(\zeta^{-1})\\O(\zeta^{-2}) & 1+O(\zeta^{-3})\end{bmatrix},\quad\zeta\to\infty,
\end{equation}
where $\mathbf{Q}$ is the matrix defined in \eqref{eq:Q-def}.
\end{itemize}
It is well-known that the unique solution of these Riemann-Hilbert conditions can be written explicitly in terms of Airy functions, and the reader can find a complete development of the solution in \cite[Appendix B]{BothnerM20}.
Next, we define the matrix function
\begin{equation}
\mathbf{H}(\lambda;\chi,\tau)\defeq\ee^{\ii M\gamma(\chi,\tau)\sigma_3}\dot{\mathbf{W}}^{(k),\mathrm{out}}(\lambda;\chi,\tau)\ee^{-\ii M\gamma(\chi,\tau)\sigma_3}\mathbf{Q}f_{\lambda_0}(\lambda;\chi,\tau)^{-\frac{1}{4}\sigma_3}\ee^{\frac{1}{4}\ii\pi\sigma_3}
\end{equation}
and note that it follows from the definition of the conformal map $\lambda\mapsto f_{\lambda_0}(\lambda;\chi,\tau)$ and the definition \eqref{eq:outer-parametrix-Schi-Stau-ALT} of $\dot{\mathbf{W}}^{(k),\mathrm{out}}(\lambda;\chi,\tau)$ that $\mathbf{H}(\lambda;\chi,\tau)$ is analytic for $\lambda\in D_{\lambda_0}(\delta)$ and is independent of $M$.
We use $\mathbf{H}(\lambda;\chi,\tau)$ and $\mathbf{A}(\zeta)$ to define an \emph{inner parametrix} on $D_{\lambda_0}(\delta)$ as follows:
\begin{equation}
\dot{\mathbf{W}}^{(k),\lambda_0}(\lambda;\chi,\tau):=\ee^{-\ii M\gamma(\chi,\tau)\sigma_3}\mathbf{H}(\lambda;\chi,\tau)M^{-\frac{1}{6}\sigma_3}\mathbf{A}(M^\frac{2}{3}f_{\lambda_0}(\lambda;\chi,\tau))\ee^{-\frac{1}{4}\ii\pi\sigma_3}\ee^{\ii M\gamma(\chi,\tau)\sigma_3},\quad\lambda\in D_{\lambda_0}(\delta).
\end{equation}
It is easy to check that 
$\dot{\mathbf{W}}^{(k),\lambda_0}(\lambda;\chi,\tau)$
takes continuous boundary values that satisfy exactly the same jump conditions within $D_{\lambda_0}(\delta)$ as do those of 
$\mathbf{W}^{(k)}(\lambda;\chi,\tau)$
itself. Also, since $\zeta$ is large of size 
$M^\frac{2}{3}$
when $\lambda\in \partial D_{\lambda_0}(\delta)$, 
\begin{multline}
\dot{\mathbf{W}}^{(k),\lambda_0}(\lambda;\chi,\tau)\dot{\mathbf{W}}^{(k),\mathrm{out}}(\lambda;\chi,\tau)^{-1} \\
\begin{aligned}
&= \ee^{-\ii M\gamma(\chi,\tau)\sigma_3}\mathbf{H}(\lambda;\chi,\tau)M^{-\frac{1}{6}\sigma_3}\mathbf{A}(\zeta)\ee^{-\frac{1}{4}\ii\pi\sigma_3}\mathbf{Q}\ee^{\frac{1}{4}\ii\pi\sigma_3}\zeta^{-\frac{1}{4}\sigma_3}M^{\frac{1}{6}\sigma_3}\mathbf{H}(\lambda;\chi,\tau)^{-1}\ee^{\ii M\gamma(\chi,\tau)\sigma_3}\\
&=\ee^{-\ii M\gamma(\chi,\tau)\sigma_3}\mathbf{H}(\lambda;\chi,\tau)\begin{bmatrix}1+O(\zeta^{-3}) & O(M^{-\frac{1}{3}}\zeta^{-1})\\O(M^\frac{1}{3}\zeta^{-2}) & 1+O(\zeta^{-3})\end{bmatrix}
\mathbf{H}(\lambda;\chi,\tau)^{-1}\ee^{\ii M\gamma(\chi,\tau)\sigma_3}\\
&=\mathbb{I}+O(M^{-1}),\quad\lambda\in\partial D_{\lambda_0}(\delta).
\end{aligned}
\label{eq:Dplus-mismatch-Schi-Stau-ALT}
\end{multline}
Since the matrix $\mathbf{W}^{(k)}(\lambda;\chi,\tau)$ satisfies $\mathbf{W}^{(k)}(\lambda^*;\chi,\tau)=\sigma_2\mathbf{W}^{(k)}(\lambda;\chi,\tau)^*\sigma_2$, 
we may define a second inner parametrix for $\lambda\in D_{\lambda_0^*}(\delta)$ to respect this symmetry.

We combine the inner and outer parametrices into a \emph{global parametrix} by setting
\begin{equation}
\dot{\mathbf{W}}^{(k)}(\lambda;\chi,\tau)\defeq\begin{cases}
\dot{\mathbf{W}}^{(k),\lambda_0}(\lambda;\chi,\tau),&\quad\lambda\in D_{\lambda_0}(\delta),\\
\sigma_2\dot{\mathbf{W}}^{(k),\lambda_0}(\lambda^*;\chi,\tau)^*\sigma_2,&\quad\lambda\in D_{\lambda_0^*}(\delta),\\
\dot{\mathbf{W}}^{(k),\mathrm{out}}(\lambda;\chi,\tau),&\quad\lambda\in\mathbb{C}\setminus(\overline{D_{\lambda_0}(\delta)\cup D_{\lambda_0^*}(\delta)}\cup\Sigma_g).
\end{cases}
\label{eq:global-parametrix-Schi-Stau-ALT}
\end{equation}

\subsection{Error analysis and asymptotic formula for $\psi_k(M\chi,M\tau)$ for $(\chi,\tau)\in \exterior$.}
As in Section~\ref{sec:small-norm-channels}, we define an error matrix to compare $\mathbf{W}^{(k)}(\lambda;\chi,\tau)$ with its global parametrix defined in \eqref{eq:global-parametrix-Schi-Stau-ALT}:
\begin{equation}
\mathbf{F}^{(k)}(\lambda;\chi,\tau):=\mathbf{W}^{(k)}(\lambda;\chi,\tau)\dot{\mathbf{W}}^{(k)}(\lambda;\chi,\tau)^{-1}.
\end{equation}
This matrix can be considered to be analytic in $\lambda$ except on a contour $\Sigma_\mathbf{F}$ consisting of the union of (i) those arcs of the jump contour for $\mathbf{W}^{(k)}(\lambda;\chi,\tau)$ other than $\Sigma_g$ outside the disks $D_{\lambda_0}(\delta)$ and $D_{\lambda_0^*}(\delta)$ and (ii) the disk boundaries $\partial D_{\lambda_0}(\delta)$ and $\partial D_{\lambda_0^*}(\delta)$, which we take to have clockwise orientation.  Also, $\mathbf{F}^{(k)}(\lambda;\chi,\tau)$ takes continuous boundary values on $\Sigma_\mathbf{F}$ from each connected component of $\mathbb{C}\setminus\Sigma_\mathbf{F}$, and $\mathbf{F}^{(k)}(\lambda;\chi,\tau)\to\mathbb{I}$ as $\lambda\to\infty$. Because $\delta>0$ is held fixed as $M\to\infty$ and because the outer parametrix is uniformly bounded on arcs of type (i), there is a constant $\nu>0$ such that on those arcs we have the uniform estimate $\mathbf{F}^{(k)}_+(\lambda;\chi,\tau)=\mathbf{F}^{(k)}_-(\lambda;\chi,\tau)(\mathbb{I}+O(\ee^{-\nu M}))$.  On the circular arcs of type (ii), the estimate \eqref{eq:Dplus-mismatch-Schi-Stau-ALT} and its Schwarz reflection guarantee that on those arcs we have the uniform estimate $\mathbf{F}^{(k)}_+(\lambda;\chi,\tau)=\mathbf{F}^{(k)}_-(\lambda;\chi,\tau)(\mathbb{I}+O(M^{-1}))$. By small-norm theory it then follows that $\mathbf{F}^{(k)}_-(\lambda;\chi,\tau)=\mathbb{I}+O(M^{-1})$ holds in the $L^2$ sense on the union of arcs of types (i) and (ii) as $M\to+\infty$.  Using the Cauchy integral representation \eqref{eq:F-Cauchy-channels} then shows that $\mathbf{F}^{(k)}(\lambda;\chi,\tau)=\mathbb{I}+\lambda^{-1}\mathbf{F}_1^{(k)}(\chi,\tau) + O(\lambda^{-2})$ as $\lambda\to\infty$ where $\mathbf{F}_1^{(k)}(\chi,\tau)=O(M^{-1})$ holds uniformly for $(\chi,\tau)$ in compact subsets of $\exterior$.

Using the fact that for $|\lambda|$ sufficiently large, 
$\mathbf{S}^{(k)}(\lambda;\chi,\tau)=\mathbf{W}^{(k)}(\lambda;\chi,\tau)\ee^{-\ii Mg(\lambda;\chi,\tau)\sigma_3}$ while $\dot{\mathbf{W}}^{(k)}(\lambda;\chi,\tau)=\dot{\mathbf{W}}^{(k),\mathrm{out}}(\lambda;\chi,\tau)$,
from \eqref{eq:psi-k-S} we have the following exact formula
\begin{equation}
\begin{split}
\psi_k(M\chi,M\tau)&=2\ii\ee^{-\ii M\tau}\lim_{\lambda\to\infty}\lambda W^{(k)}_{12}(\lambda;\chi,\tau)\ee^{\ii Mg(\lambda;\chi,\tau)}\\
&=2\ii\ee^{-\ii M\tau}\lim_{\lambda\to\infty}\lambda\left[F^{(k)}_{11}(\lambda;\chi,\tau)\dot{W}^{(k),\mathrm{out}}_{12}(\lambda;\chi,\tau)\right.\\
&\qquad\qquad\qquad\qquad\qquad\qquad{}\left.+F^{(k)}_{12}(\lambda;\chi,\tau)\dot{W}^{(k),\mathrm{out}}_{22}(\lambda;\chi,\tau)\right]\ee^{\ii Mg(\lambda;\chi,\tau)}.
\end{split}
\end{equation}
Since $\mathbf{F}^{(k)}(\lambda;\chi,\tau)\to\mathbb{I}$, $\dot{\mathbf{W}}^{(k),\mathrm{out}}(\lambda;\chi,\tau)\to\mathbb{I}$, and $g(\lambda;\chi,\tau)\to 0$ as $\lambda\to\infty$, this simplifies to
\begin{equation}
\psi_k(M\chi,M\tau)=2\ii\ee^{-\ii M\tau}\lim_{\lambda\to\infty}\lambda\left[\dot{W}^{(k),\mathrm{out}}_{12}(\lambda;\chi,\tau) +F^{(k)}_{12}(\lambda;\chi,\tau)\right].
\end{equation}
Using \eqref{eq:outer-parametrix-Schi-Stau-ALT} and that $\mathbf{F}_1^{(k)}(\chi,\tau)=O(M^{-1})$, recalling $B(\chi,\tau)=\mathrm{Im}(\lambda_0(\chi,\tau))>0$ we obtain 
\begin{equation}
\psi_k(M\chi,M\tau)=B(\chi,\tau)\ee^{-\ii M\tau}\ee^{-2\ii M\gamma(\chi,\tau)} + O(M^{-1}),
\label{eq:psi-k-shelves-chi-tau-ALT}
\end{equation}
which completes the proof of Theorem~\ref{thm:exterior}.

\section{Far-Field Asymptotic Behavior in the Domain $\shelves$}
In this section, we prove Theorem~\ref{thm:shelves} and its corollaries. 
Our analysis is valid for all $M\in \mathbb{Z}_{>0}$, with $\mathbf{G}=\mathbf{Q}^{-s}$, $s=\pm 1$ in contrast to that in the preceding section. The analysis will be guided by the sign chart of $\Re(\ii h(\lambda;\chi,\tau))$, $h(\lambda;\chi,\tau) = g(\lambda;\chi,\tau) +  \vartheta(\lambda;\chi,\tau)$.
Recall from the discussion in Section~\ref{sec:GenusZeroModification} that for $(\chi,\tau)\in\shelves$, $h'(\lambda;\chi,\tau)^2$ has 2 real double roots denoted by $a(\chi,\tau)<b(\chi,\tau)$, and two simple roots $A(\chi,\tau)\pm \ii B(\chi,\tau)$ for which we write $\lambda_0(\chi,\tau)\defeq  A(\chi,\tau) + \ii B(\chi,\tau)$, where $B(\chi,\tau)>0$, and we also have $A(\chi,\tau)<0$ because the endpoints $A(\chi,\tau) \pm \ii B(\chi,\tau)$ of $\Sigma_g$, along which $g(\lambda;\chi,\tau)$ has a jump discontinuity, lie in the left half-plane for all $(\chi,\tau)$ in the interior of $\shelves$. We recall from the beginning of Section~\ref{sec:zero-level-curve} that $\mathrm{Re}(\ii h(\lambda;\chi,\tau))=0$ holds for all $\lambda\in\mathbb{R} \setminus \Sigma_g$. Therefore, both $h_{-}(a(\chi,\tau);\chi,\tau)$ and $h(b(\chi,\tau);\chi,\tau)$ are real-valued. 
We also note the facts
\begin{equation}
h''_-(a(\chi,\tau);\chi,\tau) < 0 \quad\text{and}\quad h''(b(\chi,\tau);\chi,\tau) > 0,
\label{eq:h-double-prime-a-b-signs}
\end{equation}
which follow from the formula \eqref{eq:hprime-formula} since $\tau>0$ and $a(\chi,\tau),b(\chi,\tau)$ are real roots of the quadratic in the numerator of \eqref{eq:hprime-formula} for $(\chi,\tau)$ in the interior of $\shelves$.


Recall from Section~\ref{sec:zero-level-curve} that there is a Schwarz-symmetric arc of the zero level curve of $\lambda\mapsto \Re(\ii h(\lambda;\chi,\tau))$ that connects $\lambda_0(\chi,\tau)$ to $\lambda_0(\chi,\tau)^*$ and passes through the point $a(\chi,\tau)\in\mathbb{R}$. We place the branch cut $\Sigma_g$ on this curve, denote by $\Sigma_g^{\pm}$ its subarcs that lie in the half-planes $\mathbb{C}^{\pm}$, and orient $\Sigma_g$ from $\lambda_0^*$ to $\lambda_0$. The other bounded trajectory of the zero level curve in the upper half plane is one that connects $\lambda_0(\chi,\tau)$ to $b(\chi,\tau)$. We denote this arc by $\Gamma^+ = \Gamma^+(\chi,\tau)$, and its Schwarz reflection by $\Gamma^- = \Gamma^-(\chi,\tau)$, both with downward orientation. We also set $\Gamma\defeq  \Gamma^+ \cup \Gamma^- \cup \{  b(\chi,\tau) \}$. We set $I\defeq [a(\chi,\tau), b(\chi,\tau)]$ to
denote the only remaining bounded component of $\Re(\ii h(\lambda;\chi,\tau))=0$.
For the analysis that follows we take $\Sigma_\circ$ to be the clockwise-oriented loop $\Sigma_{\circ} = \Sigma_g \cup \Gamma$. We choose $\Sigma_\mathrm{c}$  to be a Schwarz-symmetric arc that connects the points $\lambda=\pm \ii$ while passing through the point $\lambda=\tfrac{1}{2}(a(\chi,\tau)+b(\chi,\tau))$, say (any point in $(a(\chi,\tau),b(\chi,\tau))$ would suffice), with upward orientation. See the left-hand panel of Figure~\ref{fig:SB1} for an illustration of these arcs.

\begin{figure}[h]
\begin{center}
\includegraphics{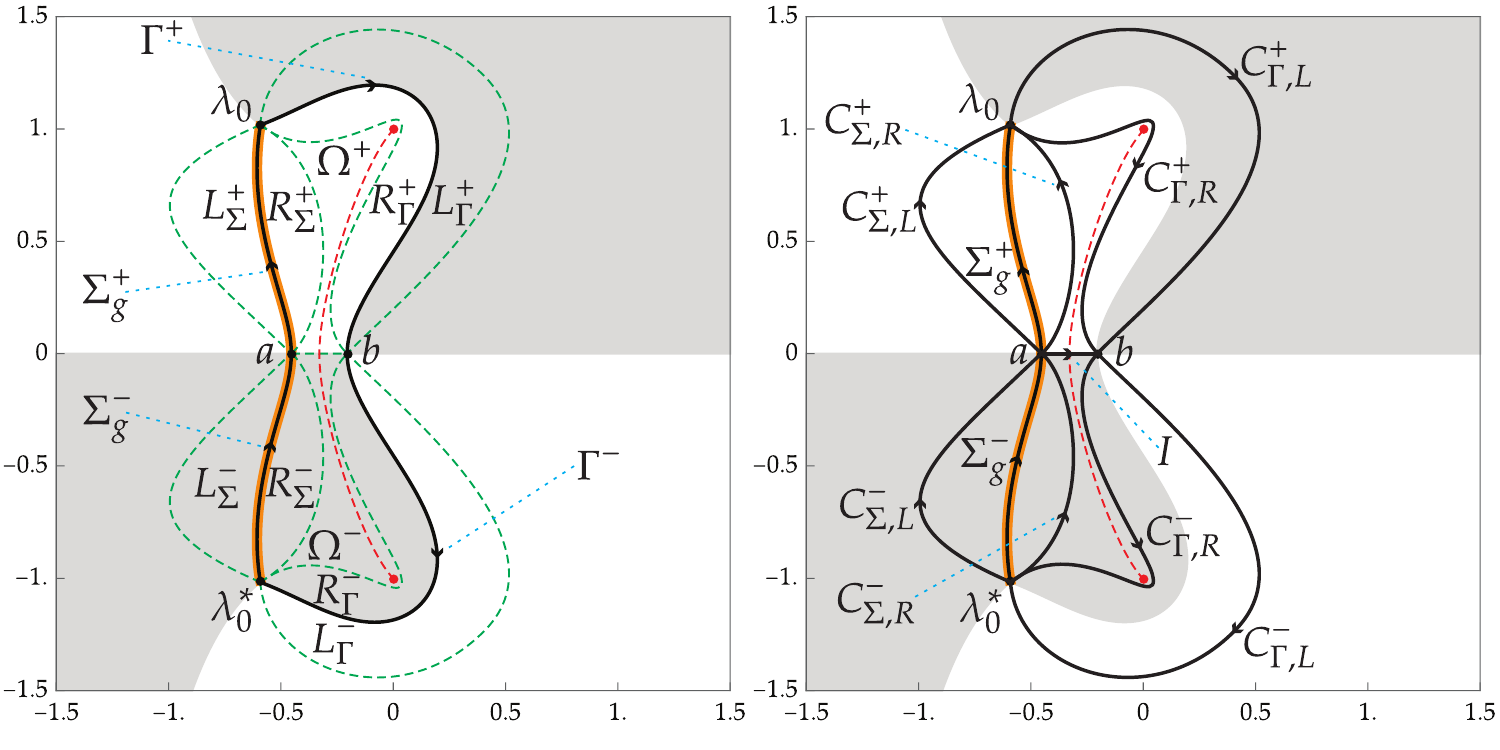}
\end{center}
\caption{Left:  the initial jump contour and sign chart of $\mathrm{Re}(\ii h(\lambda;\chi,\tau))$ (shaded for negative, unshaded for positive) for $(\chi,\tau)=(2,0.8)\in \shelves$, showing also the regions where explicit transformations are made.  Right:  the resulting jump contour after the transformations.}
\label{fig:SB1}
\end{figure}

\subsection{Introduction of $g$ and steepest descent deformation of the Riemann-Hilbert problem}
With contours chosen this way, we have $\Sigma_g \cup \Sigma_\mathrm{c} = \emptyset$; therefore, we need the version of the construction of $g(\lambda;\chi,\tau)$ described in Section~\ref{sec:g-function-loop}. We introduce the $g$-function and the matrix function $\mathbf{T}(\lambda;\chi,\tau,\mathbf{Q}^{-s}, M)$ by the global substitution \eqref{eq:T-to-S}.
We let $\Omega^+$ denote the domain enclosed by $\Sigma_g^+ \cup \Gamma^+ \cup I $ which contains $\lambda=\ii$ and $\Sigma_\mathrm{c} \cap \mathbb{C}^+$, and let $\Omega^{-}$ be its Schwarz reflection. 
We let $L_\Sigma^\pm$ and $L_\Gamma^\pm$ (resp., $R_\Sigma^\pm$ and $R_\Gamma^\pm$) denote lens-shaped regions lying to the left (resp., right) of $\Sigma_g^\pm$ and $\Gamma^\pm$ with respect to orientation, as depicted in the left-hand panel of Figure~\ref{fig:SB1}. 
The lens-shaped regions are chosen to be so thin as to exclude the points $\lambda=\pm \ii$ and $\Sigma_\mathrm{c}$ while supporting a fixed sign of $\Re(\ii h(\lambda;\chi,\tau))$. 
On $\Gamma^\pm$, we will use the two-factor factorizations of the central factor $\mathbf{Q}^{-s}$, exactly as written in \eqref{eq:Q-factorizations}.
However, along $\Sigma_g$, we will employ the following additional factorizations
\begin{equation}
\mathbf{Q}^{-s} = 
\begin{cases}
2^{\frac{1}{2}\sigma_3} \begin{bmatrix} 1 & -\frac{1}{2}s \\ 0 & 1 \end{bmatrix} \begin{bmatrix} 0 & s \\ -s & 0 \end{bmatrix} \begin{bmatrix} 1 & -s \\ 0 & 1 \end{bmatrix}, \quad& \lambda\in \Sigma_g^+,\vspace{0.33em}\\
 2^{-\frac{1}{2}\sigma_3} \begin{bmatrix} 1 & 0 \\ \frac{1}{2}s & 1 \end{bmatrix} \begin{bmatrix} 0 & s \\ -s & 0 \end{bmatrix} \begin{bmatrix} 1 & 0 \\ s & 1 \end{bmatrix},\quad& \lambda\in \Sigma_g^-.
\end{cases}
\label{eq:Q-triple-factorization}
\end{equation}
Taking advantage of the two-factor matrix factorizations \eqref{eq:Q-factorizations}, we define $\mathbf{W}(\lambda)=\mathbf{W}(\lambda;\chi,\tau,\mathbf{Q}^{-s},M)$ in the lens-shaped regions surrounding $\Gamma^\pm$ by:

\begin{equation}
\mathbf{W}(\lambda)\defeq \mathbf{T}(\lambda;\chi,\tau,\mathbf{Q}^{-s},M) 
\begin{bmatrix} 1 & 0 \\ s \ee^{2\ii M h(\lambda;\chi,\tau)}& 1 \end{bmatrix},\quad \lambda\in L^+_{\Gamma},
\label{eq:T-to-W-L-plus-Gamma}
\end{equation}
\begin{equation}
\mathbf{W}(\lambda)\defeq \mathbf{T}(\lambda;\chi,\tau,\mathbf{Q}^{-s},M) 
2^{\frac{1}{2}\sigma_3} \begin{bmatrix} 1 & \frac{1}{2} s\ee^{-2\ii M h(\lambda;\chi,\tau)} \\ 0 & 1 \end{bmatrix},\quad \lambda \in R^+_{\Gamma},
\end{equation}
%
\begin{equation}
\mathbf{W}(\lambda)\defeq \mathbf{T}(\lambda;\chi,\tau,\mathbf{Q}^{-s},M) 
2^{\frac{1}{2}\sigma_3},\quad \lambda \in \Omega^+,
\end{equation}
\begin{equation}
\mathbf{W}(\lambda)\defeq \mathbf{T}(\lambda;\chi,\tau,\mathbf{Q}^{-s},M) 
2^{-\frac{1}{2}\sigma_3},\quad \lambda \in \Omega^-,
\end{equation}
\begin{equation}
\mathbf{W}(\lambda)\defeq \mathbf{T}(\lambda;\chi,\tau,\mathbf{Q}^{-s},M) 
2^{-\sigma_3/2}  \begin{bmatrix} 1 & 0 \\ -\frac{1}{2}s \ee^{2\ii M h(\lambda;\chi,\tau)} & 1 \end{bmatrix} ,\quad \lambda \in R^-_{\Gamma},\quad\text{and}
\end{equation}
\begin{equation}
\mathbf{W}(\lambda)\defeq \mathbf{T}(\lambda;\chi,\tau,\mathbf{Q}^{-s},M) 
 \begin{bmatrix} 1 & - s \ee^{-2\ii M h(\lambda;\chi,\tau)} \\ 0 & 1 \end{bmatrix} ,\quad \lambda \in L^-_{\Gamma}.
\end{equation}
Note how the definitions above compare with \eqref{eq:W-def-Schi-Stau-Lplus-ALT}--\eqref{eq:W-def-Schi-Stau-Lminus-ALT} which use the same factorizations of $\mathbf{Q}^{-s}$: the regions $L^\pm$ and $R^\pm$ are merely replaced with $L_\Gamma^\pm$ and $R_\Gamma^\pm$, respectively. In the lens-shaped regions surrounding $\Sigma_g$ on the other hand, we make use of the factorizations \eqref{eq:Q-triple-factorization} and define:
\begin{equation}
\mathbf{W}(\lambda)\defeq \mathbf{T}(\lambda;\chi,\tau,\mathbf{Q}^{-s},M) 
 \begin{bmatrix} 1 & s  \ee^{-2\ii M h(\lambda;\chi,\tau)} \\ 0 & 1 \end{bmatrix},\quad \lambda \in L^+_{\Sigma},
\end{equation}
\begin{equation}
\mathbf{W}(\lambda)\defeq \mathbf{T}(\lambda;\chi,\tau,\mathbf{Q}^{-s},M) 
2^{\frac{1}{2}\sigma_3} \begin{bmatrix} 1 & -\frac{1}{2} s \ee^{-2\ii M h(\lambda;\chi,\tau)} \\ 0 & 1 \end{bmatrix},\quad \lambda \in R^+_{\Sigma},
\end{equation}
\begin{equation}
\mathbf{W}(\lambda)\defeq \mathbf{T}(\lambda;\chi,\tau,\mathbf{Q}^{-s},M) 
 \begin{bmatrix} 1 & 0 \\ -s  \ee^{2\ii M h(\lambda;\chi,\tau)} & 1 \end{bmatrix},\quad \lambda \in L^-_{\Sigma},\quad\text{and}
\end{equation}
\begin{equation}
\mathbf{W}(\lambda)\defeq \mathbf{T}(\lambda;\chi,\tau,\mathbf{Q}^{-s},M) 
 2^{-\frac{1}{2}\sigma_3} \begin{bmatrix} 1 & 0 \\ \frac{1}{2} s \ee^{2\ii M h(\lambda;\chi,\tau)} & 1 \end{bmatrix}
,\quad \lambda \in R^-_{\Sigma}.
\label{eq:T-to-W-R-minus-Sigma-ALT}
\end{equation}
We simply leave $\mathbf{W}(\lambda) \defeq  \mathbf{T}(\lambda;\chi,\tau,\mathbf{Q}^{-s},M) $ elsewhere. It is now easy to see that $\mathbf{W}(\lambda)$ extends to $\lambda\in \Gamma^+ \cup \Gamma^-$ as an analytic function, so that $\mathbf{W}(\lambda)$ is analytic in the complement of the jump contour $C^{+}_{\Sigma,L} \cup C^{+}_{\Sigma,R} \cup C^{-}_{\Sigma,L} \cup C^{-}_{\Sigma,R} \cup \Sigma_g \cup I \cup C^{+}_{\Gamma,L} \cup C^{+}_{\Gamma,R} \cup C^{-}_{\Gamma,L} \cup C^{-}_{\Gamma,R}$, the arcs of which are depicted in the right-hand panel of Figure~\ref{fig:SB1}. Across these arcs $\mathbf{W}(\lambda) $ satisfies the following jump relations:
\begin{equation}
\mathbf{W}_{+}(\lambda) = \mathbf{W}_{-}(\lambda)
\begin{bmatrix} 1 & 0 \\ -s \ee^{2\ii M h(\lambda;\chi,\tau)}& 1 \end{bmatrix},\quad \lambda\in C^+_{\Gamma,L},
\label{eq:W-jump-Gamma-plus-L}
\end{equation}
\begin{equation}
\mathbf{W}_{+}(\lambda) = \mathbf{W}_{-}(\lambda)
\begin{bmatrix} 1 & \frac{1}{2} s\ee^{-2\ii M h(\lambda;\chi,\tau)} \\ 0 & 1 \end{bmatrix},\quad \lambda\in C^+_{\Gamma,R},
\end{equation}
\begin{equation}
\mathbf{W}_{+}(\lambda) = \mathbf{W}_{-}(\lambda)
\begin{bmatrix} 1 & 0 \\ -\frac{1}{2}s  \ee^{2\ii M h(\lambda;\chi,\tau)} & 1 \end{bmatrix} ,\quad \lambda\in C^-_{\Gamma,R},\quad\text{and}
\end{equation}
\begin{equation}
\mathbf{W}_{+}(\lambda) = \mathbf{W}_{-}(\lambda)
\begin{bmatrix} 1 &  s \ee^{-2\ii M h(\lambda;\chi,\tau)} \\ 0 & 1 \end{bmatrix} ,\quad \lambda\in C^-_{\Gamma,L},
\end{equation}
which are again in parallel with the jump conditions \eqref{eq:Wjump-exterior-CLplus}--\eqref{eq:Wjump-exterior-CLminus}, and
\begin{equation}
\mathbf{W}_{+}(\lambda) = \mathbf{W}_{-}(\lambda)
2^{\sigma_3},\quad \lambda\in I ,
\label{eq:W-jump-I}
\end{equation}
\begin{equation}
\mathbf{W}_{+}(\lambda) = \mathbf{W}_{-}(\lambda)
\begin{bmatrix} 1 & -s  \ee^{-2\ii M h(\lambda;\chi,\tau)} \\ 0 & 1 \end{bmatrix},\quad \lambda\in C^+_{\Sigma,L},
\end{equation}
\begin{equation}
\mathbf{W}_{+}(\lambda) = \mathbf{W}_{-}(\lambda)
\begin{bmatrix} 1 & -\frac{1}{2} s \ee^{-2\ii M h(\lambda;\chi,\tau)} \\ 0 & 1 \end{bmatrix},\quad \lambda\in C^+_{\Sigma,R},
\end{equation}
\begin{equation}
\mathbf{W}_{+}(\lambda) = \mathbf{W}_{-}(\lambda)
\begin{bmatrix} 1 & 0 \\ s  \ee^{2\ii M h(\lambda;\chi,\tau)} & 1 \end{bmatrix},\quad \lambda\in C^-_{\Sigma,L},\quad\text{and}
\end{equation}
\begin{equation}
\mathbf{W}_{+}(\lambda) = \mathbf{W}_{-}(\lambda)
\begin{bmatrix} 1 & 0 \\ \frac{1}{2} s  \ee^{2\ii M h(\lambda;\chi,\tau)} & 1 \end{bmatrix},\quad \lambda\in C^-_{\Sigma,R}.
\label{eq:W-jump-Sigma-minus-R-ALT}
\end{equation}
Finally, along the branch cut $\Sigma_g$ we have
\begin{equation}
\mathbf{W}_{+}(\lambda) = \mathbf{W}_{-}(\lambda)
\begin{bmatrix} 0 & s  \ee^{-2\ii M \kappa(\chi,\tau))} \\ -s  \ee^{2 \ii M \kappa(\chi,\tau)} & 0 \end{bmatrix},\quad \lambda\in \Sigma_g = \Sigma_g^+ \cup \Sigma_g^-,
\label{eq:W-jump-Sigma-g}
\end{equation}
where $2\kappa(\chi,\tau)$ is the real constant value of $h_{+}(\lambda;\chi,\tau)+h_{-}(\lambda;\chi,\tau)$ for $\lambda\in\Sigma_g$ and $\kappa(\chi,\tau)$ is given in \eqref{eq:kappa-formula}. 
It follows from the sign chart of $\Re(\ii h(\lambda;\chi,\tau))$ as shown in Figure~\ref{fig:SB1} that all of the jump matrices above except for those supported on $I\cup\Sigma_g$ tend to the identity matrix exponentially fast as $M\to +\infty$ for $\lambda$ on the relevant supporting arcs away from the points $\lambda = \lambda_{0}(\chi,\tau)$, $\lambda_{0}(\chi,\tau)^{*}$, $a(\chi,\tau)$ and $b(\chi,\tau)$. In sufficiently small neighborhoods of these points, we will construct local parametrices that satisfy the jump conditions exactly.

\subsection{Parametrix construction}
\subsubsection{Outer parametrix construction} 
We start with construction of an outer parametrix denoted $\dot{\mathbf{W}}^\mathrm{out}(\lambda)\defeq \dot{\mathbf{W}}^\mathrm{out}(\lambda;\chi,\tau,\mathbf{Q}^{-s},M)$ satisfying exactly the jump conditions on $I$ and $\Sigma_g$ (cf., \eqref{eq:W-jump-I} and \eqref{eq:W-jump-Sigma-g}) that do not become asymptotically trivial as $M\to +\infty$. The procedure follows closely the construction in \cite[Section 4.2.2]{BilmanLM20} and it can be viewed as a combination of the outer parametrices constructed in Section~\ref{sec:channels} and in Section~\ref{sec:Schi-Stau} together with a new diagonal factor which is intrinsic to \shelves. Indeed, the jump condition \eqref{eq:W-jump-I} on $I$ is identical for $(\chi,\tau)\in\channels$ and $(\chi,\tau)\in\shelves$, hence we employ the outer parametrix \eqref{eq:Channels-Tout} to write $\dot{\mathbf{W}}^\mathrm{out}(\lambda)$ as
\begin{equation}
\dot{\mathbf{W}}^\mathrm{out}(\lambda)=\mathbf{J}(\lambda)\left(\frac{\lambda-a(\chi,\tau)}{\lambda-b(\chi,\tau)}\right)^{\ii p\sigma_3},
\label{eq:W-out-G}
\end{equation}
where the power function is defined as the principal branch and $p>0$ was defined in \eqref{eq:Channels-Tout}.
Then, $\mathbf{J}(\lambda)=\mathbf{J}(\lambda;\chi,\tau,\mathbf{Q}^{-s},M)$ extends analytically to $I$, and we will assume that it is bounded near $\lambda=a(\chi,\tau),b(\chi,\tau)$ in particular making it analytic at $\lambda=b(\chi,\tau)$.  Therefore, $\mathbf{J}(\lambda)$ is analytic for $\lambda\in\mathbb{C}\setminus\Sigma_g$ and tends to the identity as $\lambda\to\infty$.
Across $\Sigma_g$, the constant jump condition \eqref{eq:W-jump-Sigma-g} required of $\dot{\mathbf{W}}^\mathrm{out}(\lambda)$ becomes modified for $\mathbf{J}(\lambda)$:
\begin{equation}
\mathbf{J}_+(\lambda)=\mathbf{J}_-(\lambda)\left(\frac{\lambda-a(\chi,\tau)}{\lambda-b(\chi,\tau)}\right)^{\ii p\sigma_3}
\begin{bmatrix}0 & s \ee^{-2\ii M\kappa(\chi,\tau)}\\
- s \ee^{ 2\ii M\kappa(\chi,\tau)} & 0\end{bmatrix}
\left(\frac{\lambda-a(\chi,\tau)}{\lambda-b(\chi,\tau)}\right)^{-\ii p\sigma_3},\quad \lambda\in \Sigma_g,
\end{equation}
and we will convert this back into a constant jump condition on $\Sigma_g$ alone by introducing a \emph{Szeg{\H o} function} $K(\lambda;\chi,\tau)$, which we define by
\begin{equation}
K(\lambda;\chi,\tau)\defeq  p\log\left(\frac{\lambda-a(\chi,\tau)}{\lambda-b(\chi,\tau)}\right)+ p R(\lambda;\chi,\tau) \int_{a(\chi,\tau)}^{b(\chi,\tau)}\frac{\dd \eta}{R(\eta;\chi,\tau)(\eta-\lambda)}, 
\label{eq:K-def}
\end{equation}
in which the logarithm is taken to be the principal branch, $-\pi<\mathrm{Im}(\log(\cdot))<\pi$. 
It is straightforward to confirm that $K(\lambda;\chi,\tau)$ has the following properties.
Recalling the definition \eqref{eq:mu-formula-intro} of the constant $\mu(\chi,\tau)$, $K(\lambda;\chi,\tau) = -\mu(\chi,\tau) + O(\lambda^{-1})$ as $\lambda\to\infty$. Despite appearances, $K(\lambda;\chi,\tau)$ does not have a jump across $I$ as is easily confirmed by comparing the boundary values of the logarithm and using the Plemelj formula.  
The apparent singularities at $\lambda=a(\chi,\tau),b(\chi,\tau)$ are removable, so the domain of analyticity for $K(\lambda;\chi,\tau)$ is $\lambda\in\mathbb{C}\setminus\Sigma_g$, and $K(\lambda;\chi,\tau)$ takes continuous boundary values on $\Sigma_g$, including at the endpoints.  These boundary values are related by the jump condition
\begin{equation}
K_+(\lambda;\chi,\tau)+K_-(\lambda;\chi,\tau)=2 p\log\left(\frac{\lambda-a(\chi,\tau)}{\lambda-b(\chi,\tau)}\right),\quad \lambda\in\Sigma_g.
\label{eq:K-jump}
\end{equation}
One can indeed check that $K(\lambda;\chi,\tau)$ has the alternate representation obtained by the Plemelj formula:
\begin{equation}
K(\lambda;\chi,\tau) =  \frac{R(\lambda;\chi,\tau)}{2\pi \ii} \int_{\Sigma_g} \log\left( \frac{\eta - a(\chi,\tau)}{\eta - b(\chi,\tau)}\right) \frac{2 p  \dd \eta}{R_+(\eta; \chi,\tau)(\eta - \lambda)},
\label{eq:K-Plemelj}
\end{equation}
which confirms the properties stated above. The values of $K(\lambda;\chi,\tau)$ at $\lambda=a(\chi,\tau),b(\chi,\tau)$ will be useful in obtaining the asymptotic formula for $q( M \chi, M\tau;\mathbf{Q}^{-s},M)$, and they can easily be computed from the representation \eqref{eq:K-Plemelj}. 
Using
\begin{equation}
R(b(\chi,\tau);\chi,\tau) = |b(\chi,\tau)-\lambda_0(\chi,\tau)|\quad\text{and}\quad R_-(a(\chi,\tau);\chi,\tau) = |a(\chi,\tau)-\lambda_0(\chi,\tau)|,
\label{eq:R-a-b}
\end{equation}
we arrive at the formul\ae\ \eqref{eq:intro-Ka}--\eqref{eq:intro-Kb} for $K_a(\chi,\tau)\defeq K_-(a(\chi,\tau);\chi,\tau)$ and $K_b(\chi,\tau)\defeq K(b(\chi,\tau);\chi,\tau)$ respectively.

%
Preserving the normalization at infinity, we introduce $K(\lambda;\chi,\tau)$ in the construction of $\dot{\mathbf{W}}^\text{out}(\lambda)$ by
\begin{equation}
\mathbf{J}(\lambda)=\mathbf{L}(\lambda)\ee^{-\ii (K(\lambda;\chi,\tau)+ \mu(\chi,\tau) )\sigma_3}.
\label{eq:G-H}
\end{equation}
It then follows that $\mathbf{L}(\lambda)=\mathbf{L}(\lambda;\chi,\tau,\mathbf{Q}^{-s},M)$ is a matrix function analytic for $\lambda\in\mathbb{C}\setminus\Sigma_g$ that tends to $\mathbb{I}$ as $\lambda\to\infty$, and that satisfies the jump condition
\begin{equation}
\mathbf{L}_+(\lambda)=\mathbf{L}_-(\lambda)
\begin{bmatrix}0 & s \ee^{-2\ii (M\kappa(\chi,\tau) + \mu(\chi,\tau) )}\\
-s \ee^{2\ii (M\kappa(\chi,\tau) + \mu(\chi,\tau) )} & 0\end{bmatrix}
,\quad \lambda\in\Sigma_g.
\label{eq:H-jump-Sigma-g}
\end{equation}
$\mathbf{L}(\lambda)$ is given by
\begin{equation}
\mathbf{L}(\lambda)\defeq 
\ee^{-\ii \frac{1}{4}s \pi \sigma_3}
\ee^{-\ii (M\kappa(\chi,\tau) + \mu(\chi,\tau) )\sigma_3}
\mathbf{Q}y(\lambda;\chi,\tau)^{\sigma_3}\mathbf{Q}^{-1}
\ee^{\ii (M\kappa(\chi,\tau)+  \mu(\chi,\tau) )\sigma_3}
\ee^{\ii \frac{1}{4}s \pi \sigma_3},
\label{eq:H-def}
\end{equation}
where $y(\lambda;\chi,\tau)$ is defined in \eqref{eq:y-def}.
This solution clearly relates to (232) in exactly the same way that \eqref{eq:outer-parametrix-Schi-Stau-ALT} relates to \eqref{eq:T-jump-N-Schi-Stau-ALT} with $\mathbf{T}$ replaced by $\mathbf{W}$, and obviously $\mathbf{L}(\lambda)\to\mathbb{I}$ as $\lambda\to\infty$.
Combining \eqref{eq:W-out-G}, \eqref{eq:G-H}, and \eqref{eq:H-def} completes the construction of the outer parametrix $\dot{\mathbf{W}}^\mathrm{out}(\lambda)$:
\begin{equation}
\dot{\mathbf{W}}^\mathrm{out}(\lambda) = \mathbf{L}(\lambda) \ee^{-\ii ( K(\lambda;\chi,\tau) + \mu(\chi,\tau) )\sigma_3}\left(\frac{\lambda-a(\chi,\tau)}{\lambda-b(\chi,\tau)}\right)^{\ii p\sigma_3}.
\label{eq:W-out-full}
\end{equation}
Note that the only dependence on $M$ enters via the oscillatory factors $\ee^{\pm \ii M \kappa(\chi,\tau)\sigma_3}$ in $\mathbf{L}(\lambda)$. Thus, $\dot{\mathbf{W}}^{\text{out}}(\lambda)=\dot{\mathbf{W}}^{\text{out}}(\lambda;\chi,\tau,\mathbf{Q}^{-s},M)$ is bounded as $M\to+\infty$, provided that $\lambda$ is bounded away from $\lambda_0(\chi,\tau)$ and $\lambda_0(\chi,\tau)^*$.

While the outer parametrix exactly satisfies the same jump conditions satisfied by $\mathbf{W}(\lambda)$ on $\Sigma_g$ and $I$, it is discontinuous near the endpoints of these arcs. Thus, the problem at hand
requires four inner parametrices $\dot{\mathbf{W}}^a(\lambda)$, $\dot{\mathbf{W}}^b(\lambda)$, $\dot{\mathbf{W}}^{\lambda_0}(\lambda)$, and $\dot{\mathbf{W}}^{\lambda_0^*}(\lambda)$ to be defined in disks $D_{\lambda}(\delta)$, centered at the points $\lambda=a(\chi,\tau)$, $b(\chi,\tau)$, $\lambda_0(\chi,\tau)$, $\lambda_0^*(\chi,\tau)$, respectively, where $\delta=\delta(\chi,\tau)>0$ is chosen sufficiently small but independent of $M$. We take the circular boundaries of these disks to have clockwise orientation.

\subsubsection{Inner parametrix construction near the points $a(\chi,\tau)$ and $b(\chi,\tau)$} 

To define an inner parametrix $\dot{\mathbf{W}}^b(\lambda)=\dot{\mathbf{W}}^b(\lambda;\chi,\tau,\mathbf{Q}^{-s},M)$ in $D_b(\delta)$, first note that the properties of $h(\lambda;\chi,\tau)$ summarized at the beginning of this section imply that
${h}(\lambda;\chi,\tau)-{h}(b(\chi,\tau);\chi,\tau)$ is an analytic function of $\lambda$ that vanishes precisely to second order as $\lambda \to b(\chi,\tau)$. 
We introduce an $M$-independent conformal coordinate $f_b$ by setting
\begin{equation}
f_b(\lambda;\chi,\tau)^2 = 2({h}(\lambda;\chi,\tau)-{h}_b(\chi,\tau)),\quad \lambda\in D_b(\delta),
\label{eq:fb-def}
\end{equation}
where ${h}_b(\chi,\tau)\defeq {h}(b(\chi,\tau);\chi,\tau)$,
and choose the solution with $f_b'(b(\chi,\tau);\chi,\tau)>0$. To see why this choice is possible, note that repeated differentiation in \eqref{eq:fb-def} results in the relation
\begin{equation}
f_b'(b(\chi,\tau);\chi,\tau)^2 = h''(b(\chi,\tau);\chi,\tau) > 0.
\label{eq:fb-prime-h-double-prime}
\end{equation}
With this choice
the arc $I\cap D_b(\delta)$ is mapped by $\lambda \mapsto f_b(\lambda;\chi,\tau)$ locally to the negative real axis. Then, in the rescaled conformal coordinate $\zeta_b \defeq  M^\frac{1}{2} f_b$, the jump conditions satisfied by the matrix function
\begin{equation}
\mathbf{U}^b(\lambda) \defeq  \mathbf{W}(\lambda) \ii^{\frac{1}{2}(1-s)\sigma_3}
\ee^{-\ii M {h}_b(\chi,\tau)\sigma_3},\quad \lambda\in D_b(\delta)
\label{eq:W-transformation-b}
\end{equation}
coincide exactly with those of $\mathbf{U}(\zeta)$ described right before \eqref{eq:PCU-asymp} when expressed in terms of the variable $\zeta=\zeta_b$ and the jump contours are locally taken to coincide with the five rays $\arg(\zeta)=\pm \tfrac{1}{4}\pi$, $\arg(\zeta)=\pm \tfrac{3}{4}\pi$, and $\arg(-\zeta)=0$ as shown in \cite[Figure 9]{BilmanLM20}. Therefore, the construction of $\dot{\mathbf{W}}^b(\lambda)$ follows \emph{mutatis mutandis} that of the local parametrix near $b$ in Section~\ref{sec:channels-parametrix}. Indeed, replacing $\vartheta_b$ with $h_b$ and taking into account the fact that $\dot{\mathbf{W}}^\text{out}(\lambda)$ in this section differs from the outer parametrix \eqref{eq:Channels-Tout} in Section~\ref{sec:channels-parametrix} by multiplication on the left by $\mathbf{J}(\lambda)$, one obtains (compare with \eqref{eq:Channels-Tb-ALT})
\begin{equation}
\dot{\mathbf{W}}^b(\lambda)\defeq \mathbf{Y}^b(\lambda)\mathbf{U}(\zeta_b) \ii^{-\frac{1}{2}(1-s)\sigma_3} \ee^{\ii M{h}_b(\chi,\tau)\sigma_3},
\end{equation}
where $\mathbf{Y}^b(\lambda)$ is the prefactor that is holomorphic in the disk $D_b(\delta)$ and is given by
\begin{equation}
\mathbf{Y}^b(\lambda)\defeq 
\mathbf{J}(\lambda)
M^{\frac{1}{2} \ii p \sigma_3}
 \ee^{-\ii M {h}_b(\chi,\tau)\sigma_3} \ii^{\frac{1}{2}(1-s)\sigma_3}  \mathbf{H}^{b}(\lambda),
\label{eq:A-b}
\end{equation}
in which $\mathbf{H}^b(\lambda)$ is given \emph{exactly} as in \eqref{eq:Channels-Hb} with the conformal map $f_b(\lambda;\chi,\tau)$ being based on $h$ as in \eqref{eq:fb-def} rather than on $\vartheta$ as in Section~\ref{sec:channels}. $\mathbf{H}^b(\lambda)$ is holomorphic in the disk $D_b(\delta)$.
It is now easy to verify (by drawing a comparison with the construction in Section~\ref{sec:channels-parametrix}) that $\dot{\mathbf{W}}^b(\lambda)$ exactly satisfies the jump conditions for $\mathbf{W}(\lambda)$ in $D_b(\delta)$.
Comparing the this parametrix with the outer parametrix (which is unimodular) on the boundary of $D_b(\delta)$, we see that
\begin{equation}
\dot{\mathbf{W}}^b(\lambda)\dot{\mathbf{W}}^{\mathrm{out}}(\lambda)^{-1}
=\mathbf{Y}^b(\lambda)\mathbf{U}(\zeta_b) \zeta_b^{\ii p \sigma_3}\mathbf{Y}^b(\lambda)^{-1},
\label{eq:error-PC-b}
\end{equation}
and using the asymptotic expansion \eqref{eq:PCU-asymp} in \eqref{eq:error-PC-b} yields the estimate
\begin{equation}
\sup_{\lambda\in \partial D_b(\delta)}\| \dot{\mathbf{W}}^b(\lambda)\dot{\mathbf{W}}^{\mathrm{out}}(\lambda)^{-1} -\mathbb{I} \| = O(M^{-\frac{1}{2}}),\quad M\to+\infty,
\label{eq:error-PC-disk-b-large-M}
\end{equation}
where $\| \cdot \|$ denotes the matrix norm induced from an arbitrary vector norm on $\mathbb{C}^2$.


Constructing an inner parametrix $\dot{\mathbf{W}}^a(\lambda)=\dot{\mathbf{W}}^a(\lambda)(\lambda;\chi,\tau,\mathbf{Q}^{-s},M)$ in the disk $D_a(\delta)$ requires a bit more work due to the presence of the cut $\Sigma_g$ inside $D_a(\delta)$. Note that for $\lambda\in D_a$, ${h}(\lambda;\chi,\tau)$ comprises two different functions that are both analytic in the entire disk $D_a(\delta)$. We will use $\pm$ subscripts to denote these functions: ${h}_{-}(\lambda;\chi,\tau)$ (resp., ${h}_{+}(\lambda;\chi,\tau)$) coincides with ${h}(\lambda;\chi,\tau)$ for $\lambda$ to the right (resp., left) of $\Sigma_g$ with respect to (upward) orientation. These functions are of course related by ${h}_+(\lambda;\chi,\tau) + {h}_-(\lambda;\chi,\tau) = 2 \kappa(\chi,\tau)$ for $ \lambda\in D_{a}(\delta)$,
where the (real-valued) constant $\kappa(\chi,\tau)$ is given in \eqref{eq:kappa-formula}. By analogy, we denote by $D_{a,-}(\delta)$ (resp., $D_{a,+}(\delta)$) the part of $D_a(\delta)$ that lies to the right (resp., left) of $\Sigma_g$ with respect to orientation. We use the same notational convention for the boundaries of these half-disks: $\partial D_{a,\pm}(\delta)$ denotes the circular boundary of $D_{a,\pm}(\delta)$ (omitting $\Sigma_g$).

We base the definition of a conformal mapping on the analytic function $\lambda\mapsto {h}_-(\lambda;\chi,\tau)$. Again by the properties of $h$ summarized at the beginning of this section, ${h}_-(\lambda;\chi,\tau)-{h}_-(a(\chi,\tau);\chi,\tau)$ vanishes to second order as $\lambda\to a(\chi,\tau)$. We introduce an $M$-independent conformal coordinate $f_a$ by setting
\begin{equation}
f_a(\lambda;\chi,\tau)^2 = 2(h_a(\chi,\tau) - h_-(\lambda;\chi,\tau)),\quad \lambda\in D_a(\delta),
\label{eq:fa-def}
\end{equation}
where $h_a(\chi,\tau)\defeq h_{-}(a(\chi,\tau);\chi,\tau)$,
and choose the solution with $f'_a(a(\chi,\tau);\chi,\tau)<0$. This is again possible as one obtains by repeated differentiation in \eqref{eq:fa-def} the relation
\begin{equation}
f_a'(a(\chi,\tau);\chi,\tau)^2= - h_-''(a(\chi,\tau);\chi,\tau) >0,
\label{eq:fa-prime-h-double-prime}
\end{equation}
see \eqref{eq:h-double-prime-a-b-signs}.
With this choice,
the arc $I \cap D_a(\delta)$ is mapped by $\lambda \mapsto f_a(\lambda ;\chi,\tau)$ locally to the negative real axis. In the rescaled conformal coordinate $\zeta_a \defeq M^\frac{1}{2} f_a$, the jump conditions satisfied by the piecewise-defined matrix function
\begin{equation}
\mathbf{U}^a(\lambda) \defeq \begin{cases} 
\mathbf{W} (\lambda)
\ii^{\frac{1}{2}(1-s)\sigma_3}  \ii^{\sigma_3} \ee^{-\ii M {h}_a(\chi,\tau)\sigma_3}(\ii \sigma_2),\quad& \lambda\in D_{a,-}(\delta),\\
\mathbf{W}(\lambda) \ii^{\frac{1}{2}(1-s)\sigma_3} \ee^{-\ii M \kappa(\chi,\tau)\sigma_3}(-\ii \sigma_2)  \ee^{\ii M \kappa(\chi,\tau)\sigma_3}  \ii^{\sigma_3} \ee^{-\ii M {h}_a(\chi,\tau)\sigma_3}(\ii \sigma_2), \quad & \lambda \in D_{a,+}(\delta)
\end{cases}
\label{eq:W-transformation-a}
\end{equation}
coincide exactly with those of $\mathbf{U}(\zeta)$ described right before \eqref{eq:PCU-asymp} again when expressed in terms of the variable $\zeta=\zeta_a$ and the jump contours are locally taken to coincide with the five rays $\arg(\zeta)=\pm \tfrac{1}{4}\pi$, $\arg(\zeta)=\pm \tfrac{3}{4}\pi$, and $\arg(-\zeta)=0$ as shown in \cite[Figure 9]{BilmanLM20}. 
In light of the transformation \eqref{eq:W-transformation-a}, for $\mathbf{W}(\lambda)\dot{\mathbf{W}}^{a}(\lambda)^{-1}$ to be analytic in $D_{a}(\delta)$ we take the parametrix $\dot{\mathbf{W}}^{a}(\lambda)$=$\dot{\mathbf{W}}^{a}(\lambda;\chi,\tau,\mathbf{Q}^{-s},M)$ to be of the form
\begin{equation}
\dot{\mathbf{W}}^{a}(\lambda)\defeq \begin{cases}
\mathbf{Y}^{a}(\lambda) \mathbf{U}(\zeta_a)(-\ii \sigma_2) \ee^{\ii M {h}_a(\chi,\tau)\sigma_3} (\ii^{-\sigma_3})  \ii^{-\frac{1}{2}(1-s)\sigma_3}, \quad & \lambda\in D_{a,-}(\delta),\\
\mathbf{Y}^{a}(\lambda) \mathbf{U}(\zeta_a)(-\ii \sigma_2) \ee^{\ii M {h}_a(\chi,\tau)\sigma_3}(\ii^{-\sigma_3})\mathbf{K}(\chi,\tau) \ii^{-\frac{1}{2}(1-s)\sigma_3} , \quad & \lambda\in D_{a,+}(\delta),
\end{cases}
\label{eq:W-a}
\end{equation}
where we have set
\begin{equation}
\mathbf{K}(\chi,\tau) \defeq   \ee^{-\ii M\kappa(\chi,\tau)\sigma_3}(\ii \sigma_2)  \ee^{\ii M \kappa(\chi,\tau)\sigma_3}
\label{eq:K-mat}
\end{equation}
for brevity in the expressions, and where $\mathbf{Y}^a(\lambda)$ is a matrix function that is holomorphic in $D_a(\delta)$, to be determined by requiring $\dot{\mathbf{W}}^{a}(\lambda)\dot{\mathbf{W}}^{\mathrm{out}}(\lambda)^{-1}=\mathbb{I}+o(1)$ for $\lambda\in \partial D_a(\delta)$ as $M\to+\infty$.
We note that $\mathbf{J}(\lambda)$ given in \eqref{eq:W-out-G} defines two functions that are analytic in the entire disk $D_a(\delta)$, and we again use $\pm$ subscripts consistent with their boundary values taken on $\Sigma_g$ to label them:
$\mathbf{J}_\pm(\lambda)$ coincides with $\mathbf{J}(\lambda)$ for $\lambda\in D_{a,\pm}(\delta)$.
It follows from \eqref{eq:W-out-G} and \eqref{eq:W-a} that for $\lambda\in \partial D_{a,-}(\delta)$ we have
\begin{equation}
\dot{\mathbf{W}}^a(\lambda)\dot{\mathbf{W}}^\mathrm{out}(\lambda)^{-1} 
=\mathbf{Y}^{a}(\lambda) \mathbf{U}(\zeta_a)(-\ii\sigma_2) \ee^{\ii M {h}_a(\chi,\tau)\sigma_3}(\ii^{-\sigma_3})\ii^{-\frac{1}{2}(1-s)\sigma_3} \left(\frac{\lambda-a(\chi,\tau)}{\lambda-b(\chi,\tau)}\right)^{-\ii p\sigma_3} \mathbf{J}_-(\lambda)^{-1}.
\label{eq:error-right-half-disk-What}
\end{equation}
On the other hand, the definition \eqref{eq:W-out-G} also yields for $\lambda\in D_{a,-}(\delta)\setminus I$
\begin{equation}
\dot{\mathbf{W}}^\mathrm{out}(\lambda)  \ii^{\frac{1}{2}(1-s)\sigma_3}  \ii^{\sigma_3}  \ee^{- \ii M {h}_a(\chi,\tau)\sigma_3}(\ii\sigma_2) = 
\mathbf{J}_{-}(\lambda)   \ii^{\frac{1}{2}(1-s)\sigma_3} \ii^{\sigma_3}  \ee^{- \ii M {h}_a(\chi,\tau)\sigma_3}
M^{-\frac{1}{2}\ii p \sigma_3}\mathbf{H}^a(\lambda) \zeta_a^{-\ii p \sigma_3},
\label{eq:W-hat-R-out}
\end{equation}
where $\mathbf{H}^{a}(\lambda)$ is given \emph{exactly} by the formula \eqref{eq:Channels-Ha} except with the different conformal map $f_a(\lambda;\chi,\tau)$ whose construction \eqref{eq:fa-def} is based on $h_{-}(\lambda;\chi,\tau)$ rather than $\vartheta(\lambda;\chi,\tau)$ as in Section~\ref{sec:channels}. 
$\mathbf{H}^{a}(\lambda)$ is holomorphic in the entire disk $D_{a}(\delta)$.
Using \eqref{eq:W-hat-R-out} in \eqref{eq:error-right-half-disk-What} guides us to choose the prefactor $\mathbf{Y}^a(\lambda)$ to be
\begin{equation}
\mathbf{Y}^a(\lambda)\defeq \mathbf{J}_{-}(\lambda)  M^{-\frac{1}{2}\ii p \sigma_3} \ee^{- \ii M {h}_a(\chi,\tau)\sigma_3} \ii^{\frac{1}{2}(1-s)\sigma_3} \ii^{\sigma_3}  
\mathbf{H}^a(\lambda),
\label{eq:A-a}
\end{equation}
which is holomorphic in the entire disk $D_a(\delta)$ and unimodular.
Since $\mathbf{Y}^a(\lambda)$ is now determined, $\dot{\mathbf{W}}^a(\lambda)$ is determined according to \eqref{eq:W-a} and it follows that the mismatch \eqref{eq:error-right-half-disk-What} 
between the inner parametrix and the outer parametrix along $\partial D_{a,-}(\delta)$ reads:
\begin{equation}
\dot{\mathbf{W}}^{a}(\lambda) \dot{\mathbf{W}}^{\text {out }}(\lambda)^{-1}=\mathbf{Y}^{a}(\lambda)\mathbf{U}\left(\zeta_{a}\right) \zeta_{a}^{\mathrm{i} p \sigma_{3}} \mathbf{Y}^{a}(\lambda )^{-1},\quad \lambda \in \partial D_{a,-}(\delta).
\label{eq:mismatch-right-half-disk}
\end{equation}
However, we have used only the information in the right half-disk $D_{a,R}(\delta)$ to construct $\dot{\mathbf{W}}^a(\lambda)$, and one needs to check whether \eqref{eq:mismatch-right-half-disk} actually holds on the entire disk boundary $\partial D_{a}(\delta)$. This can be verified by a direct calculation using the relation \eqref{eq:W-jump-Sigma-g},
and we indeed have
 \begin{equation}
\dot{\mathbf{W}}^a(\lambda)\dot{\mathbf{W}}^\mathrm{out}(\lambda)^{-1} = \mathbf{Y}^a(\lambda)\mathbf{U}(\zeta_a) \zeta_a^{\ii p \sigma_3} \mathbf{Y}^a(\lambda)^{-1},\quad \lambda\in\partial D_a(\delta).
\label{eq:error-PC-a}
\end{equation}
By construction, $\dot{\mathbf{W}}^a(\lambda)$ exactly satisfies the jump conditions for $\mathbf{W}^a(\lambda)$ in $D_a(\delta)$.
Then
using the asymptotic expansion \eqref{eq:PCU-asymp} in \eqref{eq:error-PC-a}, we obtain the estimate
\begin{equation}
\sup_{\lambda \in \partial D_a(\delta)}\| \dot{\mathbf{W}}^a(\lambda) \dot{\mathbf{W}}^{\mathrm{out}}(\lambda)^{-1} -\mathbb{I} \| = O(M^{-\frac{1}{2}}),\quad M\to+\infty.
\label{eq:error-PC-disk-a-large-M}
\end{equation}

\subsubsection{Inner parametrix construction near the points $\lambda_0(\chi,\tau)$ and $\lambda_0(\chi,\tau)^*$} We now let $D_{\lambda_0}(\delta)$ and $D_{\lambda_0^*}(\delta) = D_{\lambda_0}(\delta)^*$ denote disks of small radius $\delta$ independent of $M$ centered at $\lambda=\lambda_0(\chi,\tau)$ and $\lambda=\lambda_0(\chi,\tau)^*$ respectively. Recalling that $h(\lambda_0(\chi,\tau);\chi,\tau) = \kappa(\chi,\tau)$ and $h'(\lambda;\chi,\tau)$ vanishes like a square root as $\lambda\to\lambda_0(\chi,\tau)$, a procedure almost exactly like the one following \eqref{eq:Airy-map-Schi-Stau-ALT} in Section~\ref{sec:Airy-parametrix} (replacing both of the boundary values $h_\pm$ in \eqref{eq:Airy-map-Schi-Stau-ALT} with $h$) leads to the construction of an inner parametrix $\dot{\mathbf{W}}^{\lambda_0}(\lambda)$ on $D_{\lambda_0}(\delta)$ in terms of Airy functions which takes continuous boundary values and satisfies exactly the same jump conditions within $D_{\lambda_0}(\delta)$ as $\mathbf{W}(\lambda)$. Moreover, across the boundary $\partial D_{\lambda_0}(\delta)$ this inner parametrix satisfies
\begin{equation}
\sup_{\lambda \in \partial D_{\lambda_0} } \| \dot{\mathbf{W}}^{\lambda_0}(\lambda)\dot{\mathbf{W}}^{\mathrm{out}}(\lambda)^{-1} -\mathbb{I} \| = O(M^{-1}),\quad M\to+\infty.
\label{eq:error-Airy-disk-plus}
\end{equation}
Since the matrix $\mathbf{W}(\lambda)$ satisfies $\mathbf{W}(\lambda^*)=\sigma_2\mathbf{W}(\lambda)^*\sigma_2$, we may define as in Section~\ref{sec:Airy-parametrix} a second inner parametrix for $\lambda\in D_{\lambda_0^*}(\delta)$ to respect this symmetry, which, of course, satisfies
\begin{equation}
\sup_{\lambda \in \partial D_{\lambda_0^*} } \| \dot{\mathbf{W}}^{\lambda_0^*}(\lambda)\dot{\mathbf{W}}^{\mathrm{out}}(\lambda)^{-1} -\mathbb{I} \| = O(M^{-1}),\quad M\to+\infty.
\label{eq:error-Airy-disk-minus}
\end{equation}

A \emph{global parametrix} $\dot{\mathbf{W}}(\lambda)=\dot{\mathbf{W}}(\lambda;\chi,\tau,\mathbf{Q}^{-s},M)$ is finally constructed by assembling the outer and inner parametrices as follows:
\begin{equation}
\dot{\mathbf{W}}(\lambda)\defeq 
\begin{cases}
\dot{\mathbf{W}}^{\lambda_0} (\lambda),&\quad \lambda\in D_{\lambda_0}(\delta),\\
\dot{\mathbf{W}}^{\lambda_0^*} (\lambda),&\quad \lambda\in D_{\lambda_0^*}(\delta),\\
\dot{\mathbf{W}}^a (\lambda),&\quad \lambda\in D_a(\delta),\\
\dot{\mathbf{W}}^b (\lambda),&\quad \lambda\in D_b(\delta),\\
\dot{\mathbf{W}}^\mathrm{out} (\lambda),&\quad \lambda\in \mathbb{C}\setminus(\Sigma_g \cup I \cup \overline{D_{\lambda_0}(\delta)\cup D_{\lambda_0^*}(\delta) \cup D_a(\delta) \cup D_b(\delta)}).
\end{cases}
\label{eq:W-dot-bun}
\end{equation}

\subsection{Small norm problem for the error and large-$M$ expansion}
To analyze the accuracy of the global parametrix $\dot{\mathbf{W}}(\lambda)$ for $(\chi,\tau)\in \shelves$, we define the \emph{error}
\begin{equation}
\mathbf{F}(\lambda) \defeq  \mathbf{W}(\lambda) \dot{\mathbf{W}}(\lambda)^{-1}.
\label{eq:F-bun}
\end{equation} 
As $\dot{\mathbf{W}}(\lambda)$ satisfies exactly the same jump conditions as $\mathbf{W}(\lambda)$ inside the disks $D_\lambda(\delta)$, $\lambda=a,b,\lambda_0,\lambda_0^*$, and on portions of the arcs $\Sigma_g$ and $I$ exterior to these disks, $\mathbf{F}(\lambda)$ can be taken as an analytic function of $\lambda\in\mathbb{C}$ with the exception the contour system $\Sigma_\mathbf{F}$, which consists of the portions of the arcs $C^\pm_{\Gamma, L}, C^\pm_{\Gamma, R}, C^\pm_{\Sigma, L}, C^\pm_{\Sigma, R}$ lying outside the disks $D_{a,b}(\delta)$ and $D_{\lambda_0,\lambda_0^*}(\delta)$ along with the four disk boundaries $\partial D_{a,b}(\delta)$ and $\partial D_{\lambda_0,\lambda_0^*}(\delta)$. We denote by $\mathbf{V}^{\mathbf{F}}(\lambda)$ the jump matrix for $\mathbf{F}(\lambda)$, which is supported on $\Sigma_\mathbf{F}$. On the arcs $C^\pm_{\Gamma, L}, C^\pm_{\Gamma, R}, C^\pm_{\Sigma, L}, C^\pm_{\Sigma, R}$ \emph{outside} the four disks, we can express $\mathbf{V}^{\mathbf{F}}(\lambda)$ as
\begin{equation}
\begin{split}
\mathbf{V}^{\mathbf{F}}(\lambda) &= \mathbf{F}_-(\lambda)^{-1}\mathbf{F}_+ (\lambda)\\
&=\dot{\mathbf{W}}^\mathrm{out}(\lambda)\mathbf{W}_-(\lambda)^{-1}\mathbf{W}_+ (\lambda)\dot{\mathbf{W}}^\mathrm{out}(\lambda)^{-1}.
\end{split}
\end{equation}
Since $\dot{\mathbf{W}}^\mathrm{out}(\lambda)$ remains bounded with unit determinant as $M\to+\infty$ and $\delta$ is fixed, there exists a positive constant $\nu>0$ such that $\mathbf{V}^\mathbf{F}(\lambda) - \mathbb{I} = O(\ee^{-\nu M})$ holds uniformly on the jump contour $\Sigma_\mathbf{F}$ for $\mathbf{F}(\lambda)$ except on the circles $\partial D_{a,b}(\delta)$ and $\partial D_{\lambda_0,\lambda_0^*}(\delta)$. On the circles, the jump matrix for $\mathbf{F}(\lambda)$ takes the form:
\begin{alignat}{2}
 \mathbf{F}_+(\lambda) &= \mathbf{F}_-(\lambda) \cdot \dot{\mathbf{W}}^{a,b}(\lambda)\dot{\mathbf{W}}^\mathrm{out}(\lambda)^{-1},&&\quad \lambda \in \partial D_{a,b}(\delta),\label{eq:F-jump-bun-a-b}\\
  \mathbf{F}_+(\lambda) &= \mathbf{F}_+(\lambda) \cdot \dot{\mathbf{W}}^{\lambda_0,\lambda_0^*}(\lambda)\dot{\mathbf{W}}^\mathrm{out}(\lambda)^{-1},&&\quad \lambda \in \partial D_{\lambda_0,\lambda_0^*}(\delta),
  \label{eq:F-jump-bun-lambda0}
\end{alignat}
because $\mathbf{W}(\lambda)$ is continuous across each of the four circles. Recalling the estimates \eqref{eq:error-PC-disk-b-large-M} and \eqref{eq:error-PC-disk-a-large-M}, it is seen from \eqref{eq:F-jump-bun-a-b} that $\mathbf{V}^\mathbf{F}(\lambda) - \mathbb{I} = O(M^{-\frac{1}{2}})$ holds uniformly on the circles $\partial D_{a,b}(\delta)$. Similarly, we see from \eqref{eq:error-Airy-disk-plus}-\eqref{eq:error-Airy-disk-minus} and \eqref{eq:F-jump-bun-lambda0} that $\mathbf{V}^\mathbf{F}(\lambda) - \mathbb{I} = O(M^{-1})$ holds uniformly on the circles $\partial D_{\lambda_0, \lambda_0^*}(\delta)$. Thus, it follows that $ \mathbf{F}_+(\lambda) = \mathbf{F}_-(\lambda)( \mathbb{I} + O(M^{-\frac{1}{2}}))$ holds uniformly as $M\to+\infty$ on the compact jump contour $\Sigma_\mathbf{F}$. Standard small-norm theory for such Riemann-Hilbert problems implies that $\mathbf{F}_-(\lambda) = \mathbb{I} + O(M^{-\frac{1}{2}})$ holds in the $L^2$ sense on $\Sigma_\mathbf{F}$, in the limit $M\to +\infty$.

\subsection{Asymptotic formula for $q(x,t;\mathbf{Q}^{-s},M)$ and fundamental rogue waves for $(\chi,\tau)\in \shelves$} We note that for the matrix function $\mathbf{W}(\lambda)=\mathbf{W}(\lambda;\chi,\tau,\mathbf{Q}^{-s},M)$
\begin{equation}
\mathbf{W} (\lambda)  = \mathbf{T} (\lambda) = \mathbf{S} (\lambda) \ee^{M g(\lambda;\chi,\tau)\sigma_3} 
\end{equation}
holds for $|\lambda|$ sufficiently large; therefore, from \eqref{eq:q-S} we have the formula
\begin{equation}
q(M\chi, M\tau; \mathbf{Q}^{-s}, M ) = 2\ii  \lim_{\lambda\to\infty} \left( \lambda W_{12}(\lambda) \ee^{\ii M g(\lambda;\chi,\tau)}\right).
\label{eq:psi-k-W}
\end{equation}
On the other hand, we see from the definitions \eqref{eq:W-dot-bun} and \eqref{eq:F-bun} that 
\begin{equation}
\mathbf{W}(\lambda) = \mathbf{F}(\lambda) \dot{\mathbf{W}}^{\mathrm{out}}(\lambda)
\end{equation}
also holds for $|\lambda|$ sufficiently large; therefore, \eqref{eq:psi-k-W} is expressed as:
\begin{equation}
q(M\chi, M\tau; \mathbf{Q}^{-s}, M ) = 2\ii  \lim_{\lambda\to\infty} \left( \lambda F_{11}(\lambda) \dot{W}^{\mathrm{out}}_{12}(\lambda) +  \lambda F_{12}(\lambda) \dot{W}^{\mathrm{out}}_{22}(\lambda) \right).
\label{eq:psi-k-F}
\end{equation}
Now $\dot{\mathbf{W}}^\mathrm{out}(\lambda)$ tends to the identity as $\lambda\to\infty$, and from \eqref{eq:F-bun} so does $\mathbf{F}(\lambda)$.  Therefore,
\begin{equation}
q(M\chi, M\tau; \mathbf{Q}^{-s}, M ) = 2\ii  \lim_{\lambda\to\infty} \left( \lambda  \dot{W}^{\mathrm{out}}_{12}(\lambda) +   \lambda F_{12}(\lambda) \right).
\label{eq:psi-k-F-simp}
\end{equation}
Recalling that $K(\lambda;\chi,\tau) +\mu(\chi,\tau)= O(\lambda^{-1})$ as $\lambda\to\infty$, it is easily seen from the definitions \eqref{eq:H-def} and \eqref{eq:W-out-full} that
\begin{equation}
2\ii \lim_{\lambda\to\infty} \lambda \dot{W}^{\mathrm{out}}_{12}(\lambda) 
=B(\chi,\tau) \ee^{-2\ii (M \kappa(\chi,\tau) +  \mu(\chi,\tau) + \frac{1}{4} s \pi)} = \mathfrak{L}_s^{[\shelves]}(\chi,\tau;M),
\label{eq:psi-k-bun-outer}
\end{equation}
producing the leading term for $q(M\chi,M\tau; \mathbf{Q}^{-s},M)$ given in \eqref{eq:leading-term-shelves-q}.

It now remains to compute the contribution in \eqref{eq:psi-k-F} coming from $\lambda F_{12}(\lambda;\chi,\tau)$ as $\lambda\to\infty$. Formulating the jump condition for $\mathbf{F}(\lambda)$ in the form $\mathbf{F}_+ - \mathbf{F}_- = \mathbf{F}_- (\mathbf{V}^{\mathbf{F}}-\mathbb{I})$ and using the fact that $\mathbf{F}(\lambda)\to\mathbb{I}$ as $\lambda\to\infty$, we obtain from the Plemelj formula
the same representation as in \eqref{eq:F-Cauchy-channels} for $\mathbf{F}(\lambda)$. It then follows that
$\mathbf{F}(\lambda)$ has the Laurent series expansion which is convergent for sufficiently large $|\lambda|$:
\begin{equation}
\mathbf{F}(\lambda) = \mathbb{I} - \frac{1}{2\pi \ii} \sum_{m=1}^{\infty} \lambda^{-m} \int_{\Sigma_\mathbf{F}} \mathbf{F}_-(\eta)(\mathbf{V}^{\mathbf{F}}(\eta)-\mathbb{I}) \eta^{m-1}\, \dd \eta,\quad |\lambda|>|\Sigma_\mathbf{F}|\defeq \sup_{\eta\in\Sigma_\mathbf{F}}|\eta|.
\label{eq:F-Laurent-bun}
\end{equation}
We obtain from this expansion the integral representation 
\begin{multline}
\lim_{\lambda\to\infty} \lambda F_{12}(\lambda) =
-\frac{1}{2\pi \ii} \left \lbrace \int_{\Sigma_\mathbf{F}} ( F_{11-}(\eta) - 1)V^{\mathbf{F}}_{12}(\eta)\dd \eta\right. \\
\left. +\int_{\Sigma_\mathbf{F}} V^{\mathbf{F}}_{12}(\eta)\dd \eta
+ \int_{\Sigma_\mathbf{F}} F_{12-}(\eta)(V^{\mathbf{F}}_{22}(\eta) -1 ) \dd \eta
 \right\rbrace.
 \label{eq:lambda-F-12-bun}
\end{multline}
We recall that $\mathbf{F}(\lambda)-\mathbb{I} = O(M^{-\frac{1}{2}})$ in the $L^2$ sense and $\mathbf{V}^{\mathbf{F}}(\lambda)-\mathbb{I} = O(M^{-\frac{1}{2}})$ in the $L^\infty$ sense on $\Sigma^{\mathbf{F}}$, in the limit $M\to+\infty$. As the $L^1$ norm is subordinate to the $L^2$ norm on the compact contour $\Sigma^\mathbf{F}$, direct application of Cauchy-Schwarz inequality shows that the first and the last integrals in \eqref{eq:lambda-F-12-bun} are both of size $O(M^{-1})$ as $M\to+\infty$. Combining this fact with \eqref{eq:psi-k-bun-outer} in the formula \eqref{eq:psi-k-F-simp} yields
\begin{equation}
q(M\chi, M\tau; \mathbf{Q}^{-s},M) =  \mathfrak{L}_s^{[\shelves]}(\chi,\tau;M) -\frac{1}{\pi}\int_{\Sigma_\mathbf{F}} V^{\mathbf{F}}_{12}(\eta)\dd \eta  + O(M^{-1}),
\quad M\to+\infty.
\end{equation}
Note that $V^\mathbf{F}_{12}(\lambda)$ is $O(M^{-1})$ on the circles $\partial D_{\lambda_0,\lambda_0^*}(\delta)$ as $M\to+\infty$ and it is  $O(\ee^{-\nu M})$ on the portions of the arcs $C^\pm_{\Gamma, L}, C^\pm_{\Gamma, R}, C^\pm_{\Sigma, L}, C^\pm_{\Sigma, R}$ lying outside the four disks, whereas $V^\mathbf{F}_{12}(\lambda)$ is $O(M^{-\frac{1}{2}})$ on the circles $\partial D_{a,b}(\delta)$. Therefore, the same formula as above holds with a different error of the same size when the integration contour $\Sigma_{\mathbf{F}}$ is replaced with $\partial D_{a}(\delta) \cup \partial D_{b}(\delta)$:
\begin{equation}
q(M\chi, M\tau; \mathbf{Q}^{-s},M)=\mathfrak{L}_s^{[\shelves]}(\chi,\tau;M) - \frac{1}{\pi}\int_{\partial D_{a}(\delta) \cup \partial D_{b}(\delta)} V^{\mathbf{F}}_{12}(\eta)\dd \eta 
 + O(M^{-1}), \quad M\to+\infty.
\label{eq:psi-k-V-12-bun}
\end{equation}
Using the asymptotic expansion \eqref{eq:PCU-asymp} in the formul\ae{} \eqref{eq:error-PC-b} and \eqref{eq:error-PC-a} and recalling that $\mathbf{V}^{\mathbf{F}}(\lambda) = \dot{\mathbf{W}}^{a,b}(\lambda)\dot{\mathbf{W}}^{\mathrm{out}}(\lambda)^{-1}$ for $\lambda\in \partial D_{a,b}(\delta)$, we see that
\begin{equation}
V^{\mathbf{F}}_{12}(\lambda) = \frac{1}{2\ii M^{\frac{1}{2}}} 
\left( 
\frac{\alpha Y_{11}^{a,b}(\lambda)^2 + \beta Y_{12}^{a,b}(\lambda)^2}{f_{a,b}(\lambda;\chi,\tau)}  
\right)
+ O(M^{-1}),\quad M\to+\infty, \quad \lambda\in\partial D_{a,b}(\delta).
\label{eq:VF-12-bun}
\end{equation}
The definition \eqref{eq:A-b} for $\mathbf{Y}^b(\lambda)$ together with the fact that $\mathbf{H}^b(\lambda)$ is a diagonal matrix directly gives
\begin{align}
Y^b_{11}(\lambda)^2 &= s L_{11}(\lambda)^2 \ee^{-2 \ii (K(\lambda;\chi,\tau) + \mu(\chi,\tau) )} \ee^{-2\ii M {h}_b(\chi,\tau)} M^{\ii p} (\lambda-a(\chi,\tau))^{2\ii p} \left( \frac{f_b(\lambda;\chi,\tau)}{\lambda-b(\chi,\tau)} \right)^{2\ii p},
\label{eq:A-b-11}\\
Y^b_{12}(\lambda)^2 &= s L_{12}(\lambda)^2 \ee^{2 \ii( K(\lambda;\chi,\tau) +\mu(\chi,\tau) )} \ee^{2\ii M {h}_b(\chi,\tau)} M^{-\ii p}(\lambda-a(\chi,\tau))^{-2\ii p} \left( \frac{f_b(\lambda;\chi,\tau)}{\lambda-b(\chi,\tau)} \right)^{-2\ii p},
\label{eq:A-b-12}
\end{align}
where to arrive at the latter formula we have used $s=\pm 1$. Similarly, the definition \eqref{eq:A-a} for $\mathbf{Y}^{a}(\lambda)$, this time together with the fact that $\mathbf{H}^a(\lambda)$ is off-diagonal and with perhaps more tedious arithmetic gives
\begin{align}
Y^a_{11}(\lambda)^2 &= - s L_{-,12}(\lambda)^2 \ee^{2\ii (K_-(\lambda;\chi,\tau) + \mu(\chi,\tau) )} \ee^{2\ii M {h}_a(\chi,\tau)} M^{\ii p} (b(\chi,\tau)-\lambda)^{2\ii p} \left( \frac{a(\chi,\tau) - \lambda}{f_a(\lambda;\chi,\tau)} \right)^{-2\ii p},
\label{eq:A-a-11}\\
Y^a_{12}(\lambda)^2 &= - s L_{-,11}(\lambda)^2 \ee^{-2 \ii (K_-(\lambda;\chi,\tau) + \mu(\chi,\tau) ) } \ee^{-2\ii M {h}_a(\chi,\tau)} M^{-\ii p}(b(\chi,\tau)-\lambda)^{-2\ii p} \left( \frac{a(\chi,\tau) - \lambda}{f_a(\lambda;\chi,\tau)} \right)^{2\ii p},
\label{eq:A-a-12}
\end{align}
and to obtain the former formula we have again used $s=\pm 1$. 
Here $\mathbf{L}_-(\lambda)$ and $K_-(\lambda;\chi,\tau)$ are the functions analytic for $\lambda\in D_a(\delta)$ coinciding with $\mathbf{L}(\lambda)=\mathbf{L}(\lambda;\chi,\tau,\mathbf{Q}^{-s},M)$ and  $K(\lambda;\chi,\tau)$, respectively, in $D_{a,-}(\delta)$. 
Thus, with \eqref{eq:A-b-11} and \eqref{eq:A-b-12}, we see from \eqref{eq:VF-12-bun} that $V^{\mathbf{F}}_{12}(\lambda)$ on the circle $\partial D_b(\delta)$ is given by:
\begin{equation}
\begin{split}
V^{\mathbf{F}}_{12}(\lambda) &=
\frac{s \alpha M^{\ii p} \ee^{-2\ii M {h}_b(\chi,\tau)} }{2\ii M^{\frac{1}{2}} f_b(\lambda;\chi,\tau)} L_{11}(\lambda)^2 \ee^{-2\ii ( K(\lambda;\chi,\tau) +\mu(\chi,\tau) ) }    (\lambda-a(\chi,\tau))^{2\ii p}  \left( \frac{f_b(\lambda;\chi,\tau)}{\lambda-b(\chi,\tau)} \right)^{2\ii p}\\
&\quad - 
\frac{s \alpha^* M^{-\ii p} \ee^{2\ii M {h}_b(\chi,\tau)} }{2\ii M^{\frac{1}{2}} f_b(\lambda;\chi,\tau)} L_{12}(\lambda)^2 \ee^{2\ii (K(\lambda;\chi,\tau) + \mu(\chi,\tau) )}  (\lambda-a(\chi,\tau))^{-2\ii p} \left( \frac{f_b(\lambda;\chi,\tau)}{\lambda-b(\chi,\tau)} \right)^{-2\ii p} \\
&\quad+ O(M^{-1}),\quad \text{in $L^\infty(\partial D_b(\delta))$ as $M\to+\infty$},
\end{split}
\label{eq:V-F-12-b}
\end{equation}
where we have used the property $\beta=-\alpha^*$. Similarly, we see from \eqref{eq:A-a-11} and \eqref{eq:A-a-12} that $V^{\mathbf{F}}_{12}(\lambda)$ on the circle $\partial D_a(\delta)$ is given by:
\begin{equation}
\begin{split}
V^{\mathbf{F}}_{12}(\lambda) &=
\frac{-s \alpha M^{\ii p} \ee^{2\ii M {h}_a(\chi,\tau)} }{2\ii M^{\frac{1}{2}} f_a(\lambda;\chi,\tau)} L_{-,12}(\lambda)^2 \ee^{2\ii (K_-(\lambda;\chi,\tau) +\mu(\chi,\tau) )} (b(\chi,\tau)-\lambda)^{2\ii p} 
\left( \frac{a(\chi,\tau)-\lambda}{f_a(\lambda;\chi,\tau)} \right)^{-2\ii p}\\
&\quad + 
\frac{s \alpha^* M^{-\ii p} \ee^{-2\ii M {h}_a(\chi,\tau)} }{2\ii M^{\frac{1}{2}} f_a(\lambda;\chi,\tau)} L_{-,11}(\lambda)^2 \ee^{-2\ii (K_-(\lambda;\chi,\tau)+\mu(\chi,\tau) )}  (b(\chi,\tau)-\lambda)^{-2\ii p} 
\left( \frac{a(\chi,\tau)-\lambda}{f_a(\lambda;\chi,\tau)} \right)^{2\ii p}\\
&\quad+ O(M^{-1})\quad \text{in $L^\infty(\partial D_a(\delta))$ as $M\to+\infty$}.
\end{split}
\label{eq:V-F-12-a}
\end{equation}
Note that $f_b(\lambda;\chi,\tau)$ in the leftmost factor of \eqref{eq:V-F-12-b} has a simple zero at $\lambda=b(\chi,\tau)$ and $f_a(\lambda;\chi,\tau)$ in the leftmost factor of \eqref{eq:V-F-12-a} has a simple zero at $\lambda=a(\chi,\tau)$, while the rest of the factors in \eqref{eq:V-F-12-b} and \eqref{eq:V-F-12-a} are holomorphic within the relevant disks.
Recalling the clockwise orientation of the circles $\partial D_{a,b}(\delta)$ and using 
\begin{align}
f_a'(a(\chi,\tau);\chi,\tau)& = - \left(-h_a''(\chi,\tau)\right)^{\frac{1}{2}}  \defeq  -\sqrt{-h''_-(a(\chi,\tau);\chi,\tau)}<0,\\
f_b'(b(\chi,\tau);\chi,\tau)& = h_b''(\chi,\tau)^{\frac{1}{2}}  \defeq  \sqrt{h''(b(\chi,\tau);\chi,\tau)}>0,
\end{align}
a simple residue calculation in \eqref{eq:V-F-12-b}--\eqref{eq:V-F-12-a} yields
\begin{multline}
-\frac{1}{\pi} \int_{D_{a}(\delta)} V^{\mathbf{F}}_{12}(\eta)\dd \eta = \frac{s}{M^{\frac{1}{2}} \left(-h_a''(\chi,\tau)\right)^{\frac{1}{2}}}
\left[ 
\alpha M^{\ii p} \ee^{2\ii M {h}_a(\chi,\tau)} L_{a,12}(\chi,\tau)^2 \ee^{2\ii (K_a(\chi,\tau) + \mu(\chi,\tau) )} X_a(\chi,\tau)^{\ii p}
\right.\\
\left. -\alpha^* M^{-\ii p} \ee^{-2\ii M {h}_a(\chi,\tau)}  
L_{a,11}(\chi,\tau)^2 \ee^{-2 \ii (K_a(\chi,\tau) + \mu(\chi,\tau) )}  X_a(\chi,\tau)^{-\ii p}
\right] + O(M^{-1})\quad\text{and}
\label{eq:V-F-12-a-residue}
\end{multline}
\begin{multline}
-\frac{1}{\pi} \int_{D_{b}(\delta)} V^{\mathbf{F}}_{12}(\eta)\dd \eta =
\frac{s}{M^{\frac{1}{2}}  h_b''(\chi,\tau)^{\frac{1}{2}}}
 \left[ 
\alpha M^{\ii p} \ee^{-2\ii M {h}_b(\chi,\tau)} 
L_{b,11}(\chi,\tau)^2 \ee^{-2 \ii (K_b(\chi,\tau) + \mu(\chi,\tau) )}  X_b(\chi,\tau)^{\ii p}
\right.\\
\left. - \alpha^{*} M^{-\ii p} \ee^{2\ii M {h}_b(\chi,\tau)}
L_{b,12}(\chi,\tau)^2 \ee^{2\ii (K_b(\chi,\tau) + \mu(\chi,\tau) )}  X_b(\chi,\tau)^{-\ii p}
\right] + O(M^{-1}),
\label{eq:V-F-12-b-residue}
\end{multline}
where we have set
\begin{equation}
X_a(\chi,\tau) \defeq - (b(\chi,\tau)-a(\chi,\tau))^2 h_a''(\chi,\tau),\quad\text{and}\quad X_b(\chi,\tau) \defeq  (b(\chi,\tau)-a(\chi,\tau))^2 h_b''(\chi,\tau),
\end{equation}
\begin{equation}
\mathbf{L}_a(\chi,\tau) \defeq  \mathbf{L}_-(a(\chi,\tau);\chi,\tau,\mathbf{Q}^{-s},M), \quad \text{and}\quad \mathbf{L}_b(\chi,\tau) \defeq  \mathbf{L}(b(\chi,\tau);\chi,\tau,\mathbf{Q}^{-s},M),
\end{equation}
and used the notation in \eqref{eq:intro-Ka} and \eqref{eq:intro-Kb}. Recalling the definitions of the four positive modulation factors in \eqref{eq:m-a-b-shelves}, we obtain from \eqref{eq:H-def} the (well-defined) expressions
\begin{equation}
L_{b,11}(\chi,\tau)^2 = \frac{1}{2} + \frac{1}{4} \big( y(b(\chi,\tau);\chi,\tau)^{-2} + y(b(\chi,\tau);\chi,\tau)^{2} \big)
=m^+_b(\chi,\tau),
\label{eq:H-b-11}
\end{equation}
and similarly
\begin{align}
L_{b,12}(\chi,\tau)^2 
&= - m^+_b(\chi,\tau) \ee^{-4\ii (M\kappa(\chi,\tau)+ \mu(\chi,\tau) + \frac{1}{4}s \pi )},
\label{eq:H-b-12}\\
L_{a,11}(\chi,\tau)^2 &=
m^+_a(\chi,\tau),\quad\text{and}
\label{eq:H-a-11}\\
L_{a,12}(\chi,\tau)^2 &= - m^-_a(\chi,\tau) \ee^{-4\ii(M\kappa(\chi,\tau)+ \mu(\chi,\tau) + \frac{1}{4}s \pi )}.\label{eq:H-a-12}
\end{align}
Recalling that $p=\tfrac{\ln(2)}{2\pi}$ and $a(\chi,\tau)<b(\chi,\tau)$ for $(\chi,\tau) \in \shelves$ together with the signs \eqref{eq:fb-prime-h-double-prime} and \eqref{eq:fa-prime-h-double-prime}, we write
\begin{align}
M^{\pm \ii p} &= \ee^{\pm \ii \ln(M) \frac{\ln(2)}{2\pi}},\\
X_a(\chi,\tau)^{\pm \ii p} &= \ee^{\pm \ii \frac{\ln(2)}{2\pi} \ln\left( - (b(\chi,\tau) - a(\chi,\tau) )^2h''_a(\chi,\tau) \right)},\quad\text{and}
\label{eq:log-b-a-1}\\
X_b(\chi,\tau)^{\pm \ii p} &= \ee^{\pm \ii \frac{\ln(2)}{2\pi} \ln\left( (b(\chi,\tau) - a(\chi,\tau) )^2h''_b(\chi,\tau) \right)}.
\label{eq:log-b-a-2}
\end{align}
Now substituting \eqref{eq:H-b-11}--\eqref{eq:log-b-a-2} along with \eqref{eq:Channels-alpha-beta} for $\alpha$ in \eqref{eq:V-F-12-a-residue} and \eqref{eq:V-F-12-b-residue} yields
\begin{equation}
- \frac{1}{\pi}\int_{\partial D_{a}(\delta) \cup \partial D_{b}(\delta)} V^{\mathbf{F}}_{12}(\eta)\dd \eta  = \mathfrak{S}^{[\shelves]}_s(\chi,\tau;M) + O(M^{-1}), \quad M\to+\infty,
\label{eq:subleading-term-shelves-q-proof}
\end{equation}
in which $\mathfrak{S}_s^{[\shelves]}(\chi,\tau;M)$ is the sub-leading term defined in \eqref{eq:subleading-term-shelves-q}. Substituting \eqref{eq:subleading-term-shelves-q-proof} back in \eqref{eq:psi-k-V-12-bun} 
finishes the proof of Theorem~\ref{thm:shelves}. 
\begin{remark} In practice, one needs to evaluate the derivatives $h''_a(\chi,\tau)=h''_-(a(\chi,\tau);\chi,\tau)$ and $h''_b(\chi,\tau)=h''(b(\chi,\tau);\chi,\tau)$ to 
use the approximation given in Theorem~\ref{thm:shelves}. These can be computed in a straightforward manner from \eqref{eq:hprime-formula} and using \eqref{eq:R-a-b}:
\begin{align}
{h}''(b(\chi,\tau);\chi,\tau) &= \frac{2\tau(b(\chi,\tau) - a(\chi,\tau))}{b(\chi,\tau)^2+1}|b(\chi,\tau) - \lambda_0(\chi,\tau)|, \label{eq:h-double-prime-b} \\
{h}_{-}''(a(\chi,\tau);\chi,\tau) &= \frac{- 2\tau(b(\chi,\tau) - a(\chi,\tau))}{a(\chi,\tau)^2+1} | a(\chi,\tau) - \lambda_0(\chi,\tau)|.
\label{eq:h-double-prime-a}
\end{align}
\label{rem:h-double-prime}
\end{remark}

\subsection{Wave-theoretic interpretation of the asymptotic formula for $q(M\chi, M\tau; \mathbf{Q}^{-s}, M)$ in $\shelves$.} 
\label{sec:wave-theoretic-interpretation}
In this subsection we prove Corollary~\ref{cor:shelves-local}.
As we will be working in a relative perturbation regime of the leading term in the large-$M$ asymptotic expansion of $q(M\chi,M\tau; \mathbf{Q}^{-s}, M)$, we compare with the formula \eqref{eq:leading-and-subleading-shelves-rewritten} and write the leading term in the form
\begin{equation}
\mathfrak{L}_s^{[\shelves]}(\chi,\tau;M) =  -\ii s B(\chi,\tau) \ee^{-2\ii \phi(\chi,\tau;M)}.
\label{eq:q-0-bun}
\end{equation}
We then fix $(\chi_0,\tau_0)\in\shelves$, and write $\chi = \chi_0 + \Delta \chi$ and $ \tau = \tau_0 + \Delta \tau$.
Noting that $\Delta \chi = M^{-1} \Delta x$ and $\Delta \tau = M^{-1} \Delta t$, and recalling the assumptions $\Delta x = O(1)$ and $\Delta t = O(1)$, it is easy to see that the phase $\phi(\chi,\tau;M)$ (see \eqref{eq:symmetrical-phases}) admits the following Taylor series expansion about $(\chi,\tau) = (\chi_0, \tau_0)$
\begin{equation}
\begin{split}
\phi(\chi,\tau;M) 
&= M \kappa(\chi_0,\tau_0) + \kappa_\chi (\chi_0,\tau_0)\Delta x + \kappa_\tau (\chi_0,\tau_0)\Delta t + O(M^{-1}\Delta x^2) + O(M^{-1}\Delta x \Delta t) + O(M^{-1}\Delta t^2) \\
&\quad+ \mu(\chi_0,\tau_0) + O(M^{-1}\Delta x) + O(M^{-1}\Delta t)\\
&=\phi(\chi_0,\tau_0; M) + \kappa_\chi (\chi_0,\tau_0)\Delta x + \kappa_\tau (\chi_0,\tau_0)\Delta t + O(M^{-1}),\quad M\to +\infty,
\end{split}
\label{eq:Omega-0-expand-1}
\end{equation}
which implies
\begin{equation}
\ee^{-2\ii \phi(\chi,\tau;M)} = \ee^{-2\ii \phi (\chi_0,\tau_0;M)} \ee^{\ii(\xi_0\Delta x - \Omega_0 \Delta t)} + O(M^{-1}),\quad M\to +\infty,
\label{eq:Omega-0-Taylor-bun}
\end{equation}
in which real local wavenumber $\xi_0$ and real local frequency $\Omega_0$ are defined in \eqref{eq:wavenumbers-intro}--\eqref{eq:frequencies-intro}. On the other hand, Taylor expansion of $B(\chi,\tau)$ in \eqref{eq:q-0-bun} around the same point $(\chi_0,\tau_0)$ gives
\begin{equation}
B(\chi,\tau) = B(\chi_0,\tau_0) + O(M^{-1}),\quad M\to +\infty.
\label{eq:B-Taylor-bun}
\end{equation}
Combining \eqref{eq:Omega-0-Taylor-bun} and \eqref{eq:B-Taylor-bun} in \eqref{eq:q-0-bun} yields the expansion
\begin{equation}
\mathfrak{L}_s^{[\shelves]}(\chi_0+M^{-1}\Delta x,\tau_0+M^{-1}\Delta t;M)
=Q(\Delta x, \Delta t) + O(M^{-1}),\quad M\to +\infty,
\label{eq:Q-0-expand}
\end{equation}
valid uniformly for $(\Delta x,\Delta t)$ bounded, where $Q(\Delta x, \Delta t)$ is given in \eqref{eq:leading-plane-wave-intro}. 
We proceed in a similar manner and obtain Taylor series expansions of the terms in the sub-leading term in  \eqref{eq:leading-and-subleading-shelves-rewritten} around the same fixed point $(\chi_0,\tau) \in \shelves$. Recall the definitions \eqref{eq:symmetrical-phases} of the symmetrical phases $\phi_a$ and $\phi_b$. For bounded $(\Delta x, \Delta t)$ as before, we have
\begin{multline}
\phi_a(\chi,\tau;M) = \phi_a(\chi_0,\tau_0; M) + \left(\Phi^{[\shelves]}_{a,\chi}(\chi_0,\tau_0) + 2\kappa_\chi(\chi_0,\tau_0) \right)\Delta x \\ + \left( \Phi^{[\shelves]}_{a,\chi}(\chi_0,\tau_0) + 2\kappa_\tau(\chi_0,\tau_0) \right) \Delta t + O(M^{-1}).
\end{multline}
Substituting \eqref{eq:Phis-shelves} in this expression and recalling the definitions  \eqref{eq:wavenumbers-intro}--\eqref{eq:frequencies-intro} for the real local wavenumber $\xi_a$ and real local frequency $\Omega_a$ gives
\begin{equation}
\ee^{\pm \ii \phi_a(\chi,\tau;M)} = \ee^{\pm \ii \phi_a(\chi_0,\tau_0;M)} \ee^{\pm \ii (\xi_a \Delta x - \Omega_a \Delta t)} + O(M^{-1}),\quad M\to +\infty.
\label{eq:expand-phi-a}
\end{equation}
An identical calculation for the phase $\phi_b(\chi,\tau;M)$ gives
\begin{equation}
\ee^{\pm \ii \phi_b(\chi,\tau;M)} = \ee^{\pm \ii \phi_b(\chi_0,\tau_0;M)} \ee^{\pm \ii (\xi_b \Delta x - \Omega_b \Delta t)} + O(M^{-1}),\quad M\to +\infty.
\label{eq:expand-phi-b}
\end{equation}
On the other hand, for the amplitude factors in \eqref{eq:leading-and-subleading-shelves-rewritten} we have the expansions 
\begin{align}
m_a^{\pm}(\chi,\tau)F_a^{[\shelves]}(\chi,\tau)&=  m_a^{\pm}(\chi_0,\tau_0)F_a^{[\shelves]}(\chi_0,\tau_0) + O(M^{-1}),\label{eq:expand-m-F-a}\\
m_b^{\pm}(\chi,\tau)F_b^{[\shelves]}(\chi,\tau)&=  m_b^{\pm}(\chi_0,\tau_0)F_b^{[\shelves]}(\chi_0,\tau_0) + O(M^{-1})\label{eq:expand-m-F-b}.
\end{align}
Using \eqref{eq:expand-phi-a}--\eqref{eq:expand-m-F-b} in the sub-leading term $\mathfrak{S}_s^{[\shelves]}(\chi,\tau;M)$ written in the form \eqref{eq:leading-and-subleading-shelves-rewritten}, taking into account the overall multiplicative factor $M^{-\frac{1}{2}}$ in \eqref{eq:leading-and-subleading-shelves-rewritten} for the error terms in \eqref{eq:expand-phi-a}--\eqref{eq:expand-m-F-b}, and factoring out $Q(\Delta x, \Delta t)$ to express the sub-leading term as a relative perturbation results in the expansion \eqref{eq:Q-perturbation-shelves}, which proves the first statement in Corollary~\ref{cor:shelves-local}.

To show that $Q(\Delta x, \Delta t)$ is a plane-wave solution of \eqref{eq:NLS-Deltas}, we need the following lemma concerning the partial derivatives $g_\chi(\lambda;\chi,\tau)$ and $g_\tau(\lambda;\chi,\tau)$.
\begin{lemma}
The partial derivatives $g_{\chi}(\lambda;\chi,\tau)$ and $g_{\tau}(\lambda;\chi,\tau)$ are given for $(\chi,\tau)\in \shelves$ by
\begin{align}
g_{\chi}(\lambda;\chi,\tau) &=  A(\chi,\tau) - \lambda + R(\lambda;\chi,\tau),\label{eq:g-chi}\\
g_{\tau}(\lambda;\chi,\tau) &= A(\chi,\tau)^2 - \frac{1}{2}B(\chi,\tau)^2 - \lambda^2 + (A(\chi,\tau) + \lambda) R(\lambda;\chi,\tau)\label{eq:g-tau}.
\end{align}
Also, the partial derivatives $\kappa_\chi(\chi,\tau)$ and $\kappa_\tau(\chi,\tau)$ are given by
\begin{align}
\kappa_\chi(\chi,\tau) &= A(\chi,\tau)\label{eq:kappa-chi},\\
\kappa_\tau(\chi,\tau) &= A(\chi,\tau)^2 - \frac{1}{2} B(\chi,\tau)^2\label{eq:kappa-tau}.
\end{align}
\label{lemma:g-derivatives}
\end{lemma}
\begin{proof} 
As $\lambda_0(\chi,\tau)$ is a real analytic function of $\chi$ and $\tau$ for $(\chi,\tau)\in \shelves$, it follows from Morera's theorem that $g_\chi(\lambda;\chi,\tau)$ and $g_\tau(\lambda;\chi,\tau)$ are functions that are analytic for $\lambda\in \mathbb{C}\setminus \Sigma_g$. 
Recall the definition of $\vartheta(\lambda;\chi,\tau)$ from \eqref{eq:vartheta}, and also recall that $g(\lambda;\chi,\tau)$ behaves like the sum of a constant and the product of $(\lambda-\lambda_0)^{\frac{3}{2}}$ with an analytic function in a neighborhood of $\lambda_0$ in $\mathbb{C}\setminus \Sigma_g$ (with the same behavior near $\lambda_0^*$ by symmetry). 
Now it is seen from \eqref{eq:hpm-kappa} that $g_\chi$ can be obtained as the function analytic for $\lambda\in\mathbb{C}\setminus\Sigma_g$ satisfying the jump condition
\begin{equation}
g_{\chi+}(\lambda;\chi,\tau) + g_{\chi-}(\lambda;\chi,\tau) = 2 \kappa_\chi(\chi,\tau) - 2\vartheta_\chi(\lambda;\chi,\tau),\quad \lambda\in\Sigma_g,
\label{eq:g-chi-jump}
\end{equation}
that is bounded at the endpoints $\lambda=\lambda_0(\chi,\tau), \lambda_0(\chi,\tau)^*$ and is normalized as $g_\chi(\lambda;\chi,\tau)=O(\lambda^{-1})$ as $\lambda\to\infty$. Similarly, $g_\tau$ can be obtained as the function analytic in the same domain satisfying the jump condition
\begin{equation}
g_{\tau+}(\lambda;\chi,\tau) + g_{\tau-}(\lambda;\chi,\tau) = 2 \kappa_\tau(\chi,\tau) - 2 \vartheta_\tau(\lambda;\chi,\tau),\quad \lambda\in\Sigma_g,
\label{eq:g-tau-jump}
\end{equation}
and that is again bounded at the endpoints and normalized as $g_\tau(\lambda;\chi,\tau)=O(\lambda^{-1})$ as $\lambda\to\infty$. 
It is easy to see that the unique functions satisfying the analyticity, jump conditions, and boundedness conditions alone are
\begin{align}
g_\chi (\lambda;\chi,\tau) &=  \kappa_{\chi}(\chi,\tau)-\lambda + R(\lambda;\chi,\tau),\quad\lambda\in\mathbb{C}\setminus\Sigma_g,\label{eq:g-chi-integrated}\\
g_\tau(\lambda;\chi,\tau) &=  \kappa_{\tau}(\chi,\tau) - \lambda^{2}  + ( A(\chi,\tau) + \lambda ) R(\lambda;\chi,\tau),\quad\lambda\in\mathbb{C}\setminus\Sigma_g,\label{eq:g-tau-integrated}
\end{align}
for instance by writing $g_\chi$ and $g_\tau$ as $R(\lambda;\chi,\tau)$ times an unknown function in \eqref{eq:g-chi-jump} and \eqref{eq:g-tau-jump} and solving the resulting jump conditions for the new unknowns by a Cauchy integral that can be evaluated by residues. Enforcing the heretofore neglected normalization conditions 
by using the expansion $R(\lambda;\chi,\tau) = \lambda - A(\chi,\tau) + \tfrac{1}{2}B(\chi,\tau)^2 \lambda^{-1} + O(\lambda^{-2})$
as $\lambda\to\infty$ in \eqref{eq:g-chi-integrated} and \eqref{eq:g-tau-integrated} results in the formul\ae\ \eqref{eq:kappa-chi} and \eqref{eq:kappa-tau}.
\end{proof}
\begin{remark}
Since $\kappa(\chi,\tau)$ is a smooth function of both variables, its first order partial derivatives with respect to $\chi$ and $\tau$ commute, which in light of the expressions \eqref{eq:kappa-chi} and \eqref{eq:kappa-tau} implies the partial differential equation
\begin{equation}
\frac{\partial A}{\partial\tau}=\frac{\partial}{\partial\chi}(A^2-\tfrac{1}{2}B^2).
\label{eq:Whitham1}
\end{equation}
Similarly, taking the coefficients $c^{(\chi)}(\chi,\tau)$ and $c^{(\tau)}(\chi,\tau)$ of the (leading) term proportional to $\lambda^{-1}$ in $g_\chi(\lambda;\chi,\tau)$ and $g_\tau(\lambda;\chi,\tau)$ respectively, we obtain the consistency relation $c^{(\chi)}_\tau=c^{(\tau)}_\chi$ which is equivalent to the partial differential equation
\begin{equation}
\frac{\partial}{\partial\tau}B^2=\frac{\partial}{\partial\chi}(2AB^2).
\label{eq:Whitham2}
\end{equation}
Under the identifications $\rho=B^2$ and $U=-2A$, the two equations \eqref{eq:Whitham1}--\eqref{eq:Whitham2} are equivalent to the dispersionless nonlinear Schr\"odinger system (or genus-zero Whitham system) written in \eqref{eq:dispersionless-NLS}.  As a $2\times 2$ quasilinear system, it can be written in Riemann invariant (diagonal) form; 
in particular, the variables $\lambda_0=A+\ii B$ and $\lambda_0^*=A-\ii B$ are Riemann invariants for this system, in terms of which it becomes
\begin{equation}
\begin{split}
\lambda_{0,\tau} + \left(-\tfrac{3}{2}\lambda_0-\tfrac{1}{2}\lambda_0^*\right)\lambda_{0,\chi}^*&=0\\
\lambda_{0,\tau}^* + \left(-\tfrac{1}{2}\lambda_0-\tfrac{3}{2}\lambda_0^*\right)\lambda_{0,\chi}^*&=0.
\end{split}
\end{equation}
Since $\kappa(\chi,\tau)$ and $\gamma(\chi,\tau)$ differ by a constant according to \eqref{eq:kappa-gamma}, this proves Corollary~\ref{cor:Whitham}.
\label{rem:Whitham}
\end{remark}
Substituting \eqref{eq:kappa-chi} and \eqref{eq:kappa-tau} in \eqref{eq:wavenumbers-intro} and \eqref{eq:frequencies-intro}, we see that
\begin{align}
\xi_0 &= -2 A(\chi_0,\tau_0)\label{eq:xi-0-explicit},\\
\Omega_0 &= 2 A(\chi_0,\tau_0)^2 - B(\chi_0,\tau_0)^2\label{eq:Omega-0-explicit}.
\end{align}
Using these expressions and noting from \eqref{eq:leading-plane-wave-intro} that $|\mathcal{A}| = B(\chi_0,\tau_0)$, it is straightforward to show that  the wavenumber $\xi_0$, the frequency $\Omega_0$, and the modulus $|\mathcal{A}|$ of the complex amplitude for $Q(\Delta x, \Delta t)$ satisfy the nonlinear dispersion relation
\begin{equation}
\Omega_0 - \tfrac{1}{2}\xi_0^2 +|\mathcal{A}|^2 = 0.
\label{eq:nls-dispersion}
\end{equation}
This proves that $Q(\Delta x, \Delta t)$ is a plane-wave solution of \eqref{eq:NLS-Deltas}.

We will now prove the claim that each of the functions $p_{a}(\Delta x, \Delta t)$ and $p_b(\Delta x, \Delta t)$ in the expansion \eqref{eq:Q-perturbation-shelves} defines a solution of the linearization \eqref{eq:linearization-intro} of \eqref{eq:NLS-Deltas} about the plane-wave solution $Q(\Delta x,\Delta t)$.
Observe that the expansion \eqref{eq:Q-perturbation-shelves} is of the form \eqref{eq:q-perturb-appendix} in the treatment given in Appendix \ref{A:perturbations} and hence gives a relative perturbation expansion of $Q(\Delta x,\Delta t)$ for $M\gg 1$. We let $r_{a,b}(\Delta x, \Delta t)$ and $s_{a,b}(\Delta x, \Delta t)$ denote the real and imaginary parts of $p_{a,b}(\Delta x, \Delta t)$, respectively. For convenience and brevity in the calculations to come, we set
\begin{equation}
Z_{a,b}^\pm \defeq \pm F_{a,b}^{[\shelves]}(\chi_0,\tau_0)m_{a,b}^\pm(\chi_0,\tau_0).
\label{eq:Z-a-b}
\end{equation}
Then $r_{a,b}(\Delta x, \Delta t)$ and $s_{a,b}(\Delta x, \Delta t)$ are expressed in terms of these quantities as
\begin{align}
r_{a,b}(\Delta x,\Delta t) &= \frac{Z_{a,b}^- - Z_{a,b}^+}{B(\chi_0,\tau_0)}\sin\left(\phi_{a,b}(\chi_0,\tau_0;M)+\xi_{a,b}\Delta x -\Omega_{a,b}\Delta t\right),\\
s_{a,b}(\Delta x,\Delta t) &= \frac{Z_{a,b}^- + Z_{a,b}^+}{B(\chi_0,\tau_0)}\cos\left(\phi_{a,b}(\chi_0,\tau_0;M)+\xi_{a,b}\Delta x -\Omega_{a,b}\Delta t\right).
\end{align}
In view of Appendix~\ref{A:perturbations}, to prove that $p_a(\Delta x, \Delta t)$ and $p_b(\Delta x,\Delta t)$ solve \eqref{eq:linearization-intro},
it suffices to show that the pairs $(r_{a}(\Delta x, \Delta t), s_{a}(\Delta x, \Delta t))$ and $(r_{b}(\Delta x, \Delta t), s_{b}(\Delta x, \Delta t))$ satisfy \eqref{eq:linearized-NLS-r-s} (written in the variables $(\Delta x,\Delta t)$ instead of $(x,t)$). Suppressing the dependencies on the fixed point $(\chi_0,\tau_0)$ for brevity, it is easy to see using $|\mathcal{A}|^2=B^2=B(\chi_0,\tau_0)^2$ that $(r_{a,b}(\Delta x, \Delta t), s_{a,b}(\Delta x, \Delta t))$ satisfies \eqref{eq:linearized-NLS-r-s} if and only if 
\begin{align}
\left( \xi_{a,b}^2 + 2(\xi_0 \xi_{a,b} - \Omega_{a,b}) \right)Z_{a,b}^+ +\left( \xi_{a,b}^2 - 2(\xi_0 \xi_{a,b} - \Omega_{a,b}) \right) Z_{a,b}^- &= 0,\label{eq:F-a-b-pm-sys-1}\\
\left( \xi_{a,b}^2 + 2(\xi_0 \xi_{a,b} - \Omega_{a,b})  - 4 B^2 \right) Z_{a,b}^+ +\left( - \xi_{a,b}^2 + 2(\xi_0 \xi_{a,b} - \Omega_{a,b})  + 4 B^2 \right)Z_{a,b}^- &= 0.\label{eq:F-a-b-pm-sys-2}
\end{align}
Note that \eqref{eq:F-a-b-pm-sys-1}--\eqref{eq:F-a-b-pm-sys-2} 
constitute two homogeneous systems of linear equations for $(Z_{a,b}^+, Z_{a,b}^-)$, one for each choice of subscript $a$, $b$.
These systems have nontrivial solutions if and only if they are singular, which amount to the conditions
\begin{equation}
4 (\xi_0 \xi_{a,b}-\Omega_{a,b} )^2 = \xi_{a,b}^2 \left(\xi_{a,b}^2 - 4 B^2 \right).\label{eq:linearized-dispersion-a-b-bun}
\end{equation}
Again recalling that $B^2=|\mathcal{A}|^2$, these are precisely two instances of the linearized dispersion relation \eqref{eq:linearized-dispersion} to be satisfied by the pairs $(\xi_{a,b}, \Omega_{a,b})$ of relative local wavenumbers and frequencies. We will first show that the conditions \eqref{eq:linearized-dispersion-a-b-bun} hold, and then show that the pairs $(Z_{a,b}^+, Z_{a,b}^-)$
lie in the (nontrivial) nullspaces of the coefficient matrices for the systems \eqref{eq:F-a-b-pm-sys-1}--\eqref{eq:F-a-b-pm-sys-2}.
 To prove \eqref{eq:linearized-dispersion-a-b-bun},
we refer back to Lemma~\ref{lemma:g-derivatives} and use the expression \eqref{eq:g-chi} for $g_\chi(\lambda;\chi,\tau)$ together with $\vartheta_\chi(\lambda;\chi,\tau)=\lambda$ and \eqref{eq:kappa-chi} in the definitions \eqref{eq:wavenumbers-intro} of $\xi_{a,b}$ to see that
\begin{equation}
\xi_a = -2 R_-(a(\chi_0;\tau_0);\chi_0,\tau_0)\quad \text{and}\quad
\xi_b = -2 R(b(\chi_0;\tau_0);\chi_0,\tau_0).
\label{eq:xi-a-b-explicit}
\end{equation}
Similarly, using the expression \eqref{eq:g-tau} for $g_\tau(\chi,\tau)$ together with $\vartheta_\tau(\lambda;\chi,\tau)=\lambda^2$ and \eqref{eq:kappa-chi} in the definitions \eqref{eq:frequencies-intro} of $\Omega_{a,b}$ yields
\begin{equation}
\begin{split}
\Omega_a &= 2(A(\chi_0,\tau_0) + a(\chi_0,\tau_0)) R_-(a(\chi_0;\tau_0);\chi_0,\tau_0),\\
\Omega_b &= 2(A(\chi_0,\tau_0) + b(\chi_0,\tau_0))  R(b(\chi_0;\tau_0);\chi_0,\tau_0).
\label{eq:Omega-a-b-explicit}
\end{split}
\end{equation}
Now, to show that \eqref{eq:linearized-dispersion-a-b-bun} holds, we recall the definition of $\xi_0$ in \eqref{eq:wavenumbers-intro}, and use \eqref{eq:xi-a-b-explicit} and \eqref{eq:Omega-a-b-explicit} to observe that
\begin{equation}
\begin{split}
4 (\xi_0 \xi_a-\Omega_a )^2 
&= 4\left(4 A(\chi_0,\tau_0) R(a(\chi_0,\tau_0); \chi_0, \tau_0) - 2(A(\chi_0, \tau_0)+a(\chi_0, \tau_0)) R_-(a(\chi_0, \tau_0); \chi_0, \tau_0)\right)^{2}\\
&=16 R_-(a(\chi_0, \tau_0); \chi_0, \tau_0)^2 \left(a(\chi_0, \tau_0) - A(\chi_0, \tau_0) \right)^2.
\end{split}
\label{eq:linearized-dispersion-a-LHS}
\end{equation}
Next, the right-hand side of \eqref{eq:linearized-dispersion-a-b-bun} reads
\begin{equation}
\begin{split}
\xi_a^2 \left(\xi_a^2 - 4 B(\chi_0,\tau_0)^2 \right)
&= 4 R_-(a(\chi_0,\tau_0);\chi_0,\tau_0)^2 \left(4 R_-(a(\chi_0,\tau_0);\chi_0,\tau_0)^2 - 4 B(\chi_0,\tau_0)^2 \right)\\
&= 16 R_-(a(\chi_0,\tau_0);\chi_0,\tau_0)^2 \left( a(\chi_0,\tau_0) - A(\chi_0,\tau_0) \right)^2,
\end{split}
\label{eq:linearized-dispersion-a-RHS}
\end{equation}
since $R(\lambda;\chi,\tau)^2 = (\lambda- A(\chi,\tau))^2 + B(\chi,\tau)^2$. The identities \eqref{eq:linearized-dispersion-a-LHS}--\eqref{eq:linearized-dispersion-a-RHS} prove that the linearized dispersion relation \eqref{eq:linearized-dispersion-a-b-bun} holds for $(\xi_a, \Omega_a)$. A completely analogous calculation having the point $b(\chi_0,\tau_0)$ in place of $a(\chi_0,\tau_0)$ shows that \eqref{eq:linearized-dispersion-a-b-bun} holds for $(\xi_b, \Omega_b)$.

As we have now established that the linear systems \eqref{eq:F-a-b-pm-sys-1}--\eqref{eq:F-a-b-pm-sys-2}
 are both singular, it remains to show that the pairs of quantities $(Z_{a,b}^{+}, Z_{a,b}^{-})$
lie in the corresponding nullspaces. 
To do so, it suffices to verify that \eqref{eq:F-a-b-pm-sys-1} 
holds. 
Using the definitions \eqref{eq:m-a-b-shelves} in \eqref{eq:Z-a-b}
and noting that $F^{[\shelves]}_a(\chi_0,\tau_0)$ and $F^{[\shelves]}_b(\chi_0,\tau_0)$ are nonzero, it is seen that verifying \eqref{eq:F-a-b-pm-sys-1}
amounts to showing that
\begin{align}
\cos( \arg(a(\chi_0,\tau_0)-\lambda_0(\chi_0,\tau_0)) ) &= \frac{-2(\xi_0\xi_a - \Omega_a)}{\xi_a^2},\label{eq:show-cos-arg-a-bun}\\
\cos( \arg(b(\chi_0,\tau_0)-\lambda_0(\chi_0,\tau_0)) ) &= \frac{-2(\xi_0\xi_b - \Omega_b)}{\xi_b^2}.\label{eq:show-cos-arg-b-bun}
\end{align}
However, according to 
\eqref{eq:xi-0-explicit}, \eqref{eq:xi-a-b-explicit}, and \eqref{eq:Omega-a-b-explicit}, we obtain for the right-hand side of the purported identities \eqref{eq:show-cos-arg-a-bun}--\eqref{eq:show-cos-arg-b-bun} that
\begin{align}
\frac{-2(\xi_0\xi_a - \Omega_a)}{\xi_a^2}&=\frac{a(\chi_0,\tau_0) - A(\chi_0,\tau_0)}{R_-(a(\chi_0,\tau_0);\chi_0,\tau_0)},\label{eq:show-cos-arg-a-simp}\\
\frac{-2(\xi_0\xi_b - \Omega_b)}{\xi_b^2}&=\frac{b(\chi_0,\tau_0) - A(\chi_0,\tau_0)}{R(b(\chi_0,\tau_0);\chi_0,\tau_0)}.\label{eq:show-cos-arg-b-simp}
\end{align}
Since $R_-(a(\chi_0,\tau_0);\chi_0,\tau_0)$ and $R(b(\chi_0,\tau_0);\chi_0,\tau_0)$ are both positive, while $a(\chi_0,\tau_0)$ and $b(\chi_0,\tau_0)$ are real and $A(\chi_0,\tau_0)=\mathrm{Re}(\lambda_0(\chi_0,\tau_0))$, 
we see that the identities \eqref{eq:show-cos-arg-a-bun}--\eqref{eq:show-cos-arg-b-bun} indeed both hold:
\begin{align}
\frac{-2(\xi_0\xi_a - \Omega_a)}{\xi_a^2}&=\frac{\mathrm{Re}(a(\chi_0,\tau_0)-\lambda_0(\chi_0,\tau_0))}{|a(\chi_0,\tau_0)-\lambda_0(\chi_0,\tau_0)|}=\cos(\arg(a(\chi_0,\tau_0)-\lambda_0(\chi_0,\tau_0))),\\
\frac{-2(\xi_0\xi_b - \Omega_b)}{\xi_b^2}&=\frac{\mathrm{Re}(b(\chi_0,\tau_0)-\lambda_0(\chi_0,\tau_0))}{|b(\chi_0,\tau_0)-\lambda_0(\chi_0,\tau_0)|}=\cos(\arg(b(\chi_0,\tau_0)-\lambda_0(\chi_0,\tau_0))).
\end{align}
Thus, we have shown that $(Z_a^+,Z_a^-)$ and $(Z_b^+,Z_b^-)$ are non-trivial solutions of the linear homogeneous systems \eqref{eq:F-a-b-pm-sys-1}--\eqref{eq:F-a-b-pm-sys-2}.
This implies that  $p_a(\Delta x, \Delta t)$ and  $p_b(\Delta x, \Delta t)$ are solutions of \eqref{eq:linearization-intro}.
\begin{remark}
Requiring instead any of the individual plane waves 
\begin{equation}
\begin{split}
(\Delta x, \Delta t) &\mapsto \frac{\ii F_a^{[\shelves]}(\chi_0,\tau_0)}{B(\chi_0,\tau_0)}m_a^{\pm}(\chi_0,\tau_0)\ee^{\pm \ii\phi_a(\chi_0,\tau_0;M)} \ee^{\pm \ii(\xi_a \Delta x - \Omega_a \Delta t)}\quad\text{or}\\
(\Delta x, \Delta t) &\mapsto \frac{\ii F_b^{[\shelves]}(\chi_0,\tau_0)}{B(\chi_0,\tau_0)}m_b^{\pm}(\chi_0,\tau_0)\ee^{\pm \ii\phi_b(\chi_0,\tau_0;M)} \ee^{\pm \ii(\xi_b \Delta x - \Omega_b \Delta t)}
\end{split}
\end{equation}
in \eqref{eq:p-a-b-shelves} to be solutions of \eqref{eq:linearization-intro} forces $B(\chi_0,\tau_0)=0$, which is a contradiction. Therefore, one indeed needs to form the combinations $p_a(\Delta x, \Delta t)$ and $p_b(\Delta x, \Delta t)$ as in \eqref{eq:p-a-b-shelves}.
\end{remark}
Finally, we show that the relative wavenumbers $\xi_{a,b}$ do not lie in the band of modulational instability $(-2 B(\chi_0,\tau_0), 2 B(\chi_0,\tau_0))$ in view of the well-known theory summarized in Appendix~\ref{A:perturbations}. This result follows from the identities \eqref{eq:xi-a-b-explicit} in a straightforward manner. Indeed,
\begin{equation}
\begin{alignedat}{3}
\xi_a^2 &= 4 \left( a(\chi_0,\tau_0) - A(\chi_0,\tau_0) \right)^2 + 4 B(\chi_0,\tau_0)^2 &&> 4 B(\chi_0,\tau_0)^2,
\\
\xi_b^2 &= 4 \left( b(\chi_0,\tau_0) - A(\chi_0,\tau_0) \right)^2 + 4 B(\chi_0,\tau_0)^2 &&> 4 B(\chi_0,\tau_0)^2.
\end{alignedat}
\end{equation}

\newpage
\appendix
%
%
\section{Relative Perturbations of Plane Waves and Linear Instability Bands}
\label{A:perturbations}
In this section of the Appendix we consider relative perturbations of a plane-wave solution $q=Q(x,t):= \mathcal{A} \ee^{\ii (\xi_0 x - \Omega_0 t)}$ of the focusing nonlinear Schr\"odinger equation in the form \eqref{eq:NLS-ZBC}, 
having complex-valued amplitude $\mathcal{A}$, wavenumber $\xi_0\in\mathbb{R}$, and frequency $\Omega_0\in\mathbb{R}$ necessarily linked by the nonlinear dispersion relation $\Omega_0 -\tfrac{1}{2}\xi_0^2 + |\mathcal{A}|^2=0$. Requiring more generally that 
\begin{equation}
q(x,t) = Q(x,t)(1 + \varepsilon p(x,t) + o(\varepsilon)),\quad 0< \varepsilon \ll 1,
\label{eq:q-perturb-appendix}
\end{equation}
also solves \eqref{eq:NLS-ZBC} 
and formally retaining terms up to $o(\varepsilon)$ yields the differential equation
\begin{equation}
\ii p_t + \ii \xi_0 p_x + \tfrac{1}{2}p_{xx} |\mathcal{A}|^2(p + p^*) = 0,
\label{eq:linearized-NLS}
\end{equation} 
which is real-linear, but not complex-linear. We split $p(x,t)$ into its real and imaginary parts:
\begin{equation}
r(x,t) := \tfrac{1}{2}(p(x,t) + p(x,t)^*)\quad \text{and} \quad s(x,t) := -\ii\tfrac{1}{2}(p(x,t) - p(x,t)^*),
\end{equation}
giving rise to the following system of coupled linear differential equations with real-valued coefficients for $(r,s)$:
\begin{equation}
\begin{aligned}
r_t + \xi_0 r_x + \tfrac{1}{2} s_{xx} &=0,\\
s_t + \xi_0 s_x - \tfrac{1}{2} r_{xx} - 2 |\mathcal{A}|^2 r &=0.\\
\end{aligned}
\label{eq:linearized-NLS-r-s}
\end{equation}
We will now carry out a Fourier analysis to determine the instability bands for relative perturbations $p(x,t)$ of $Q(x,t)$. To this end, we suppose $p(x,t)$ is a plane-wave solution of \eqref{eq:linearized-NLS}. For convenience we drop for the moment the reality condition for $(r,s)$ and work with the ansatz
\begin{equation}
\begin{bmatrix}
r(x,t) \\ s(x,t)
\end{bmatrix}
:= \begin{bmatrix}
\alpha \\ \beta
\end{bmatrix}\ee^{\ii (\xi x - \Omega t)}
\label{eq:r-s-ansatz}
\end{equation}
for some complex constants $\alpha$ and $\beta$, and $\xi\in\mathbb{R}$, $\Omega\in\mathbb{R}$. Substituting \eqref{eq:r-s-ansatz} in \eqref{eq:linearized-NLS-r-s} yields the homogenous linear algebraic system
\begin{equation}
\begin{bmatrix}
-\ii(\Omega -\xi_0 \xi) & -\frac{1}{2}\xi^2 \\  -\frac{1}{2}\xi^2 + 2 |\mathcal{A}|^2 & \ii(\Omega - \xi_0 \xi)
\end{bmatrix}
\begin{bmatrix}
\alpha \\ \beta
\end{bmatrix}
=
\begin{bmatrix}
0 \\ 0
\end{bmatrix}.
\end{equation}
This system has a nontrivial solution $(\alpha,\beta)$ if and only if 
\begin{equation}
4 (\Omega - \xi_0 \xi)^2 = \xi^2 \left(\xi^2 - 4|\mathcal{A}|^2 \right),
\label{eq:linearized-dispersion}
\end{equation}
which is the linearized dispersion relation for the relative wavenumber $\xi$ and the relative frequency $\Omega$. As $\xi,\xi_0\in\mathbb{R}$, we see that $\Im(\Omega) \neq 0$ if $\xi^2< 4|\mathcal{A}|^2$. Therefore, the plane-wave solutions of \eqref{eq:linearized-NLS-r-s} with relative wavenumbers $\xi$ lying in the band $(-2 |\mathcal{A}|, 2 |\mathcal{A}|)$ exhibit exponential growth in time $t>0$.  This is the well-known modulational (or sideband, or Benjamin-Feir) instability of plane-wave solutions for the focusing nonlinear Schr\"odinger equation.

\section{Proofs of Some Elementary Results}
\label{A:Proofs}
\subsection{Symmetries of $q(x,t;\mathbf{Q}^{-s},M)$:  Proof of Proposition~\ref{prop:symmetry}}
\begin{proof}[Proof of Proposition~\ref{prop:symmetry}]
For this proof, we assume without loss of generality that $\Sigma_\circ$ is a circle centered at the origin of arbitrary radius $r$ greater than $1$ with clockwise orientation.
Taking $\mathbf{G}=\mathbf{Q}^{-s}$ for $s=\pm 1$, the jump condition \eqref{eq:P-bulk-jump} in Riemann-Hilbert Problem~\ref{rhp:rogue-wave-reformulation} for $\mathbf{P}(\lambda;x,t)=\mathbf{P}(\lambda;x,t,\mathbf{Q}^{-s},M)$ can be written as 
\begin{equation}
\mathbf{P}_+(\lambda;x,t)=\mathbf{P}_-(\lambda;x,t)\ee^{-\ii\theta(\lambda;x,t)\sigma_3}
B(\lambda)^{M\sigma_3}
\mathbf{Q}^{-s}
B(\lambda)^{-M\sigma_3}
\ee^{\ii\theta(\lambda;x,t)\sigma_3},
\label{eq:P-bulk-jump-rewrite}
\end{equation}
where $\theta(\lambda;x,t)\defeq\lambda x+\lambda^2t$. 
Define $\mathbf{R}(\lambda;x,t)$ in terms of $\mathbf{P}(\lambda;x,t)$ by
\begin{equation}
\mathbf{R}(\lambda;x,t)\defeq\begin{cases}
-s\sigma_3\mathbf{P}(\lambda;x,t)\ee^{-2\ii\theta(\lambda;x,t)\sigma_3}\sigma_1,&\quad |\lambda|<r,
\\
\sigma_3\mathbf{P}(\lambda;x,t)
B(\lambda)^{2M\sigma_3}
\sigma_3,&\quad |\lambda|>r.
\end{cases}
\end{equation} 
$\mathbf{R}(\lambda;x,t)$ is obviously analytic for $\lambda\in\mathbb{C}\setminus\Sigma_\circ$, and since 
powers of $B(\lambda)$ tend to $1$ 
as $\lambda\to\infty$ we have $\mathbf{R}(\lambda;x,t)\to\mathbb{I}$ as $\lambda\to\infty$.  To compute the jump across the circle $\Sigma_\circ$, we use the jump condition \eqref{eq:P-bulk-jump} for $\mathbf{P}(\lambda;x,t)$ to obtain, using \eqref{eq:Q-def} in the last step,
\begin{equation}
\begin{split}
\mathbf{R}_+(\lambda;x,t)&=\sigma_3\mathbf{P}_+(\lambda;x,t)
B(\lambda)^{2M\sigma_3}
\sigma_3\\
&=\sigma_3\mathbf{P}_-(\lambda;x,t)\ee^{-\ii\theta(\lambda;x,t)\sigma_3}
B(\lambda)^{M\sigma_3}
\mathbf{Q}^{-s}\sigma_3 
B(\lambda)^{M\sigma_3}
\ee^{\ii\theta(\lambda;x,t)\sigma_3}\\
&=-s\mathbf{R}_-(\lambda;x,t)\sigma_1\ee^{\ii\theta(\lambda;x,t)\sigma_3}
B(\lambda)^{M\sigma_3}
\mathbf{Q}^{-s}\sigma_3
B(\lambda)^{M\sigma_3}
\ee^{\ii\theta(\lambda;x,t)\sigma_3}\\
&=\mathbf{R}_-(\lambda;x,t)\ee^{-\ii\theta(\lambda;x,t)\sigma_3}
B(\lambda)^{-M\sigma_3}
\left[-s\sigma_1\mathbf{Q}^{-s}\sigma_3\right]
B(\lambda)^{M\sigma_3}
\ee^{\ii\theta(\lambda;x,t)\sigma_3}\\
&=\mathbf{R}_-(\lambda;x,t)\ee^{-\ii\theta(\lambda;x,t)\sigma_3}
B(\lambda)^{-M\sigma_3}
\mathbf{Q}^{-s}
B(\lambda)^{M\sigma_3}
\ee^{\ii\theta(\lambda;x,t)\sigma_3}.
\end{split}
\end{equation}
Now $\theta(\lambda;-x,t)=\theta(-\lambda;x,t)$
and $B(\lambda)=B(-\lambda)^{-1}$.
Therefore, we see that $\mathbf{P}(\lambda;-x,t)$ and $\mathbf{R}(-\lambda;x,t)$ satisfy exactly the same analyticity, jump, and normalization conditions and therefore by uniqueness $\mathbf{P}(\lambda;-x,t)=\mathbf{R}(-\lambda;x,t)$.  Thus,
\begin{equation}
\begin{split}
q(-x,t;\mathbf{Q}^{-s},M)&=2\ii\lim_{\lambda\to\infty}\lambda P_{12}(\lambda;-x,t)\\
&=2\ii\lim_{\lambda\to\infty}\lambda R_{12}(-\lambda;x,t)\\
&=-2\ii\lim_{\lambda\to\infty}\lambda R_{12}(\lambda;x,t)\\
&=2\ii\lim_{\lambda\to\infty}\lambda P_{12}(\lambda;x,t)\\
&=q(x,t;\mathbf{Q}^{-s},M).
\end{split}
\end{equation}
Since 
$B(-\lambda^*)=B(\lambda)^*$,
it is even easier to see that $\mathbf{P}(\lambda;x,-t)$ and $\mathbf{P}(-\lambda^*;x,t)^*$ solve the same Riemann-Hilbert problem and hence are equal.  Therefore
\begin{equation}
\begin{split}
q(x,-t;\mathbf{Q}^{-s},M)&=2\ii\lim_{\lambda\to\infty}\lambda P_{12}(\lambda;x,-t)\\
&=2\ii\lim_{\lambda\to\infty}\lambda P_{12}(-\lambda^*;x,t)^*\\
&=-2\ii\lim_{\lambda\to\infty}\left[\lambda^* P_{12}(\lambda^*;x,t)\right]^*\\
&=\left[2\ii\lim_{\lambda\to\infty}\lambda P_{12}(\lambda;x,t)\right]^*\\
&=q(x,t;\mathbf{Q}^{-s},M)^*.
\end{split}
\end{equation}
This completes the proof of Proposition~\ref{prop:symmetry}.
\end{proof}

\subsection{Continuation of $u(\chi,\tau)$ to $\overline{\exterior\cup\shelves}$:  Proof of Proposition~\ref{prop:u}}
\begin{proof}[Proof of Proposition~\ref{prop:u}]
We first examine $P(u;\chi,\tau)$ near the positive $\chi$ and $\tau$ axes.  
\begin{lemma}
Fix $\chi>0$.  Then for $\tau>0$ sufficiently small there exists a unique and simple real root of $P(u;\chi,\tau)$.
\label{lem:tau-small}
\end{lemma}
\begin{proof}
Since the simple root at $u=\chi$ for $\tau=0$ persists for small $\tau$, we need to show that the four roots of $P(u;\chi,0)$ near $u=\tfrac{1}{3}\chi$ and the two roots of $P(u;\chi,0)$ near $u=0$ become complex roots of $P(u;\chi,\tau)$ for $\tau\neq 0$ small.  

To study the roots of $P(u;\chi,\tau)$ near $u=\tfrac{1}{3}\chi$, we set $u=\tfrac{1}{3}\chi + \Delta u$ and then express $P(u;\chi,\tau)$ in terms of $\Delta u$:
\begin{multline}
27 P(\tfrac{1}{3}\chi+\Delta u;\chi,\tau)=2187\Delta u^7 - 243(3\chi^2-8\tau^2)\Delta u^5-162\chi^3\Delta u^4 \\{}+ 216\tau^2(54-3\chi^2+2\tau^2)\Delta u^3
-144\chi^3\tau^2\Delta u^2-144\chi^2\tau^4\Delta u-32\chi^3\tau^4.
\end{multline}
For small $\tau$, the dominant balance in which $\Delta u$ is also small is $\Delta u=\tau\zeta$ for a new unknown $\zeta=O(1)$ as $\tau\to 0$.  Then we find that we can divide by $\tau^4$ for $\tau\neq 0$ and obtain
\begin{equation}
27 \tau^{-4}P(\tfrac{1}{3}\chi +\tau\zeta;\chi,\tau)=-2\chi^3 (9\zeta^2+4)^2 + O(\tau),\quad\tau\to 0,\quad\tau\neq 0.
\end{equation}
So, to leading order, we have double purely imaginary roots at $\zeta=\pm\tfrac{2}{3}\ii$, which can split apart at higher order in $\tau$.  This shows that the four roots of $P(u;\chi,\tau)$ near $u=\tfrac{1}{3}\chi$ for small $\tau$ have nonzero imaginary parts.  

To study the roots of $P(u;\chi,\tau)$ near $u=0$ we observe that the dominant balance in $P(u;\chi,\tau)=0$ in which $u$ and $\tau$ are both small for $\chi>0$ fixed occurs with $u=\tau\zeta$ with new unknown $\zeta=O(1)$ as $\tau\to 0$.  Dividing by $\tau^2$ after the substitution yields
\begin{equation}
\tau^{-2}P(\tau\zeta;\chi,\tau)=-(\chi^5+ 8\chi(54+\chi^2)\tau^2+16\chi\tau^4)\zeta^2-16\chi^3 + O(\tau),\quad\tau\to 0,\quad\tau\neq 0.
\end{equation}
So, to leading order, we have a purely imaginary pair of simple roots at $\zeta=\pm 4\ii\chi^{-1}$.  This shows that the two roots of $P(u;\chi,\tau)$ near $u=0$ for small $\tau$ have nonzero imaginary parts.
\end{proof}

\begin{lemma}
Fix $\tau>0$.  Then for $\chi>0$ sufficiently small there exists a unique and simple real root of $P(u;\chi,\tau)$.
\label{lem:chi-small}
\end{lemma}
\begin{proof}
It is easy to see that since $P(u;0,\tau)=u^3(81u^4+72\tau^2u^2 + 432\tau^2+16\tau^4)$, for all $\tau>0$, $P(u;0,\tau)$ has a triple root at $u=0$ and no other real roots.  To unfold the triple root for small $\chi$, set $u=\chi \zeta$ and assume that the new unknown $\zeta$ is bounded as $\chi\downarrow 0$.  Thus one finds that one may divide by $\chi^3$ for $\chi\neq 0$ and obtain
\begin{equation}
\chi^{-3}P(\chi \zeta;\chi,\tau)=P_0(\zeta;\tau) + O(\chi^2),\quad\chi\downarrow 0,\quad \chi\neq 0,
\end{equation}
where $P_0$ is a cubic polynomial in $\zeta$:
\begin{equation}
P_0(\zeta;\tau)\defeq (432\tau^2+16\tau^4)\zeta^3-(432\tau^2+16\tau^4)\zeta^2+144\tau^2\zeta-16\tau^2.
\end{equation}
The discriminant of $P_0(\zeta;\tau)$ is proportional to $27\tau^{12}+\tau^{14}$ which vanishes for no $\tau>0$.  Therefore as $\tau$ varies between $\tau=0$ and $\tau=+\infty$, the root configuration of $P_0(\zeta;\tau)$ (i.e., three real roots or one real root with a complex-conjugate pair with nonzero imaginary part) persists for all $\tau>0$.  In the limit $\tau\to\infty$, the dominant terms in $P_0(\zeta;\tau)$ are $16\tau^4(\zeta^3-\zeta^2)$, so there is one real root near $\zeta=1$ and two small roots of size $\zeta=O(\tau^{-1})$.  We may write $\zeta=\tau^{-1}\zeta_1$ to separate them:
\begin{equation}
P_0(\tau^{-1}\zeta_1;\tau)=-16\tau^2(\zeta_1^2+1)+O(\tau),\quad\tau\to\infty,
\end{equation}
so $\zeta_1=\pm\ii + O(\tau^{-1})$ as $\tau\to\infty$.  Therefore, $P_0(\zeta;\tau)$ has a unique simple real root denoted $\zeta(\tau)$ and a conjugate pair of complex roots for all $\tau>0$.  It is easy to see that
\begin{equation}
\zeta(\tau)=\frac{1}{3}+O(\tau^2),\quad\tau\downarrow 0\quad\text{and}\quad
\zeta(\tau)=1+O(\tau^{-2}),\quad\tau\uparrow\infty.
\end{equation}
Since neither $P_0(\tfrac{1}{3};\tau)=-\tfrac{32}{27}\tau^4$ nor $P_0(1;\tau)=128\tau^2$ can vanish for any $\tau>0$, it then follows that $\tfrac{1}{3}<\zeta(\tau)<1$ holds for all $\tau>0$.  
This proves that the triple root of $P(u;0,\tau)$ originates in the limit $\chi\downarrow 0$ as the collision of a conjugate pair of roots and a real simple root having the expansion
\begin{equation}
u(\chi,\tau)=\chi \zeta(\tau)+O(\chi^3),\quad\chi\downarrow 0,\quad\tau>0.
\end{equation}
Since the remaining quartet of complex roots of $P(u;\chi,\tau)$ for $\chi=0$ remains complex for small $\chi$, the proof is finished.
\end{proof}
It then follows via \eqref{eq:eliminate-v}
that 
\begin{equation}
v(\chi,\tau)=2\tau\frac{2\zeta(\tau)-1}{3\zeta(\tau)-1}+O(\chi^2),\quad\chi\downarrow 0,\quad\tau>0,
\end{equation}
and then from \eqref{eq:eliminate-AB},
\begin{equation}
A(\chi,\tau)=O(\chi)\quad\text{and}\quad B(\chi,\tau)^2=\frac{2\zeta(\tau)}{3\zeta(\tau)-1} + O(\chi^2),\quad\chi\downarrow 0,\quad \tau>0.
\end{equation}
Note that $B(0,\tau)^2>1$ for all $\tau>0$, and that $B(0,\tau)^2\to 1$ as $\tau\uparrow +\infty$ while $B(0,\tau)^2\to +\infty$ as $\tau\downarrow 0$.

By Lemma~\ref{lem:tau-small} and Lemma~\ref{lem:chi-small}, $P(u;\chi,\tau)$ has a unique simple real root $u=u(\chi,\tau)$ for $(\chi,\tau)$ in the open first quadrant near each of the coordinate axes.  Next we show that this situation persists throughout $(\mathbb{R}_{>0}\times\mathbb{R}_{>0})\setminus\overline{\channels}$ by studying the resultant $\delta(\chi,\tau)$ of $P(u;\chi,\tau)$ and $P'(u;\chi,\tau)$ with respect to $u$, a polynomial in $(\chi,\tau)$ the zero locus of which detects repeated roots $u$ of $P(u;\chi,\tau)$. 
We consider the renormalized resultant
\begin{equation}
\delta^\mathrm{R}\defeq\frac{\delta(\chi,\tau)}{c \tau^{16}\chi^{6}},
\end{equation}
where $c \defeq 539122498937926189056$ is a constant that factors out of $\delta(\chi,\tau)$ along with the product $\tau^{16}\chi^{6}$. The renormalized resultant $\delta^\mathrm{R}$ is even in $\chi$ and $\tau$ and so can be expressed as
\begin{equation}
\begin{aligned}
\delta^\mathrm{R}(X,T) =& 16 X^{7} + 304 T X^{6} + 24 T (98 T + 1011 ) X^{5} + T (9488 T^2  - 380376 T -19683) X^4\\
&+64 T^2 (332 T^2 - 18009 T + 57645)X^3+ 
384T^2 (68 T^3 - 2553 T^2 + 159246 T - 59049 )X^2\\
&+16384  T^3  (T-54)(T+27)^2 X + 4096 T^3 (T+27)^4,
\end{aligned}
\end{equation}
where $X=\chi^2$ and $T=\tau^2$.
Since we have already shown that $P(u;\chi,\tau)$ has a unique and simple real root for $(\chi,\tau)$ near the coordinate axes, we can study the equation $\delta^\mathrm{R}(\chi^2,\tau^2)=0$ rather than $\delta(\chi,\tau)=0$.  If, as $(\chi,\tau)$ is taken out of one or the other region near the axes where it is known that $P(u;\chi,\tau)$ has a unique real and simple root, $P(u;\chi,\tau)$ does not acquire any repeated roots then in particular it does not acquire any repeated real roots and hence the number of real roots cannot change.  

Therefore, it would be sufficient to prove that the renormalized resultant $\delta^\mathrm{R}(\chi^2,\tau^2)$ does not vanish in the unbounded region $(\chi,\tau)\in(\mathbb{R}_+\times\mathbb{R}_+)\setminus\overline{\channels}$.  With this goal in mind, we view $\delta^\mathrm{R}(X,T)$ as a polynomial in $X$ with coefficients polynomial in $T$.
Then it is easy to see that for $T=\tau^2$ sufficiently large, all of the coefficients of powers of $X$ in $\delta^\mathrm{R}(X,T)$ are positive, so there are no nonnegative roots $X\ge 0$ of $\delta^\mathrm{R}(X,T)$.  Looking on the $T$-axis, we see that $\delta^\mathrm{R}(0,T)=4096T^3(T+27)^4$, which does not vanish for any $T=\tau^2>0$.  Therefore, as $\tau$ is decreased, the only way a positive value of $X>0$ for which $\delta^\mathrm{R}(X,T)=0$ can occur is if first there is a positive repeated root, i.e., a positive value $X>0$ for which both $\delta^\mathrm{R}(X,T)=0$ and $\delta^\mathrm{R}_X(X,T)=0$.  Setting to zero the resultant of the latter two polynomial equations with respect to $X$ gives the condition on $T=\tau^2$ for which there exist repeated roots $X$ (possibly negative or complex) of $\delta^\mathrm{R}(X,T)$.  This condition factors as:
\begin{equation}
T^{16}(T+27)^5(243T+1)^3(64T-27)^3Q_5(T)^2=0,
\label{eq:resultant-of-resultant}
\end{equation}
where $Q_5(T)$ is a quintic polynomial:
\begin{multline}
Q_5(T)\defeq 663552T^5+954511200T^4+829508109289T^3-14696124806763T^2\\
{}+82617806699739T-205891132094649.
\end{multline}
For $T=\tau^2>0$, the first three factors on the left-hand side of \eqref{eq:resultant-of-resultant} are nonzero, and the fourth factor vanishes exactly for $T=T^\sharp\defeq (\tau^\sharp)^2$, i.e., for $\tau=\pm\tau^\sharp$, where $\tau^\sharp$ is defined in \eqref{eq:corner-point}.  So it remains to determine whether $Q_5(T)=0$ holds for any $T>0$.  In fact $Q_5(0)\neq 0$ and $Q_5(T^\sharp)\neq 0$ by direct computation, so we will apply the theory of Sturm sequences (see Definition~\ref{def:Sturm-sequence}) to count the number of real roots $T$ in the intervals $(0,T^\sharp)$ and $(T^\sharp,+\infty)$.
We thus obtain the following sign sequences at the points $T=0$, $T=T^\sharp=(\tau^\sharp)^2$, and $T=\infty$:
\begin{equation}
\begin{split}
\Xi[Q_5](0)&=(-,+,+,-,-,+),\\
\Xi[Q_5](T^\sharp)&=(-,+,+,-,-,+),\\
\Xi[Q_5](+\infty)&=(+,+,-,-,-,+).
\end{split}
\end{equation}
Since $\#(\Xi[Q_5](0))- \#(\Xi[Q_5](T^\sharp))=3-3=0$, by Sturm's theorem (see Theorem~\ref{t:Sturm}) $Q_5(T)$ has no real root in $[0,T^\sharp]$. We also see that $\#(\Xi[Q_5](T^\sharp))- \#(\Xi[Q_5](+\infty))=3-2=1$, which similarly proves that there exists exactly one real root of $Q_5(T)$ in the interval 
$(T^\sharp,+\infty)$; we denote it by $T_1$. One can easily check numerically that $T_1\approx 10.232235>T^\sharp$, which gives $\tau_1\defeq \sqrt{T_1} \approx 3.198786$.  So as $T$ decreases from $T=+\infty$, the first possible bifurcation point at which positive solutions $X$ of $\delta^\mathrm{R}(X,T)=0$ might appear is $T=T_1$.  

Next, one checks directly that $\delta^\mathrm{R}(X,T^\sharp)$ factors as the product of a quintic polynomial in $X$ with strictly positive coefficients and $(16X-81)^2$.  Referring to \eqref{eq:corner-point}, this means that $X^\sharp=(\chi^\sharp)^2$ is a positive double root of $X\mapsto\delta^\mathrm{R}(X,T^\sharp)$, and that there are no other positive roots.  We now show that this double root splits into a pair of real simple roots as $T$ decreases from $T^\sharp$ and into a pair of complex-conjugate simple roots as $T$ increases from $T^\sharp$.  Indeed, if we write $X=X^\sharp+\Delta X$ and $T=T^\sharp+\Delta T$ for $\Delta X$ and $\Delta T$ small, then the dominant terms in $\delta^\mathrm{R}(X,T)$ are those homogeneous in $(\Delta X,\Delta T)$ of degree $2$ and these terms turn out to be proportional to a perfect square:  $(\Delta X-4\Delta T)^2$.  Therefore $\Delta X=4\Delta T+o(\Delta T)$ as $\Delta T\to 0$.  To split the double root present for $\Delta T=0$ therefore requires continuing the calculation to higher order; for this purpose we write $X=X^\sharp+\Delta X$ with $\Delta X=4\Delta T + \zeta$ and discover that the dominant terms in $\delta^\mathrm{R}(X,T)$ are now proportional to $3645\zeta^2+8192\Delta T^3$.  Setting these to zero gives distinct real solutions for $\zeta$ only if $\Delta T<0$.  This perturbative analysis proves that near $T=T^\sharp$ there only exist positive real solutions $X$ of $\delta^\mathrm{R}(X,T)=0$ for $T\le T^\sharp$, and these roots satisfy
\begin{equation}
X=X^\sharp+4(T-T^\sharp)\pm\sqrt{\tfrac{8192}{3645}}(T^\sharp-T)^\frac{3}{2} + o((T^\sharp-T)^\frac{3}{2}),\quad T\uparrow T^\sharp.
\label{eq:discriminant-roots-near-T0}
\end{equation}
Then, since we have already shown that there can be no repeated roots of $X\mapsto\delta^\mathrm{R}(X,T)$ for $T^\sharp<T<T_1$, there are no positive roots $X$ at all for $T$ in this range.  Therefore, for $T>T^\sharp$, only for $T=T_1$ is it possible for there to be any positive roots $X$ of $X\mapsto\delta^\mathrm{R}(X,T)$, and no such root can be simple.  Since $\delta^\mathrm{R}(X,T_1)$ is a polynomial in $X$ of degree $7$, for this special value of $T=T_1$ there are at most finitely many positive and necessarily repeated roots $X=X_i>0$, $i\le 3$, corresponding to $\chi_i=\sqrt{X_i}$.  Numerically, one sees that in fact the only repeated root of $X\mapsto\delta^\mathrm{R}(X,T_1)$ (recall that this map must have one or more repeated roots, possibly negative real or complex, by choice of $T_1$) is a positive number $X=X_1\approx31.8597$ corresponding to $\chi_1=\sqrt{X_1}\approx5.64444$ and that there are no other positive roots.

Finally, we consider the range $T<T^\sharp$.  The two simple roots of $X\mapsto\delta^\mathrm{R}(X,T)$ with the expansions \eqref{eq:discriminant-roots-near-T0} cannot coalesce, nor can any new roots appear, for $0<T<T^\sharp$ as has already been shown.  We will show that the two simple roots with the expansions \eqref{eq:discriminant-roots-near-T0} are contained within the domain $\channels$ for all $T\in (0,T^\sharp)$. To show this, we look for simultaneous solutions of the condition \eqref{eq:boundary-curve} describing the boundary of $\channels$, expressed as a polynomial condition in $(X,T)$, and $\delta^\mathrm{R}(X,T)=0$ by computing the resultant with respect to $X$.  The latter resultant is proportional to $T^9(64T-27)^6Q_9(T)$ where $Q_9(T)$ is a ninth-degree polynomial having the Sturm sequences
\begin{equation}
\Xi[Q_9](0)=\Xi[Q_9](T^\sharp)=(+,+,+,-,-,+,+,-,+,+)
\end{equation}
from which it follows by Sturm's theorem that there are no values of $T\in (0,T^\sharp)$ for which roots $X$ of $X\mapsto\delta^\mathrm{R}(X,T)$ can coincide with points of the boundary of $\channels$.  It therefore remains to determine whether the expansions \eqref{eq:discriminant-roots-near-T0} give values of $X$ that lie in the interior of $\channels$.  But near $(X,T)=(X^\sharp,T^\sharp)$ a similar local analysis of the condition \eqref{eq:boundary-curve} as already performed for the condition $\delta^\mathrm{R}(X,T)=0$ shows that \eqref{eq:boundary-curve} only has real solutions $X$ for $T\le T^\sharp$ and that these solutions have the expansions
\begin{equation}
X=X^\sharp+4(T-T^\sharp)\pm\sqrt{\tfrac{8192}{729}}(T^\sharp-T)^\frac{3}{2}+o((T^\sharp-T)^\frac{3}{2}),\quad T\uparrow T^\sharp.
\label{eq:ChannelsBoundaryNearX0T0}
\end{equation}
For $T^\sharp-T$ small and positive, the interior of $\channels$ lies between these latter two curves.
Comparing with \eqref{eq:discriminant-roots-near-T0} we then see that locally the roots $X$ of $X\mapsto\delta^\mathrm{R}(X,T)$ are indeed contained within $\channels$, and this necessarily persists throughout the whole interval $T\in (0,T^\sharp)$.

Therefore, the only points $(\chi,\tau)\in\mathbb{R}_{> 0}\times\mathbb{R}_{> 0}$ in the exterior of $\overline{\channels}$ where the resultant $\delta(\chi,\tau)$ of $P(u;\chi,\tau)$ and $P'(u;\chi,\tau)$ vanishes are $(\chi_i,\tau_1)$, $i\le 3$.  Since by Lemmas~\ref{lem:tau-small} and ~\ref{lem:chi-small} it is known that $P(u;\chi,\tau)$ has a unique real and simple root for points $(\chi,\tau)$ in the exterior sufficiently close to the coordinate axes, and since complex-conjugate roots of $P(u;\chi,\tau)$ are prevented from bifurcating onto the real axis under continuation in $(\chi,\tau)$ unless $\rho(\chi,\tau)$ vanishes, it follows that $P(u;\chi,\tau)$ has a unique real and simple root for all $(\chi,\tau)$ in the part of the open first quadrant exterior to $\overline{\channels}$ with the possible exception of only the points $(\chi_i,\tau_1)$, $i\le 3$. For these exceptional isolated points it can in principle happen that one or more complex-conjugate pairs of roots of $P(u;\chi,\tau)$ coalesce on the real axis, but these are either roots of even multiplicity or in the case of a collision with the simple root they may add an even number to its multiplicity.

Letting $u(\chi,\tau)$ denote the unique real root of odd multiplicity, we extend $u(\chi,\tau)$ to the coordinate axes within $(\mathbb{R}_{>0}\times\mathbb{R}_{>0})\setminus\overline{\channels}$ by continuity:  $u(0,\tau)=0$ for $\tau>0$ and $u(\chi,0)=\chi$ for $\chi>2$.  Note that $u(0,\tau)$ is non-simple root of $P(u;0,\tau)$, but $u(\chi,0)$ is a simple root of $P(u;\chi,0)$.  This completes the proof of Proposition~\ref{prop:u}.
\end{proof}

\begin{remark}
As pointed out earlier, numerics suggest that there is only one positive value of $\chi=\chi_1$ for which there are repeated roots of $P(u;\chi,\tau)$ for $\tau=\tau_1$.  Numerical calculations also show that there are two repeated roots of $P(u;\chi_1,\tau_1)$ forming a complex-conjugate pair.  Therefore $P(u;\chi_1,\tau_1)$ also has just one real root and it is simple.  Thus apparently there is just one exceptional point, and in fact it is not really exceptional after all.
\end{remark}

\section{Some Useful Facts About Polynomials with Real Coefficients}
\label{A:Sturm}
We remind the reader that the discriminant $\Delta_f$ of a polynomial $f(z)= a_n z^n + a_{n-1} z^{n-1} + \cdots + a z+ a_0$, $n\geq 1$, $a_n\neq 0$, with roots (counted with multiplicity) $\xi_1, \xi_2,\dots,\xi_n\in\mathbb{C}$ can be expressed as
\begin{equation}
\Delta_f = a_n^{2n-2}\prod_{1 \leq j < k \leq n}\left(\xi_{j}-\xi_{k}\right)^{2}.
\label{eq:discriminant-roots}
\end{equation}
We assume that the coefficients $a_k$, $k=1,\dots,n$ of the polynomial $f$ are real in the rest of this appendix. In this case, the representation \eqref{eq:discriminant-roots} provides information about the number of non-real roots of $f$. Since the non-real roots of $f$ come in complex conjugate pairs, it is seen from \eqref{eq:discriminant-roots} that $\Delta_f>0$ if and only if $f$ has all distinct real roots or the number of non-real roots are a multiple of $4$. On the other hand, in case $n\geq 2$, $\Delta_f<0$ if and only if the number of non-real roots of $f$ is $2~\mathrm{mod}(4)$. 

The following method is useful for obtaining information about the real roots of a univariate polynomial $f(z)$. We first give a definition (\cite{Sturm1829}, see also \cite[Section 1.3]{Sturmfels02}).
\begin{definition}[Sturm sequence]
Given a polynomial $f(z)$ of degree $n$, define polynomials $f_k(z)$, $k=0,1,2,\dots$ by
\begin{equation}
\begin{split}
f_0(z) &:= f(z),\\
f_1(z) &:= f'(z),\\
f_k(z) &:= -\rem(f_{k-2}(z), f_{k-1}(z)),\quad \text{for $k\geq 2$},
\end{split}
\end{equation}
where $\rem(f_{k-2}(z), f_{k-1}(z))$ denotes the remainder arising in the division of $f_{k-2}(z)$ by $f_{k-1}(z)$. For sufficiently large $k$ we have $f_k(z)\equiv 0$, so let $m$ be the index of the last non-trivial polynomial $f_m(z)$. The \emph{Sturm sequence} of $f(z)$ is the finite sequence of polynomials $(f_0(z), f_1(z),\ldots, f_m(z))$, where necessarily $m\leq n = \deg(f)$.
\label{def:Sturm-sequence}
\end{definition}
We denote by $\Xi[f](a)$ the sequence of \emph{signs} of the Sturm sequence of $f(z)$ evaluated at a point $a\in\mathbb{R}$:
\begin{equation}
\Xi[f](a) := (\sign(f_0(a)),\sign(f_1(a)),\sign(f_2(a)),\ldots, \sign(f_m(a)) ),
\end{equation}
and we let $\#(\Xi[f](a))$ denote the number of sign variations in $\Xi[f](a)$, i.e., the number of sign changes ignoring any zeros when counting. For instance, for $f(z)=4z^3 + z^2 -2$, we have the Sturm sequence
\begin{align}
f_0(z)=4z^3 + z^2 -2,\quad
f_1(z):=12 z^2 +2z,\quad
f_2(z):= \frac{1}{18}z +2,\quad
f_3(z):=-15480,\quad
f_4(z):=0,
\end{align}
and hence at $z=4$, for example, we have
\begin{equation}
 \Xi[f](4) = (\sign(-30), \sign(44), \sign(17/9), \sign(-15480))=(-,+,+,-),
\end{equation}
which gives $\#(\Xi[f](-2))=2$. As a more complicated example, we obtain $\#(\Xi[f](a))= 3$ if $\Xi[f](a)=(+,+,0,+,-,0,+,+,0,-)$. The following theorem (\cite{Sturm1829}, see also \cite[Theorem 1.4]{Sturmfels02}) gives an \emph{exact} count of real zeros of $f(z)$ weighted by multiplicity in an interval using $\#(\Xi[f](\cdot))$.
\begin{theorem}[Sturm's Theorem] Suppose that $a<b$ and neither $a$ nor $b$ is a zero of $f(z)$. Then $\#(\Xi[f](a))\geq \#(\Xi[f](b))$, and the number of real zeros, weighted by multiplicity, of the polynomial $f(z)$ in the interval $[a,b]$ is equal to $\#(\Xi[f](a)) - \#(\Xi[f](b))$.
\label{t:Sturm}
\end{theorem}
The theorem also applies to the case where $a=-\infty$ or $b=+\infty$ by considering the asymptotic behavior of the polynomials in the Sturm sequence, which amounts to looking at the signs of the leading coefficients of the polynomials $f_k(z)$ in the Sturm sequence of $f$.

Another result on the real roots of a polynomial is the following:

\begin{theorem}[D\'escartes' Rule of Signs] Let $f(z)=a_n z^n + a_{n-1} z^{n-1} + \cdots + a_1 z + a_0$, $a_n\neq 0$. The number of positive real roots of $f$ is at most the number of sign variations in its coefficient sequence $(a_n,a_{n-1},\ldots, a_1,a_0)$. Moreover, the number of positive real roots of $f$ differs from the the number of sign variations of the coefficients sequence by an even (nonnegative) integer.
\label{t:Descartes}
\end{theorem}


\begin{thebibliography}{99}
\bibitem{AAS09}
N.\@ Akhmediev, A.\@ Ankiewicz, and J.\@ M.\@ Soto-Crespo, ``Rogue waves and rational solutions of the nonlinear Schr\"odinger equation,'' \textit{Phys.\@ Rev.\@ E} \textbf{80}, art.ID 026601, 2009.

\bibitem{BilmanB19}
D. Bilman and R. J. Buckingham, 
``Large-order asymptotics for multiple-pole solitons of the focusing nonlinear Schr\"odinger equation,'' 
\textit{J. Nonlinear Sci.\@} \textbf{29}, 2185--2229, 2019. 

\bibitem{BilmanBW19} D. Bilman, R. J. Buckingham, and D.-S. Wang, ``Large-order asymptotics for multiple-pole solitons of the focusing nonlinear Schr\"odinger equation II: far-field behavior,'' \texttt{arXiv:1911.04327}, 2019. 

\bibitem{BilmanM19} D. Bilman and P. D. Miller, ``A robust inverse scattering transform for the focusing nonlinear Schr\"odinger equation,'' \textit{Comm.\@ Pure Appl.\@ Math.\@} \textbf{72}, 1722--1805, 2019.

\bibitem{BilmanLM20} D. Bilman, L. Ling and P. D. Miller, ``Extreme superposition:  rogue waves of infinite order and the Painlev\'e-III hierarchy,'' \textit{Duke Math.\@ J.\@} \textbf{169}, 671--760, 2020.

\bibitem{BiondiniM17} G. Biondini and D. Mantzavinos, 
``Long-time asymptotics for the focusing nonlinear Schr\"odinger equation with nonzero boundary conditions at infinity and asymptotic stage of modulational instability,'' 
\textit{Comm.\@ Pure Appl.\@ Math.\@} \textbf{70}, 2300--2365, 2017.

\bibitem{BothnerM20} T.\@ Bothner and P.\@ D.\@ Miller, ``Rational solutions of the Painlev\'e-III equation: Large parameter asymptotics,'' \textit{Constr.\@ Approx.\@} \textbf{51}, 123--225, 2020.

\bibitem{BuckinghamJM21} R. J. Buckingham, R. M. Jenkins, and P. D. Miller, ``Talanov self-focusing and its non-generic character,'' in preparation, 2021.

\bibitem{Jenkins58} J. A. Jenkins, \textit{Univalent Functions and Conformal Mapping}, Springer-Verlag, Berlin, 1958.

\bibitem{LiM21} S. Li and P. D. Miller, ``On the Maxwell-Bloch system in the sharp-line limit without solitons,'' in preparation, 2021.

\bibitem{Miller18} P.\@ D.\@ Miller, ``On the increasing tritronqu\'ee solutions of the Painlev\'e-II equation,'' \textit{SIGMA} \textbf{14}, 125, 38 pages, 2018.

\bibitem{Strebel84} K. Strebel, \textit{Quadratic Differentials}, Springer-Verlag, Berlin, 1984.

\bibitem{Sturm1829}
J.\@ C.\@ F.\@ Sturm, ``Analyse d'un m\'emoire sur la r\'esolution des \'equations num\'eriques,'' \textit{Bulletin des Sciences de F\'erussac} \textbf{11}, 419--422, 1928.

\bibitem{Sturmfels02}
B.\@ Sturmfels, \textit{Solving Systems of Polynomial Equations}, \textit{CBMS Regional Conference Series in Mathematics} \textbf{97}, 152 pp., Published for the Conference Board of the Mathematical Sciences, Washington, DC; by the American Mathematical Society, Providence, RI, 2002. ISBN:978-0-8218-3251-6.

\bibitem{Suleimanov17}
B. I. Suleimanov, ``Effect of a small dispersion on self-focusing in a spatially one-dimensional case,''
\textit{JETP Lett.\@} \textbf{106}, 400--405, 2017.

\bibitem{WangYWH17}
L.\@ Wang, C.\@ Yang, J.\@ Wang, and J.\@ He, ``The height of an $n$th-order fundamental rogue wave for the nonlinear Schr\"odinger equation,'' \textit{Phys.\@ Lett.\@ A} \textbf{381},  1714--1718, 2017.

\bibitem{Zhou89}
X. Zhou, ``The Riemann-Hilbert problem and inverse scattering,'' \textit{SIAM J. Math.\@ Anal.\@}, \textbf{20}, 966--986, 1989.
\end{thebibliography}
\end{document}